\documentclass[11pt]{article}
\usepackage[left=1in, right=1in, top=1in, bottom=1in, margin=1in]{geometry}
\usepackage{zi4}
\usepackage[utf8]{inputenc} 
\usepackage{authblk}

\usepackage{amsmath,amssymb,xspace,graphicx,relsize,bm,breqn}
\usepackage{amsthm}
\usepackage{braket}

\usepackage{microtype} 
\usepackage{multirow}
\usepackage{multicol}
\usepackage{tikz}
\usepackage{adjustbox}

\usepackage{dsfont}
\usepackage{array}
\usepackage{makecell}
\usepackage{nicematrix}

\usepackage{tcolorbox}
\usepackage{ninecolors}\NineColors{saturation=high}
\usepackage{xcolor}

\usepackage{float}
\usepackage[linesnumbered,ruled,vlined]{algorithm2e}
\SetKwInput{KwInput}{Input}
\SetKwInput{KwOutput}{Output}
\SetKwInOut{Promise}{Promise}
\SetKwInput{Goal}{Goal}
\SetKwProg{Fn}{function}{}{}
\SetKwFor{RepTimes}{repeat}{times}{}
\SetKwFunction{Ver}{Verify}
\SetKwFunction{Prep}{PrepareState}
\SetKwComment{Comment}{/* }{ */}
\usepackage{algorithmic}

\newtcolorbox{myalgorithm}[1][]{
    colback=gray!10, 
    colframe=black, 
    arc=5pt, 
    boxrule=0.5pt, 
    left=0pt, right=0pt, top=0pt, bottom=0pt 
}
\usepackage[margin=1in]{geometry}

\newcommand{\poly}{\ensuremath{\mathsf{poly}}}

\newcommand{\id}{\ensuremath{\mathbb{I}}}

\usepackage{upgreek}
\usepackage{enumerate}

\usepackage[usestackEOL]{stackengine}
\usepackage{thm-restate,mathrsfs}
\usepackage{enumerate}
\usepackage{array}
\usepackage{parskip}
\def\01{\{0,1\}}

\usepackage[margin=1in]{geometry}

\definecolor{citegreen}{HTML}{208054}
\definecolor{citeblue}{HTML}{0055cc}
\definecolor{CeruleanBlue}{RGB}{42,82,190}
\definecolor{SlateBlue}{RGB}{106,90,205}
\definecolor{EmeraldGreen}{RGB}{80,200,120}

\usepackage[hypertexnames=false]{hyperref}
\NineColors{saturation=high}
\hypersetup{
    breaklinks=true,   
    colorlinks=true, 
    urlcolor=blue3, 
    linkcolor=red3, 
    citecolor=CeruleanBlue, 
}

\newcommand{\calQ}{{\cal Q }}

\newcommand{\be}{\begin{equation}}
\newcommand{\ee}{\end{equation}}
\newcommand{\ba}{\begin{array}}
\newcommand{\ea}{\end{array}}
\newcommand{\bea}{\begin{eqnarray}}
\newcommand{\eea}{\end{eqnarray}}

\usepackage{mathtools}
\DeclarePairedDelimiter\ceil{\lceil}{\rceil}

\DeclareMathOperator{\sgn}{sgn}
\DeclareMathOperator{\polylog}{\mathsf{polylog}}

\newcommand{\ra}{\rangle}
\newcommand{\la}{\langle}

\newcommand{\spec}{\mathrm{spec}}
\newcommand{\expnorm}{\textsf{C}}
\newcommand{\junk}{\mathrm{junk}}

\newcommand{\norm}[1]{\left\lVert#1\right\rVert}

\newcommand{\calA}{{\cal A }}

\newcommand{\calL}{{\cal L }}
\newcommand{\calN}{{\cal N }}

\newcommand{\calG}{{\cal G }}
\newcommand{\calV}{{\cal V }}
\newcommand{\calE}{{\cal E }}
\newcommand{\calC}{{\cal C }}
\newcommand{\calS}{{\cal S }}

\newcommand{\calM}{{\cal M }}
\newcommand{\calZ}{{\cal Z }}

\newcommand{\rlcA}{\mathbf{A}}
\newcommand{\rlcC}{\mathbf{C}}
\newcommand{\rlcR}{\mathbf{R}}
\newcommand{\rlcG}{\mathbf{G}}
\newcommand{\rlcL}{\mathbf{L}}
\newcommand{\rlcM}{\mathbf{M}}
\newcommand{\rlcK}{\mathbf{K}}

\newcommand{\rlcP}{\mathbf{P}}
\newcommand{\rlcQ}{\mathbf{Q}}
\newcommand{\rlcI}{\mathbf{I}}
\newcommand{\rlcB}{\mathbf{B}}

\newcommand{\rlcres}{\mathrm{r}}
\newcommand{\rlccap}{\mathrm{c}}
\newcommand{\rlcind}{\mathrm{\ell}}
\newcommand{\rlcvs}{\mathrm{v}}
\newcommand{\rlcjs}{\mathrm{s}}

\newcommand{\rlc}{\textsf{RLC}}
\newcommand{\kcl}{\textsf{KCL}}
\newcommand{\kvl}{\textsf{KVL}}
\newcommand{\dae}{\textsf{DAE}}
\newcommand{\ode}{\textsf{ODE}}
\newcommand{\dc}{\textsf{DC}}

\newcommand{\mna}{\textsf{MNA}}
\newcommand{\bqp}{\textsf{BQP}}
\newcommand{\SWAP}{\textsf{SWAP}}

\newcommand{\supp}{\mathrm{supp}}

\newcommand{\psihist}{\psi_{\mathrm{hist}}}
\newcommand{\normhist}{\mathcal{Z}}

\newtheorem{theorem}{Theorem}[section]

\newtheorem{definition}[theorem]{Definition}

\newtheorem{lemma}[theorem]{Lemma}
\newtheorem{remark}{Remark}
\newtheorem{corollary}[theorem]{Corollary}

\newtheorem{fact}[theorem]{Fact}
\newtheorem{claim}[theorem]{Claim}

\newtheorem{problem}[theorem]{{Problem}}
\newtheorem{result}[theorem]{{Result}}

\usepackage{thm-restate} 

\definecolor{olivegreen}{rgb}{0.42, 0.56, 0.14}
\definecolor{darkolivegreen}{rgb}{0.25, 0.4, 0.15}

\makeatletter
\def\widebreve{\mathpalette\wide@breve}
\def\wide@breve#1#2{\sbox\z@{$#1#2$}%
     \mathop{\vbox{\m@th\ialign{##\crcr
\kern0.08em\brevefill#1{0.8\wd\z@}\crcr\noalign{\nointerlineskip}%
                    $\hss#1#2\hss$\crcr}}}\nolimits}
\def\brevefill#1#2{$\m@th\sbox\tw@{$#1($}%
  \hss\resizebox{#2}{\wd\tw@}{\rotatebox[origin=c]{90}{\upshape(}}\hss$}
\makeatletter

\makeatletter

\title{Simulating dynamics of RLC circuits with a \\ quantum differential-algebraic equations solver}

\author{Arkopal Dutt}
\author{Anirban Chowdhury}
\author{Kristan Temme}
\author{Hari Krovi}
\affil[]{IBM Research}

\begin{document}
\NiceMatrixOptions{code-for-first-col=\scriptstyle,code-for-first-row=\scriptstyle}

\date{}

\maketitle

\begin{abstract}

We introduce a quantum algorithm for simulating the dynamics of electrical circuits consisting of resistors, inductors and capacitors (aka \textsf{RLC} circuits) along with power sources. Given oracle access to the connectivity of the circuit and values of the electrical elements, our algorithm prepares a quantum state that encodes voltages and current values either at a specified time or the history of their evolution over a time-interval. We show this state can be used to estimate physically relevant quantities, such as the power loss, or the stored electrical energy. For an \textsf{RLC} circuit with $N$ components, our algorithm runs in time $\polylog(N)$ under mild assumptions on the connectivity of the circuit and values of its components. This provides an exponential speed-up over classical algorithms that take $\poly(N)$ time in the worst-case. Our algorithm can be used to estimate energy across a set of components or dissipated power in $\polylog(N)$ time, a problem that we prove is $\bqp$-hard and therefore unlikely to be efficiently solved by classical algorithms.

The main challenge in simulating the dynamics of \textsf{RLC} circuits is that they are governed by differential-algebraic equations (\textsf{DAE}s), a coupled system of differential equations with hidden algebraic constraints. Consequentially, existing quantum algorithms for ordinary differential equations (\textsf{ODE}s) cannot be directly utilized. We therefore develop a quantum $\dae$ solver for simulating the time-evolution of \textsf{DAE}s. For \textsf{RLC} circuits, we employ modified nodal analysis (\textsf{MNA}) to create a system of \textsf{DAE}s compatible with our quantum algorithm. We establish $\bqp$-hardness by demonstrating that any network of classical harmonic oscillators, for which an energy-estimation problem is known to be $\bqp$-hard, is a special case of an LC circuit. Our work gives theoretical evidence of quantum advantage in simulating \textsf{RLC} circuits and we expect that our quantum \textsf{DAE} solver will find broader use in the simulation of dynamical systems. 

\end{abstract}

\newpage 

\setcounter{tocdepth}{2}
\tableofcontents

\newpage
\section{Introduction}

Simulating the time evolution of large-scale dynamical systems is one of the most dominant computational tasks across science and engineering e.g., in weather modeling~\cite{durran1998numerical,lynch2008origins}, chemical processes~\cite{luyben1989process,omar2017crystal,bui2022continuum}, and robotics~\cite{siciliano2009robotics,xu2012space,huang2020dynamic}. Classical simulation methods are often limited by the dimensionality of the system and stiffness of the underlying dynamics~\cite{wanner1996solving,ferziger2002computational}. This presents an opportunity for investigating whether quantum algorithms could provide provable speedups over classical algorithms, and therefore identifying new applications of quantum computation. Efforts towards this end include the development of quantum algorithms for fluid mechanics \cite{li2025fluid,jennings2025end}, systems of harmonic oscillators \cite{PhysRevX.13.041041, krovi2024quantum}, and an extensive body of work on quantum algorithms for differential equations \cite{Ber14, BCOW17, Krovi2023improvedquantum, Fang2023timemarchingbased, Berry2024quantumalgorithm,LCHS,PhysRevX.13.041041,bravyi2025quantumsimulationnoisyclassical}.

Electrical circuits are examples of dynamical systems that have wide-ranging applications, from power-distribution networks~\cite{chen2001efficient,yang2012powerrush} and analog filtering to integrated circuit design~\cite{nagel1973spice}. The accurate prediction of transient behavior of electrical circuits is essential for analyzing resonance, signal integrity, and energy dissipation prior to hardware fabrication. In this work, we focus on simulating the dynamics of $\rlc$ circuits, which consist of passive resistors, inductors, and capacitors along with $\dc$ voltage and current sources. From a modeling perspective, $\rlc$ circuits form a canonical and expressive subclass of general electrical circuits. Under standard approximations and in certain regimes, complex circuit components\,---\,including op-amps, transformers, and interconnects\,---\,can be modelled as $\rlc$ circuits. Simulation of $\rlc$ circuits thus captures many computational aspects of general circuit simulation

and motivates the central question of this work:
\begin{quote}
\centering
\emph{Can we design a quantum algorithm to simulate RLC circuits with a provable speedup over current classical algorithms?}
\end{quote}

The immediate challenge in tackling this problem is that the dynamics of \textsf{RLC} circuits are governed by linear \emph{differential-algebraic} (\dae) equations i.e., a coupled system of equations involving both differential equations describing the temporal evolution of the state and hidden algebraic constraints that restrict the solution to a submanifold of the state space~\cite{riaza2008differential}. The algebraic constraints follow from the fact that currents and voltages in an electrical network must obey Kirchhoff's current and voltage laws at all times. In principle, it is possible to classically eliminate these constraints and obtaine a system of ordinary differential equations (\textsf{ODE}s) with a minimal set of variables, but the computational cost of doing so will scale polynomially with the size of the electrical circuit in the worst case, and thus become expensive for very large circuits. 
This prevents us, for example, from simulating even purely inductor and capacitor (\textsf{LC}) circuits by using the algorithm in~\cite{PhysRevX.13.041041} as \textsf{LC} circuits without capacitive spanning trees correspond to $\dae$s by standard modelling analysis. Similar considerations apply to constrained mechanical systems with a large number of components where too \dae's find use \cite{dejalon2012kinematic,sol1983kinematics}. As a result, quantum algorithms for differential equations cannot directly be applied to $\dae$s, and a new framework has to be developed.

\subsection{Problem description}
\label{subsec:problem}

We consider an electrical circuit represented by a graph $\calG(\calV,\calE)$ where $\calE$ are the set of edges corresponding to branches containing circuit elements and $\calV$ are the set of nodes corresponding to the interconnections of the circuit elements. Each of these branches will have either a passive element, e.g., a resistor, inductor or capacitors, or power source, e.g., independent voltage or current sources.
The electrical circuits of interest in this work will be those with $N+1$ nodes ($|\calV|=N+1$) where one node is called a \emph{reference node}. The reference node can be highly connected and have degree up to $N$, but every other node is restricted to a maximum degree of $d$. The connectivity of the circuit is specified via the \emph{incidence matrix} $\mathbf{A} \in \{-1,0,1\}^{N \times |\calE|}$ (the reference node is not included in the definition of the incidence matrix) which is $d$-sparse. We will also need to specify the components in the network. Let $\mathbf{R}$ be the diagonal matrix of resistance values in each branch with a resistor and let $\mathbf{C}$ be the diagonal matrix of capacitance values in the branches with capacitors. Further, let $\mathbf{L}$ be the matrix of inductance values, where the diagonal entries are the inductors in the branches (self-inductances) and the off-diagonal terms are mutual-inductances; we assume that $\mathbf{L}$ is $d_\ell \leq d$-sparse. We will assume oracle access to $\mathbf{A}$, $\mathbf{R}$, $\mathbf{C}$ and $\mathbf{L}$.

\subsubsection{Differential-algebraic equations of RLC circuits}\label{subsec:tech_MNA}

The state of an $\rlc$ circuit $\calG$ can be described through its node voltages $\vec{u}$, voltage drops along branches $\vec{v}$ and currents through branches $\vec{i}$. Since the branch voltages can be related to the node voltages as $\vec{v} = \rlcA^T \vec{u}$, one can work with the state description of either $(\vec{u},\vec{i})$ or $(\vec{v},\vec{i})$. 

However, either description may be redundant since the voltage drops and branch currents must satisfy Kirchoff's current law (KCL), Kirchoff's voltage law (KVL) and branch constitutive relations (BCR) at any time. To systematically handle the redundancy in the state description, we model $\rlc$ electrical circuits using \emph{modified nodal analysis} ($\mna$) which was introduced by Ho et al.~\cite{ho1975mna} and has since become one of the most widely used modeling approaches.

We denote this state description as $\vec{x}(t) = (\vec{u}(t), \vec{i}_{\rlcind}(t), \vec{i}_{\rlcvs}(t))^T$. 
The governing equations of the dynamics of $\rlc$ circuits corresponding to $\mna$ are then\footnote{The form of $\rlcK$ as written in Eq.~\eqref{eq:mna} is convenient as the Hermitian part of $\rlcK$ i.e., $(\rlcK + \rlcK^\dag)/2$ will then just involve the resistive sub-matrix.}
\begin{align} 
    \label{eq:mna}
    \underbrace{
    \begin{bmatrix}
        \rlcA_{\rlccap} \rlcC \rlcA_{\rlccap}^T & 0 & 0 \\
        0 & \rlcL & 0 \\
        0 & 0 & 0 
    \end{bmatrix}}_{:=\rlcM}
    \frac{d}{dt} \begin{bmatrix}
        \vec{u}(t) \\ \vec{i}_\rlcind(t) \\ \vec{i}_{\rlcvs}(t)
    \end{bmatrix} + 
    \underbrace{
    \begin{bmatrix}
        \rlcA_\rlcres \rlcG \rlcA_\rlcres^T & \rlcA_{\rlcind} & \rlcA_{\rlcvs} \\
        -\rlcA_{\rlcind}^T & 0 & 0 \\
        -\rlcA_{\rlcvs}^T & 0 & 0 
    \end{bmatrix}}_{:=\rlcK}
    \begin{bmatrix}
        \vec{u}(t) \\ \vec{i}_{\rlcind}(t) \\ \vec{i}_{\rlcvs}(t)
    \end{bmatrix} = 
    \underbrace{
    \begin{bmatrix}
        - \rlcA_{\rlcjs} \vec{i}_{\rlcjs} \\ 0 \\ -\vec{v}_{\rlcvs}
    \end{bmatrix}}_{:=\vec{f}}
\end{align}
where we denote the matrices $\rlcM, \rlcK$ as indicated along with the vector $\vec{f}$ which contains information regarding the independent $\dc$ sources, $\rlcG = \rlcR^{-1}$ is the conductance matrix and with initial conditions on $\vec{x}(0)$. The $\mna$ equation simplifies to
\begin{align}\label{eq:dae-general}
    \rlcM \dot{\vec{x}} + \rlcK \vec{x} = \vec{f}\;
\end{align}
which is precisely a $\dae$ specified by the matrices $\rlcM$ and $\rlcK$. We note straightaway that the matrix $\rlcM$ is singular for $\rlc$ circuits with voltage sources. Even if no voltage sources are present, $\rlcM$ can $\rlcA_{\rlccap} \rlcC \rlcA_{\rlccap}^T$ is singular as is the case for many circuits (except those satisfying certain connectivity conditions among their components). This issue does not arise for example in the mechanical system simulated in~\cite{PhysRevX.13.041041} because there, the $\rlcM$ matrix is diagonal and positive. The fact $\rlcM$ can be singular is one of the main challenges in simulating $\dae$s; it may not be possible to obtain a differential equation in $\vec{x}(t)$ by, e.g., simply multiplying both sides of Eq.~\ref{eq:dae-general} by $\rlcM^{-1}$. The difficulty of simulating $\dae$s is often characterized by its tractability index~\cite{lamour2013differential}, which quantifies the hardness of decoupling the algebraic constraints from the underlying differential equation. We will discuss this in more detail in Section~\ref{subsec:tech_approach}. For linear $\rlc$ circuits, the tractability index can be at most two~\cite{estevez2000structural,riaza2008differential}.

\subsubsection{Problem statements}
Our quantum algorithm will use a quantum state-vector encoding where the amplitudes of the quantum state prepared will correspond to voltages and currents of the RLC network, up to an overall normalization, similar to other well-known quantum algorithms \cite{Ber14,harrow2009linear,gilyen2019qsvt}. More precisely, let $\vec{x}(t)$ be  the vector of node voltages and currents across inductors and voltage sources as described before, at a time $t$. We denote $\ket{x(t)} := \vec{x}(t)/\norm{\vec{x}(t)}$ as the corresponding quantum state whose amplitudes in the computational basis correspond to the elements of $\vec{x}(t)$ up to a normalization. Then, our goal is to prepare the \emph{history state} $\ket{\Psi}$ that encodes a discretized time-evolution of these variables over a time interval $[0,T]$
\begin{align}\label{def:hist_state}
\ket{\Psi} = \frac{1}{\normhist} \sum_{k=0}^{m} \|x(k\Delta t)\| \ket{k} \ket{x(k\Delta t)}
\end{align}
for appropriate choices of $m$ and $\Delta t$. From the history-state, we could extract both the state $\ket{x(T)}$ encoding the solution at the final time $T$, and quantities such as the energy stored in a subset of in inductors or capacitors, or the power dissipated across certain resistors. We state informal versions of these problems below.

\begin{problem}[State preparation]\label{prob:RLC_dynamics}
Let $\mathbf{A}$ be the incidence matrix of the circuit and assume it to be $d$-sparse and let $\mathbf{R},\mathbf{L},\mathbf{C}$ be as defined above. Given oracle access to $\mathbf{A}, \mathbf{R}, \mathbf{L}, \mathbf{C}$, an oracle preparing the initial state describing the $\rlc$ circuit at time $0$, oracle preparing the current and voltage sources, the tractability index of the $\mna$ system, the goal is to output 
\begin{enumerate}[$(a)$]
    \item \textbf{History state}: a state $\ket{\Phi}$ that is $\varepsilon$-close (in $\ell_2$ norm) to the history state of the $\rlc$ circuit.
    \item \textbf{Normalized state at a particular time:} a state $\ket{\phi(T)}$ that is $\varepsilon$-close (in $\ell_2$ norm) to the normalized solution $\ket{x(T)}$ at time $T$.
\end{enumerate}
\end{problem}

Extracting the complete state description of the circuit $x(t)$ from the quantum-encoded solution $\ket{\psi(t)}$ using quantum state tomography would take $\mathsf{poly}(N)$ time and thus negate any quantum-speedup. Thus, it is desirable to estimate physically relevant properties from the quantum state while retaining the $\mathsf{polylog}(N)$ running time. 
We formalize this problem below for estimating the energy stored in capacitors and energy dissipated by resistors.
\begin{problem}[Estimating stored and dissipated energy]\label{prob:classical_output}
Given the same access as Problem~\ref{prob:RLC_dynamics}, an oracle for $S \subseteq \calE$ which is a subset of only capacitors and inductors or only resistors and simulation time $t \in (0,T]$, output an $\varepsilon$-approximate estimate of the energy $E_S(t)$ being stored (in any subset of inductors and capacitors) or power dissipated (in any subset of resistors) at time $t$.
\end{problem}

\subsection{Main results}
We give a positive answer to the questions posed above by $(i)$ providing an efficient quantum algorithm for simulating the time dynamics of $\rlc$ circuits with sources, $(ii)$ showing that relevant quantities of stored and dissipated energies at different simulation times can be outputted efficiently while still retaining a quantum speedup, and that $(iii)$ this latter problem is in fact $\bqp$-complete for a family of circuits involving inductors and capacitors. 

\begin{result}[Efficient preparation of history-state]\label{res:RLC_sim}
Problem~\ref{prob:RLC_dynamics} can be solved with a quantum algorithm that makes $\widetilde{O}(T \cdot \poly(\log N, d, \log(1/\varepsilon))$ queries to the oracles for the matrices of $\mathbf{A}$, $\rlcM$, $\rlcR$, $\rlcC$, $\rlcL$, and the oracles preparing the initial state and power sources. The algorithm uses an additional gate complexity of $O(\log N)$ times the number of overall queries.
\end{result}
The $\widetilde{O}(\cdot)$ above hides dependencies of the algorithm on the structure of the underlying graph and spectrum of the component matrices $\rlcR$, $\rlcL$, and $\rlcC$. These dependencies are made precise in the theorem statements in the main text.

The query and gate complexities of the quantum algorithm are $\polylog(N)$, providing an exponential speedup over classical methods that would require $\poly(N)$ time in the worst-case to determine the state of the electrical circuit at any time. Of course, reading out all of the voltages and currents at any time from the quantum state would still incur an exponential overhead. But, as the next result shows, we can extract physical quantities of interest from it.

\begin{result}[Estimating stored and dissipated energy]\label{res:quantum_speedup}
Problem~\ref{prob:classical_output} can be solved with a quantum algorithm that makes $\widetilde{O}(\poly(\log N, T, d))$ queries to the oracles for $\rlcA,\rlcR,\rlcC,\rlcL$, and uses $\widetilde{O}(\poly(\log N, T, d))$ two-qubit gates.
\end{result}
The above result shows that an end-to-end algorithm can be obtained for the task of simulating dynamics of $\rlc$ circuits and output classically relevant quantities. This is significant as this would enable large circuit simulations where the primary goal is to obtain speedups in the computation of specific quantities that summarize the state description such as the energy being stored (or power dissipated across) a set of components in the circuit.

\paragraph{\textsf{BQP}-completeness.}
Finally, we show that solving Problem~\ref{prob:classical_output} is in fact $\bqp$-complete.
\begin{result}\label{res:bqp_completeness}
We show that given an arbitrary network consisting of inductors and capacitors, the problem of deciding whether the energy in a single capacitor (normalized by the total energy) is above $2/3$ or below $1/3$ is $\bqp$-complete.
\end{result}
The above result implies that it is unlikely that classical algorithms can achieve the same efficiency as the quantum algorithms to estimate quantities such as the energy in a capacitor or an inductor (unless $\bqp\subseteq\textsf{BPP}$). Our proof also gives a way to embed an arbitrary instance of a coupled oscillator network into an instance of an $\mathsf{LC}$ network.

\section{Technical overview}\label{sec:tech_overview}
In this section, we give an overview of our approach and the main technical methods behind our main results. We obtain Result~\ref{res:RLC_sim} by solving the more general problem of simulating linear time-independent $\dae$s. This is discussed in Section~\ref{subsec:tech_approach} and is then followed by how the quantum algorithm for simulating general $\dae$s is applied to $\rlc$ circuits in Section~\ref{subsec:qsim_rlc} and \ref{res:quantum_speedup}. Finally, we comment on the classical hardness of simulating $\rlc$ circuits in Section~\ref{subsec:quantum_speedup}.

\subsection{Quantum DAE solver}\label{intro:QDAE_solver}
To tackle the problem of simulating the dynamics of $\rlc$ circuits, we first tackle the more general problem of simulating linear time-independent $\dae$s:
$$
\rlcM \dot{\vec{x}} + \rlcK \vec{x} = \vec{f},
$$
given initial conditions $\vec{x}_0$ and where $\rlcM, \rlcK$ are now time-independent general matrices. For the quantum algorithm which we describe next, we assume we are given approximate block-encodings (which are unitary matrices encoding these matrices up to some normalization~\cite{gilyen2019qsvt}) of $\rlcM$ and $\rlcK$. When applying the quantum $\dae$ solver to $\rlc$ circuits, we will describe how these are constructed from appropriate oracle access.

\subsubsection{Approach}\label{subsec:tech_approach}
Our quantum algorithm is based on the projector-based methods established by M\"{a}rz~\cite{marz2002index,lamour2013differential}. On a high level, our algorithm proceeds in the following fashion. Let $\{\vec{x}(t)\}_{t \in [0,T]}$ be the state trajectory that would have been obtained by simulating the $\dae$ across time. The algorithm first decomposes the state vector into its \emph{differential} and \emph{algebraic} components, denoted by $\vec{y}$ and $\vec{z}$ respectively. The differential components' trajectory $\vec{y}(t)_{[0,T]}$ lies on a submanifold of the state space determined by the \emph{hidden} algebraic constraints, also called the constraint submanifold. The differential components $\vec{y}$ capture the dynamical behavior of the state description and can be solved over the time interval $(0,T]$ by solving an explicit regular $\ode$ obtained by projecting the $\dae$ on to the constraint submanifold. 

The main idea behind our quantum algorithm is to thus decouple the problem of simulating the $\dae$ over the entire state space into $(i)$ first solving for $\vec{y}(t)$ by simulating an $\ode$ on the constraint submanifold over the time interval $[0,T]$, and $(ii)$ then constructing the entire state trajectory $\{\vec{x}(t)\}_{t \in [0,T]}$ by solving for the algebraic components $\vec{z}(t)$ via linear algebraic relations. 

\subsubsection{Overview of algorithm}\label{}
The above decoupling is achieved by defining a sequence of \emph{admissible} projectors $\{\rlcQ_{i}\}_{i \in [k]}$. Starting form $\rlcM_0 = \rlcM$, a sequence of matrices $\{\rlcM_{i-1}\}_{i \in [k+1]}$ are defined such that $\rlcQ_{i}$ is the projector on to the kernel of $\rlcM_{i}$ and satisfies $\rlcQ_i \rlcQ_{i-1} = 0$ (which is the admissibility condition). The matrices $\rlcM_i$ are defined by setting $\rlcP_i = \id - \rlcQ_i$ and then $\rlcM_i = \rlcM_{i-1} + \rlcK_{i-1} \rlcQ_{i-1}$ where $\rlcK_i = \rlcK_{i-1} \rlcP_{i-1}$. The tractability index is then defined as the first $i$ for which $\rlcM_i$ is invertible. The admissibility condition is to ensure decoupling of the explicit $\ode$ from the linear-algebraic relations. We now describe how the explicit induced $\ode$ is obtained for linear $\dae$s of tractability index up to two. This also coincides with the highest possible index achievable by the $\mna$ equations corresponding to linear $\rlc$ circuits which we describe next. 

A key subroutine in our quantum algorithm is a quantum subroutine for solving $\ode$s $\dot{\vec{x}} = \calA \vec{x} + \vec{f}$, and another for implementing linear transformations efficiently. For the former, we choose the quantum $\ode$ solver of \cite{Krovi2023improvedquantum}. The latter can be implemented efficiently once we have a block-encoding of the corresponding differential operator $\calA$.

\paragraph{Index zero.}
In this case, the explicit $\ode$ to be solved is then directly 
$$
\dot{\vec{x}} = \rlcM^{-1} \rlcK \vec{x} + \vec{f},
$$
given initial conditions of $\vec{x}(0)$. The differential variables in this case is the entire state $\vec{x}$ and the state trajectory $\{\vec{x}(t)\}_{t \in [0,T]}$ can be obtained by simulating the above $\ode$ across time.

\paragraph{Index one.}

When the $\mna$ system has index $1$, the matrix $\rlcM$ is singular.

Here, denote by $\rlcQ_0$ the projector onto the kernel of $\rlcM$ and let $\rlcP_0 = \rlcI - \rlcQ$. Decompose the state vector $\vec{x}$ as $\vec{x} = \rlcP_0 \vec{x} + \rlcQ_0 \vec{x}$, where $\vec{y} := \rlcP_0 \vec{x}$ are the differentiable components and $\vec{z} := \rlcQ_0 \vec{x}$ are the algebraic components. Defining $\rlcM_1 = \rlcM + \rlcK \rlcQ_0$, the explicit $\ode$ governing the dynamical behavior of $\vec{y}$ is given by $\dot{\vec{y}} = -\rlcP_0 \rlcM_1^{-1} \rlcK \vec{y} + \rlcM_1^{-1} \vec{f}$ and the algebraic components can be recovered through the equation $\vec{z} = - \rlcQ_0 \rlcM_1^{-1} \rlcK \vec{y}$. Accordingly, we simulate the corresponding $\dae$ by first obtaining $\vec{y}(t)$ via the quantum $\ode$ algorithm and $\vec{z}(t)$ via an appropriate linear transformation implemented through block-encoding techniques \cite{gilyen2019qsvt,chakraborty_et_al:LIPIcs:2019:10609}. We quickly illustrate this below.

In our quantum algorithm, we first prepare the initial states corresponding to the differential and algebraic parts i.e., a state of the form $\ket{0,\psi_{y(0)}}+\ket{1,\psi_{z(0)}}$. Next, we implement the quantum ODE solver conditioned on the first register being in $\ket{0}$. This gives a state of the form 
\begin{equation}
    \sum_t\ket{0,t, \psi_{y(t)}} + \ket{1,\psi_{z(0)}}\,.
\end{equation}
Next, we implement a linear-algebraic transformation conditioned on the first register being $\ket{0}$ to solve for algebraic part and coherently combine with the input $\ket{\phi_{z,0}}$ to produce a state of the form
\begin{equation}
    \sum_t\ket{0,t, \psi_{x(t)}} + \ket{1,\psi_\perp}\,,
\end{equation}
which is the history state encoding the evolution when the first register is $\ket{0}$. 

\paragraph{Index two.} In this case, one needs to apply the projection steps twice. We define $\rlcQ_0$ as the projector on to $\ker(\rlcM)$ and $\rlcQ_1$ as the projector on to $\ker(\rlcM_1)$ where $\rlcM_1 = \rlcM + \rlcK \rlcQ_0$. We then set the projectors $\rlcP_0 = \rlcI - \rlcQ_0$ and $\rlcP_1 = \rlcI - \rlcQ_1$. For index $2$, we are guaranteed that $\rlcM_2$ is non-singular. The differentiable variables here are then $\vec{y} = \rlcP_0 \rlcP_1 \vec{x}$, and the algebraic variables occur at different layers: $\vec{z} = \rlcQ_0 \rlcP_1 \vec{x}$, and $\vec{w} = \rlcQ_1 \vec{x}$. 

The induced explicit $\ode$ is obtained by projecting the $\dae$ on the constraint submanifold using the projector $\rlcP_0 \rlcP_1$.
\begin{align}
\dot{\vec{y}} &=  - \rlcP_0 \rlcP_1 \rlcM_2^{-1} \rlcK \vec{y}(t) + \rlcP_0 \rlcP_1 \rlcM_2^{-1} \vec{f}.
\end{align}

\paragraph{The finer technical details.}

\begin{enumerate}[$1)$] 
    \item \emph{Block-encodings:} For our quantum algorithm, we need to implement matrices corresponding to the dynamics of the explicit $\ode$ obtained and the algebraic relations. As mentioned earlier, these are obtained from the matrices $\rlcM$ and $\rlcK$ via the sequence of admissible projectors. Given access to block-encodings of $\rlcM$ and $\rlcK$, we show that these can be implemented efficiently. For instance, we need to block-encode the $\ode$ operator $\rlcM^{-1}\rlcK$ in the case of index 0. For index $1$, we need to block-encode $\ode$ operator $\rlcP_0 \rlcM_1^{-1} \rlcK$ as well as the algebraic equation operator $\rlcQ\rlcM_1^{-1} \rlcK$ in the case of index 1. Starting from block-encodings of $\rlcM$ and $\rlcK$, we use fairly standard quantum linear algebra techniques and QSVT~\cite{gilyen2019qsvt,martyn2021grand} to construct approximate block-encodings of all relevant operators for $\dae$s of indices 0, 1 and 2. The main challenge here is to carefully account for approximation errors while chaining together several block-encoding subroutines. 
    \item \emph{Propagation of errors:} As we use block-encodings of the projectors to implement the initial state and the forcing required for the $\ode$ solve over $\vec{y}$, we are only able to do this approximately. As a consequence, we need to account for the corresponding errors in the quantum state that the $\ode$ solver outputs. We show that the error only show grows linearly in time with the magnitude of approximation error of the forcing and has a constant factor dependence on the magnitude of the approximation error in the initial condition. We also need to account for errors in the approximations of the applied block-encodings on the output state which we do by standard techniques.
    \item \emph{Complexity of history state and final state preparation:} The main contribution to the complexity of our quantum algorithm is dictated by the quantum $\ode$ solver. The complexity of solving an $\ode$  $\dot{\vec{x}} = \mathbf{H} \vec{x} + \vec{b}$ depends on the norm of the involved matrix $\mathbf{H}$, its condition number $\kappa(\mathbf{H})$, and norm of the exponential $\expnorm(\mathbf{H}) := \max_{t \in [0,T]} \norm{\exp(\mathbf{H} t)}$, all of which we bound for $\rlc$ circuits and also for a general $\dae$. The linear algebraic operations that implemented the projection steps have gate complexity that depend on the condition number of the associated matrices. In addition, we must also account for the gate-complexity of block-encoding $\mathbf{H}$ and other associated matrices for every index. 
\end{enumerate}

\subsection{Quantum simulation of RLC circuits}\label{subsec:qsim_rlc}
In this section, we describe how the quantum $\dae$ solver from the previous section can be applied to simulate the differential-algebraic equations ($\dae$) of modified nodal analysis ($\mna$) governing the dynamics of $\rlc$ circuits, leading to Result~\ref{res:RLC_sim}. 

\subsubsection{Applying the quantum DAE solver: Block encoding constructions}
Our quantum algorithm (aka quantum $\dae$ solver) for simulating $\dae$'s of the form $\rlcM \dot{\vec{x}} + \rlcK \vec{x} = \vec{f}$ requires access to the approximate block-encodings for the matrices $\rlcM$ and $\rlcK$. Moreover, the time complexity of the quantum $\dae$ solver is dictated by the cost of these block-encodings. For $\rlc$ circuits with degree-$d$ (excluding the reference node), we show how to construct these using standard sparse oracle access~\cite{gilyen2019qsvt} to the reduced incidence matrices corresponding to the different components and power sources, and the component matrices ($\rlcG$, $\rlcL$, $\rlcC$) themselves.  Using standard routines, we show that the cost of the block-encodings of $\rlcM$ and $\rlcK$ involves $O(1)$ uses of the oracles and $O(\log N)$ additional gate complexity.

The cost of the block-encodings of the operators corresponding to the dynamics and projectors depend on the component matrices and the structure of the underlying graph. Particularly, for linear $\rlc$ circuits, there is a characterization of the tractability index of the $\mna$ equations completely in terms of the topological structure of the graph~\cite{estevez2000structural,riaza2008differential}, which we discuss next.
The matrix $\rlcM$ is non-singular when the $\mna$ system (Eq.~\eqref{eq:mna}) has index $0$. This occurs if and only if the $\rlc$ circuit has no voltage sources and there exists a capacitive spanning tree. The former condition is clearly required as the sub-matrix corresponding to voltage sources in $\rlcM$ would be zero and the latter condition ensures that $\rlcA_{\rlccap}^T$ (and thereby the submatrix $\rlcA_{\rlccap} \rlcC \rlcA_{\rlccap}^T$) is non-singular. The $\mna$ system over the state $\vec{x} = (\vec{u}, \vec{i}_\rlcind)^T$ is then
\begin{align} 
\label{eq:mna_index0}
\underbrace{
\begin{bmatrix}
    \rlcA_{\rlccap} \rlcC \rlcA_{\rlccap}^T & 0 \\
    0 & \rlcL 
\end{bmatrix}}_{:=\rlcM}
\frac{d}{dt} \begin{bmatrix}
    \vec{u}(t) \\ \vec{i}_\rlcind(t)
\end{bmatrix} + 
\underbrace{
\begin{bmatrix}
    \rlcA_\rlcres \rlcG \rlcA_\rlcres^T & \rlcA_{\rlcind} \\
    -\rlcA_{\rlcind}^T & 0
\end{bmatrix}}_{:=\rlcK}
\begin{bmatrix}
    \vec{u}(t) \\ \vec{i}_{\rlcind}(t)
\end{bmatrix} = 
\underbrace{
\begin{bmatrix}
    - \rlcA_{\rlcjs} \vec{i}_{\rlcjs} \\ 0
\end{bmatrix}}_{:=\vec{f}},
\end{align}
given initial conditions of $\vec{u}(0)$ and $\vec{i}_\rlcind(0)$ and the vector $\vec{f}$ only is non-zero at components corresponding to the independent current sources. In this case, the cost of the block-encoding of the relevant differential operator $\rlcM^{-1}\rlcK$ is dictated by the condition number of $\rlcM$ and $\norm{\rlcK}$, which go as $O(d,\rlcres_\min^{-1}, \rlccap_\max, \rlcind_\max, \lambda_{\min}(\rlcA_\rlccap \rlcA_\rlccap^T)^{-1})$ where $\rlcres$ is the minimum resistance, $\rlccap_\max$ is the maximum capacitance, and $\rlcind_\max$ is the maximum capacitance. The structural dependence here enters through the minimum non-singular value of $\rlcA_\rlccap \rlcA_\rlccap^T$.

It is known that circuits without loops over voltage sources and capacitors, and cutsets over inductors and current sources, corresponding to $\dae$ index $1$. For all other circuits, we obtain a $\mna$ $\dae$ of tractability index $2$. For the former, we show that the cost of the block-encoding of the relevant dynamics $\rlcP_0 \rlcM_1^{-1} \rlcK$ can again be connected to the underlying structure of the graph.

\subsubsection{Analysis}

To comment on the complexity of our quantum $\dae$ solver when applied to $\rlc$ circuits, we need to comment on the gate complexity of the involved block-encoding constructions, which we discussed above, and the norm of the exponential for the $\dae$s of $\rlc$ circuits, which we discuss next. 

\paragraph{Norm of the exponential.}

We bound the norm of the exponential $\expnorm(\mathbf{H}) = \max_{t \in [0,T]} \|\exp(\mathbf{H}t)\|$ for each index by commenting on the Lyapunov stability of the $\ode$ under consideration.

An equilibrium point is considered stable in the sense of Lyapunov if the norm of this exponential remains bounded for all $t \geq 0$, ensuring that trajectories starting near the origin stay within a fixed neighborhood. If the system is asymptotically stable (meaning all eigenvalues of $\mathbf{H}$ have strictly negative real parts), the norm $\|\exp(\mathbf{H}t)\|$ does not just stay bounded but decays to zero as $t \to \infty$ \cite{mehrmann2023control}.
Mathematically, starting from the Lyapunov equation
$$
\mathbf{H}^T \mathbf{X} + \mathbf{X} \mathbf{H} = \mathbf{Y}.
$$
and choosing a suitable negative definite matrix $\mathbf{Y}$, the unique positive definite solution $\mathbf{X}$ allows us to establish an exponential bound of the form:
\begin{equation}
    \expnorm(\mathbf{H}) \leq \sqrt{\kappa(\mathbf{X})} e^{\alpha t},
\end{equation}
where $\alpha$ is the log-norm of $\mathbf{Y}$. This lets us bound the ``worst-case'' growth of the system\,---\,often caused by non-normal matrices where $\|\exp(\mathbf{H}t)\|$ might transiently grow before decaying\,---\,without needing to compute the matrix exponential or its Jordan form directly. We show how to choose $\mathbf{X}$ and $\mathbf{Y}$ for $\mathbf{H}$ arising from $\dae$s of index 0, 1 and 2. In particular, we obtain that the exponential norm for $\rlc$ circuits for each index can be upper bounded by the condition number of a regularized $\rlcM$ matrix, which itself we shown can be in terms of the degree and values of the component matrices.

\paragraph{Overall complexity.}
We determine our overall complexity for each $\dae$ index (i.e., $0$, $1$, $2$) to be $\widetilde{O}(d, \log N)$ with dependencies on the structure of the graph and values of the components. For example, our algorithms involve the dependency on $1/\lambda_{\min}^+(\rlcA_\rlccap \rlcA_\rlccap^T)$ (or the minimum non-zero eigenvalue of $\rlcA_\rlccap \rlcA_\rlccap^T$). In this case, the matrix $\rlcA_\rlccap \rlcA_\rlccap^T$ is nothing but the Laplacian of the subgraph corresponding to capacitors. Here minimum non-zero eigenvalue corresponds to the spectral gap since the smallest eigenvalue of the Laplacian is zero. From the theory of expander graphs \cite{hoory2006expander}, it is known that the spectral gap of the Laplacian of a graph $\calG$ can be lower bounded by $\phi(\calG)^2/2$, where $\phi(\calG)$ is the conductance of the graph. Constructions of $d$-regular graphs are known where the spectral gap is a constant. It is important to note that here, we do not need the spectral gap to be constant. It is sufficient for it to be $1/\polylog(N)$, where $N$ is the number of vertices.

\subsection{Estimating energy and other observables}\label{subsec:quantum_speedup}

Our quantum algorithm for simulating $\dae$s prepares a state of the following form before any measurement is carried out:
\begin{equation}\label{eq:history}
    \ket{\phi} = \frac{\norm{\widetilde{\Psi}}}{\beta}\ket{0^a}\ket{\widetilde{\Psi}} + \ket{\perp}\,,
\end{equation}
where $\widetilde{\Psi}$ is the history state that is close to the desired history state or has the promise that its final time component is close to the desired solution at the final time, $\ket{\perp}$ refers to some state that is orthogonal to every state with $\ket{0^a}$ in the first register, $a$ is the number of ancillas used and $\beta$ is the norm of the initial state that is specified and known to us. To estimate observables of the form $\vec{x}(T)^T \mathbf{O} \vec{x}(T)$ (also called quadratic form when $\mathbf{O}$ is Hermitian), we prepare $\ket{\phi}$ considering the task of final time preparation (Problem~\ref{prob:RLC_dynamics}$(ii)$) and in which case, the encoded history state has the following form 
\begin{equation}
    \ket{\widetilde{\psi}} \propto \sum_{k=0}^{m}\|\widetilde{x}(k\Delta t)\| \ket{k} \ket{\widetilde{x}(k\Delta t)}\,, \text{ where } \norm{\widetilde{x}(T) - \vec{x}(T)} \leq \varepsilon \norm{\vec{x}(T)},
\end{equation}
or in other words the encoded $\ket{\widetilde{x}(T)}$ component is close to the true final state $\ket{\vec{x}(T)}$. An estimate of $\vec{x}(T)^T \mathbf{O} \vec{x}(T)$ is then obtained by carrying out a Hadamard test of $\ket{\phi}$ and a block-encoding of $\mathbf{D} = \ket{0}^a\bra{0}^a \otimes \ket{T}\bra{T} \otimes \mathbf{O}$ that selects for the final time. Note that we do not measure the state $\ket{\phi}$ which itself is obtained after applying various block-encodings, especially the block-encoding corresponding to the quantum linear system algorithm. Since we know that the correct part of the state is when the flag qubits are $\ket{0}^a$, we simply block-encode this information into the observable.

\paragraph{RLC circuits.} Particularly, this enables us to compute quantities relevant to RLC circuits, such as the energy stored in a set of capacitors or inductors, or the power dissipated across a set of resistors in the circuit. For example, the energy across capacitors of the circuit takes the form $(\vec{u}^T \rlcA_\rlccap\rlcC \rlcA_{\rlccap}^T \vec{u})/2$. This corresponds to measuring against the state $\vec{x} = [\vec{u}, \vec{i}_\rlcind, \vec{i}_\rlcvs]^T$ with an observable $\mathbf{O}$ that has the sub-matrix $\rlcA_\rlccap\rlcC \rlcA_{\rlccap}^T$ corresponding to $\vec{u}$ part of the state and $0$ elsewhere. Moreover, $\rlcA_\rlccap\rlcC \rlcA_{\rlccap}^T$ is efficient to block-encode and is required to be block-encoded as part of setting up the $\mna$ equations.

\subsection{BQP-completeness}
Finally, we show that estimation of energy from electrical circuits is $\bqp$-hard. In particular, we consider the following problem.
\begin{problem}[\textsf{BQP}-completeness]\label{prob:BQP}
Given a network of inductors and capacitors with no resistors or sources and oracle access to the matrices $\rlcA$, $\rlcL$ and $\rlcC$, decide if the normalized energy in a single capacitor is above $2/3$ or below $1/3$.
\end{problem}

To show this, we show that an arbitrary instance of the coupled oscillator problem of \cite{PhysRevX.13.041041} can be reduced to the above $\mathsf{LC}$ decision problem. Since it was shown that the analogous decision version for the coupled oscillator problem is $\bqp$-complete, the hardness the $\mathsf{LC}$ problem follows. Our reduction gives a concrete construction of the folklore intuition in analog circuit theory that $\mathsf{LC}$ networks are related to springs and masses (without damping). An equivalent way of stating this result is as an estimation problem (which may be more natural from a practical point of view). Our results show that estimating the energy of a single capacitor to a constant additive precision is $\bqp$-hard. This indicates that it is unlikely that there is an efficient (i.e., $\polylog N$) classical algorithm for this problem whereas our algorithms show there there is a $\polylog N$ quantum algorithm (here $N$ is the number of components in the circuit). Another interesting aspect of this reduction is that it gives a way to construct an $\mathsf{LC}$ circuit corresponding to a coupled oscillator network.\footnote{We remark that not all $\mathsf{LC}$ circuits however can be mapped to an arbitrary coupled oscillator network in the sense of \cite{PhysRevX.13.041041} e.g., those without a capacitive tree.} Other than for small examples, we do not know of any work that gives a way to construct such a circuit for a general oscillator network. 

\section{Discussion}\label{sec:discussion}
\subsection{Contextualizing our work}\label{subsec:related_works}
In this section, we quickly go through some design choices as part of the proposed quantum $\dae$ solver and contextualize the current work with regards to prior work.

\paragraph{Choice of differential equation solver.} There are two main approaches to solving linear ordinary differential equations ($\ode$s). The first converts the problem into that of solving a linear system by discretizing time and then using a quantum linear systems algorithm \cite{harrow2009linear, CKS}. This approach has been developed in several papers for time-independent and time-dependent coefficients \cite{Ber14, BCOW17, Krovi2023improvedquantum, Fang2023timemarchingbased, Berry2024quantumalgorithm}. The second approach, called linear combination of Hamiltonian simulation (LCHS), uses Hamiltonian simulation rather than linear systems \cite{LCHS}. The time complexity of these algorithms depends on the properties of the matrix (corresponding to the differential operator) involved in the linear $\ode$. Many algorithms assume that it has a negative log-norm (where the log-norm is a quantity that measures the stability of the $\ode$). However, the $\ode$s for $\rlc$ circuits obtained from $\mna$ in the current work have differential operators with a positive log-norm (while still being stable in a Lyapunov sense). Therefore, it is essential to use a quantum ODE solver that allows for a positive log-norm. Here, we use the version presented in \cite{Krovi2023improvedquantum} because it does not assume that the log-norm is negative but rather quantifies the run-time using the exponential norm.

\paragraph{Quantum algorithms for classical dynamical systems.}
The speed-up in quantum algorithms for linear $\ode$s occur for high dimensional problems under certain conditions such as norms and sparsity of involved operators being well-behaved. These algorithms also have potential applications to simulation of classical dynamical systems. Some examples include simulating the wave equation \cite{PhysRevA.99.012323}, computing the energy of a network of coupled oscillators \cite{PhysRevX.13.041041}, and coupled oscillators with damping and forcing \cite{krovi2024quantum}. More recently, under some conditions, these applications have been extended to nonlinear oscillators and equations with stochasticity \cite{jennings2025end, bravyi2025quantumsimulationnoisyclassical}. Motivated by the need to expand the set of applications to high-utility practical problems, we consider the problem of simulating the dynamics of RLC circuits with current and voltage sources in this work.

\paragraph{Prior work on electrical networks.}
Prior work has studied steady-state problems for electrical circuits, specifically that of computing effective resistance in purely resistive circuits \cite{wang2017efficient}. The connection between effective resistance, graph theory and quantum algorithms has also been studied extensively in a number of works \cite{Belovs2013QuantumWA,Apers2022ElfsTA,Jarret2018QuantumAF,chakraborty_et_al:LIPIcs:2019:10609}.
While \cite{PhysRevX.13.041041} mentions that $\mathsf{LC}$ circuits can be modelled with classical harmonic oscillators, our $\bqp$-hardness proof indicates that only a restricted class of $\mathsf{LC}$ circuits can be simulated by their algorithm.

\paragraph{Oracle access.}
Beyond the spectral and dynamical hypotheses that are needed for a super-polynomial speed-up to survive, the quantum \textsf{DAE} solver additionally requires two ingredients that classical transient solvers do not: (i)~the circuit connectivity must be
specified through a sparse oracle $\mathcal{O}$ that, on input of a node index, returns its neighbors and the resistance, inductance, capacitance values of the incident branches\footnote{For the block-encoding constructions required for setting up the $\mna$ equations for the $\rlc$ circuits, we assume standard sparse oracle access~\cite{gilyen2019qsvt} to the incidence submatrices and the component matrices. These oracles could however be constructed using the oracle access described here.}; and (ii)~the quantity of interest must be representable as an efficiently measurable observable. Whether any physically meaningful network admits such an oracle is therefore a precondition for the speed-up to be practically relevant. However, these preconditions are in fact met in practice: the canonical example is the parasitic-$\rlc$ model of a VLSI clock-distribution network~\cite{kahng2011vlsi}, whose $10^{6}$--$10^{9}$ element values are never stored as an explicit netlist but instead computed \emph{on the fly} from the layout geometry by multipole-accelerated field solvers~\cite{nabors1991fastcap,kamon1994fasthenry}. The same pattern---layout-driven, on-the-fly computation of lumped
resistors, inductors, or capacitors ---is standard for on-chip power- and ground-delivery networks~\cite{qian2003randomwalk}, for full-chip thermal networks  (which share the $\rlc$ form under the electro-thermal duality)~\cite{zhan2007thermal}.
\subsection{Open questions}\label{sec:open_questions}
Our work opens up several interesting directions for future work.
\begin{enumerate}
    \item \textbf{More general electrical circuits:} We have just scratched the surface of electrical circuits that would be amenable to quantum simulation. It would be worthwhile to investigate if methods in the current work can be extended to include active elements (e.g, op-amps, transistors), nonlinear elements (e.g., diodes, transistors, varistors) as well as voltage (or current) controlled elements and time-dependent power sources.
    \item \textbf{Alternate circuit models:} In this work, we utilized modified nodal analysis ($\mna$) in modeling $\rlc$ circuit dynamics. However, alternate models of $\rlc$ circuits~\cite{riaza2008differential} exist which consider different descriptions of the state and hence lead to different $\dae$s. For example, sparse tableau formulations use branch voltages and currents. These are very systematic and topology-transparent, but introduce more variables than nodal methods. Charge/flux-based MNA instead uses capacitor charges and inductor fluxes, which is often preferred when circuits contain nonlinear elements. A natural extension of the current work would be to then consider applying the proposed quantum $\dae$ solver to these alternate models.
    \item \textbf{Quantum simulation strategies:} One of the main subroutines of our quantum $\dae$ solver is a quantum $\ode$ solver~\cite{Krovi2023improvedquantum} which is based on solving linear systems corresponding to the $\ode$. This $\ode$ solver could be replaced by perhaps directly mapping the problem to Hamiltonian simulation or by utilizing alternate $\ode$ solvers such as those based on linear combination of Hamiltonian simulation (LCHS)~\cite{LCHS}. To ensure applicability, one would then need to be able to handle dynamical operators that may have positive log-norm but bounded exponential norm.
    \item \textbf{Resource estimation:} To make simulating $\rlc$ circuits a viable practical application of fault-tolerant quantum computers, we need to investigate efficient quantum circuit synthesis and resource estimation for specific problem instances against different architectures. Moreover, one could consider building benchmarks for resource estimation of quantum differential equation solvers based on simulating $\rlc$ circuits as has been done for quantum phase estimation considering ground state estimation of chemistry Hamiltonians (e.g., \cite{lee2021thc}).
    \item \textbf{Simulation of general $\dae$s:} In this work, we primarily investigated the application of the proposed quantum $\dae$ solver to simulation of $\rlc$ circuits. It is then natural to consider applications to other systems governed by $\dae$s such as those in control theory~\cite{athans2013optimal}, robotics~\cite{campbell2019applications}, mechanical models of constrained systems~\cite{campbell1995high,wanner1996solving}, dynamics of gas or fluid networks~\cite{jansen2014unified,grundel2014model}, fluid dynamics~\cite{emmrich2013operator}, and chemical engineering~\cite{stechlinski2018nonsmooth}. This will also require extension of the analysis to linear time-dependent $\dae$s as well as nonlinear $\dae$s. 
\end{enumerate}


\section{Preliminaries}
Let $\mathbb N_0$, $\mathbb{N}$, $\mathbb Z$, $\mathbb R$, $\mathbb R_+$, and $\mathbb C$ denote the set of nonnegative integers, positive integers, integers, real numbers, positive real numbers, and complex numbers, respectively. For $n \in \mathbb N$, we define $[n]_0 \coloneq \{0,1,\dots,n\}$ and $[n]\coloneq \{1,\dots,n\}$. We let $\| \cdot \|$ denote the 2-norm, unless otherwise specified by an appropriate subscript. 

We will often use the following fact.
\begin{fact} \label{fact:normalized_vector_error}
Let $\vec{p},\vec{q} \in \mathbb{C}^N$. Define $\ket{p} := \vec{p}/\|\vec{p}\|$ and $\ket{q} := \vec{q}/\|\vec{q}\|$. Then, 
$$
\norm{\ket{p} - \ket{q} } \leq \frac{2\|\vec{p} - \vec{q}\|}{\|\vec{p}\|}.
$$
\end{fact}
\begin{proof}
Starting from the definitions of $\ket{p}$ and $\ket{q}$, we obtain
\begin{align*}
    \left \| \frac{\vec{p}}{\|\vec{p}\|} - \frac{\vec{q}}{\|\vec{q}\|} \right \| \leq \left \| \frac{\vec{p}}{\|\vec{p}\|} - \frac{\vec{q}}{\|\vec{p}\|} \right \| + \left \| \frac{\vec{q}}{\|\vec{p}\|} - \frac{\vec{q}}{\|\vec{q}\|} \right \|
    &= \frac{\|\vec{p} - \vec{q} \|}{\|\vec{p}\|} +  \frac{\norm{\vec{q} (\|\vec{q}\| - \|\vec{p}\|)}}{\|\vec{p}\| \|\vec{q}\|} \\
    &= \frac{\|\vec{p} - \vec{q} \|}{\|\vec{p}\|} +  \norm{\ket{q}} \frac{\Big| \|\vec{q}\| - \|\vec{p}\| \Big|}{\|\vec{p}\|} \\
    &\leq \frac{2\|\vec{p} - \vec{q}\|}{\|\vec{p}\|},
\end{align*}
where we added and subtracted $\vec{q}/\|\vec{p}\|$ before applying the triangle inequality in the second inequality, and used the reverse triangle inequality of $| \norm{\vec{q}} - \norm{\vec{p}}| \leq \norm{\vec{q} - \vec{p}}$ in the last line. This gives us the desired result.
\end{proof}

\subsection{Electrical circuits}
We will work with a multi-graph $\calG=(\calV,\calE)$ describing the connectivity of an electrical circuit with nodes corresponding to intersections of wires between distinct components and edges along branches on which components lie. Note that there may be more than one branch and hence edge between two nodes.  

We will denote the voltage drop across branch $b$ (also called branch voltages) as $v_b$ and the corresponding vector of voltage drops across all branches in $\calG$ as $\vec{v}$. Similarly, the current through branch $b$ will be denoted as $i_b$ and the corresponding vector $\vec{i}$. Moreover, after setting a particular node in the electrical circuit $\calG$ as the reference node $a_0$ and considered grounded (i.e., voltage is zero), we can assign voltages to different nodes $a$ by considering any path from $a_0$ to $a$ and summing up the voltage drops across the branches lying along this path. We will denote the node voltage at node $a$ as $u_a$ and the corresponding vector as $\vec{u} \in \mathbb{R}^{N+1}$.
 
Moreover, we will associate some directions to the edges $\calE=\{(n_i,n_j)\}$ with $n_i \rightarrow n_j$ i.e., $n_i$ is the source of the edge (or the edge leaves node $n_i$) and $n_j$ is the sink of the edge (or edge enters node $n_j$). Let us define the incidence matrix $\rlcA$ as follows:
\begin{equation}\label{eq:incidence_defn}
    \rlcA_{jk} = \begin{cases}
        +1 & \text{ if edge $k$ leaves node $j$} \\
        -1 & \text{ if edge $k$ enters node $j$} \\
        0 & \text{otherwise}
    \end{cases}
\end{equation}
Moreover, as the node voltages are determined with respect to a particular reference node, we will omit the row in $\rlcA$ corresponding to the reference node. Then, $\rlcA \in \mathbb{Z}^{(|\calV| - 1)\times |\calE|}$. We have the following facts commenting on the reduced incidence matrix $\rlcA$ and its submatrices.
\begin{fact}
For a connected graph $\calG$, the incidence matrix $A$ is full rank i.e.,  $\mathrm{rank}(\rlcA) = |\calV|-1$.
\end{fact}

Let us assume that the edges $\calE$ are ordered such that we can write the incidence matrix as
$$
\rlcA = \begin{bmatrix}
        \rlcA_\rlcres & \rlcA_\rlccap & \rlcA_\rlcind & \rlcA_{\rlcjs} & \rlcA_{\rlcvs}
    \end{bmatrix}, \quad
$$
with $\rlcA_\rlcres$ (and similarly others) corresponding to the sub-matrix of $\rlcA$ whose columns correspond to edges with resistors (or capacitors, inductors, current sources, voltage sources). We can then accordingly partition the node voltages, branch voltages or branch currents with subscripts $\rlcres,\rlcind,\rlcind,\rlcvs,\rlcjs$ corresponding to branches containing resistors, inductors, capacitors, current sources, and voltage sources, respectively.
\begin{align}
    \vec{u} = (\vec{u}_\rlcres, \vec{u}_\rlccap, \vec{u}_\rlcind, \vec{u}_{\rlcjs}, \vec{u}_{\rlcvs})^T, \quad
    \vec{v} = (\vec{v}_\rlcres, \vec{v}_\rlccap, \vec{v}_\rlcind, \vec{v}_{\rlcjs}, \vec{v}_{\rlcvs})^T, \quad
    \vec{i} = (\vec{i}_\rlcres, \vec{i}_\rlccap, \vec{i}_\rlcind, \vec{i}_{\rlcjs}, \vec{i}_{\rlcvs})^T.
\end{align}

We have the following facts regarding properties of the incidence matrix $\rlcA$ in relation to the structure of the graph of the electrical circuit $\calG$. For that, we have the following definitions. Within a given graph, a path connecting two vertices $a_0$ and $a_m$ is a sequence
of nodes $(a_0, a_1, \ldots, a_m)$ in which the consecutive nodes are connected to each other by an edge. A path is said to be closed if $a_0 = a_m$. A closed path with $m \geq 1$ in which an edge is not repeated in the path $a_i \neq a_j, \forall 1 \leq i < j \leq m$ called a loop. We will call a subset $F$ of the set of edges $\calE$ of the connected graph $\calG$, a cutset if the removal of $F$ from $\calG$ results in a disconnected graph. Finally, given a connected graph, we will refer to a connected subgraph which contains all nodes and has no loops, as a tree \footnote{As is common in circuit theory~\cite{riaza2008differential}, the definition of tree implicitly assumes that it is a spanning tree.}.

\begin{fact}[Lemma~5.3,\cite{riaza2008differential}]\label{fact:rel_loop_A}
A subset $F$ of the set of edges of a connected digraph $\calG$ does not contain loops if and only if $A_F$ has full column rank.
\end{fact}

\begin{fact}[Lemma~5.4,\cite{riaza2008differential}]\label{fact:rel_cutset_A}
A subset $F$ of the set of edges of a connected digraph $\calG$ does not contain cutsets if and only if $A_{\calG-F}$ has full row rank.
\end{fact}

\begin{fact}[Lemma~5.5,\cite{riaza2008differential}]\label{fact:rel_tree_A}
Let $F$ be a set of $N-1$ edges of a connected digraph $\calG$, which does not contain cutsets. Then $A_F$ is non-singular if and only if $F$ defines a tree. In this case, $\det(A_F) = \pm 1$.
\end{fact}

\paragraph{Kirchhoff laws.}
The branch currents and voltages in any electrical circuit have to satisfy the so called Kirchhoff laws. The Kirchhoff's current law ($\kcl$) states that the sum of currents leaving any node is zero. This can be summarized as
\begin{equation}\label{eq:KCL}
    \rlcA\vec{i} = 0 \implies \rlcA_{\rlcres} \vec{i}_{\rlcres} + \rlcA_{\rlccap} \vec{i}_\rlccap + \rlcA_{\rlcind}\vec{i}_\rlcind + \rlcA_{\rlcvs} \vec{i}_{\rlcvs} + \rlcA_{\rlcjs} \vec{i}_{\rlcjs} = 0,
\end{equation}
where the implication follows from using the reduced incidence matrices corresponding to different electrical components. Kirchhoff's voltage law ($\kvl$) states that the sum of voltage drops across branches of any loop is zero. This can be recast in terms of node voltages by setting the voltage of the reference node to be zero and then noting that the $\kvl$ enforces that the voltage at any node $a$ be the sum of the branch voltages along any path from the reference node to $a$, as we had defined earlier. It can then be shown that we have the following constraint corresponding to $\kvl$
\begin{equation}\label{eq:KVL}
\vec{v} = \rlcA^T \vec{u}. 
\end{equation}

\paragraph{Constitutive equations.} 
We will assume the elements in the $\rlc$ circuits are passive. The constitutive equations relating the branch current and voltage for each of the components - resistors, capacitors and inductors - are as follows respectively:
\begin{equation}
    \rlcR \vec{i}_{\rlcres} = \vec{v}_{\rlcres} = \rlcA_{\rlcres}^T \vec{u}, \quad \vec{i}_{\rlccap} = \rlcC \frac{d}{dt} \vec{v}_{\rlccap} = \rlcC \rlcA_{\rlccap}^T \frac{d}{dt} \vec{u}, \quad \rlcL \frac{d}{dt} \vec{i}_{\rlcind} = \vec{v}_{\rlcind} = \rlcA_{\rlcind}^T \vec{u},
    \label{eq:component_iv_relations}
\end{equation}
where $\rlcR$ is a positive diagonal matrix containing resistances, $\rlcC$ is a positive diagonal matrix containing capacitances and $\rlcL$ is a symmetric positive definite matrix with self-inductances along the diagonal and mutual inductances as the off-diagonal terms.

\paragraph{Well-posed circuits.} A circuit is said to be \emph{well-posed} if it has neither $V$-loops (loops of independent voltage sources) nor $I$-cutsets (cutsets over independent current sources). This is to prevent short-circuits. This is formally stated in the following lemma.
\begin{lemma}[{\cite[Lemma~3.45]{lamour2013differential}}]\label{lem:well-posed-circs}
Given an $\rlc$ circuit with independent voltage sources and current sources, then, the
following relations are satisfied.
\begin{enumerate}[$(i)$]
\item The matrix $\begin{bmatrix}\rlcA_\rlcres & \rlcA_\rlccap & \rlcA_\rlcind & \rlcA_{\rlcvs} \end{bmatrix}$ has full row rank, since cutsets of current sources are forbidden.
\item The matrix $\rlcA_{\rlcvs}$ has full column rank, since loops of voltage sources are forbidden.
\end{enumerate}
\end{lemma}
The relations in the above lemma relating conditions on the incidence matrix to conditions of the graph of the $\rlc$ circuit can be obtained via Facts~\ref{fact:rel_loop_A}--\ref{fact:rel_cutset_A}. Note that well-posed circuits may contain $VC$-loops with at least one capacitor and $IL$-cutsets with at least one inductor. 

\paragraph{Structured $\rlc$ circuits.} 
In this work, we will focus on $\rlc$ circuits over $N+1$ nodes where each node has degree at most $d$ (excluding the reference node). In particular, we will work with following definition.
\begin{definition}[Degree-$d$ $\rlc$ Circuits]\label{def:structure_RLC_circs}
Let $d, N \in \mathbb{N}$, $\rlccap_\max \geq \rlccap_\min > 0$, $\rlcres_\max \geq \rlcres_\min > 0$, $\rlcind_\max > 0$, $\mathbb{R} \ni \rlcind_\min < \rlcind_\max$. Suppose that $\calG$ is an $\rlc$ circuit over $N+1$ nodes. We say that $\calG$ has degree-$d$ if each node (excluding the reference node) has degree at most $d$. Additionally, we will assume that the capacitance matrix $\calC$ is a diagonal matrix with elements taking values in $[\rlccap_\min,\rlccap_\max]$, the conductance matrix $\rlcG = \rlcR^{-1}$ is a diagonal matrix with elements taking values in $[\rlcres_\max^{-1},\rlcres_\min^{-1}]$, and the inductance matrix $\rlcL$ is symmetric positive definite with elements taking values in $[\rlcind_\min,\rlcind_\max]$. Moreover, any inductor in $\calG$ has mutual inductances with at most $d_\rlcind \leq d - 1$ other inductors in $\calG$.
\end{definition}
\emph{Remark:} Note that the reference node can have very high degree and possibly be connected to all the other nodes in $\calG$. This occurs for example in a star graph. Additionally, for $\calG$s satisfying Definition~\ref{def:structure_RLC_circs}, $\rlcA$ is $d$-row-sparse and 2-column-sparse. For degree-$d$ LC circuits, the conductance matrix $\rlcG$ is just the zero matrix.

\begin{fact}\label{fact:spectrum_RLC_components}
Let $\mu \in \mathbb{R}, |\mu| \leq \rlcind_\max/d$. Consider Definition~\ref{def:structure_RLC_circs}. The spectrum of $\rlcC$ and $\rlcG$ satisfies $\spec(\calC) \in [\rlccap_\min, \rlccap_\max]$ and $\spec(\rlcG) \in [\rlcres_\max^{-1}, \rlcres_\min^{-1}]$ respectively. For the inductance matrix, $\norm{\rlcL} \leq d \rlcind_\max$. If the magnitudes of all mutual inductances in $\rlcL$ are less than $\mu$, then $\norm{\rlcL} \geq \rlcind_\max/d$.
\end{fact}
\begin{proof}
The spectrum of $\rlcC$ and $\rlcG$ follows from Definition~\ref{def:structure_RLC_circs}. The upper bound on $\norm{\rlcL}$ follows from the application of Gershgorin circle theorem and noting that elements take a maximum value of $\rlcind_\max$. Given the promise on mutual inductances, the lower bound on $\norm{\rlcL}$ also follows from Gershgorin circle theorem as $|\lambda_\min - \rlcind_\max| \leq (d-1) \rlcind_\max/d \implies \lambda_\min \geq \rlcind_\max/d$.
\end{proof}

\paragraph{Graph Laplacian.}
We define the \emph{reduced} Laplacian of the graph $\calG(\calE,\calV)$ as
\begin{equation}\label{eq:reduced_laplacian}
    \calL_\calG = \rlcA \rlcA^T,
\end{equation}
where the $\rlcA$ is the reduced incidence matrix (which does not contain the row corresponding to the reference node of the graph) and hence the name. We have the following useful fact regarding the spectrum of $\calL_\calG$ that we will use throughout this work.
\begin{fact}\label{fact:spectrum_laplacian_G}
The reduced Laplacian of a degree-$d$ graph $\calG$ (excluding the reference node) satisfies
$$\lambda_{\max}(\calL_\calG) \leq 2d.$$
\end{fact}
\begin{proof}
Let $\calV'$ be the vertex set without the reference node. From the definition of $\calL_{\calG} = \rlcA\rlcA^T$, the elements of $\calL_{\calG}$ are $\calL_{\calG,ij} = \sum_{e \in \calE} \rlcA_{ie} \rlcA_{je}$. We observe that the diagonal elements of $\calL_{\calG}$ satisfy
$$
\calL_{\calG,ii} = \sum_{e \in \calE} \rlcA_{ie}^2 = \sum_{e \in \calE} \sgn(\rlcA_{ie}) = \deg(i) \leq d,
$$
where we have used that the only edges contributing to the sum are those that are incident on node $i \in \calV'$ in the third equality and denoted the degree of node $i \in \calV'$ by $\deg(i)$. The off-diagonal elements satisfy (where $i \neq j$)
$$
\calL_{\calG,ij} = \sum_{e \in \calE} \rlcA_{ie} \rlcA_{je} = - (\# (\text{edges between } i \text{ and } j)),
$$
where we have used in the second equality that the edge $e$ must be incident on both nodes $i$ and $j$ for the term to be non-zero and that the edge is then the source of one node while the sink for the other which implies that $\rlcA_{ie}\rlcA_{je} = -1$. We thus have that for any $i \in \calV'$
$$
\sum_{j \in \calV', j \neq i} |\calL_{\calG,ij}| = \sum_{j \in \calV', j \neq i} (\# (\text{edges between } i \text{ and } j)) \leq \deg(i) \leq d.
$$
From the Gershgorin circle theorem, we then have that
$$
\lambda_{\max}(\calL_{\calG}) \leq \max_{i \in \calV'} \calL_{\calG,ii} + \sum_{j \in \calV', j \neq i} |\calL_{\calG,ij}| \leq 2d.
$$
This completes the proof.
\end{proof}

\subsection{Differential-algebraic equations}
The dynamics of $\rlc$ circuits are described by \emph{differetial-algebraic equations} ($\dae$) which also arise for other constrained dynamical systems including chemical
reaction networks and mechanical systems under control. A $\dae$ is an implicit system of the form
\begin{equation}\label{eq:dae-vgeneral}
    F(t, \vec{x}(t), \dot{\vec{x}}(t)) = 0,
\end{equation}
where time $t \in (0,T]$, $\vec{x}(t) \in \mathbb{R}^N, \forall t$ is the unknown state trajectory and
$F : (0,T] \times \mathbb{R}^N \times \mathbb{R}^N \to \mathbb{R}^M$ is a sufficiently
smooth mapping. 
Unlike ordinary differential equations (ODEs), the relation of Eq.~\eqref{eq:dae-general} may contain both differential constraints (involving $\dot{\vec{x}}$) and algebraic constraints (involving only $\vec{x}$ and $t$). This implicit structure also imposes compatibility conditions on admissible initial values. A DAE is said to be regular on a domain if the Jacobian $\partial F / \partial \dot{x}$ has constant rank, ensuring the existence of a smooth solution manifold. We provide a brief background on $\dae$s below and refer the reader to standard texts~\cite{brenan1995numerical,wanner1996solving,kunkel2006differential} for further details.

\subsubsection{Linear time-independent $\dae$}
In this work, we will consider linear time-independent $\dae$s taking the following form
\begin{equation}\label{eq:dae-linear}
\mathbf{M} \dot{\vec{x}}(t) + \mathbf{K}\vec{x}(t) = \vec{f},
\end{equation}
where $\mathbf{M},\mathbf{K} \in \mathbb{R}^{N \times N}$ are constant matrices, $\vec{x} \in \mathbb{R}^{N}$ is the unknown time-evolving state description and $\vec{f} \in \mathbb{R}^{N}$ is a known constant forcing. If $\mathbf{M}$ is non-singular, Eq.~\eqref{eq:dae-linear} reduces to solving the ODE $\dot{\vec{x}} = -\mathbf{M}^{-1}\mathbf{K} \vec{x} - \mathbf{M}^{-1}\vec{f}$. When $\mathbf{M}$ is singular, Eq.~\eqref{eq:dae-linear} contains algebraic constraints
and its solution set is restricted to those states $\vec{x}(t)$ satisfying
\begin{equation}
    - \mathbf{K} \vec{x}(t) + \vec{f}(t) \in \operatorname{range}(\mathbf{M}).
\end{equation}
Well-posedness of Eq.~\eqref{eq:dae-linear} typically requires the \emph{regularity} of the matrix pencil $\lambda \mathbf{M} - \mathbf{K}$, meaning $\det(\lambda \mathbf{M} - \mathbf{K})$ is not identically zero as a polynomial in~$\lambda$~\cite{kunkel2006differential}.

\subsubsection{Index of DAE}\label{subsec:prelims_index_DAE}

We now introduce the notion of the \emph{index} of a $\dae$, which characterizes the hardness of solving a $\dae$. While there exist many different notions such as the differentiation index, Kronecker index and strangeness (see~\cite{schwarz2024common}), we will focus on the \emph{differentiation index} and \emph{tractability index} here. 

\paragraph{Differentiation index.}

The differentiation index quantifies how many times the hidden algebraic constraints of a $\dae$ must be differentiated in time before an explicit ODE representation emerges. Formally, for the implicit system Eq.~\eqref{eq:dae-general}, the index is the smallest integer $\mu$ for which repeated differentiation of Eq.~\eqref{eq:dae-general} up to order $\mu$ allows all $\dot{x}$ components to be solved explicitly as functions of
$(t,x)$ on a neighborhood of the solution manifold. This is formally stated below.
\begin{definition}[Differentiation index.]
The $\dae$ $F(t, \vec{x}(t), \dot{\vec{x}}(t)) = 0$ has differentiation index $\mu_D$ if $\mu_D$ is the minimal number of differentiations required to define a system of equations
$$
\begin{bmatrix}
    F(t, \vec{x}(t), \dot{\vec{x}}(t)) \\
    \frac{d}{dt} F(t, \vec{x}(t), \dot{\vec{x}}(t)) \\
    \vdots \\
    \frac{d^\mu}{dt^\mu} F(t, \vec{x}(t), \dot{\vec{x}}(t))
\end{bmatrix} = 0
$$
that allow for extraction of an explicit ordinary differential equation of form $\dot{\vec{x}}(t) = g(\vec{x}(t),t)$ using only algebraic manipulations.
\end{definition}

One limitation in using the differential index is that it explicitly considers the original state variables (or coordinates) and differentiates against them to go from the implicit system $F(\cdot)$ to an $\ode$.

\paragraph{Tractability index.} 

The tractability index counts how many steps of constraint reduction are needed to obtain an $\ode$ system on the constraint manifold. To define these constraint reductions and thereby tractability index, we will specialize to linear time-independent $\dae$s (Eq.~\eqref{eq:dae-linear}). For a more general description, we refer the reader to \cite{lamour2013differential}.

Consider the linear $\dae$ of Eq.~\eqref{eq:dae-linear}
$$
\rlcM \dot{\vec{x}}(t) + \rlcK \vec{x}(t) = \vec{f}(t).
$$
Let us define $\rlcM_0 = \rlcM$ and $\rlcK_0 = \rlcK$. We then define the following sequence of matrices recursively starting from $i=0$ and then in order $i=1,2,\ldots$
\begin{align}\label{eq:hierarchy_matrix_pencil}
\rlcM_{i+1} = \rlcM_i + \rlcK_i \rlcQ_i, \quad \rlcK_{i+1} = \rlcK_i \rlcP_i    
\end{align}
where $\rlcQ_i$ is a projector onto the kernel of $\rlcM_i$, which we denote by $\mathbf{N}_i = \ker \rlcM_i$, and $\rlcP_i$ is set as $\rlcP_i = \mathbf{I} - \rlcQ_i$. Each of these projection steps correspond to a step of constraint reduction. We also need to ensure admissibility~\cite{lamour2013differential} i.e., $\rlcQ_i \rlcQ_{i-1} = 0$ to ensure the decoupling behavior. We stop at $\rlcM_k$ when it is non-singular and $k$ is then defined to be the tractability index of the $\dae$ at hand. At this stage, the reduced equation can be solved for $\dot{\vec{x}}$, and an $\ode$ is obtained on the underlying constraint manifold. 

\emph{Remark.} For constant-coefficient linear DAEs, the differentiation index and the tractability index are equivalent~\cite{riaza2008differential}. In other words, the number of differentiations required to solve for $\dot{\vec{x}}$ is the same as the number of projection-based steps required to obtain a regularized system. In nonlinear time-dependent DAEs, the differentiation index and tractability index may be different. Moreover, the differentiation index requires that the $\vec{f}(t)$ term be (continuously) differentiable.

\subsection{Modified Nodal Analysis for RLC circuits}
We will model $\rlc$ circuits using modified nodal analysis ($\mna$) which was introduced by Ho et al.~\cite{ho1975mna}. Notably, it is utilized in the popular classical simulators of $\textsf{SPICE}$~\cite{nagel1973spice} and $\textsf{TITAN}$.

We denote the state description of the $\rlc$ circuit as $\vec{x}(t) = (\vec{u}(t), \vec{i}_{\rlcind}(t), \vec{i}_{\rlcvs}(t))^T$. The governing equations of the dynamics of $\rlc$ circuits corresponding to $\mna$ are then
\begin{align} 
    \label{eq:mna_main}
    \underbrace{
    \begin{bmatrix}
        \rlcA_{\rlccap} \rlcC \rlcA_{\rlccap}^T & 0 & 0 \\
        0 & \rlcL & 0 \\
        0 & 0 & 0 
    \end{bmatrix}}_{:=\rlcM}
    \frac{d}{dt} \begin{bmatrix}
        \vec{u}(t) \\ \vec{i}_\rlcind(t) \\ \vec{i}_{\rlcvs}(t)
    \end{bmatrix} + 
    \underbrace{
    \begin{bmatrix}
        \rlcA_\rlcres \rlcG \rlcA_\rlcres^T & \rlcA_{\rlcind} & \rlcA_{\rlcvs} \\
        -\rlcA_{\rlcind}^T & 0 & 0 \\
        -\rlcA_{\rlcvs}^T & 0 & 0 
    \end{bmatrix}}_{:=\rlcK}
    \begin{bmatrix}
        \vec{u}(t) \\ \vec{i}_{\rlcind}(t) \\ \vec{i}_{\rlcvs}(t)
    \end{bmatrix} = 
    \underbrace{
    \begin{bmatrix}
        - \rlcA_{\rlcjs} \vec{i}_{\rlcjs} \\ 0 \\ -\vec{v}_{\rlcvs}
    \end{bmatrix}}_{:=\vec{f}}
\end{align}
where we denote the matrices $\rlcM \in \mathbb{R}^{}, \rlcK \in \mathbb{R}^{}$ as indicated along with the vector $\vec{f}$ which contains information regarding the independent $\dc$ sources, $\rlcG = \rlcR^{-1}$ is the conductance matrix and with initial conditions of $\vec{x}(0)$. The $\mna$ equations are then simply $\rlcM \dot{\vec{x}} + \rlcK \vec{x} = \vec{f}$. 
\footnote{The form of $\rlcK$ as written in Eq.~\eqref{eq:mna} is convenient as the Hermitian part of $\rlcK$ will then just involve the resistive sub-matrix.}

It has been shown that the $\mna$ equations have a maximum tractability index of $2$. This is formally stated below.
\begin{theorem}[Theorem~5.2,\cite{riaza2008differential}]\label{thm:index_mna}
Consider a well-posed connected circuit $\calG$ with a positive definite capacitance matrix $C$, positive definite conductance matrix $\rlcG$ ($=\rlcR^{-1}$), and positive definite inductance matrix $\rlcL$. Then, the $\mna$ system of Eq.~\eqref{eq:mna} has tractability index
\begin{enumerate}[$(i)$]
    \item zero if and only if there are no voltage sources and there exists a capacitive tree,
    \item one if and only if $\calG$ contains neither $VC$-loops (except for $C$-loops) nor $IL$-cutsets,
    \item two otherwise.
\end{enumerate}
\end{theorem}

\subsection{Inherent $\ode$ of $\dae$s}\label{subsec:ODE_of_DAE}
As alluded to in Section~\ref{sec:tech_overview} and in the definition of the tractability index, we can determine an explicit $\ode$ on a submanifold of the state space after steps of constraint reduction. This is made explicit below.

\paragraph{Index zero.} As we are promised that $\rlcM$ is non-singular, the induced explicit $\ode$ for index $0$ can be obtained directly over the entire state space:
$$
\dot{\vec{x}} = -\rlcM^{-1} \rlcK\vec{x} + \rlcM^{-1}\vec{f}
$$

\paragraph{Index one.}
We will now derive the induced explicit $\ode$ for index $1$. Towards that end, let $\rlcQ_0$ be the projector onto $\ker{\rlcM}$ and define $\rlcP_0 := \id - \rlcQ$. We have the following claim.
\begin{claim}\label{claim:P_Q_dae}
We have the following:
\begin{itemize}
    \item $\rlcM_1^{-1} \rlcM = \rlcP_0$
    \item $\rlcM_1^{-1} \rlcK \rlcQ_0 = \rlcQ_0$
\end{itemize}    
\end{claim}
\begin{proof}
We first note that $\rlcM_1 \rlcP_0 = (\rlcM + \rlcK\rlcQ_0) \rlcP_0 = \rlcM \rlcP_0 = \rlcM$. This then implies that $\rlcM_1^{-1} \rlcM = \rlcP_0$ proving the first item. The second result follows from $\rlcM_1 \rlcQ_0 = (\rlcM+\rlcK\rlcQ_0)\rlcQ_0 = \rlcK \rlcQ_0$.
\end{proof}
The induced explicit $\ode$ can then be obtained by projecting the $\dae$ on the constraint submanifold using the projector $\rlcP$, as described below.

\begin{theorem}\label{thm:dae_index1_ode}
Let the linear time-independent $\dae$ of $\mathbf{M} \dot{\vec{x}}(t) + \mathbf{K}\vec{x}(t) = \vec{f}$ have index $1$. Let $\vec{y} = \rlcP_0 \vec{x}$ and $\vec{z} = \rlcQ_0 \vec{x}$. The state $\vec{x}(t)$ can be determined across time $t \in [0,T]$ by solving
\begin{align*}
    \dot{\vec{y}} &= -\rlcP_0 \rlcM_1^{-1} \rlcK \vec{y} + \rlcP_0 \rlcM_1^{-1} \vec{f}, &\text{ with } \,\, \vec{y}(0) = \rlcP_0 \vec{x}(0) \\
    \vec{z} &= - \rlcQ_0 \rlcM_1^{-1} \rlcK \vec{y} + \rlcQ_0 \rlcM_1^{-1} \vec{f}, \,\, \forall t \in (0,T]  &\text{ and } \,\, \vec{z}(0) = \rlcQ_0 \vec{x}(0).
\end{align*}
\end{theorem}
\begin{proof}
Multiplying the given $\dae$ by $\rlcM_1^{-1}$ on both sides, we obtain
\begin{align}
\rlcM_1^{-1} \rlcM \dot{\vec{x}} + \rlcM_1^{-1} \rlcK \vec{x} &= \rlcM_1^{-1} \vec{f} \\    
\rlcM_1^{-1} \rlcM \dot{\vec{x}} + \rlcM_1^{-1} \rlcK \rlcP_0 \vec{x} + \rlcM_1^{-1} \rlcK \rlcQ_0 \vec{x} &= \rlcM_1^{-1} \vec{f} \\    
\rlcP_0 \dot{\vec{x}} + \rlcM_1^{-1} \rlcK \rlcP_0 \vec{x} + \rlcQ_0 \vec{x} &= \rlcM_1^{-1} \vec{f} \\    
\rlcP_0 \dot{\vec{x}} + \rlcP_0 \rlcM_1^{-1} \rlcK \rlcP_0 \vec{x} + \rlcP_0 \rlcQ_0 \vec{x} &= \rlcP_0 \rlcM_1^{-1} \vec{f} \\    
\dot{\vec{y}} + \rlcP_0 \rlcM_1^{-1} \rlcK \vec{y} &= \rlcP_0 \rlcM_1^{-1} \vec{f},
\label{eq:diff_eq_Px}
\end{align}
where we decomposed $\vec{x} = \rlcP_0 \vec{x} + \rlcQ_0 \vec{x}$ in the second line, applied Claim~\ref{claim:P_Q_dae} in the third line, pre-multiplied on both sides by $\rlcP$ in the fourth line and denoted $\vec{y} := \rlcP_0 \vec{x}$ in the last line while noting that $\rlcP_0 \rlcQ_0=0$. The explicit induced $\ode$ for the differentiable variables $\vec{y}$ is then
\begin{equation}\label{eq:ODE-index1}
\dot{\vec{y}} = -\rlcP_0 \rlcM_1^{-1} \rlcK \vec{y} + \rlcP_0 \rlcM_1^{-1} \vec{f}.    
\end{equation}
We can also obtain an equation for $\vec{z} := \rlcQ_0 \vec{x}$ by pre-multiplying by $\rlcQ_0$ on the both sides in the third line of the above set of equations to obtain:
\begin{align}\label{eq:linear_system_Qx}
    \rlcQ_0 \rlcM_1^{-1} \rlcK \rlcP_0 \vec{x} + \rlcQ_0 \vec{x} = \rlcQ_0 \rlcM_1^{-1} \vec{f} \implies \vec{z} &= - \rlcQ_0 \rlcM_1^{-1} \rlcK \vec{y} + \rlcQ_0 \rlcM_1^{-1} \vec{f}.
\end{align}
This concludes the proof.
\end{proof}
Theorem~\ref{thm:dae_index1_ode} already suggests an approach for tackling $\dae$s of index $1$ and in particular, $\mna$ systems, by solving for the differential variables $\vec{y}$ and then the algebraic variables $\vec{z}$. This was described as part of Section~\ref{subsec:tech_approach}. This will be made formal in Section~\ref{sec:master_qa} and will constitute the main idea behind our quantum algorithm for solving $\dae$s.

\paragraph{Index two.} The induced explicit $\ode$ is obtained by projecting the $\dae$ on the constraint submanifold using the projector $\rlcP_0 \rlcP_1$. This is formally described below.
\begin{restatable}{theorem}{daeindextwo}\label{thm:dae_index2_ode}
Let the linear time-independent $\dae$ of $\mathbf{M} \dot{\vec{x}}(t) + \mathbf{K}\vec{x}(t) = \vec{f}$ have index $2$. Let $\vec{y} = \rlcP_0 \rlcP_1 \vec{x}$, $\vec{z}_1 = \rlcQ_0 \rlcP_1 \vec{x}$, and $\vec{z}_2 = \rlcQ_1 \vec{x}$. The state $\vec{x}(t)$ can be determined across time $t \in [0,T]$ by solving
\begin{align*}
    \dot{\vec{y}} &=  - \rlcP_0 \rlcP_1 \rlcM_2^{-1} \rlcK \vec{y} + \rlcP_0 \rlcP_1 \rlcM_2^{-1} \vec{f}, \\
    \vec{z}_2 &= - \rlcQ_1 \rlcM_2^{-1} \rlcK \vec{y} + \rlcQ_1 \rlcM_2^{-1} \vec{f}, \\ 
    \vec{z}_1 &= \rlcQ_0 \dot{\vec{z}}_2 - \rlcQ_0 \rlcP_1 \rlcM_2^{-1} \rlcK \vec{y} - \rlcQ_0 \vec{z}_2 + \rlcQ_0 \rlcP_1 \rlcM_2^{-1} \vec{f} \\
    &= \rlcQ_0 \rlcQ_1 \left(\rlcM_2^{-1}\rlcK_2\right)^2 \vec{y} + \rlcQ_0 (2 \rlcQ_1 - \id) \rlcM_2^{-1} \rlcK\vec{y} - \rlcQ_0 \rlcQ_1 \rlcM_2^{-1}\rlcK_2 \rlcM_2^{-1} \vec{f} + \rlcQ_0 (\id - 2\rlcQ_1) \rlcM_2^{-1} \vec{f}.
\end{align*}
\end{restatable}
The proof of the above theorem is provided in Appendix~\ref{appsec:proof_dae_index2}. Note that $\vec{z}_1$ depends on the time-derivative of $\vec{z}_2$ and therefore indirectly on the time-derivative of $\vec{y}$. This is reduced to an expression on just $\vec{y}$ using the $\ode$ governing the evolution of $\dot{\vec{y}}$. 

\subsection{Quantum routines}
\subsubsection{Block-encodings}
The usual definition of block-encoding of square matrices is as follows.
\begin{definition}[\cite{gilyen2019qsvt}]\label{def:BE}
Given a matrix $\rlcA \in \mathbb{C}^{N \times N}$ with $N=2^n$, if we can find $\alpha \in \mathbb{R}^+$, and an $(m+n)$-qubit unitary matrix $U_{\rlcA}$ such that
$$
\norm{A - \alpha (\la 0^m | \otimes \id_n) U_{\rlcA} (\ket{0^m} \otimes \id_n)}_2 \leq \varepsilon,
$$
then $U_{\rlcA}$ is called an $(\alpha,m,\varepsilon)$-block-encoding of $\rlcA$. In particular, when the block encoding is exact with $\varepsilon=0$, $U_{\rlcA}$ is called an $(\alpha,m)$-block-encoding of $\rlcA$.
\end{definition}

Applying an approximate block-encoding $U_A$ of a linear operator $\mathbf{A}$ has the following effect on an initial state.
\begin{claim}\label{claim:state_approx_BE}
Let $\alpha_A \geq 1, a_A \in \mathbb{N}, \varepsilon_A \in (0,1)$ and $0 < \sigma_\min \leq \sigma_\max$. Let $\mathbf{A}$ be an $n$-qubit operator with a $(\alpha_A,a_A,\varepsilon_A)$-block-encoding $U_A$ and let its singular values lie in $[\sigma_\min,\sigma_\max]$. Suppose $\ket{\phi}$ is an arbitrary $n$-qubit state. If $\varepsilon_A \leq \sigma_\min \min(\varepsilon,1/\alpha_A)/2$, then $U_A$ can be used to prepare a state that is $\varepsilon$-close to $\rlcA \ket{\phi}/\norm{\rlcA \ket{\phi}}$ with success probability $\geq \sigma_\min^2/(4 \alpha_A^2)$.
\end{claim}
\begin{proof}
Let $\widetilde{A} = (\bra{0}^a \otimes \id) U_A (\ket{0}^a \otimes \id)$. Starting from the definition of $U_A$, we have that
\begin{equation}\label{eq:interim_action_UA}
\norm{\rlcA - \alpha_A \widetilde{A}} \leq \varepsilon_A \implies \norm{\frac{\rlcA}{\alpha_A} \ket{\phi} - \widetilde{A} \ket{\phi}} \leq \frac{\varepsilon_A}{\alpha_A}.    
\end{equation}
Applying $U_A$ to $\ket{\phi}$ produces the state $\ket{0}^a \widetilde{A} \ket{\phi} + \ket{\perp}$ where $\ket{\perp}$ is some unnormalized state that is orthogonal to every state with $\ket{0}^a$ in the first register. After measuring and post-selecting for $\ket{0}^a$ on the first register, we obtain the state $\ket{\psi} = \widetilde{A} \ket{\phi}/\norm{\widetilde{A} \ket{\phi}}$. Using Fact~\ref{fact:normalized_vector_error}, we find that this state satisfies
\begin{align*}
\norm{\ket{\psi} - \frac{\mathbf{A} \ket{\phi}}{\norm{\rlcA \ket{\phi}}}} = \norm{\frac{\widetilde{A} \ket{\phi}}{\norm{\widetilde{A} \ket{\phi}}} - \frac{(\mathbf{A}/\alpha_A) \ket{\phi}}{\norm{(\rlcA/\alpha_A) \ket{\phi}}}} 
& \leq \frac{2}{\norm{(\rlcA/\alpha_A) \ket{\phi}}} \norm{\widetilde{A} \ket{\phi} - \frac{\mathbf{A}}{\alpha_A} \ket{\phi} } \\
& \leq \frac{2 \alpha_A}{\sigma_\min} \cdot \frac{\varepsilon_A}{\alpha_A} \\
&= 2\frac{\varepsilon_A}{\sigma_\min},
\end{align*}
where we used Eq.~\eqref{eq:interim_action_UA} in the second line and $\norm{\rlcA \ket{\phi}} \geq \sigma_\min$. Setting $\varepsilon_A \leq \varepsilon \sigma_\min/2$ ensures the desired accuracy. The success probability of obtaining $\ket{\psi}$ is
$$
p_{\mathrm{succ}} = \norm{\widetilde{A}\ket{\phi}}_2^2 \geq \left(\norm{\frac{\rlcA}{\alpha_A} \ket{\phi}} - \norm{\frac{\rlcA}{\alpha_A} \ket{\phi} - \widetilde{A}\ket{\phi}} \right)^2 \geq \left(\frac{\sigma_\min}{\alpha_A} - \frac{\varepsilon_A}{\alpha_A} \right)^2,
$$
where we used the reverse triangle inequality in the second inequality and used Eq.~\eqref{eq:interim_action_UA} in the final inequality. If we set $\varepsilon_A \leq \sigma_\min/(2\alpha_A)$, we get that the success probability satisfies 
$$
p_{\mathrm{succ}} \geq \sigma_\min^2/(4 \alpha_A^2)
$$
Combining the above two conditions on $\varepsilon_A$ gives us the desired result.
\end{proof}

\paragraph{Sparse matrices.} We will rely on the following result~\cite[Lemma~48]{gilyen2019qsvt} for constructing block-encodings of sparse matrices.
\begin{lemma}[{\cite[Lemma~48]{gilyen2019qsvt}}]\label{lem:sparse-block-enc}
Let $\rlcA \in \mathbb{C}^{2^w \times 2^w}$ be a matrix that is $s_r$-row-sparse and $s_c$-column-sparse, and each element of $A$ has absolute value at most $1$. Suppose that we have access to the following sparse-access oracles acting on two $(w + 1)$ qubit registers
\begin{align*}
&O_r : \ket{i}\ket{k} \rightarrow \ket{i} \ket{r_{ik}} \enspace \enspace \forall i \in [2^w]-1, k \in [s_r], \\
&O_c : \ket{\ell}\ket{j} \rightarrow \ket{c_{\ell j}} \ket{j} \enspace \enspace \forall \ell \in [s_c], j \in [2^w]-1,    
\end{align*}
where $r_{ij}$ is the index for the $j$-th non-zero entry of the $i$-th row of $\rlcA$, or if there are less than $i$ non-zero entries, then it is $j + 2w$, and similarly $c_{ij}$ is the index for the $i$-th non-zero entry of the $j$-th column of $\rlcA$, or if there are less than $j$ non-zero entries, then it is $i+ 2w$. Additionally assume that we have access to an oracle $O_{\rlcA}$ that returns the entries of $\rlcA$ in a binary description
\begin{align*}
O_{\rlcA} : \ket{i}\ket{j}\ket{0}^{\otimes b} \rightarrow \ket{i}\ket{j}\ket{a_{ij}} \enspace \forall i,j \in [2^w] -1,
\end{align*}
where $a_{ij}$ is a b-bit binary description of $\rlcA_{ij}$. Then, we can implement a $(\sqrt{s_r s_c}, w+3, \varepsilon)$-block-encoding of $\rlcA$ with a single use of $O_r$, $O_c$, two uses of $O_{\rlcA}$ and additionally using $O(w + \log^{2.5}(s_r s_c /\varepsilon))$ one and two qubit gates while using $O(b,\log^{2.5}(s_r s_c/\varepsilon))$ ancilla qubits.
\end{lemma}

We now state a few relevant results for manipulating block-encodings such as performing arithematic operations and obtaining the inverse of matrices. 
\paragraph{Arithematic operations.} We will require the ability to perform matrix arithematic operations on block-encodings. Towards that, we state a collection of relevant results. Below, we define $\overline{\rlcA} \in \mathbb{C}^{(M+N) \times (M+N)}$ as the Hermitian dilation of $\rlcA \in \mathbb{C}^{M \times N}$:
\begin{equation}\label{eq:hermitian_dilation}
    \overline{\rlcA} = \begin{bmatrix}
        \mathbf{0} & \rlcA \\ \rlcA^\dagger & \mathbf{0}
    \end{bmatrix}.
\end{equation}

\begin{lemma}[{\cite[Lemma~17]{chakraborty2023quantum}}]\label{lem:sum_BEs}
Let $m \in \mathbb{N}$. For each $j \in [m]$, let $\rlcA_j$ be an $s$-qubit operator and $y_j \in \mathbb{R}^{+}$. Let $U_j$ be a $(\alpha_j,a_j,\varepsilon_j)$-block-encoding of $\rlcA_j$, implemented in time $T_j$. Define the matrix $\rlcA = \sum_{j=1}^m y_j \rlcA_j$ and the vector $\eta \in \mathbb{R}^m$ s.t. $\eta_j = y_j \alpha_j$. Let $U_\eta$ be a state-preparation unitary of the state $\gamma^{-1} \sum_j \sqrt{\eta_j}\ket{j}$ where $\gamma = \sum_j \eta_j$, implemented in time $T_\eta$. Then, we can implement a
$$
\left(\sum_j y_j \alpha_j, \max_j(a_j) + s, \sum_j y_j \varepsilon_j \right)
$$
\end{lemma}
block-encoding of $\rlcA$ at a cost $O(T_\eta + \sum_j T_j)$.

\begin{lemma}[{\cite[Lemma~54]{gilyen2019qsvt}},{\cite[Lemma 4]{chakraborty2019power}}]\label{lem:prod_BEs}
If $U_{\rlcA}$ is an $(\alpha,a,\varepsilon)$ block-encoding of $\rlcA$ with gate
complexity $T_A$ and $U_B$ is a $(\beta,b,\delta)$ block-encoding of $\mathbf{B}$ with gate complexity $T_B$, then
\begin{enumerate}[(i)]
    \item $\overline{\rlcA}$ has an $(\alpha,a+1,\varepsilon)$ block-encoding that can be implemented with gate complexity $O(T_{\rlcA})$.
    \item $(I_b \otimes U_A)(I_a \otimes U_B)$ is an $(\alpha \beta, a+b, \alpha \delta + \beta \varepsilon)$ block-encoding of $\rlcA \rlcB$ with gate complexity $O(T_A + T_B)$.
\end{enumerate}
\end{lemma}
We will also require the ability to take tensor products of matrices.
\begin{lemma}[{\cite[Lemma~21]{chakraborty2023quantum}}]\label{lem:tensor_prod_BEs}
Let $U_A$ and $U_B$ be $(\alpha,a,\varepsilon_1)$ and $(\beta,b,\varepsilon_2)$-block-encodings of $A$ and $B$, $s$ and $t$-qubit operators, implemented in time $T_A$ and $T_B$ respectively. Define $S:= \prod_{i=1}^s \SWAP_{a+b+i}^{a+i}$. Then, $S(U_A \otimes U_B)S^\dagger$ is an $(\alpha \beta, a+b, \alpha \varepsilon_2 + \beta \varepsilon_1 + \varepsilon_1 \varepsilon_2)$-block-encoding of $A \otimes B$, implemented at a cost of $O(T_A + T_B)$.
\end{lemma}
\paragraph{Matrix pseudoinverse.} Lastly, we will require the ability to implement a block-encoding of the pseudoinverse $\rlcA^{+}$ of $\rlcA$. Among the results available~(e.g., \cite{harrow2009linear,gilyen2019qsvt,chakraborty2019power,chakraborty2023quantum}), we will use the following result from \cite[Theorem~26]{chakraborty2023quantum} which allows us to implement $\rlcA^{+}$ given an approximate block-encoding of $\rlcA$ and not an exact block-encoding.
\begin{lemma}\label{lem:pseudoinverse_BE}
Let $\varepsilon \in (0,1]$. Let $\rlcA$ be a normalized matrix with non-zero singular values in the range $[1/\kappa_A,1]$ for some $\kappa_A \geq 1$. For some $\varepsilon_1 = o(\varepsilon/(\kappa_A^2\log(\kappa_A/\varepsilon)))$ and $\alpha \geq 2$, let $U_{\rlcA}$ be an $(\alpha,a,\varepsilon_1)$-block-encoding of $\rlcA$ implemented in time $T_A$. Then, we can implement a $(2 \kappa_A,a+1,\varepsilon)$-block-encoding of $\rlcA^{+}$ at a cost of
$$
O\Big(\kappa_A \alpha T_A \log(\kappa_A/\varepsilon)\Big).
$$
\end{lemma}

\paragraph{Implementing sequences of operators.} Often, we will need to implement a sequence of operators on a given quantum state. We can do this efficiently given efficient block-encodings for each operator. This is described formally below.
\begin{claim}
\label{claim:prod_list_BEs}
Let $\{\rlcA_i\}_{i \in [m]}$ be a list of $m$ linear operators in $\mathbb{C}^{N \times N}$. Suppose each $\rlcA_i$ has an $(\alpha_i,a_i,\varepsilon_i)$-block-encoding $U_{A_i}$ implementable in time $T_i$. Then, $\mathbf{B} = \rlcA_m \rlcA_{m-1} \ldots \rlcA_1$ has a $(\widetilde{\alpha}, \widetilde{a}, \widetilde{\varepsilon})$-block-encoding $U_B$ where 
$$\widetilde{\alpha} = \prod_{i=1}^m \alpha_i, \enspace \widetilde{a} = \sum_{i=1}^m a_i, \enspace \widetilde{\varepsilon} = \sum_{i=1}^m \varepsilon_i (\widetilde{\alpha}/\alpha_i),$$ 
and is implementable in time $O(\sum_{i=1}^m T_i)$
\end{claim}
\begin{proof}
We will argue that $U_B := U_{A_m} U_{A_{m-1}} \ldots U_{A_1}$ is the desired block-encoding of $\mathbf{B}$. Towards that, let $V_j = (\id^{a_j} \otimes U_{A_j})(\id^{a_{j-1}} \otimes U_{A_{j-1}}) \ldots (\id^{a_1} \otimes U_{A_1})$ be an $(\widetilde{\alpha}_j, \widetilde{a}_j, \widetilde{\varepsilon}_j)-$block-encoding for $\mathbf{A}_j \ldots \mathbf{A}_1$ with parameters $\widetilde{\alpha}_j > 0$, $\widetilde{a}_j \in \mathbb{N}$ and $\widetilde{\varepsilon}_j \in (0,1)$ to be determined later. We note that $V_{j+1} = (\id^{a_j+1} \otimes U_{A_{j+1}})V_j$ is a $(\widetilde{\alpha}_{j+1}, \widetilde{a}_{j+1}, \widetilde{\varepsilon}_{j+1})$-block-encoding of $U_{A_{j+1}}$ with
\begin{equation}\label{eq:interim_recurrence}
\widetilde{\alpha}_{j+1} = \alpha_{j+1} \widetilde{\alpha}_j, \quad \widetilde{a}_{j+1} = a_{j+1} + \widetilde{a}_j, \quad \widetilde{\varepsilon}_{j+1} = \alpha_{j+1} \widetilde{\varepsilon}_j + \widetilde{\alpha}_j \varepsilon_{j+1}
\end{equation}
using Claim~\ref{lem:prod_BEs}. Unrolling the recursion, we obtain $\widetilde{\alpha}_j = \prod_{i=1}^j \alpha_i$, $\widetilde{a}_j = \prod_{i=1}^j a_i$ and
$$
\widetilde{\varepsilon}_j = \sum_{i=1}^j \varepsilon_i \prod_{k \in [j], k \neq i} \alpha_k = \sum_{i=1}^j \varepsilon_i (\widetilde{\alpha}_j/\alpha_i).
$$
Considering the block-encoding parameters for $j=m$ gives us the desired result.
\end{proof}

To boost the success probability of preparing states from application of block-encodings, we will often use amplitude amplification~\cite{brassard2002quantum,gilyen2019qsvt}.
\begin{theorem}\label{thm:amplitude_amplification}
Let $p,a,a' > 0$. There is a quantum algorithm $\calQ$ that satisfies the following: given access to a unitary $U$ such that $U \ket{0^n} = \ket{\psi}$ where $\ket{\psi} = \sqrt{p} \ket{\psi_0} + \sqrt{1-p}\ket{\psi_1}$ for an unknown $p > a$ (with known $a$) and $\ket{\psi_0},\ket{\psi_1}$ are orthogonal quantum states, $\calQ$ makes $O(\sqrt{a'/a})$ queries to $U,U^\dagger$ and outputs $\ket{\psi_0}$ with probability $a' > 0$.
\end{theorem}

Finally, we will often require a block-encoding for $\ket{0}^m\bra{0}^m$ for some $m \in \mathbb{N}$. This can be obtained from the following claim.
\begin{claim}\label{claim:proj_0state_BE}
Let $m \in \mathbb{N}$. There is $(1,1,0)$-block-encoding for $\ket{0}^m\bra{0}^m$ implementable in gate complexity $O(m)$ and using $2$ ancilla qubits.
\end{claim}
\begin{proof}
Let $R_0 = 2\ket{0^m}\bra{0^m} - \id$ be the reflection about $\ket{0^m}$. Adding this to $\id$ via Lemma~\ref{lem:sum_BEs} as $R_0/2 + \id/2$ gives us the desired block-encoding. The main contribution to the gate complexity is due to $R_0$ which can be implemented with $O(m)$ elementary gates if given access to an additional ancilla qubit.
\end{proof}

\subsubsection{Quantum $\ode$ solver}\label{subsec:prelims_ODE_solver}
We will also require a quantum algorithm for solving linear $\ode$s. We adapt the result from \cite{Krovi2023improvedquantum} which gave an algorithm assuming sparse block-encodings corresponding to the dynamics. We show the following result that allows us to solve a first order $\ode$ given any type of approximate block-encoding.

\begin{restatable}{theorem}{qodesolver}\label{thm:ODE_solver}
Let $\varepsilon \in (0,1)$. Consider the linear ordinary differential equation
$$
\frac{d\vec{x}}{dt} = \bm{\calA}\vec{x} + \vec{f}, \qquad \vec{x}(0) = \vec{x}_0,
$$
for $\bm{\calA} \in \mathbb{C}^{N \times N}$ and $\vec{f}, \vec{x}_0 \in \mathbb{C}^N$.
We are given an $(\alpha_{\calA}, a_{\calA}, \varepsilon_{\calA})$ block-encoding $U_{\calA}$ for $\bm{\calA}$ and unitaries $O_x, O_f$ that prepare states proportional to $\vec{x}_0$ and $\vec{f}$. Let the history state encoding the time-evolution of $\vec{x}(t)$ across the time-interval $[0,T]$ be
$$
\ket{\Psi} = \frac{1}{\normhist} \sum_{k=0}^{m} \norm{x (k \Delta t)} \ket{k, x(k\Delta t)},
$$
where $\Delta t = O(1/\|\bm{\calA}\|)$ and $m = \lceil T/\Delta t \rceil$. Then, there exists a quantum algorithm that outputs $\ket{\widehat{\Psi}}$ such that $\|\ket{\widehat{\Psi}} - \ket{\Psi}\| \leq \varepsilon$. Define
$$
\expnorm(\bm{\calA}) := \sup_{t \in [0,T]}\|\exp(\bm{\calA}t)\|, \quad \kappa := T \|\bm{\calA}\| \expnorm(\bm{\calA}), \quad \mu^2 = m^{-1} \sum_{j=1}^m \norm{\vec{x}(t)}^2, \quad  \mathcal{C}_f := \log\left(1 + \frac{T e^2 \|\vec{f}\|}{\mu}\right).
$$
The complexity is as follows:
\begin{itemize}
    \item Uses of $U_{\calA}$: $\widetilde{O}\left(\kappa^2 \cdot \poly(\mathcal{C}_f, \log(1/\varepsilon), \log(\kappa)) \right)$
    \item Uses of $O_x, O_f$: $\widetilde{O}\left( \kappa \, \calC_f \log(1/\varepsilon)\right)$
    \item Additional gates: $\widetilde{O}\left(\kappa^2 \cdot \poly(\mathcal{C}_f, \log(1/\varepsilon), \log(\kappa)) \right)$
\end{itemize}

It suffices to set the block-encoding precision as 
$\varepsilon_{\calA} = o\left(\varepsilon \cdot \kappa^{-3} \cdot \poly(\calC_f, \log(\kappa), \log(1/\varepsilon))^{-1} \right)$.
\end{restatable}
The proof of the above theorem is postponed to Appendix~\ref{appsec:qalgo_ode_la} as it follows along similar lines of \cite{Krovi2023improvedquantum} but with a slightly modified condition on the choice of $\calC_f$ to ensure the output of the algorithm is close to the desired history state. The reader is referred to Section~\ref{sec:discussion} for a discussion regarding the choice of the quantum $\ode$ solver and applicability of other solvers. The algorithm of Theorem~\ref{thm:ODE_solver} involves a quantum linear system solve and we will often require the corresponding block-encoding in our algorithms, which we expose in the following remark.
\begin{remark}\label{remark:BE_ODE_solver}
Let $\varepsilon \in (0,1)$. Consider the context of Theorem~\ref{thm:ODE_solver}. The algorithm of Theorem~\ref{thm:ODE_solver} involves solving a linear system $\calL \vec{x} = \Psi_0$ where $\vec{x} \in \mathbb{C}^{Nm}$, $\calL \in \mathbb{C}^{Nm \times Nm}$ and
\begin{equation}\label{def:ODE_psi0}
\Psi_0 = \left[\norm{\vec{x}_0} \ket{0,0,\vec{x}_0} + (\Delta t) \norm{\vec{f}} \sum_{j=0}^{m-1} \ket{j,1,f} \right],    
\end{equation}
where the first register corresponds to the time index $j$, second register corresponds to the Taylor index, and the third register is the system register holding the state description at time $j\Delta t$. If $\varepsilon_\calA = o(\varepsilon/(\kappa^2 \log(\kappa/\varepsilon) \poly(\calC_f)))$, then there exists an $(\alpha, a, \varepsilon)$-block-encoding $U_{\calL^{-1}}$ of $\calL$ with
$$
\alpha = 4 e \kappa, a = O(a_\calA + \log(T \norm{\bm{\calA}}).
$$
The gate complexity of $U_{\calL^{-1}}$ is $\widetilde{O}\left(\kappa \cdot \poly(\mathcal{C}_f \log(\kappa/\varepsilon)) \right)$ uses of $U_{\calA}$ and $\widetilde{O}\left(\kappa \cdot \poly(\mathcal{C}_f, \log(\kappa/\varepsilon)) \right)$ additional gates.
\end{remark}

\begin{restatable}{theorem}{qodesolverfinalstate}\label{thm:ODE_solver_final_state}
Let $\varepsilon \in (0,1)$. Consider the context of Theorem~\ref{thm:ODE_solver}. Then, there exists a quantum algorithm that outputs $\ket{\phi}$ such that $\left\| \ket{\phi} - \ket{x(T)} \right\| \leq \varepsilon$. Define
$$
\expnorm(\bm{\calA}) := \sup_{t \in [0,T]}\|\exp(\bm{\calA}t)\|, \quad \kappa := T \|\bm{\calA}\| \expnorm(\bm{\calA}), \quad  g := \frac{\max_{t\in[0,T]} \norm{\vec{x}(t)}}{\norm{\vec{x}(T)}}, \quad \mathcal{C}_f := \log\left(1 + \frac{T e^2 \|\vec{f}\|}{\norm{\vec{x}(T)}}\right).
$$
The complexity of the algorithm is:
\begin{itemize}
    \item Uses of $U_{\calA}$: $O\Big(g \kappa \cdot \poly\left(\log(1/\varepsilon), \log(\kappa), \mathcal{C}_f)\right) \Big)$
    \item Uses of $O_x, O_f$: $O\Big(g \kappa \cdot \log(1/\varepsilon)\Big)$,
    \item Additional gates: $O\Big(g \kappa \cdot \poly\left(\log(1/\varepsilon), \log(\kappa), \mathcal{C}_f)\right) \Big)$
\end{itemize}
It suffices to set the block-encoding precision as
$\varepsilon_\calA = o\left(\varepsilon \kappa^{-1} \cdot \poly\left(\log\left(1/\varepsilon\right), \log(\kappa), \calC_f\right)^{-1} \right)$.
\end{restatable}

The complexity of the quantum $\ode$ solvers depend on the exponential norm of the matrix. To bound this quantity, we will often use the following fact from \cite{bhatia2013matrix}.
\begin{fact}\label{fact:ub_exp_norm}
For any matrix $\mathbf{B}$, the norm of the exponential of $\mathbf{B}$ is upper bounded by the norm of the exponential of its Hermitian part i.e., 
$$
\norm{\exp(\mathbf{B})} \leq \norm{\exp( (\mathbf{B} + \mathbf{B}^\dagger)/2 )}.
$$
\end{fact}

\section{Quantum differential-algebraic equations solver}\label{sec:master_qa}
In this section, we propose our quantum algorithm for solving general linear time-independent differential-algebraic equations ($\dae$s) with specified tractability index $k$. In particular, we will consider the following problem setup and access for the algorithm.

\begin{figure}[!ht]
\begin{tcolorbox}
\begin{definition}\label{def:prob_setup_DAEs}
\textbf{Problem setup for simulating $\dae$s}\vspace{2mm}

\noindent\emph{
Let $\varepsilon \in (0,1)$, $N \in \mathbb{N}$, $k \in \{0,1,2\}$ and $T, f_\max > 0$. Suppose we are given a linear time-independent $\dae$ 
$$
\rlcM\dot{\vec{x}} + \rlcK\vec{x} = \vec{f} \enspace \text{with initial condition } \vec{x}(0) = \vec{x}_0,
$$
of tractability index $k$, where $\rlcM, \rlcK \in \mathbb{C}^{N \times N}$, and $\vec{x},\vec{f} \in \mathbb{C}^N$. We are given access to 
\vspace{1mm}
\begin{enumerate}[$(i)$]
    \item $(\alpha_M,a_M,\varepsilon_M)$-block-encoding $U_M$ of $\rlcM$ and $(\alpha_K,a_K,\varepsilon_K)$-block-encoding $U_K$ of $\rlcK$, \vspace{1mm}
    \item oracles $O_f$ and $O_x$ that prepare the forcing state $\ket{f} := \vec{f}/\norm{\vec{f}}$ and the initial state $\ket{x(0)} := \vec{x}(0)/\norm{\vec{x}(0)}$, respectively.
\end{enumerate} 
}
\end{definition}
\end{tcolorbox}
\end{figure}
Our goal is to prepare a state $\varepsilon$-close to the history state or the normalized solution at time $T$. We will introduce our quantum $\dae$ solver in Section~\ref{subsec:QDAE_solver_algo_analysis} considering problem setup of Definition~\ref{def:prob_setup_DAEs} along with concrete theorems for tractability index up to two. Before doing so, we first discuss the approach behind the solver in Section~\ref{subsec:approach} and then the main subroutines involved in Sections~\ref{subsec:useful_BEs}-\ref{subsec:initial_state_forcing}. We refer the reader to Section~\ref{subsec:tech_approach} for a high-level description and intuition behind the approach. We will later apply the quantum $\dae$ solver to simulate the $\mna$ equations (Eq.~\eqref{eq:mna}) of $\rlc$ circuits leading to Result~\ref{res:RLC_sim} in Section~\ref{sec:sim_RLC}. Note that in Definition~\ref{def:prob_setup_DAEs}, we assume that block-encodings of the matrices $\rlcM,\rlcK$ are provided to us but we will show these can be explicitly constructed for $\rlc$ circuits in Section~\ref{sec:algo_components_RLC} given appropriate classical oracle access.

\subsection{Approach}\label{subsec:approach}
We now discuss how the above high-level approach of Section~\ref{subsec:tech_approach} may be implemented on a quantum computer. Note that the $\dae$ of index $0$ is equivalent to solving an $\ode$ and thus, Theorem~\ref{thm:ODE_solver} in principle can be utilized. We thus discuss how to approach simulating $\dae$s of index $1$ and $2$ here. We will give the corresponding quantum algorithms and its theoretical guarantees in Section~\ref{subsec:QDAE_solver_algo_analysis}.

\subsubsection{Index $1$}\label{subsec:ind1_algo}
We will now discuss how to solve Theorem~\ref{thm:dae_index1_ode}. Let $U_1$ be the unitary that prepares the initial state $\ket{\psi_0}$ to the algorithm of Theorem~\ref{thm:ODE_solver}
\begin{equation}\label{def:ind1_psi0}
\ket{\psi_0} = U_1 \ket{0,0,0} = \frac{1}{\normhist_0} \left[\norm{\vec{x}_0} \ket{0,0,\vec{x}_0} + h \norm{\vec{f}} \sum_{j=0}^{m-1} \ket{j,1,f} \right],    
\end{equation}
where the first register corresponds to the time index $j$, second register corresponds to the Taylor index, and the third register is the system register holding the state description at time $jh$, $h := \Delta t$ is the time-step corresponding to the quantum $\ode$ solver, $m = \ceil{T/h}$ is the number of time-steps and $\normhist_0 = \sqrt{\norm{\vec{x}(0)}^2 + m h^2 \norm{\vec{f}}^2}$ is the normalization to ensure it is a valid quantum state. 

Additionally, consider the unitary $U_2$ that prepares the state $\ket{\phi_0}$ defined as
\begin{equation}\label{def:ind1_phi0}
\ket{\phi_0} = U_2 \ket{0,0,0} = \frac{1}{\normhist'_0} \left[\norm{\vec{x}_0} \ket{0,0,\vec{x}_0} + \norm{\vec{f}} \sum_{j=1}^{m} \ket{j,0,f} \right],    
\end{equation}
which is just a slight variant of the state $\ket{\psi_0}$ and notably holds the force vector across time indices $j\in [m]$ without a pre-factor of $h$. Moreover, since the state corresponding to the Taylor index remains in $\ket{0}$, the unitary $U_2$ can be prescribed to only include gates acting on the first and third registers. The normalization is $\normhist'_0 = \sqrt{\norm{\vec{x}(0)}^2 + m \norm{\vec{f}}^2}$.

For index $1$, we know from Theorem~\ref{thm:dae_index1_ode} that the $\ode$ solve over $\vec{y} := \rlcP_0 \vec{x}$ requires inputs corresponding to the initial state $\rlcP_0 \vec{x}_0$ and forcing $\rlcP_0 \rlcM_1^{-1} \vec{f}$. The corresponding input state, which we denote by $\ket{\psi_{y,0}}$ to be able to use Theorem~\ref{thm:ODE_solver} is then
\begin{equation}\label{def:ind1_psiy0}
\ket{\psi_{y,0}} = \frac{1}{\normhist_{y,0}} \left[\norm{\rlcP_0 \vec{x}_0} \ket{0,0,\rlcP_0 \vec{x}_0} + h \norm{\rlcP_0 \rlcM_1^{-1} \vec{f}} \sum_{j=0}^{m-1} \ket{j,1, \rlcP_0 \rlcM_1^{-1} f} \right],    
\end{equation}
where $\normhist_{y,0} = \sqrt{\norm{\rlcP_0 \vec{x}(0)}^2 + m h^2 \norm{\rlcP_0 \rlcM_1^{-1} \vec{f}}^2}$. We will presently show how $\ket{\psi_{y,0}}$ can be obtained from $\ket{\psi_0}$. Additionally, for the linear-algebraic equation solve for $\vec{z} := \rlcQ_0 \vec{x}$, we need to hold the relevant inputs corresponding to the initial state $\rlcQ_0 \vec{x}_0$ and forcing $\rlcQ_0 \rlcM_1^{-1} \vec{f}$. The corresponding input state, which we denote by $\ket{\phi_{z,0}}$, is then
\begin{equation}\label{def:ind1_phiz0}
\ket{\phi_{z,0}} = \frac{1}{\normhist_{z,0}} \left[\norm{\rlcQ_0 \vec{x}_0} \ket{0,0, \rlcQ_0 \vec{x}_0} + \norm{\rlcQ_0 \rlcM_1^{-1} \vec{f}} \sum_{j=1}^{m} \ket{j,0,\rlcQ_0 \rlcM_1^{-1} f} \right],    
\end{equation}
where $\normhist_{z,0} = \sqrt{\norm{\rlcQ_0 \vec{x}(0)}^2 + m \norm{\rlcQ_0 \rlcM_1^{-1} \vec{f}}^2}$.

\paragraph{Input for index $1$.} We now describe the input to the algorithm for $\dae$ of index $1$. We consider an additional register that tracks the index of the variables. For simplicity, we ignore the normalizations of the quantum states below:
\begin{align}\label{eq:ind1_input}
\ket{0,0,0,0}  \xrightarrow{H \otimes \id} & \ket{0} \ket{0,0,0} + \ket{1} \ket{0,0,0} \\
\xrightarrow{\ket{0}\bra{0} \otimes U_1 + \ket{1}\bra{1}\otimes U_2} & \ket{0}\ket{\psi_0} + \ket{1}\ket{\phi_0} \\
\xrightarrow{\ket{0}\bra{0} \otimes \Big(\id \otimes \ket{0}\bra{0} \otimes \rlcP_0 + \id \otimes \ket{1}\bra{1}\otimes \rlcP_0 \rlcM_1^{-1}\Big) + \ket{1}\bra{1} \otimes \id}   & \ket{0}\ket{\psi_{y,0}} + \ket{1}\ket{\phi_0} \\
\xrightarrow{\ket{1}\bra{1} \otimes \Big(\ket{0}\bra{0} \otimes \id \otimes \rlcQ_0 + \sum_{j=1}^m \ket{j}\bra{j} \otimes \id \otimes \rlcQ_0 \rlcM_1^{-1}\Big) + \ket{0}\bra{0} \otimes \id}   & \ket{0}\ket{\psi_{y,0}} + \ket{1}\ket{\phi_{z,0}}.
\end{align}

\paragraph{Solving the $\dae$.} We now give a rundown of the algorithm for index $1$ assuming exact implementation of the linear operators corresponding to the projectors and dynamics. As part of the analysis in the next section, we will track errors and sub-normalizations due to using their approximate block-encodings. For brevity, in the following, we will also ignore the normalizations of the quantum states as done above. 

\noindent \emph{Step 1: ODE solve. \,\,}
Let $\calL^{-1}$ be the linear operator corresponding to the linear system solve of Theorem~\ref{thm:ODE_solver} where $\bm{\calA}$ is now instantiated as $\bm{\calA} := \rlcP_0 \rlcM_1^{-1} \rlcK$. We will apply $\calL^{-1}$ (or rather its block-encoding) to solve for the differential variables $\vec{y} := \rlcP_0 \vec{x}$ as follows, staring from the input of Eq.~\eqref{eq:ind1_input}:
\begin{equation}\label{eq:ind1_ode_solve}
\ket{0}\ket{\psi_{y,0}} + \ket{1}\ket{\phi_{z,0}} \xrightarrow{\ket{0}\bra{0}\otimes \calL^{-1} + \ket{1}\bra{1}\otimes \id} \ket{0} \left[ \sum_{j=0}^m \ket{j,0,y(jh)} \right] + \ket{1} \ket{\phi_{z,0}},
\end{equation}
where we now have the desired history state over $\vec{y}(jh) \,\forall j\in [m]_0$ i.e., $\ket{\psihist^{(y)}}$ conditioned on the first register being $\ket{0}$.

\noindent \emph{Step 2: Linear-algebraic solve.\,\,} We now solve for $\vec{z} = \rlcQ_0 \vec{x}$, which recall satisfies $\vec{z} = -\rlcQ_0 \rlcM_1^{-1} \rlcK \vec{y} + \rlcQ_0 \rlcM_1^{-1} \vec{f}$ (Theorem~\ref{thm:dae_index1_ode}). As the goal is to finally obtain $\vec{x} = \vec{y} + \vec{z}$, we do this in two steps. We first combine the solution of $\vec{y}$ with $-\rlcQ_0 \rlcM_1^{-1} \rlcK \vec{y}$ to obtain an intermediate state as $\vec{w}= (\id -\rlcQ_0 \rlcM_1^{-1} \rlcK)\vec{y}$ and then finally combine $\vec{w}$ with $\rlcQ_0 \rlcM_1^{-1} \vec{f}$ to obtain $\vec{x} = \vec{w} + \rlcQ_0 \rlcM_1^{-1} \vec{f}$. Note that the initial condition is reset back as $\vec{x}(0) = \vec{y}(0) + \vec{z}(0)$. This is implemented as follows starting from the state obtained at end of Eq.~\eqref{eq:ind1_ode_solve}
\begin{align}\label{eq:ind1_LA_solve}
& \ket{0} \left[ \sum_{j=0}^m \ket{j,0,y(jh)} \right] + \ket{1} \ket{\phi_{z,0}} \\
\xrightarrow{\ket{0}\bra{0} \otimes \Big(\ket{0}\bra{0} \otimes \id + \sum_{j=1}^m \ket{j}\bra{j} \otimes (\id - \rlcQ_0 \rlcM_1^{-1} \rlcK) \Big) + \ket{1}\bra{1} \otimes \id}  & \ket{0} \left[ \ket{0,0,y(0)} + \sum_{j=1}^m \ket{j,0, w(jh)} \right] + \ket{1} \ket{\phi_{z,0}} \\
\xrightarrow{H \otimes \id}  & \ket{0} \left[\sum_{j=0}^m \ket{j,0, x(jh)} \right] + \ket{1} \ket{\mathrm{junk}},
\end{align}
where we noted that $\ket{\phi_{z,0}}$ contains the relevant initial condition and projected part of forcing to be added to the $w(jh) \, \forall j \in [m]$.

\subsubsection{Index $2$}\label{subsec:ind2_algo}
Recall from Theorem~\ref{thm:dae_index2_ode} that the equations of interest for $\vec{y} = \rlcP_0 \rlcP_1 \vec{x}$, $\vec{z}_1 = \rlcQ_0 \rlcP_1 \vec{x}$, and $\vec{z}_2 = \rlcQ_1 \vec{x}$ are given by
\begin{align*}
    \dot{\vec{y}} &=  - \rlcP_0 \rlcP_1 \rlcM_2^{-1} \rlcK \vec{y} + \rlcP_0 \rlcP_1 \rlcM_2^{-1} \vec{f}, \\
    \vec{z}_2 &= - \rlcQ_1 \rlcM_2^{-1} \rlcK \vec{y} + \rlcQ_1 \rlcM_2^{-1} \vec{f}, \\ 
    \vec{z}_1
    &= \rlcQ_0 \rlcQ_1 \left(\rlcM_2^{-1}\rlcK_2\right)^2 \vec{y} + \rlcQ_0 (2 \rlcQ_1 - \id) \rlcM_2^{-1} \rlcK\vec{y} - \rlcQ_0 \rlcQ_1 \rlcM_2^{-1}\rlcK_2 \rlcM_2^{-1} \vec{f} + \rlcQ_0 (\id - 2\rlcQ_1) \rlcM_2^{-1} \vec{f}.
\end{align*}
In our algorithm, the goal is to obtain an approximate solution to $\vec{x}(t), \forall t \in (0,T]$ and for this, we do not need to determine $\vec{z}_1$ and $\vec{z}_2$ on their own. We thus consider $\vec{z} = \vec{z}_1 + \vec{z}_2$ which is now given by the equation
\begin{align}\label{eq:ind2_LA_z}
\vec{z} &= \left[ \rlcQ_0 \rlcQ_1 (\rlcM_2^{-1} \rlcK_2)^2 + (\rlcQ_0 \rlcQ_1 - \rlcQ_0 \rlcP_1 - \rlcQ_1)\rlcM_2^{-1} \rlcK \right] \vec{y} \nonumber \\
&\qquad + \left[- \rlcQ_0 \rlcQ_1 \rlcM_2^{-1} \rlcK_2 + \rlcQ_0 \rlcP_1 - \rlcQ_0 \rlcQ_1 + \rlcQ_1 \right] \rlcM_2^{-1}\vec{f}.
\end{align}
Let us introduce the following notation to simplify the equations:
\begin{align}\label{eq:ind2_LA_z_simple}
\rlcG_2 &:=  \left[ \rlcQ_0 \rlcQ_1 (\rlcM_2^{-1} \rlcK_2)^2 + (\rlcQ_0 \rlcQ_1 - \rlcQ_0 \rlcP_1 - \rlcQ_1)\rlcM_2^{-1} \rlcK \right] \\
\mathbf{F}_2 &:= \left[- \rlcQ_0 \rlcQ_1 \rlcM_2^{-1} \rlcK_2 + \rlcQ_0 \rlcP_1 - \rlcQ_0 \rlcQ_1 + \rlcQ_1 \right] \rlcM_2^{-1}
\end{align}
so we can write the equations over $\vec{z}$ now as 
\begin{align}\label{eq:ind2_LA_z_simple_final}
\vec{z} = \rlcG_2 \vec{y} + \mathbf{F}_2 \vec{f}.
\end{align}
We can now proceed as we had in the case of index one. We start off with states $\ket{\psi_{y,0}}$ that encodes the inputs to the $\ode$ solve over $\vec{y}$ and $\ket{\phi_{z,0}}$ that encodes the initial condition $\vec{z}_0 = (\id - \rlcQ_1) \vec{x}$ and the projected part of the force $\mathbf{F}_2 \vec{f}$. We carry out an $\ode$ solve considering $\ket{\psi_{y,0}}$ as the input and then apply the corresponding linear-algebraic transformation (that here would correspond to $(\id + \rlcG_2)$) followed by a Hadamard test.

\subsection{Admissible projectors}
We now discuss how to construct admissible projectors (Section~\ref{subsec:prelims_index_DAE}) corresponding to $\dae$s with tractability index $2$.
\paragraph{Index one.} We set $\rlcP_0 = \rlcM^+ \rlcM$ to be the orthogonal projector onto $\ker(\rlcM)^\perp$. We then obtain $\rlcQ_0$ by setting $\rlcQ_0 = \id - \rlcP_0$ to be the orthogonal projector onto $\ker(\rlcM)$. Here, we obtain admissibility trivially.

\paragraph{Index two.} We construct $\rlcQ_1$ in the following way. This can be shown to be equivalent to the constructions discussed in \cite{lamour2013differential}.
\begin{claim}\label{claim:admissible_proj_ind2}
Let $\widetilde{\rlcQ}_1 = \id - \rlcM_1^+ \rlcM_1$ be the orthogonal projector onto $\ker(\rlcM_1)$. Define 
$$ 
\rlcQ_1 = \widetilde{\rlcQ}_1 (\rlcP_0 \widetilde{\rlcQ}_1)^+ \rlcP_0.
$$
The operator $\rlcQ_1$ is $(i)$ admissible i.e., $\rlcQ_1 \rlcQ_0 = 0$, $(ii)$ a valid projector i.e., $\rlcQ_1^2 = \rlcQ_1$, and $(ii)$ projects exactly onto $\ker(\rlcM_1)$ i.e., $\mathrm{Im}(\rlcQ_1) = \ker(\rlcM_1)$.
\end{claim}
\begin{proof}
Let $\rlcA = \rlcP_0 \widetilde{\rlcQ}_1$. We can then rewrite our proposed projector as $\rlcQ_1 = \widetilde{\rlcQ}_1 \rlcA^+ \rlcP_0$.

$(i)$ We directly obtain from evaluation that $\rlcQ_1 \rlcQ_0 = \tilde{\rlcQ}_1 \rlcA^+ \rlcP_0 \rlcQ_0 = 0$ since $\rlcP_0 = \id - \rlcQ_0$. 

$(ii)$ By directly evaluating, we obtain
$$
\rlcQ_1^2 = \left( \widetilde{\rlcQ}_1 \rlcA^+ \rlcP_0 \right) \left( \tilde{\rlcQ}_1 \rlcA^+ \rlcP_0 \right) = \widetilde{\rlcQ}_1 \rlcA^+ (\rlcP_0 \tilde{\rlcQ}_1) \rlcA^+ \rlcP_0 =  \widetilde{\rlcQ}_1 \rlcA^+ (\rlcA) \rlcA^+ \rlcP_0 = \widetilde{\rlcQ}_1 (\rlcA^+) \rlcP_0 = \rlcQ_1
$$
where we used that the Moore-Penrose pseudoinverse satisfies $\rlcA^+ \rlcA \rlcA^+ = \rlcA^+$. 

$(iii)$ Finally, we need to show that this projects exactly onto $\ker(\rlcM_1)$. Since the leftmost matrix of $\rlcQ_1$ is $\widetilde{\rlcQ}_1$, the output of $\rlcQ_1$ will always be constrained to $\text{im}(\widetilde{\rlcQ}_1)$. Therefore, $\text{im}(\rlcQ_1) \subseteq \ker(\rlcM_1)$. To prove it is exactly $\ker(\rlcM_1)$, we must prove that $\rlcQ_1$ acts as the identity operator for any vector already inside $\ker(\rlcM_1)$. Let $\vec{x} \in \ker(\rlcM_1)$. By definition of the orthogonal projector $\widetilde{\rlcQ}_1$, we have $\widetilde{\rlcQ}_1 \vec{x} = \vec{x}$. We note that
$\rlcQ_1 \vec{x} = \widetilde{\rlcQ}_1 \rlcA^+ \rlcP_0 \vec{x} \implies \rlcQ_1 \vec{x} = \widetilde{\rlcQ}_1 \rlcA^+ \rlcP_0 (\widetilde{\rlcQ}_1 \vec{x}) = \widetilde{\rlcQ}_1 \rlcA^+ \rlcA \vec{x}$. Now we must evaluate the effect of $\rlcA^+ \rlcA \vec{x}$. The matrix $\rlcA^+ \rlcA$ is the orthogonal projector onto $\text{im}(\rlcA^\dagger) = \text{im}(\widetilde{\rlcQ}_1 \rlcP_0)$.To understand how this operator behaves on $\vec{x}$, we rely on a fundamental structural property of regular index-2 DAEs: $\ker(\rlcM) \cap \ker(\rlcM_1) = \{0\}$. Since $\ker(\rlcP_0) = \ker(\rlcM)$ and $\text{im}(\widetilde{\rlcQ}_1) = \ker(\rlcM_1)$, this means no non-zero vector in $\text{im}(\widetilde{\rlcQ}_1)$ belongs to the null space of $\rlcP_0$. 
The operator $\rlcA^+ \rlcA$ acts as the identity matrix for any vector residing in $\text{im}(\widetilde{\rlcQ}_1)$. Since $\vec{x} \in \text{im}(\widetilde{Q}_1)$, it follows that: $\rlcA^+ \rlcA \vec{x} = \vec{x}$. Substituting this back into our equation gives us $\rlcQ_1 \vec{x} = \widetilde{\rlcQ}_1 \vec{x} = \vec{x}$. This completes the proof.
\end{proof}

\subsection{Useful block-encodings}\label{subsec:useful_BEs}
We assume access to an $(\alpha_M, a_M, \varepsilon_M)$ block-encoding $U_M$ of $\rlcM$ and an $(\alpha_K, a_K, \varepsilon_K)$ block-encoding $U_K$ of $\rlcK$. We state the costs here in terms of the relevant block-encoding parameters and also spectral properties of $\rlcM$ and $\rlcK$. 

\paragraph{Upper bounds on norms.}
\begin{claim}\label{claim:ub_norm_Mi_Ki}
Let $\uptau > 0$. Given a linear time-independent $\dae$ of Eq.~\eqref{eq:dae-linear} with index $k \in \{1,2\}$, we have the following bounds on the norms of the matrices $\{\rlcM_j\}_{j \in [k]}$
\begin{enumerate}[$(i)$]
    \item $\norm{\rlcM_1}_2 \leq \norm{\rlcM}_2 + \norm{\rlcK}_2$
    \item $\norm{\rlcM_2}_2 \leq \norm{\rlcM}_2 + \norm{\rlcK}_2 (1+1/\uptau)$ as long as $\sigma^+_\min(\rlcP_0 \widetilde{\rlcQ}_1) \geq \uptau$.
\end{enumerate}
\end{claim}
\begin{proof}
To prove $(i)-(ii)$, we repeatedly use the triangle inequality along with sub-multiplicativity of $\norm{\cdot}_2$ along with the fact that $2$-norms of orthogonal projectors is $1$ i.e., $\norm{\rlcP_0}_2 = \norm{\rlcQ_0}_2 = \norm{\widetilde{\rlcP}_1}_2 = \norm{\widetilde{\rlcQ}_1}_2 = 1, \forall i \in [k]_0$.

$(i)$ We note that
$$
\norm{\rlcM_1}_2 \leq \norm{\rlcM}_2 + \norm{\rlcK} \norm{\rlcQ_0}_2 = \norm{\rlcM}_2 + \norm{\rlcK}_2.
$$ 

$(ii)$ We note that $\norm{\rlcQ_1} \leq \norm{\widetilde{\rlcQ}_1} \norm{(\rlcP_0 \widetilde{\rlcQ}_1)^+} \norm{\rlcP_0} \leq 1/\uptau$. We then have
$$
\norm{\rlcM_2} \leq \norm{\rlcM_1} + \norm{\rlcK} \norm{\rlcP_0} \norm{\rlcQ_1} \leq \norm{\rlcM}_2 + \norm{\rlcK}_2 (1+1/\uptau).
$$
\end{proof}

\subsubsection{DAE of index 0}
We require a block-encoding of $\rlcM^{-1}$ to implement the linear dynamical operator $-\rlcM^{-1} \rlcK$ corresponding to index zero and to implement the state corresponding to the forcing $\rlcM^{-1} \ket{f}$.

\begin{claim}\label{claim:inv_M_BE_ind0}
Let $\varepsilon, \sigma \in (0,1)$. Let $\rlcM$ be as defined in Definition~\ref{def:prob_setup_DAEs} with a $(\alpha_M,a_M,\varepsilon_M)$-block-encoding $U_M$ where $\alpha_M \geq 2\norm{\rlcM}$ and that its minimum singular value satisfies $\sigma_\min(\rlcM) \geq \sigma$. Suppose that the given $\dae$ has index 0. Then, there exists a $(2/\sigma,a_M + 1,\varepsilon)$ block-encoding of $\rlcM^{-1}$ for $\varepsilon_M = o(\sigma^2 \varepsilon/\log(1/\sigma \varepsilon))$ that requires 
$O\left((\alpha_M/\sigma) \log(1/\sigma \varepsilon)\right)$ uses of $U_M$.
\end{claim}
\begin{proof}
This follows from applying Lemma~\ref{lem:pseudoinverse_BE}. Let $\rlcA := \rlcM/\alpha$ where $\alpha \geq \norm{\rlcM}$ and assume that $\alpha_M$ is set as $\alpha_M = 2 \alpha$. Then, $\rlcA$ is a normalized matrix with $\spec(\rlcA) \in [\sigma/\alpha,1]$. We also observe that $U_M$ is a $(2,a_M,\varepsilon_M/\alpha)$ block-encoding of $\rlcA$
\begin{align*}
&\norm{\rlcM - 2\alpha (\bra{0^{a_M}} \otimes \id)U_M(\ket{0^{a_M}} \otimes \id)} \leq \varepsilon_M \\
\implies & \norm{\rlcA - 2(\bra{0^{a_M}} \otimes \id)U_M(\ket{0^{a_M}} \otimes \id)} \leq \varepsilon_M/\alpha,
\end{align*}
where we divided by $\alpha$ on both sides in the second line. Here, the relevant condition number parameter is then $\kappa_A = \alpha/\sigma$. Applying Lemma~\ref{lem:pseudoinverse_BE} then produces a $(2\kappa_A,a_M+1,\delta)$-block-encoding of $\rlcA^{-1}$, which we denote by $V$, provided $\varepsilon_M/\alpha = o(\delta/(\kappa_A^2 \log(\kappa_A/\delta)))$. The cost of implementing $V$ is $O(\kappa_A \log(\kappa_A/\delta))$ uses of $U_M$. Since $\rlcA^{-1} = \alpha \rlcM^{-1}$, $V$ is then a $(2/\sigma, a_M+1, \delta/\alpha)$-block-encoding of $\rlcM^{-1}$. Setting the desired error $\delta/\alpha = \varepsilon$ gives us that we require $\varepsilon_M = o(\alpha^2 \varepsilon/(\kappa_A^2 \log(\kappa_A/(\alpha \varepsilon)))) = o(\sigma^2 \varepsilon/\log(1/(\sigma \varepsilon))$. The cost of implementing $V$ is then $O((\alpha_M/\sigma) \log(1/(\sigma \varepsilon)))$ uses of $U_M$. This completes the proof.
\end{proof}

We can now combine the above block-encoding of $\rlcM^{-1}$ with that of $\rlcK$ to obtain a block-encoding of $\rlcM^{-1} \rlcK$.
\begin{claim}\label{claim:dynamics_BE_ind0}
Let $\varepsilon, \sigma \in (0,1)$. Let $\rlcM$ and $\rlcK$ be matrices as defined in Definition~\ref{def:prob_setup_DAEs} with a $(\alpha_M,a_M,\varepsilon_M)$-block-encoding $U_M$ implementable in time $T_M$ and a $(\alpha_K,a_K,\varepsilon_K)$-block-encoding $U_K$ implementable in time $T_K$, respectively. Let $\alpha_M \geq 2\norm{\rlcM}$ and the minimum singular value of $\rlcM$ satisfy $\sigma_\min(\rlcM) \geq \sigma$. Suppose that the given $\dae$ has index 0. Then, there exists a $(2\alpha_K/\sigma,a_M + a_K + 1,\varepsilon)$ block-encoding of the $\rlcM^{-1} \rlcK$ for $\varepsilon_M = o(\sigma^2 \varepsilon/(2 \alpha_K \log(2 \alpha_K/(\sigma \varepsilon)))$ and $\varepsilon_K = \sigma \varepsilon/4$ that requires gate complexity
$$
O\left((\alpha_M/\sigma) \log(\alpha_K/(\sigma \varepsilon)) + T_K\right).
$$
\end{claim}
\begin{proof}
Let $\varepsilon' \in (0,1)$ be an error parameter to be fixed later. From Claim~\ref{claim:inv_M_BE_ind0}, we obtain a $(2/\sigma,a_M + 1,\epsilon')$ block-encoding $U_{M^{-1}}$ of $\rlcM^{-1}$ as long as $\varepsilon_M = o(\sigma^2 \varepsilon'/\log(1/(\sigma \varepsilon')))$. We now combine this with the $(\alpha_K,a_K,\varepsilon_K)$-block-encoding $U_K$ of $\rlcK$ using Lemma~\ref{lem:prod_BEs}. This gives us a $\Big(2\alpha_K/\sigma, a_M + a_K + 1, (2/\sigma) \varepsilon_K + \alpha_K \varepsilon'\Big)$-block-encoding of $\rlcM^{-1}\rlcK$. Setting $\varepsilon'= \varepsilon/(2\alpha_K)$ and $\varepsilon_K = \sigma \varepsilon/4$ gives us a block-encoding of $\rlcM^{-1} \rlcK$ with the desired error $\varepsilon$. This in turn requires setting $\varepsilon_M = o(\sigma^2 \varepsilon/(2 \alpha_K \log(2 \alpha_K/(\sigma \varepsilon)))$. The total gate complexity is due to application of $U_{M^{-1}}$ which requires $O( (\alpha_M/\sigma) \log(\alpha_K/(\sigma \varepsilon)))$ uses of $U_M$ and one use of $U_K$.
\end{proof}

\subsubsection{DAE of index one}
To solve the inherent $\ode$ for index one, we will need block-encodings of $-\rlcP_0 \rlcM_1^{-1} \rlcK$ governing the dynamics, $\rlcP_0$ to implement the initial state $\rlcP_0 \ket{x(0)}$ and $\rlcP_0 \rlcM_1^{-1}$ to implement the forcing state $\rlcP_0 \rlcM_1^{-1} \ket{f}$. We first give block-encodings of the projectors $\rlcP_0$ and $\rlcQ_0 := \id - \rlcP_0$.

\begin{claim}\label{claim:projs_BE_ind1}
Let $\varepsilon, \sigma \in (0,1)$. Let $\rlcM$ be as defined in Definition~\ref{def:prob_setup_DAEs} with a $(\alpha_M,a_M,\varepsilon_M)$-block-encoding $U_M$ implementable in time $T_M$ where $\alpha_M \geq 2 \norm{\rlcM}$ and with minimum non-zero singular value satisfying $\sigma_\min^+(\rlcM) \geq \sigma$. Suppose that the given $\dae$ has index 1. Let $\rlcQ_0$ be the projector onto the kernel of $\rlcM$ and let $\rlcP_0=\id-\rlcQ_0$. Then, there exists the following block-encodings for $\varepsilon_M = o(\sigma^3 \varepsilon^2/\alpha_M)$:
\begin{enumerate}[$(i)$]
    \item $(2\kappa_M, 2a_M + 1,\varepsilon)$-block-encoding of $\rlcP_0$ implementable in $O(T_M \kappa_M \log(\kappa_M/\varepsilon))$ time,
    \item $(2\kappa_M + 1, 2a_M + 2,\varepsilon)$-block-encoding of $\rlcQ_0$ implementable in $O(T_M \kappa_M \log(\kappa_M/\varepsilon))$ time,
\end{enumerate}
where $\kappa_M = \alpha_M/\sigma$ is an upper bound on the condition number of $\rlcM$.
\end{claim}
\begin{proof}
Let $\varepsilon' \in (0,1)$ be an error parameter to be fixed later. We first note that Lemma~\ref{lem:pseudoinverse_BE} can be used to implement the pseudoinverse $\rlcM^{+}$. Moreover, Claim~\ref{claim:inv_M_BE_ind0} can be utilized to obtain a block-encoding of $\rlcM^{+}$ as well, where $\sigma$ is now a lower bound for the minimum non-zero singular value of $\rlcM$. We thus obtain a $(2/\sigma,a_M + 1,\varepsilon')$-block-encoding of the $\rlcM^{+}$, which we denote by $U_{M^{+}}$, for $\varepsilon_M = o(\sigma^2 \varepsilon'/\log(1/(\sigma \varepsilon')))$. Noting that $1/\log(1/x) \geq x, \forall x > 0$, it is then sufficient to set $\varepsilon_M \leq O(\sigma^3 {\varepsilon'}^2)$.

We then obtain a block-encoding of $\rlcP_0 = \rlcM^+ \rlcM$ by using Lemma~\ref{lem:prod_BEs}. This produces a $(2\alpha_M/\sigma, 2a_M + 1, 2\varepsilon_M/\sigma + \alpha_M \varepsilon')$-block-encoding of $\rlcP_0$, which we denote by $U_{P_0}$. Setting $\varepsilon' = \varepsilon/(2\alpha_M)$ and $\varepsilon_M = \sigma \varepsilon/2$ gives us a block-encoding $U_{P_0}$ with the desired error $\varepsilon$. We thus need $\varepsilon_M = \min(\sigma \varepsilon/2, O(\sigma^3 \varepsilon^2/\alpha_M)) \leq O(\sigma^3 \varepsilon^2/\alpha_M)$. The main contribution to the gate complexity of $U_{P_0}$ is due to $U_{M^{+}}$ which is $O(T_M (\alpha_M/\sigma)\log(\alpha_M/(\sigma \varepsilon)))$ from Claim~\ref{claim:inv_M_BE_ind0}. This completes the proof of item $(i)$.

To obtain a block-encoding of $\rlcQ_0$, we combine the block encoding of $\rlcP_0$ with $\id$ as $\rlcQ_0 = \id - \rlcP_0$ using Lemma~\ref{lem:sum_BEs}. This gives us a $(2\alpha_M/\sigma + 1, 2a_M + 2,\varepsilon)$-block-encoding of $\rlcQ_0$ with gate complexity $O(T_M (\alpha_M/\sigma)\log(\alpha_M/(\sigma \varepsilon)))$. This completes the proof of item $(ii)$ and the overall claim.
\end{proof}

We now give block-encodings of $\rlcM_1 := \rlcM + \rlcK \rlcQ_0$ and its inverse (which is guaranteed to exist in the case of index $1$).
\begin{claim}\label{claim:M1_BE_ind1}
Let $\varepsilon, \sigma \in (0,1)$. Let $\rlcM$ be as defined in Definition~\ref{def:prob_setup_DAEs} with a $(\alpha_M,a_M,\varepsilon_M)$-block-encoding $U_M$ implementable in time $T_M$ where $\alpha_M \geq 2 \norm{\rlcM}$ and with minimum non-zero singular value satisfying $\sigma_\min^+(\rlcM) \geq \sigma$. Let $\rlcK$ have a $(\alpha_K,a_K,\varepsilon_K)$-block-encoding $U_K$ implementable in time $T_K$. Suppose that the given $\dae$ has index 1. Let $\rlcQ_0$ be the projector onto the kernel of $\rlcM$ and let $\rlcP_0=\id-\rlcQ_0$. Denote $\kappa_M := \alpha_M/\sigma$.

Then, for $\varepsilon_M = o(\varepsilon^2/\kappa_M^3)$ and $\varepsilon_K = \varepsilon/(6\kappa_M + 3)$, there exists a
$$
\Big(2\kappa_M \alpha_K + \alpha_M + \alpha_K, 2a_M + a_K + 3, \varepsilon\Big)
$$
-block-encoding of $\rlcM_1 = \rlcM + \rlcK \rlcQ_0$ with gate complexity $O(T_M \kappa_M \log(\alpha_M \kappa_M/\varepsilon) + T_K)$.
\end{claim}
\begin{proof}
Let $\varepsilon_Q \in (0,1)$ be an error parameter to be fixed later. From Claim~\ref{claim:projs_BE_ind1}$(ii)$, we have a $(2\kappa_M + 1, 2a_M + 2,\varepsilon_Q)$ block-encoding of $\rlcQ_0$ which in turn uses $U_M$ and requires $\varepsilon_M \leq O(\sigma^3 \varepsilon_Q^2/\alpha_M)$. We can now combine this with the block-encoding of $\rlcK$ using Lemma~\ref{lem:prod_BEs} to obtain a $\Big(2\kappa_M \alpha_K + \alpha_K, 2a_M + a_K + 2, \alpha_M \varepsilon_Q + (2\kappa_M + 1)\varepsilon_K\Big)$-block-encoding of $\rlcK \rlcQ_0$, which we denote by $V$. We then implement a block-encoding of $\rlcM_1$ as $\rlcM + \rlcK \rlcQ_0$ using Lemma~\ref{lem:sum_BEs} to obtain a 
$$
\Big(2\kappa_M \alpha_K + \alpha_M + \alpha_K, 2a_M + a_K + 3, \alpha_M \varepsilon_Q + (2\kappa_M + 1)\varepsilon_K + \varepsilon_M\Big)
$$
block-encoding of $\rlcM_1$, which we denote as $U_{M_1}$. We then set $\varepsilon_Q = \varepsilon/(3\alpha_M)$ which requires $\varepsilon_M \leq O(\sigma^3 \varepsilon^2/\alpha_M^3)$. We note that $\varepsilon_M \leq O(\varepsilon^2) \leq O(\varepsilon)$ as $\alpha_M \geq 1$ and $\sigma \leq 1$. Additionally, setting $\varepsilon_K = \varepsilon/(6\kappa_M + 3)$ gives us a block-encoding $U_{M_1}$ with the desired error $\varepsilon$. The corresponding gate complexity is $O(T_M \kappa_M \log(\alpha_M \kappa_M/\varepsilon) + T_K)$ due to applications of $U_M$ for implementing block-encodings of $\rlcM,\rlcQ_0$ from Claim~\ref{claim:projs_BE_ind1}$(ii)$ and $U_K$ for implementing the block-encoding of $\rlcK$. 
\end{proof}

A block-encoding of $\rlcM_1^{-1}$ can then be obtained using Lemma~\ref{lem:pseudoinverse_BE} similar to what was done in Claim~\ref{claim:inv_M_BE_ind0}.
\begin{claim}\label{claim:inv_M1_BE_ind1}
Let $\varepsilon,\sigma, \sigma_1 \in (0,1)$. Consider the context of Claim~\ref{claim:M1_BE_ind1}. Let $\alpha_M \geq 2 \norm{\rlcM}$ and $\alpha_K \geq 2 \norm{\rlcK}$. Let the minimum singular value of $\rlcM_1$ satisfy $\sigma_\min(\rlcM_1) \geq \sigma_1$. For $\varepsilon_M \leq O(\sigma_1^6 \varepsilon^4/\kappa_M^3)$ and $\varepsilon_K = \sigma_1^3 \varepsilon^2/(6\kappa_M + 3)$, there exists a
$$
\Big(2/\sigma_1, 2a_M + a_K + 4, \varepsilon\Big)
$$
-block-encoding of $\rlcM_1^{-1}$ with gate complexity 
$$
\widetilde{O}\left( T_M \frac{\kappa_M^2 \alpha_K}{\sigma_1} \log \left(\frac{\alpha_M \kappa_M}{\sigma_1 \varepsilon} \right) + T_K \frac{\kappa_M \alpha_K}{\sigma_1} \log \left(\frac{1}{\sigma_1 \varepsilon} \right) \right).
$$
\end{claim}
\begin{proof}
We first note that $\norm{\rlcM_1} \leq \norm{\rlcM} + \norm{\rlcK}$ from Claim~\ref{claim:ub_norm_Mi_Ki}. Let $\varepsilon_1 \in (0,1)$ be an error parameter to be fixed later. From Claim~\ref{claim:M1_BE_ind1}, we have a $\Big(2\kappa_M \alpha_K + \alpha_M + \alpha_K, 2a_M + a_K + 3, \varepsilon_1\Big)$-block-encoding for $\rlcM_1$, which we denote by $U_{M_1}$, for $\varepsilon_M \leq O(\varepsilon_1^2/ \kappa_M^3)$ and $\varepsilon_K = \varepsilon_1/(6\kappa_M + 3)$.

Let $\rlcA := \rlcM_1/\alpha_1$ where $\alpha_1 = \norm{\rlcM} + \norm{\rlcK}$. We are given that $\alpha_M \geq 2 \norm{\rlcM}$ and $\alpha_K \geq 2 \norm{\rlcK}$. Then, $\rlcA$ is a normalized matrix with $\spec(\rlcA) \in [\sigma_1/\alpha_1,1]$. We then observe that $U_{M_1}$ is also a $(\beta, 2a_M + a_K + 3,\varepsilon_1/\alpha_1)$ block-encoding of $\rlcA$ where we have denoted $\beta = (2\kappa_M \alpha_K + \alpha_M + \alpha_K)/\alpha_1$ and satisfies $\beta \geq 2$. 

Here, the relevant condition number parameter is then $\kappa_A = \alpha_1/\sigma_1$. Applying Lemma~\ref{lem:pseudoinverse_BE} then produces a $(2\kappa_A, 2a_M + a_K + 4,\delta)$-block-encoding of $\rlcA^{-1}$, which we denote by $V$, provided $\varepsilon_1/\alpha_1 = o(\delta/(\kappa_A^2 \log(\kappa_A/\delta)))$. The cost of implementing $V$ is $O(\kappa_A \beta \log(\kappa_A/\delta))$ uses of $U_{M_1}$. Since $\rlcA^{-1} = \alpha_1 \rlcM_1^{-1}$, $V$ is then a $(2/\sigma_1,2a_M + a_K + 4, \delta/\alpha_1)$-block-encoding of $\rlcM_1^{-1}$. To set the desired error $\delta/\alpha_1 = \varepsilon$, we require $\varepsilon_1 = o(\sigma_1^2 \varepsilon/\log(1/(\sigma_1 \varepsilon)))$ which is satisfied by setting $\varepsilon_1 \leq O(\sigma_1^3 \varepsilon^2)$. This in turn requires $\varepsilon_M \leq O(\sigma_1^6 \varepsilon^4/ \kappa_M^3)$ and $\varepsilon_K = \sigma_1^3 \varepsilon^2/(6\kappa_M + 3)$. 

The cost of implementing $V$ is then $O((\kappa_M \alpha_K + \alpha_M + \alpha_K)/\sigma_1 \log(1/(\sigma_1 \varepsilon))) = O(\kappa_M \alpha_K/\sigma_1 \log(1/(\sigma_1 \varepsilon))$ uses of $U_{M_1}$. The gate complexity of $U_{M_1}$ from Claim~\ref{claim:M1_BE_ind1} for $\varepsilon_1 = O(\sigma_1^3 \varepsilon^2)$ is 
$$
O\Big(T_M \kappa_M \log(\alpha_M \kappa_M/(\sigma_1^3 \varepsilon^2)) + T_K\Big).
$$ 
Putting this together gives us the stated gate complexity in the theorem statement. This completes the proof.
\end{proof}

We can now use the above block-encoding of $\rlcM_1^{-1}$ to give block-encodings of $\rlcP_0 \rlcM_1^{-1} \rlcK$ and $\rlcP_0 \rlcM_1^{-1}$, which we formally describe below.
\begin{claim}\label{claim:ode_BE_ind1}
Let $\varepsilon, \sigma, \sigma_1 \in (0,1)$. Let $\rlcM$ be as defined in Definition~\ref{def:prob_setup_DAEs} with a $(\alpha_M,a_M,\varepsilon_M)$-block-encoding $U_M$ implementable in time $T_M$ where $\alpha_M \geq 2 \norm{\rlcM}$ and with minimum non-zero singular value satisfying $\sigma_\min^+(\rlcM) \geq \sigma$. Suppose that the given $\dae$ has index 1. Let $\rlcQ_0$ be the projector onto the kernel of $\rlcM$ and let $\rlcP_0=\id-\rlcQ_0$. Let $\rlcM_1 = \rlcM + \rlcK \rlcQ_0$ have its minimum singular value satisfy $\sigma_\min(\rlcM_1) \geq \sigma_1$. Then, \begin{enumerate}[$(i)$]
    \item for $\varepsilon_M = \poly(\sigma_1 \varepsilon/\kappa_M)$ and $\varepsilon_K = \poly(\sigma_1 \varepsilon/\kappa_M)$, there exists a $\Big(4\kappa_M/\sigma_1, 4a_M + a_K + 5, \varepsilon\Big)$-block-encoding of $\rlcP_0 \rlcM_1^{-1}$ implementable with gate complexity 
    $$
    \widetilde{O}\left(T_M \frac{\kappa_M^2 \alpha_K}{\sigma_1} \log \left(\frac{\alpha_M \kappa_M}{\sigma_1 \varepsilon_1} \right) + T_K \frac{\kappa_M \alpha_K}{\sigma_1} \log \left(\frac{\kappa_M}{\sigma_1 \varepsilon} \right) \right)
    $$
    \item for $\varepsilon_M = \poly(\sigma_1 \varepsilon/(\alpha_K \kappa_M))$ and $\varepsilon_K = \poly(\sigma_1 \varepsilon/(\alpha_K \kappa_M))$, there exists a 
    $$\Big(4\kappa_M \alpha_K/\sigma_1, 4a_M + 2a_K + 5, \varepsilon\Big)$$-block-encoding of $\rlcP_0 \rlcM_1^{-1} \rlcK$ implementable with gate complexity
    $$
    \widetilde{O}\left(T_M \frac{\kappa_M^2 \alpha_K}{\sigma_1} \log \left(\frac{\alpha_K \alpha_M \kappa_M}{\sigma_1 \varepsilon} \right) + T_K \frac{\kappa_M \alpha_K}{\sigma_1} \log \left(\frac{\alpha_K \kappa_M}{\sigma_1 \varepsilon} \right) \right).
    $$
\end{enumerate}
\end{claim}
\begin{proof}
$(i)$ Let $\varepsilon_0,\varepsilon_1 \in (0,1)$ be error parameters to be fixed later. From Claim~\ref{claim:projs_BE_ind1}, we have a $(2\kappa_M, 2a_M + 1,\varepsilon_0)$-block-encoding of $\rlcP_0$, which we denote by $U_0$, for $\varepsilon_M = O(\sigma^3 \varepsilon_0^2/\alpha_M)$. From Claim~\ref{claim:inv_M1_BE_ind1}, we have a $(2/\sigma_1, 2a_M + a_K + 4, \varepsilon_1)$-block-encoding of $\rlcM_1^{-1}$, which we denote by $U_1$, for $\varepsilon_M = O(\sigma_1^6 \varepsilon_1^4/\kappa_M^3)$ and $\varepsilon_K = \sigma_1^3 \varepsilon_1^2/(6\kappa_M + 3)$. We can now combine these using Lemma~\ref{lem:prod_BEs} to obtain a 
$$
\Big(4\kappa_M/\sigma_1, 4a_M + a_K + 5, 2\kappa_M \varepsilon_1 + 2\varepsilon_0/\sigma_1 \Big)
$$
-block-encoding of $\rlcP_0\rlcM_1^{-1}$, which we denote by $U_{P_0 M_1^{-1}}$. We set $\varepsilon_1 = \varepsilon/(4\kappa_M)$ and $\varepsilon_0 = \sigma_1 \varepsilon/4$ to get a block-encoding with the desired error $\varepsilon$. This requires setting $\varepsilon_M \leq \min\Big(O(\sigma^3 \sigma_1^2 \varepsilon^2/\alpha_M), O(\sigma_1^6 \varepsilon^4/ \kappa_M^7)\Big)$ and $\varepsilon_K \leq O(\sigma_1^3 \varepsilon^2/\kappa_M^3)$. The corresponding gate complexity from using Lemma~\ref{lem:prod_BEs} is
\begin{align*}
& \widetilde{O}\left(T_M \kappa_M \log\frac{\kappa_M}{\varepsilon_0} +  T_M \frac{\kappa_M^2 \alpha_K}{\sigma_1} \log \left(\frac{\alpha_M \kappa_M}{\sigma_1 \varepsilon_1} \right) + T_K \frac{\kappa_M \alpha_K}{\sigma_1} \log \left(\frac{1}{\sigma_1 \varepsilon_1} \right) \right) \\
= \, & \widetilde{O}\left(T_M \frac{\kappa_M^2 \alpha_K}{\sigma_1} \log \left(\frac{\alpha_M \kappa_M}{\sigma_1 \varepsilon} \right) + T_K \frac{\kappa_M \alpha_K}{\sigma_1} \log \left(\frac{\kappa_M}{\sigma_1 \varepsilon} \right) \right),
\end{align*}
where we have used the cost of $U_0$ from Claim~\ref{claim:projs_BE_ind1} and cost of $U_{M_1^{-1}}$ from Claim~\ref{claim:inv_M1_BE_ind1}. This proves item $(i)$.

$(ii)$ Consider the block-encoding $U_{P_0 M_1^{-1}}$ from item $(i)$ but with precision $\varepsilon' \in (0,1)$ to be decided soon. Using Lemma~\ref{lem:prod_BEs}, we can now combine $U_{P_0 M_1^{-1}}$ with $U_K$ to obtain a 
$$
\Big( 4\kappa_M \alpha_K/\sigma_1, 4a_M + 2a_K + 5, 4\kappa_M/\sigma_1 \varepsilon_K + \alpha_K \varepsilon' \Big)
$$
-block-encoding of $\rlcP_0\rlcM_1^{-1} \rlcK$, which we denote by $U_{P_0 M_1^{-1} K}$. We set $\varepsilon' = \varepsilon/(2\alpha_K)$ and $\varepsilon_K = \sigma_1 \varepsilon/(8\alpha_K \kappa_M)$ to obtain a block-encoding with the desired error $\varepsilon$. This requires setting $\varepsilon_M \leq \min\Big(O(\sigma^3 \sigma_1^2 \varepsilon^2/\alpha_K^2 \alpha_M), O(\sigma_1^6 \varepsilon^4/\alpha_K^4 \kappa_M^7)\Big)$ and $\varepsilon_K \leq O(\sigma_1^3 \varepsilon^2/\alpha_K^2 \kappa_M^3)$ as per item $(i)$ for $U_{P_0 M_1^{-1}}$ and $\varepsilon_K = \sigma_1 \varepsilon/(8\alpha_K \kappa_M)$ for $U_K$. Combining these conditions gives us the stated condition on $\varepsilon_M$ and $\varepsilon_K$. Using Lemma~\ref{lem:prod_BEs}, the gate complexity of $U_{P_0 M_1^{-1}K}$ is due to $U_{P_0 M_1^{-1}}$ (item $(i)$) and $U_K$ which is (for $\varepsilon'$ defined here)
\begin{align*}
&\widetilde{O}\left(T_M \frac{\kappa_M^2 \alpha_K}{\sigma_1} \log \left(\frac{\alpha_M \kappa_M}{\sigma_1 \varepsilon'} \right) + T_K \frac{\kappa_M \alpha_K}{\sigma_1} \log \left(\frac{\kappa_M}{\sigma_1 \varepsilon'} \right) + T_K \right)    \\
= \, & \widetilde{O}\left(T_M \frac{\kappa_M^2 \alpha_K}{\sigma_1} \log \left(\frac{\alpha_K \alpha_M \kappa_M}{\sigma_1 \varepsilon} \right) + T_K \frac{\kappa_M \alpha_K}{\sigma_1} \log \left(\frac{\alpha_K \kappa_M}{\sigma_1 \varepsilon} \right) \right),    
\end{align*}
where we used $\varepsilon' = \varepsilon/(2 \alpha_K)$. This completes the proof of item $(ii)$ and the overall claim. 
\end{proof}

We have so far discussed the block-encodings required to solve the inherent $\ode$ for index $1$. We will also require the following block-encodings to execute the linear algebraic step to obtain the desired history state and thereby solve the $\dae$ of index $1$.
\begin{claim}\label{claim:lin_alg_BE_ind1}
Let $\varepsilon, \sigma \in (0,1)$. Let $\rlcM$ be as defined in Definition~\ref{def:prob_setup_DAEs} with a $(\alpha_M,a_M,\varepsilon_M)$-block-encoding $U_M$ implementable in time $T_M$ where $\alpha_M \geq 2 \norm{\rlcM}$ and with minimum non-zero singular value satisfying $\sigma_\min^+(\rlcM) \geq \sigma$. Suppose that the given $\dae$ has index 1. Let $\rlcQ_0$ be the projector onto the kernel of $\rlcM$ and let $\rlcP_0=\id-\rlcQ_0$. Define $\rlcM_1 := \rlcM + \rlcK \rlcQ_0$. Then, 
\begin{enumerate}[$(i)$]
    \item for $\varepsilon_M = \poly(\sigma_1 \varepsilon/\kappa_M)$ and $\varepsilon_K = \poly(\sigma_1 \varepsilon/\kappa_M)$, there exists a $\Big( (4\kappa_M +2)/\sigma_1, 4a_M + a_K + 6, \varepsilon\Big)$-block-encoding of $\rlcQ_0 \rlcM_1^{-1}$ implementable with gate complexity 
    $$
    \widetilde{O}\left(T_M \frac{\kappa_M^2 \alpha_K}{\sigma_1} \log \left(\frac{\alpha_M \kappa_M}{\sigma_1 \varepsilon_1} \right) + T_K \frac{\kappa_M \alpha_K}{\sigma_1} \log \left(\frac{\kappa_M}{\sigma_1 \varepsilon} \right) \right)
    $$
    \item for $\varepsilon_M = \poly(\sigma_1 \varepsilon/(\alpha_K \kappa_M))$ and $\varepsilon_K = \poly(\sigma_1 \varepsilon/(\alpha_K \kappa_M))$, there exists a 
    $$\Big( (4\kappa_M +2)\alpha_K/\sigma_1 + 1, 4a_M + 2a_K + 7, \varepsilon\Big)$$-block-encoding of $\id - \rlcQ_0 \rlcM_1^{-1} \rlcK$ implementable with gate complexity
    $$
    \widetilde{O}\left(T_M \frac{\kappa_M^2 \alpha_K}{\sigma_1} \log \left(\frac{\alpha_K \alpha_M \kappa_M}{\sigma_1 \varepsilon} \right) + T_K \frac{\kappa_M \alpha_K}{\sigma_1} \log \left(\frac{\alpha_K \kappa_M}{\sigma_1 \varepsilon} \right) \right).
    $$
\end{enumerate}
\end{claim}
\begin{proof}
$(i)$ Let $\varepsilon_0,\varepsilon_1 \in (0,1)$ be error parameters to be fixed later. From Claim~\ref{claim:projs_BE_ind1}, we have a $(2\kappa_M + 2, 2a_M + 1,\varepsilon_0)$-block-encoding of $\rlcQ_0$, which we denote by $U_0$, for $\varepsilon_M = O(\sigma^3 \varepsilon_0^2/\alpha_M)$. From Claim~\ref{claim:inv_M1_BE_ind1}, we have a $(2/\sigma_1, 2a_M + a_K + 4, \varepsilon_1)$-block-encoding of $\rlcM_1^{-1}$, which we denote by $U_1$, for $\varepsilon_M = O(\sigma_1^6 \varepsilon_1^4/\kappa_M^3)$ and $\varepsilon_K = \sigma_1^3 \varepsilon_1^2/(6\kappa_M + 3)$. We can now combine these using Lemma~\ref{lem:prod_BEs} to obtain a 
$$
\Big( (4\kappa_M +2)/\sigma_1, 4a_M + a_K + 6, (2\kappa_M + 1) \varepsilon_1 + 2\varepsilon_0/\sigma_1 \Big)
$$
-block-encoding of $\rlcQ_0\rlcM_1^{-1}$, which we denote by $U_{Q_0 M_1^{-1}}$. We set $\varepsilon_1 = \varepsilon/(4\kappa_M + 2)$ and $\varepsilon_0 = \sigma_1 \varepsilon/4$ to get a block-encoding with the desired error $\varepsilon$. This requires setting $\varepsilon_M \leq \min\Big(O(\sigma^3 \sigma_1^2 \varepsilon^2/\alpha_M), O(\sigma_1^6 \varepsilon^4/ \kappa_M^7)\Big)$ and $\varepsilon_K \leq O(\sigma_1^3 \varepsilon^2/\kappa_M^3)$. The corresponding gate complexity from using Lemma~\ref{lem:prod_BEs} is
\begin{align*}
& \widetilde{O}\left(T_M \kappa_M \log\frac{\kappa_M}{\varepsilon_0} +  T_M \frac{\kappa_M^2 \alpha_K}{\sigma_1} \log \left(\frac{\alpha_M \kappa_M}{\sigma_1 \varepsilon_1} \right) + T_K \frac{\kappa_M \alpha_K}{\sigma_1} \log \left(\frac{1}{\sigma_1 \varepsilon_1} \right) \right) \\
= \, & \widetilde{O}\left(T_M \frac{\kappa_M^2 \alpha_K}{\sigma_1} \log \left(\frac{\alpha_M \kappa_M}{\sigma_1 \varepsilon} \right) + T_K \frac{\kappa_M \alpha_K}{\sigma_1} \log \left(\frac{\kappa_M}{\sigma_1 \varepsilon} \right) \right) \\
\end{align*}
where we have used the cost of $U_0$ from Claim~\ref{claim:projs_BE_ind1} and cost of $U_{M_1^{-1}}$ from Claim~\ref{claim:inv_M1_BE_ind1}. This proves item $(i)$.

$(ii)$ Consider the block-encoding $U_{Q_0 M_1^{-1}}$ from item $(i)$ but with precision $\varepsilon' \in (0,1)$ to be decided soon. Using Lemma~\ref{lem:prod_BEs}, we can now combine $U_{Q_0 M_1^{-1}}$ with $U_K$ to obtain a 
$$
\Big( (4\kappa_M +2)\alpha_K/\sigma_1, 4a_M + 2a_K + 6, (4\kappa_M + 2)/\sigma_1 \varepsilon_K + \alpha_K \varepsilon' \Big)
$$
-block-encoding of $\rlcQ_0\rlcM_1^{-1} \rlcK$, which we denote by $U_{Q_0 M_1^{-1} K}$. We combine this with $\id$ using Lemma~\ref{lem:sum_BEs} to obtain a 
$$
\Big( (4\kappa_M +2)\alpha_K/\sigma_1 + 1, 4a_M + 2a_K + 7, (4\kappa_M + 2)/\sigma_1 \varepsilon_K + \alpha_K \varepsilon' \Big)
$$
-block-encoding of $\id - \rlcQ_0\rlcM_1^{-1} \rlcK$, which we denote by $U_{\id - Q_0 M_1^{-1} K}$. We set $\varepsilon' = \varepsilon/(2\alpha_K)$ and $\varepsilon_K = \sigma_1 \varepsilon/(4\alpha_K (2 \kappa_M + 1))$ to obtain a block-encoding with the desired error $\varepsilon$. This requires setting $\varepsilon_M \leq \min\Big(O(\sigma^3 \sigma_1^2 \varepsilon^2/\alpha_K^2 \alpha_M), O(\sigma_1^6 \varepsilon^4/\alpha_K^4 \kappa_M^7)\Big)$ and $\varepsilon_K \leq O(\sigma_1^3 \varepsilon^2/\alpha_K^2 \kappa_M^3)$ as per item $(i)$ for $U_{Q_0 M_1^{-1}}$ and $\varepsilon_K = \sigma_1 \varepsilon/(4\alpha_K (2 \kappa_M + 1))$ for $U_K$. Combining these conditions gives us the stated condition on $\varepsilon_M$ and $\varepsilon_K$. Using Lemma~\ref{lem:prod_BEs}, the gate complexity of $U_{Q_0 M_1^{-1}K}$ is due to $U_{Q_0 M_1^{-1}}$ (item $(i)$) and $U_K$ which is (for $\varepsilon'$ defined here)
\begin{align*}
&\widetilde{O}\left(T_M \frac{\kappa_M^2 \alpha_K}{\sigma_1} \log \left(\frac{\alpha_M \kappa_M}{\sigma_1 \varepsilon'} \right) + T_K \frac{\kappa_M \alpha_K}{\sigma_1} \log \left(\frac{\kappa_M}{\sigma_1 \varepsilon'} \right) + T_K \right)    \\
= \, & \widetilde{O}\left(T_M \frac{\kappa_M^2 \alpha_K}{\sigma_1} \log \left(\frac{\alpha_K \alpha_M \kappa_M}{\sigma_1 \varepsilon} \right) + T_K \frac{\kappa_M \alpha_K}{\sigma_1} \log \left(\frac{\alpha_K \kappa_M}{\sigma_1 \varepsilon} \right) \right),    
\end{align*}
where we used $\varepsilon' = \varepsilon/(2 \alpha_K)$. This completes the proof of item $(ii)$ and the overall claim. 
\end{proof}

\subsubsection{DAE of index two.}
To solve the inherent $\ode$ for index $2$, we will need block-encodings of $-\rlcP_0 \rlcP_1 \rlcM_2^{-1} \rlcK$ governing the dynamics, $\Pi_1 = \rlcP_0 \rlcP_1$ to implement the initial state $\Pi_1 \ket{x(0)}$ and $\Pi_1 \rlcM_2^{-1}$ to implement the forcing state $\Pi_1 \rlcM_1^{-1} \ket{f}$. To give block-encodings of the projectors $\rlcP_1$ and $\rlcQ_1 := \id - \rlcP_1$, we first construct the block-encodings of $\widetilde{\rlcP}_1$ and $\widetilde{\rlcQ}_1$ as defined in Claim~\ref{claim:admissible_proj_ind2}. We omit the proofs of the claims as they proceed similarly to the ones in the previous sections.

\begin{claim}\label{claim:projs_BE_ind2}
Let $\varepsilon, \sigma \in (0,1)$. Let $\rlcM$ be as defined in Definition~\ref{def:prob_setup_DAEs} with a $(\alpha_M,a_M,\varepsilon_M)$-block-encoding $U_M$ implementable in time $T_M$ with minimum non-zero singular value satisfying $\sigma_\min^+(\rlcM) \geq \sigma$. Similarly, $\sigma_\min^+(\rlcM_1) \geq \sigma_1$. Define $\kappa_M = \alpha_M/\sigma$. Suppose that the given $\dae$ has index $2$. Then, there exists the following block-encodings for $\varepsilon_M = \poly(\sigma_1 \varepsilon/\kappa_M)$ and $\varepsilon_K = \poly(\sigma_1 \varepsilon/(\kappa_M \alpha_K))$:
\begin{enumerate}[$(i)$]
    \item $\Big(2/\sigma_1, 2a_M + a_K + 4, \varepsilon\Big)$-block-encoding of $\rlcM_1^{+}$ with gate complexity 
    $$
    \widetilde{O}\left( T_M \frac{\kappa_M^2 \alpha_K}{\sigma_1} \log \left(\frac{\alpha_M \kappa_M}{\sigma_1 \varepsilon} \right) + T_K \frac{\kappa_M \alpha_K}{\sigma_1} \log \left(\frac{1}{\sigma_1 \varepsilon} \right) \right).
    $$
    \item $(O(\alpha_K \kappa_M/\sigma_1), O(a_M + a_K),\varepsilon)$-block-encodings of $\widetilde{\rlcP}_1 := \rlcM_1^+ \rlcM_1$ and $\widetilde{\rlcQ}_1 := \id - \widetilde{\rlcP}_1$ implementable with gate complexity
    $$
    \widetilde{O}\left( T_M \frac{\kappa_M^2 \alpha_K}{\sigma_1} \log \left(\frac{\alpha_M \kappa_M}{\sigma_1 \varepsilon} \right) + T_K \frac{\kappa_M \alpha_K}{\sigma_1} \log \left(\frac{1}{\sigma_1 \varepsilon} \right) \right).
    $$
\end{enumerate}
\end{claim}

\begin{claim}\label{claim:admissible_projs_BE_ind2}
Let $\varepsilon, \sigma, \uptau \in (0,1)$. Let $\rlcM$ be as defined in Definition~\ref{def:prob_setup_DAEs} with a $(\alpha_M,a_M,\varepsilon_M)$-block-encoding $U_M$ implementable in time $T_M$ with minimum non-zero singular value satisfying $\sigma_\min^+(\rlcM) \geq \sigma$. Similarly, $\sigma_\min^+(\rlcM_1) \geq \sigma_1$. Define $\kappa_M = \alpha_M/\sigma$. Suppose that the given $\dae$ has index $2$. Let $\sigma_\min^+(\rlcP_0 \widetilde{\rlcQ}_1) \geq \uptau$. Then, there exists the following block-encodings for $\varepsilon_M = \poly(\uptau \sigma_1 \varepsilon/(\kappa_M \alpha_K))$ and $\varepsilon_K = \poly(\uptau \sigma_1 \varepsilon/(\kappa_M \alpha_K))$:
\begin{enumerate}[$(i)$]
    \item $(2/\uptau, O(a_M + a_K),\varepsilon)$-block-encoding of $(\rlcP_0 \widetilde{\rlcQ}_1)^+$ implementable with gate complexity 
    $$
    \widetilde{O}\left( (T_M + T_K) \poly\left(\frac{\kappa_M \alpha_K}{\sigma_1 \uptau} \log \left(\frac{\alpha_K \kappa_M}{\sigma_1 \uptau \varepsilon} \right)\right) \right).
    $$
    \item $(O(\alpha_K \kappa_M^2/(\sigma_1 \uptau), O(a_M + a_K),\varepsilon)$-block-encodings of $\rlcQ_1$ and $\rlcP_1$ implementable with gate complexity
    $$
    \widetilde{O}\left( (T_M + T_K) \poly\left(\frac{\kappa_M \alpha_K}{\sigma_1 \uptau} \log \left(\frac{\alpha_K \kappa_M}{\sigma_1 \uptau \varepsilon} \right)\right) \right).
    $$
\end{enumerate}
\end{claim}

We now give block-encoding of the inverse $\rlcM_2 := \rlcM_1 + \rlcK_1 \rlcQ_1$ (which is guaranteed to exist in the case of index $2$).
\begin{claim}\label{claim:inv_M2_BE_ind2}
Let $\varepsilon, \sigma, \sigma_1, \sigma_2, \uptau \in (0,1)$. Let $\rlcM$ be as defined in Definition~\ref{def:prob_setup_DAEs} with a $(\alpha_M,a_M,\varepsilon_M)$-block-encoding $U_M$ implementable in time $T_M$ with minimum non-zero singular value satisfying $\sigma_\min^+(\rlcM) \geq \sigma$. Similarly, $\sigma_\min^+(\rlcM_1) \geq \sigma_1$ and $\sigma_\min^+(\rlcM_2) \geq \sigma_2$. Define $\kappa_M = \alpha_M/\sigma$. Suppose that the given $\dae$ has index $2$. Let $\sigma_\min^+(\rlcP_0 \widetilde{\rlcQ}_1) \geq \uptau$. For $\varepsilon_M = \varepsilon_K = \poly(\uptau \sigma_1 \sigma_2 \varepsilon/(\kappa_M \alpha_K))$, there exists a
$$
\Big(2/\sigma_2, O(a_M + a_K), \varepsilon\Big)
$$
-block-encoding of $\rlcM_2^{-1}$ with gate complexity 
$$
\widetilde{O}\left( (T_M + T_K) \poly\left(\frac{\kappa_M \alpha_K}{\sigma_1 \sigma_2 \uptau} \log \left(\frac{\alpha_K \kappa_M}{\sigma_1 \sigma_2 \uptau \varepsilon} \right)\right) \right).
$$
\end{claim}

We can now use the above block-encoding of $\rlcM_2^{-1}$ to give block-encodings of $\rlcP_0 \rlcP_1 \rlcM_2^{-1}$ and $\rlcP_0 \rlcP_1 \rlcM_2^{-1} \rlcK$, which we formally describe below.
\begin{claim}\label{claim:ode_BE_ind2}
Let $\varepsilon, \sigma, \sigma_1, \sigma_2, \uptau \in (0,1)$. Let $\rlcM$ be as defined in Definition~\ref{def:prob_setup_DAEs} with a $(\alpha_M,a_M,\varepsilon_M)$-block-encoding $U_M$ implementable in time $T_M$ with minimum non-zero singular value satisfying $\sigma_\min^+(\rlcM) \geq \sigma$. Similarly, $\sigma_\min^+(\rlcM_1) \geq \sigma_1$ and $\sigma_\min^+(\rlcM_2) \geq \sigma_2$. Define $\kappa_M = \alpha_M/\sigma$. Suppose that the given $\dae$ has index $2$. Let $\sigma_\min^+(\rlcP_0 \widetilde{\rlcQ}_1) \geq \uptau$. Then, there exists the following block-encodings for $\varepsilon_M = \varepsilon_K = \poly(\uptau \sigma_1 \sigma_2 \varepsilon/(\kappa_M \alpha_K))$:
\begin{enumerate}[$(i)$]
    \item $(\poly(\alpha_K \kappa_M/(\sigma_1 \uptau)), O(a_M + a_K), \varepsilon)$-block-encoding of $\rlcP_0 \rlcP_1$ with gate complexity
    $$
    \widetilde{O}\left( (T_M + T_K) \poly\left(\frac{\kappa_M \alpha_K}{\sigma_1 \uptau} \log \left(\frac{\alpha_K \kappa_M}{\sigma_1 \uptau \varepsilon} \right)\right) \right).
    $$
    \item $(\poly(\alpha_K \kappa_M/(\sigma_1 \sigma_2 \uptau)), O(a_M + a_K), \varepsilon)$-block-encoding of $\rlcP_0 \rlcP_1 \rlcM_2^{-1}$ implementable with gate complexity 
    $$
    \widetilde{O}\left( (T_M + T_K) \poly\left(\frac{\kappa_M \alpha_K}{\sigma_1 \sigma_2 \uptau} \log \left(\frac{\alpha_K \kappa_M}{\sigma_1 \sigma_2 \uptau \varepsilon} \right)\right) \right).
    $$
    \item $(\poly(\alpha_K \kappa_M/(\sigma_1 \sigma_2 \uptau)), O(a_M + a_K), \varepsilon)$-block-encoding of $\rlcP_0 \rlcP_1 \rlcM_2^{-1} \rlcK$ implementable with gate complexity 
    $$
    \widetilde{O}\left( (T_M + T_K) \poly\left(\frac{\kappa_M \alpha_K}{\sigma_1 \sigma_2 \uptau} \log \left(\frac{\alpha_K \kappa_M}{\sigma_1 \sigma_2 \uptau \varepsilon} \right)\right) \right).
    $$
\end{enumerate}
\end{claim}

For the linear-algebraic solve, we need block-encodings for $\mathbf{F}_2$ and $\rlcG_2$ (Eq.~\eqref{eq:ind2_LA_z_simple}). Note that the operator corresponding to the initial condition of $\vec{z}$ i.e., $(\id - \rlcP_0 \rlcP_1)$ can be implemented using Lemma~\ref{lem:sum_BEs} and Claim~\ref{claim:ode_BE_ind2}.
\begin{claim}\label{claim:lin_alg_BE_ind2}
Let $\varepsilon, \sigma, \sigma_1, \sigma_2, \uptau \in (0,1)$. Let $\rlcM$ be as defined in Definition~\ref{def:prob_setup_DAEs} with a $(\alpha_M,a_M,\varepsilon_M)$-block-encoding $U_M$ implementable in time $T_M$ with minimum non-zero singular value satisfying $\sigma_\min^+(\rlcM) \geq \sigma$. Similarly, $\sigma_\min^+(\rlcM_1) \geq \sigma_1$ and $\sigma_\min^+(\rlcM_2) \geq \sigma_2$. Define $\kappa_M = \alpha_M/\sigma$. Suppose that the given $\dae$ has index $2$. Let $\sigma_\min^+(\rlcP_0 \widetilde{\rlcQ}_1) \geq \uptau$. Then, for $\varepsilon_M = \varepsilon_K = \poly(\uptau \sigma_1 \sigma_2 \varepsilon/(\kappa_M \alpha_K))$, there exists $(\poly(\alpha_K \kappa_M/(\sigma_1 \uptau)), O(a_M + a_K), \varepsilon)$-block-encoding of $\mathbf{F}_2$ and $\mathbf{G}_2$ (Eq.~\eqref{eq:ind2_LA_z_simple}) with gate complexity
$$
\widetilde{O}\left( (T_M + T_K) \poly\left(\frac{\kappa_M \alpha_K}{\sigma_1 \uptau} \log \left(\frac{\alpha_K \kappa_M}{\sigma_1 \uptau \varepsilon} \right)\right) \right).
$$
\end{claim}

\subsection{Preparation and effects of approximate input state}\label{subsec:initial_state_forcing}
For solving inherent $\ode$s of $\dae$s of index $k \in \{1,2\}$, we need to prepare input states encoding the initial state of $\rlcP_0 \vec{x}(0)$ or $\rlcP_0 \rlcP_1 \vec{x}(0)$. Similarly, we need to ensure the input state encodes the projected forcing $\rlcP_0 \rlcM_1^{-1} \vec{f}$ or $\rlcP_0 \rlcP_1 \rlcM_1^{-1} \vec{f}$. However, as we only prepare approximate block-encodings of the involved matrices, we will only prepare the corresponding states approximately. We thus need to account for errors in the $\ode$ solver due to the initial state and forcing which have not been considered as part of Theorem~\ref{thm:ODE_solver}. 

\paragraph{Input state preparation.}
We will use the following result adapted from \cite[Lemma~14]{Krovi2023improvedquantum} for preparing initial states to our algorithm. The proof follows identically to that in \cite{Krovi2023improvedquantum} and is included for completeness.
\begin{claim}\label{claim:initial_state_prep}
Let $\alpha_1, \alpha_2 \in \mathbb{R}^+$, $m \in \mathbb{N}$ and $h \in (0,1)$. Let $O_x$ and $O_f$ be oracles from Definition~\ref{def:prob_setup_DAEs}. Define the state
$$
\ket{\psi_0} = \frac{1}{\normhist_0} \left[\alpha_1 \norm{\vec{x}_0} \ket{0,0,x_0} + \alpha_2 h \norm{\vec{f}} \sum_{j=0}^{m-1} \ket{j,1,f} \right],
$$
where $\normhist_0 = \sqrt{\alpha_1 \norm{\vec{x}_0}^2 + \alpha_2 h^2 \norm{\vec{f}}^2}$. Then, $\ket{\psi_0}$ can be prepared with a single call to $O_x$ and $O_f$, and an additional $O(\log m)$ elementary gates.
\end{claim}
\begin{proof}
We prepare $\ket{\psi_0}$ from $\ket{0,0,0}$. We first apply a rotation $R$ on the second register to obtain
$$
(\id \otimes R \otimes \id) \ket{0,0,0} \rightarrow \frac{\alpha_1 \norm{\vec{x}_0}}{\normhist_0} \ket{0,0,0} + \frac{\sqrt{m}h\alpha_2 \norm{\vec{f}}}{\normhist_0} \ket{0,1,0}.
$$
We then apply the oracles $O_x$ and $O_f$ to the third register conditioned on the second register being in $\ket{0}$ and $\ket{1}$ respectively. This gives us the state
$$
\xrightarrow{\id \otimes \ket{0}\bra{0} \otimes O_x + \id \otimes \ket{1}\bra{1} \otimes O_f} \frac{\alpha_1 \norm{\vec{x}_0}}{\normhist_0} \ket{0,0,0} + \frac{\sqrt{m}h\alpha_2 \norm{\vec{f}}}{\normhist_0} \ket{0,1,0}.
$$
Finally, we then apply a rotation $R_2$ on the first register conditioned on the second register being $\ket{1}$ that takes 
$$
R_2 \ket{0} \rightarrow \frac{1}{\sqrt{m}} \sum_{j=0}^{m-1} \ket{j}.
$$
This step takes $O(\log m)$ gates using the circuit from \cite{shukla2024efficient}. This completes the proof.
\end{proof}

\paragraph{Effects of approximate input state.}
We now comment on the effect of the approximate initial state and forcing on the solution of an $\ode$ before commenting on how this error propagates to the history state outputted by the quantum $\ode$ solver.

We first have the following claim that describes the error incurred in the solution from misspecified initial condition and forcing. 
\begin{claim}\label{claim:approx_init_forcing_soln}
Let $\eta_0, \eta_f \geq 0$. Suppose we are given the linear ODE of $\dot{\vec{x}} = \mathcal{A}\vec{x} + \vec{f}$ with the initial condition $\vec{x}(0) = \vec{x}_0$ and where $\vec{x},\vec{f}  \in \mathbb{R}^{N}, \bm{\mathcal{A}} \in \mathbb{R}^{N \times N}$. If the error in the specified initial condition $\vec{w}(0)$ is $\norm{\vec{x}(0) - \vec{w}(0)} \leq \eta_0$ and the specified forcing $\vec{g}$ is $\norm{\vec{g} - \vec{f}} \leq \eta_f$, then the error in the obtained solution $\vec{w}(t)$ from the true solution $\vec{x}(t)$ satisfies
$$
\norm{\vec{w}(t) - \vec{x}(t)} \leq \expnorm(\mathcal{A})(\eta_0 + t \eta_f),\,\forall t \in (0,T].
$$
\end{claim}
\begin{proof}
Let the error vector between the true solution and estimated solution at time $t$ be $\vec{e}(t) = \vec{w}(t) - \vec{x}(t)$. Let the error vector between $\vec{f}$ and $\vec{g}$ be $\vec{\delta} = \vec{g} - \vec{f}$. The linear ODE governing the evolution of this error is
$$
\frac{d\vec{e}(t)}{dt} = \frac{d\vec{w}(t)}{dt} - \frac{d\vec{x}(t)}{dt} = \mathcal{L} \vec{e} + \vec{\delta},
$$
with the initial condition $\vec{e}(0) = \vec{w}(0) - \vec{x}(0)$. Using variation of parameters, the error vector at time $t$ satisfies
\begin{align*}
\vec{e}(t) &= \exp(t\calA) \vec{e}(0) + \int_{s=0}^t \exp\Big((t-s)\calA\Big) \vec{\delta} ds \\
\norm{\vec{e}(t)} &\leq \norm{\exp(t \calA)} \norm{\vec{e}(0)} + \int_{s=0}^t \norm{\exp\Big((t-s)\calA\Big)} \norm{\vec{\delta}} ds \\
\norm{\vec{e}(t)} &\leq \norm{\exp(t \calA)} \eta_0 + \eta_f \int_{\tau=0}^t \norm{\exp\Big(\tau \calA\Big)} d\tau \\
\norm{\vec{e}(t)} &\leq \expnorm(\calA)(\eta_0 + t \eta_f),
\end{align*}
where we have used the sub-multiplicativity of the norm and the triangle inequality in the second line followed by a change of variables as $\tau = t-s$ in the third line. Finally, we noted the definition of $\expnorm(\calA)$ giving us the desired result.
\end{proof}

We will not use the above claim but it tells us what to expect when we consider the error incurred in the solution from misspecified inputs and approximating it by a truncated Taylor series. We do this next.
\begin{claim}\label{claim:diff_taylor_truth_approx_input}
Let $\eta_0, \eta_f \in (0,1)$ and $T > 0$. Consider the $\ode$: $\dot{\vec{x}} = \calA \vec{x} + \vec{f}$. Suppose $\vec{x}(t), \forall t \in (0,T]$ is the solution to the $\ode$ at time $t$ when the initial condition is $\vec{x}_0$ and forcing is $\vec{f}$. Suppose $\widetilde{x}(t), \forall t \in (0,T]$ is the solution obtained by truncated Taylor series at $k$ terms to the $\ode$ of $\dot{\vec{x}} = \calA \vec{x} + \vec{g}$ with the approximate initial condition $\widetilde{x}_0$ satisfying $\norm{\widetilde{x}_0 - \vec{x}_0} \leq \eta_0$ and approximate forcing $\vec{g}$ satisfying $\norm{\vec{f} - \vec{g}} \leq \eta_f$. Define $\norm{\calA} h \leq 1$, and $m = \ceil{T/h}$. Then, the following are true:
\begin{enumerate}[$(i)$]
    \item If $(k + 1)! \geq \frac{m e^3}{\delta} \left(1 + \frac{T e^2 \norm{\vec{f}}}{\norm{\vec{x}(T)}}\right)$, then the solutions at the final time $T$ satisfy
    $$
    \norm{\vec{x}(T) - \widetilde{x}(T)} \leq \delta \norm{\vec{x}(T)} + \expnorm(\calA) (1 + \delta) \eta_0 + 2 T  \expnorm(\calA) (1+\delta) \eta_f.
    $$
    \item If $(k + 1)! \geq \frac{4 m e^3}{\delta} \left(1 + \frac{T e^2 \norm{\vec{f}}}{\mu}\right)$ where $\mu^2 = m^{-1} \sum_{j=1}^m \norm{\vec{x}(jh)}^2$ is the average energy across time, then the normalized root mean square error of the solutions across time $t \in (0,T]$ satisfy
    $$
    \left(\frac{\sum_{j=0}^m \norm{\vec{x}(t_j) - \widetilde{x}(t_j)}^2}{\sum_{j=0}^m \norm{\vec{x}(t_j)}^2} \right)^{1/2} \leq \frac{\delta}{2} + \frac{(1 + 2 \expnorm(\calA))\eta_0}{\mu} + \frac{4 T \expnorm(\calA) \eta_f}{\mu}.
    $$
\end{enumerate}
\end{claim}
\begin{proof}
$(i)$ The true solution $\vec{x}(T)$ and the approximate solution from Taylor series with mis-specified inputs $\widetilde{x}(T)$ and $\vec{g}$ can be expressed as
\begin{equation}\label{eq:def_true_vec_taylor_vec}
\vec{x}(T) = L_0 \vec{x}_0 + L_1 \vec{f}, \quad \widetilde{x}(T) = L_0' \widetilde{x}_0 + L_1' \vec{g},
\end{equation}
where
$$
L_0 = \exp(\calA T), \quad L_0' = T_k^m(\calA h), \quad L_1 = \int_{0}^T \exp(\calA s) ds, \quad L_1' = h \sum_{j=0}^{m-1} T_k^j (\calA h) S_k(\calA h),
$$
and $T_k$ and $S_k$ are functions defined as $T_k(z) = \sum_{j=0}^k (z^j/j!)$, and $S_k(z) = \sum_{j=1}^k (z^{j-1}/j!)$ respectively. We then have that
\begin{align}
    \norm{\vec{x}(T) - \widetilde{x}(T)} &= \norm{L_0 \vec{x}_0 - L_0' \widetilde{x}_0 + L_1 \vec{f} - L_1' \vec{g}} \nonumber \\
    &= \norm{L_0 \vec{x}_0 - L_0' \vec{x}_0 + L_0' \vec{x}_0 - L_0' \widetilde{x}_0 + L_1 \vec{f} - L_1' \vec{f} + L_1'\vec{f} - L_1' \vec{g}} \nonumber \\
    &\leq \norm{(L_0 - L_0') \vec{x}_0 + (L_1 - L_1') \vec{f}} + \norm{L_0' (\vec{x}_0 - \widetilde{x}_0)} + \norm{L_1'(\vec{f} - \vec{g})} \nonumber \\
    &\leq \frac{me^3}{(k+1)!} \left(1 + \frac{T e^2 \norm{\vec{f}}}{\norm{\vec{x}(T)}}\right) + \norm{L_0' (\vec{x}_0 - \widetilde{x}_0)} + \norm{L_1'(\vec{f} - \vec{g})}
    \label{eq:interim_diff_taylor_true}
\end{align}
where we have added and subtracted $L_0'\vec{x},L_1'\vec{f}$ inside the norm in the second line, applied triangle inequality after collecting terms in the third line, and used \cite[Theorem~3]{Krovi2023improvedquantum} that showed 
$$
\norm{(L_0 - L_0') \vec{x}_0 + (L_1 - L_1') \vec{f}} \leq \frac{me^3}{(k+1)!} \left(1 + \frac{T e^2 \norm{\vec{f}}}{\norm{\vec{x}(T)}}\right)
$$
in the fourth line. We bound the second term in Eq.~\eqref{eq:interim_diff_taylor_true} as
\begin{equation}\label{eq:norm_L0prime}
\norm{L_0' (\vec{x}_0 - \widetilde{x}_0)} \leq \norm{L_0'} \norm{\vec{x}_0 - \widetilde{x}_0} \leq \expnorm(\calA) (1+\delta) \eta_0,
\end{equation}
where we used $\norm{L_0'} \leq \expnorm(\calA)(1+\delta)$ as shown in \cite[Lemma~13]{Krovi2023improvedquantum} and the promise $\norm{\widetilde{x}_0 - \vec{x}(0)} \leq \eta_0$. To bound the third term in Eq.~\eqref{eq:interim_diff_taylor_true}, we first note that
\begin{align}
\norm{L_1'} = \norm{ h \sum_{j=0}^{m-1} T_k^j (\calA h) S_k(\calA h)}  \leq h \sum_{j=0}^{m-1} \norm{T_k^j (\calA h)} \norm{S_k(\calA h)}
&\leq h \sum_{j=0}^{m-1} 2 \expnorm(\calA)(1+\delta) \nonumber \\
&\leq 2 m h \expnorm(\calA)(1+\delta) \nonumber \\
&\leq 2 T \expnorm(\calA)(1+\delta),
\label{eq:norm_L1prime}
\end{align}
where we used triangle inequality followed by submultiplicativity of norms in the second inequality, and used in the third inequality that $\norm{T_k^j (\calA h)} \leq \expnorm(\calA)(1+\delta)$ for any $j \leq m$ (as shown in \cite[Lemma~13]{Krovi2023improvedquantum}) and 
$$
\norm{S_k(\calA h)} = \norm{\sum_{j=1}^k \frac{(\calA h)^{j-1}}{j!}} \leq \sum_{j=1}^k \frac{\norm{\calA h}^{j-1}}{j!} \leq \sum_{j=1}^k \frac{1}{j!} \leq e - 1 \leq 2.
$$
We can then bound the third term in Eq.~\eqref{eq:interim_diff_taylor_true}
\begin{equation}\label{eq:norm_third_term}
\norm{L_1'(\vec{f} - \vec{g})} \leq \norm{L_1'} \norm{\vec{f} - \vec{g}} \leq 2 T \expnorm(\calA)(1+\delta) \eta_f, 
\end{equation}
where we used Eq.~\eqref{eq:norm_L1prime} and the promise that $\norm{\vec{f} - \vec{g}} \leq \eta_f$. Plugging in the upper bounds from Eq.~\eqref{eq:norm_L0prime} and Eq.~\eqref{eq:norm_third_term} into Eq.~\eqref{eq:interim_diff_taylor_true} gives us
$$
\norm{\vec{x}(T) - \widetilde{x}(T)} \leq \delta \norm{\vec{x}(T)} + \expnorm(\calA) (1 + \delta) \eta_0 + 2 T \expnorm(\calA) (1+\delta) \eta_f,
$$
which is the desired result. This completes the proof of item $(i)$.

$(ii)$ Let $\vec{x}_k(t)$ be the solution obtained from assuming a Talyor series approximation of the solution with $k$ terms. Denoting $t_j := jh$. Then, we can show that for $j \in [m]$
\begin{equation}\label{eq:diff_taylor_trunc_true}
\norm{\vec{x}(t_j) - \vec{x}_k(t_j)} \leq \frac{j e^3}{(k+1)!} \norm{\vec{x}(t_j)} + \frac{j (t_j) e^5}{(k+1)!}\norm{\vec{f}}
\end{equation}
and 
\begin{align*}
\norm{\vec{x}(t_j) - \widetilde{x}(t_j)} &\leq \frac{j e^3}{(k+1)!} \norm{\vec{x}(t_j)} + \frac{j (t_j) e^5}{(k+1)!}\norm{\vec{f}} + \expnorm(\calA) (1 + \delta) \eta_0 + 2 (t_j)  \expnorm(\calA) (1+\delta) \eta_f \\
&\leq \frac{m e^3}{(k+1)!} \Big(\norm{\vec{x}(t_j)} + T e^2 \norm{\vec{f}}\Big) + 2 \expnorm(\calA) \eta_0 + 4 T \expnorm(\calA) \eta_f.
\end{align*}
Let us denote $\rho := 2 \expnorm(\calA) \eta_0 + 4 T \expnorm(\calA) \eta_f$ and the vector $\vec{\alpha} \in \mathbb{R}^{m+1}$ where $\alpha_j$ is the upper bound on $\norm{\vec{x}(t_j) - \widetilde{x}(t_j)}$ not involving the term depending on $\rho$ i.e.,
$$
\alpha_0 = 0, \quad \alpha_j = \frac{m e^3}{(k+1)!} \Big(\norm{\vec{x}(t_j)} + T e^2 \norm{\vec{f}}\Big), \forall j \in [m]
$$
We then have $\norm{\vec{x}_0 - \widetilde{x}_0} \leq \eta_0$ and $\norm{\vec{x}(t_j) - \widetilde{x}(t_j)} \leq \alpha_j + \rho, \forall j \in [m]$. We note that for $(k+1)! \geq \frac{4 me^3}{\delta}\left(1 + \frac{T e^2 \norm{\vec{f}}}{\mu}\right)$, we have that for $j \in [m]$
$$
\alpha_j \leq \frac{\delta \mu}{4}\left(\frac{\norm{\vec{x}(t_j)}}{\mu + T e^2 \norm{\vec{f}}} + \frac{T e^2 \norm{\vec{f}}}{\mu + T e^2 \norm{\vec{f}}} \right) 
\leq \frac{\delta \mu}{4}\left(\frac{\norm{\vec{x}(t_j)}}{\mu} + \frac{T e^2 \norm{\vec{f}}}{T e^2 \norm{\vec{f}}} \right) \leq \frac{\delta}{4}(\norm{\vec{x}(t_j)} + \mu),
$$
where we have used for $a,b \in \mathbb{R}^+$, $1/(a+b) \leq 1/a$. This implies that
\begin{equation}\label{eq:inter_bound_norm_alpha}
\sum_{j=1}^m \alpha_j^2 \leq \sum_{j=1}^m \frac{\delta^2}{16}(\norm{\vec{x}(t_j)} + \mu)^2 \leq \sum_{j=1}^m \frac{\delta^2}{16}(2\norm{\vec{x}(t_j)}^2 + 2\mu^2) \leq \frac{\delta^2}{8}( \sum_{j=1}^m \norm{\vec{x}(t_j)}^2 + \sum_{j=1}^m \mu^2) \leq \frac{\delta^2}{4} m \mu^2,    
\end{equation}
where we have used $(a+b)^2 \leq 2(a^2 + b^2)$ in the third inequality and used the definition $\mu^2 = m^{-1} \sum_{j=1}^m \norm{\vec{x}(t_j)}^2$ in the final inequality.

Applying triangle inequality on $\mathbb{R}^{m+1}$, we have
$$
\left(\sum_{j=1}^m \norm{\vec{x}(t_j) - \widetilde{x}(t_j)}^2\right)^{1/2} \leq \eta_0 + \left( \sum_{j=1}^m \alpha_j^2 \right)^{1/2} + \sqrt{m} \rho \leq \eta_0 + \frac{\delta}{2} \sqrt{m} \mu + \sqrt{m} \rho,
$$
where we have used Eq.~\eqref{eq:inter_bound_norm_alpha} in the second inequality. We then have that
\begin{equation}
\left(\frac{\sum_{j=0}^m \norm{\vec{x}(t_j) - \widetilde{x}(t_j)}^2}{\sum_{j=0}^m \norm{\vec{x}(t_j)}^2} \right)^{1/2} \leq \frac{\left(\sum_{j=0}^m \norm{\vec{x}(t_j) - \widetilde{x}(t_j)}^2\right)^{1/2}}{\sqrt{m} \mu}  \leq \frac{\eta_0}{\sqrt{m} \mu} + \frac{\delta}{2} + \frac{\rho}{\mu}
\end{equation}
Expanding the definition of $\rho$ and using $\eta_0/(\sqrt{m}\mu) \leq \eta_0/\mu$ as $m \geq 1$, the above then implies
$$
\left(\frac{\sum_{j=0}^m \norm{\vec{x}(t_j) - \widetilde{x}(t_j)}^2}{\sum_{j=0}^m \norm{\vec{x}(t_j)}^2} \right)^{1/2} \leq \frac{\delta}{2} + \frac{(1 + 2 \expnorm(\calA))\eta_0}{\mu} + \frac{4 T \expnorm(\calA) \eta_f}{\mu},
$$
which is the desired result. This completes the proof of item $(ii)$ and the overall claim.
\end{proof}

Before we comment on how the errors in the input state propagate to the history state, let us prove the following quick claim that allows us to bound the distances between a target history state and prepared history state depending on the errors between their encoded time-evolved states.
\begin{claim}\label{claim:diff_hist_states}
Let $\varepsilon, h \in (0,1)$, $T > 0$, $m = \ceil{T/h}$ and consider vectors $\vec{x}(t), \vec{w}(t), \forall t \in [0,T]$. Suppose $\sum_{j=0}^m \norm{\vec{x}(jh) - \vec{w}(jh)}^2 \leq \varepsilon^2 \sum_{j=0}^m \norm{\vec{x}(jh)}^2$. Define
$$
\ket{\Psi_x} := \frac{1}{\normhist_x}\sum_{j=0}^{m} \norm{\vec{x}(jh)} \ket{j, x(jh)}, \quad \ket{\Psi_w} := \frac{1}{\normhist_w}\sum_{j=0}^{m} \norm{\vec{w}(jh)} \ket{j, w(jh)}.
$$
Then, $\norm{\ket{\Psi_x} - \ket{\Psi_w}} \leq 2\varepsilon$.
\end{claim}
\begin{proof}
Evaluating the distance between the history states, we obtain
$$
\norm{\ket{\Psi_x} - \ket{\Psi_w}} = \norm{ \frac{\sum_{j=0}^m \ket{j} \otimes \vec{x}(jh)}{\normhist_x} - \frac{\sum_{j=0}^m \ket{j} \otimes \vec{w}(jh)}{\normhist_w} } 
\leq 2 \norm{ \frac{\sum_{j=0}^m \ket{j} \otimes (\vec{x}(jh) - \vec{w}(jh)}{\normhist_x}},    
$$
where we used Fact~\ref{fact:normalized_vector_error} in the second inequality in the first line. Using the orthogonality over the $\ket{j}$ register, we then have
$$
\norm{\ket{\Psi_x} - \ket{\Psi_w}}^2 \leq 4 \frac{\sum_{j=0}^m \norm{\vec{x}(jh) - \vec{w}(jh)}^2}{\sum_{j=0}^m \norm{\vec{x}(jh)}^2} \leq 4 \frac{\varepsilon^2 \sum_{j=0}^m \norm{\vec{x}(j)}^2}{\sum_{j=0}^m \norm{\vec{x}(jh)}^2} = 4 \varepsilon^2 \implies \norm{\ket{\Psi_x} - \ket{\Psi_w}} \leq 2 \varepsilon,
$$
where we used the given promise $\sum_{j=0}^m \norm{\vec{x}(jh) - \vec{w}(jh)}^2 \leq \varepsilon^2 \sum_{j=0}^m \norm{\vec{x}(jh)}^2$ in the second inequality before the implication. This completes the proof.
\end{proof}

We will often use the lemma below to comment on how the input states to our projected $\ode$ solves are obtained.
\begin{lemma}\label{lem:projected_ODE_initial_state}
Suppose we are given oracles $O_x, O_f$ that prepare the states $\ket{x(0)}$ and $\ket{f}$ respectively. Suppose we are given an an $(\alpha_P, a_P, \varepsilon_P)$-block-encoding $U_P$ of $\mathbf{P}$ implementable in time $T_P$ and an $(\alpha_G, a_G, \varepsilon_G)$-block-encoding $U_G$ of $\mathbf{G}$ implementable in time $T_G$. Define
$$
\ket{\psi_0} = \frac{1}{\normhist_0} \left[\alpha_P \norm{\vec{x}_0} \ket{0,0,x_0} + \alpha_G h \norm{\vec{f}} \sum_{j=0}^{m-1} \ket{j,1,f} \right]
$$
Then, applying $W := \id\otimes \ket{0}\bra{0}\otimes U_P + \id\otimes \ket{1}\bra{1}\otimes U_G$ to $\ket{0}^a \ket{\psi_0}$ where $a = \max(a_P,a_G)$ produces
$$
\ket{0}^{a} \frac{1}{\normhist_0} \left[\norm{\widetilde{y}_0} \ket{0,0, \widetilde{y}_0} + h \norm{\widetilde{g}} \sum_{j=0}^{m-1} \ket{j,1,\widetilde{g}} \right] + \ket{\junk}
$$
where $\ket{\junk}$ is a state orthogonal to every state with $\ket{0}^a$ over the ancilla qubits, $\widetilde{y}_0$ satisfies $\norm{\widetilde{y}_0 - \rlcP\vec{x}_0} \leq \varepsilon_P \norm{\vec{x}_0}$ and the approximate forcing $\widetilde{g}$ satisfies $\norm{\widetilde{g} - \vec{g}} = \norm{\widetilde{\rlcG}\vec{f} - \rlcG\vec{f}} \leq \varepsilon_G \norm{\vec{f}}$.
\end{lemma}
\begin{proof}
Define the vectors $\vec{g} = \rlcG \vec{f}, \vec{y} = \rlcP \vec{x}$ with their corresponding states as $\ket{g} := \vec{g}/\norm{\vec{g}}$ and $\ket{y} = \vec{y}/\norm{\vec{y}}$ respectively. We use Claim~\ref{claim:initial_state_prep} to produce the following state
\begin{equation}\label{eq:interim_init0}
\ket{\psi_0} = \frac{1}{\normhist_0} \left[\alpha_P \norm{\vec{x}_0} \ket{0,0,x_0} + \alpha_G h \norm{\vec{f}} \sum_{j=0}^{m-1} \ket{j,1,f} \right],
\end{equation}
where $\alpha_P$ and $\alpha_G$ are the subnormalizations of the block-encodings $U_P$ and $U_G$ respectively. We now apply the unitary $W_1$ to $\ket{\psi_0}$, which is defined as
\begin{equation}\label{eq:interim_W1}
W_1 = \id\otimes \ket{0}\bra{0}\otimes U_P + \id\otimes \ket{1}\bra{1}\otimes U_G,  
\end{equation}
where we have assumed that one of the block-encodings of $U_P,U_G$ has been padded as the number of flag qubits required may be different. To determine the form of $W_1 \ket{\psi_0}$, let
$$
\widetilde{\rlcP} := \alpha_P(\bra{0}^{\otimes a_P}\otimes \id)\,U_P\,(\ket{0}^{\otimes a_P}\otimes \id), \quad
\widetilde{\rlcG} := \alpha_G(\bra{0}^{\otimes a_G}\otimes \id)\,U_G\,(\ket{0}^{\otimes a_G}\otimes \id).
$$
We note that the applications of $U_P$ to $\ket{x_0}$ and $U_G$ to $\ket{f}$ are
\begin{align}
U_P\ket{0^{a_1}}\ket{x_0} &= \frac{1}{\alpha_P}\ket{0^{a_1}}\,\widetilde{\rlcP}\ket{x_0} + \ket{\junk_P}, \\
U_G\ket{0^{a_1}}\ket{f} &= \frac{1}{\alpha_G}\ket{0^{a_1}}\,\widetilde{\rlcG}\ket{f} + \ket{\junk_G},
\label{eq:effects_UP_UG}
\end{align}
where $\ket{\junk_P}$ in the first line is some unnormalized state that is orthogonal to every state with $\ket{0}^{a_1}$ on the first register, and $\ket{\junk_G}$ in the first line is some unnormalized state that is orthogonal to every state with $\ket{0}^{a_1}$ on the first register. We then have that $W_1 \ket{\psi_0}$ can be written as
\begin{align}
W_1 \ket{\psi_0} &= \frac{1}{\normhist_0}
\left[ \|\vec{x}_0\| \ket{0,0} \left(\ket{0^{a_1}}\widetilde{\rlcP}\ket{x_0} + \alpha_P \ket{\junk_P}
\right) + h\|f\| \sum_{j=0}^{m-1}\ket{j,1} \left(\ket{0^{a_1}}\widetilde{\rlcG}\ket{f} +
\alpha_G \ket{\junk_G} \right)\right] \nonumber \\
&= \frac{1}{\normhist_0}
\ket{0}^{a_1} \left[ \|\widetilde{\rlcP} \vec{x}_0\| \ket{0,0,\widetilde{\rlcP} x_0}
 + h \|\widetilde{\rlcG} \vec{f}\| \sum_{j=0}^{m-1} \ket{j,1, \widetilde{\rlcG} f} \right] + \ket{\junk},
\end{align}
where we have used $\widetilde{\rlcP} \ket{x_0} = \widetilde{\rlcP} x_0/\norm{x_0} = \norm{\widetilde{\rlcP} x_0} \ket{\widetilde{\rlcP}x_0}/\norm{x_0}$ and similarly $\widetilde{\rlcG} \ket{f} = \norm{\widetilde{\rlcG} f} \ket{\widetilde{\rlcG} f}/\norm{f}$ and collected the junk states in the second line into $\ket{\junk}$ that is orthogonal to every state $\ket{a}^{a_1}$ on the first register. Let us denote $\widetilde{y}_0 = \widetilde{\rlcP} \vec{x}_0$, $\widetilde{g} = \widetilde{\rlcG} \vec{f}$ and define the state
\begin{equation}\label{eq:interim_psi1}
\ket{\psi_1} := \frac{1}{\normhist_1} \left[\norm{\widetilde{y}_0} \ket{0,0, \widetilde{y}_0} + h \norm{\widetilde{g}} \sum_{j=0}^{m-1} \ket{j,1,\widetilde{g}} \right],
\end{equation}
where $\normhist_1 = \sqrt{\norm{\widetilde{y}_0}^2 + mh^2 \norm{\widetilde{g}}^2}$. We can then write
$$
W_1 \ket{\psi_0} = \frac{\normhist_1}{\normhist_0} \ket{0}^{a_1} \ket{\psi_1} + \ket{\junk}.
$$
Note that $\ket{\psi_1}$ effectively encodes the approximate initial condition $\widetilde{y}_0$ and approximate forcing $\widetilde{g}$ that satisfy
\begin{equation}\label{eq:approx_input_encodings}
\norm{\widetilde{y}_0 - \vec{y}_0} = \norm{\widetilde{\rlcP}\vec{x}_0 - \rlcP\vec{x}_0} \leq \varepsilon_P \norm{\vec{x}_0}, \quad \norm{\widetilde{g} - \vec{g}} = \norm{\widetilde{\rlcG}\vec{f} - \rlcG\vec{f}} \leq \varepsilon_G \norm{\vec{f}},
\end{equation}
where we have used the definitions of $\widetilde{\rlcP}, \widetilde{\rlcG}$ and the precisions of their respective block-encodings $U_P,U_G$. This completes the proof.
\end{proof}

\subsection{Algorithm and Analysis}
\label{subsec:QDAE_solver_algo_analysis}
We now give the algorithms corresponding to the quantum $\dae$ solver for $\dae$s with tractability index up to $2$ using the subroutines introduced in Sections~\ref{subsec:useful_BEs}--\ref{subsec:initial_state_forcing} and approach proposed in Section~\ref{subsec:approach}. We will also comment on the guarantees of the quantum $\dae$ solver by commenting on its correctness and complexity in each case.

\subsubsection{Index zero}\label{subsec:ind0_dae_solver}
For the case of index zero, we only need to carry out an $\ode$ solve as $\rlcM$ is guaranteed to be invertible, which is $\dot{\vec{x}} = - \rlcM^{-1} \rlcK \vec{x} + \rlcM^{-1} \vec{f}$. We have the following main result.
\begin{theorem}\label{thm:sim_ind0}
Let $\varepsilon, \sigma \in (0,1)$. Consider the problem setup of Definition~\ref{def:prob_setup_DAEs} and suppose that the $\dae$ has index $0$. Let the minimum singular value of $\rlcM$ satisfy $\sigma_\min(\rlcM) \geq \sigma$. Then, there exists a quantum algorithm that outputs a state $|\widehat{\Psi}\rangle$ with success probability $\geq 2/3$ such that $\left\| |\widehat{\Psi}\rangle - \normhist^{-1} \sum_{k=0}^{m} \norm{\vec{x}(k \Delta t)} \ket{k} \ket{x(k \Delta t)} \right\| \leq \varepsilon$ where $\Delta t = \sigma/ \|\rlcK \|$ and $m = \ceil{T/\Delta t}$. The algorithm uses the following complexity 
\begin{align*}
\text{Queries to }U_{M}&: \widetilde{O}\left(T^2 \, \expnorm(-\rlcM^{-1}\rlcK)^2 \cdot \poly\Big(\kappa_M, \alpha_K, \calC_f, \log(1/\varepsilon) \Big) \right) \,,\\
\text{Queries to }U_{K}&: \widetilde{O}\left(T^2 \, \expnorm(-\rlcM^{-1}\rlcK)^2 \cdot \poly\Big(\kappa_M, \alpha_K, \calC_f, \log(1/\varepsilon) \Big)\right) \,,\\ 
\text{Queries to }O_{f},O_x&: \widetilde{O}\left(T \alpha_K \kappa_M \expnorm(-\rlcM^{-1} \rlcK) \cdot \log(1/\varepsilon)\right) \,,
\end{align*}
where $\kappa_M = \alpha_M/\sigma$ is an upper bound on the condition number of $\rlcM$ and $\mathcal{C}_f := \log\left(1 + \frac{T e^2 \|\vec{f}\|}{\sigma \|\mu\|}\right)$ where $\mu^2 = m^{-1} \sum_{j=1}^m \norm{\vec{x}(j \Delta t)}^2$. An additional gate complexity of $O(1)$ times the number of uses of $U_M$ and $U_K$ is also required. It is sufficient to set the precisions of $U_M$ and $U_K$ as
$$
\varepsilon_M = \poly\left(\frac{\sigma \varepsilon}{\alpha_K T \expnorm(-\rlcM^{-1} \rlcK) \calC_f }\right) \text{ and } \varepsilon_K = \poly\left(\frac{\sigma \varepsilon}{\alpha_K T \expnorm(-\rlcM^{-1} \rlcK) \calC_f }\right).
$$
\end{theorem}
\begin{proof}
Define the vector $\vec{g} = \rlcM^{-1} \vec{f}$ with their corresponding state as $\ket{g} := \vec{g}/\norm{\vec{g}}$. Define $h = \sigma/\norm{\rlcK}$, $m = \ceil{T/h}$, and $t_j = jh ,\forall j\in[m]_0$. Let $\varepsilon_1 \in (0,1)$ be an error parameter to be determined later. Let $U_{M^{-1}}$ be the $(\alpha_1,a_1,\varepsilon_1)$ block-encoding of $\rlcM^{-1}$ from Claim~\ref{claim:inv_M_BE_ind0} with $\alpha_1 = 2/\sigma$ and $a_1 = a_M + 1$. This requires $\varepsilon_M = o(\sigma^2 \varepsilon_1/\log(1/(\sigma \varepsilon_1)))$. Let
\begin{equation}\label{eq:ind0_tilde_invM}
    \widetilde{M}^{-1} := \alpha_1(\bra{0}^{\otimes a_1}\otimes \id)\,U_{M^{-1}}\,(\ket{0}^{\otimes a_1}\otimes \id).
\end{equation}

Let $U_0$ be the unitary from Claim~\ref{claim:initial_state_prep} that produces the following state from $\ket{0,0,0}$:
\begin{equation}\label{eq:ind0_interim_init0}
\ket{\psi_0} = \frac{1}{\normhist_0} \left[\norm{\vec{x}_0} \ket{0,0,x_0} + \alpha_1 h \norm{\vec{f}} \sum_{j=0}^{m-1} \ket{j,1,f} \right].
\end{equation}
We now apply the unitary $W_1$ to $\ket{0}^{a_1} \ket{\psi_0}$, which is defined as
\begin{equation}\label{eq:ind0_interim_W1}
W_1 = \id\otimes \ket{0}\bra{0}\otimes \id + \id\otimes \ket{1}\bra{1}\otimes U_{M^{-1}}.  
\end{equation}
From the Lemma~\ref{lem:projected_ODE_initial_state} (for the instantiations of $\rlcP := \id$ and $\rlcG := \rlcM^{-1}$ there), the form of $W_1 \ket{0}^{a_1} \ket{\psi_0}$ is given by
\begin{align}\label{eq:ind0_interim_W1psi0}
W_1 \ket{0}^{a_1} \ket{\psi_0} &= \frac{\normhist_1}{\normhist_0} \ket{0}^{a_1} \ket{\psi_1} + \ket{\junk},
\end{align}
where $\ket{\junk}$ is an (unnormalized) state orthogonal to every state with $\ket{0}^{a_1}$ on the first register and $\ket{\psi_1}$ is defined as
\begin{equation}\label{eq:ind0_interim_psi1}
\ket{\psi_1} := \frac{1}{\normhist_1} \left[\norm{\vec{x}_0} \ket{0,0, x_0} + h \norm{\widetilde{g}} \sum_{j=0}^{m-1} \ket{j,1,\widetilde{g}} \right],
\end{equation}
where $\widetilde{g} = \widetilde{\rlcM}^{-1} \vec{f}$ and $\normhist_1 = \sqrt{\norm{\widetilde{y}_0}^2 + mh^2 \norm{\widetilde{g}}^2}$. Define $\Psi_1 := \calZ_1 \ket{\psi_1}$. Note that $\ket{\psi_1}$ effectively encodes the exact initial condition $\vec{x}_0$ and approximate forcing $\widetilde{g}$ that satisfies
\begin{equation}\label{eq:ind0_approx_input_encodings}
\norm{\widetilde{g} - \vec{g}} = \norm{\widetilde{M}^{-1} \vec{f} - \rlcM^{-1} \vec{f}} \leq \varepsilon_1 \norm{\vec{f}},
\end{equation}
where we have used the definitions of $\widetilde{M}^{-1}$ (Eq.~\eqref{eq:ind0_tilde_invM} and the precisions of the corresponding block-encodings $U_{M^{-1}}$.

Let $\delta \in (0,1)$ be an error parameter to be fixed later and define
\begin{equation}\label{eq:ind0_proj_ODE_taylor_steps}
k = \left\lceil \frac{2\log \Omega}{\log \log \Omega} \right\rceil, \quad \text{ where } \Omega = \frac{4 m e^3}{\delta} \left(1 + \frac{T e^2 \norm{\vec{f}}}{\sigma \mu} \right),
\end{equation}
where $\mu^2 := m^{-1} \sum_{j=1}^m \norm{\vec{x}(t_j)}^2$. This choice ensures that $(k+1)! \geq \Omega$.

Let $\vec{x}(t), \forall t \in [0,T]$ be the true solution of the given $\ode$ and let $\widetilde{x}(t), \forall t \in [0,T]$ be the solution obtained from approximating the solution by a truncated series over $k$ terms with the mis-specified forcing $\widetilde{g}$ (instead of $\vec{g}$). Using Claim~\ref{claim:diff_taylor_truth_approx_input}$(ii)$ with instantiations of $\calA := -\rlcM^{-1} \rlcK$, $\eta_0 := 0$ and $\eta_f := \varepsilon_1 \norm{\vec{f}}$, we obtain
\begin{equation}\label{eq:ind0_int1_diff_truey_taylory}
\left(\frac{\sum_{j=0}^m \norm{\vec{x}(t_j) - \widetilde{y}(t_j)}^2}{\sum_{j=0}^m \norm{\vec{x}(t_j)}^2} \right)^{1/2} \leq \frac{\delta}{2} + \frac{4 T \expnorm(\mathbf{-\rlcM^{-1} \rlcK}) \varepsilon_1 \norm{\vec{f}}}{\mu},
\end{equation}
Setting the block-encoding precisions $\varepsilon_1$ as
\begin{equation}\label{eq:ind0_eps1_choice}
\varepsilon_1 = \frac{\delta \cdot \mu}{8 T \expnorm(\mathbf{-\rlcM^{-1} \rlcK}) \norm{\vec{f}}}
\end{equation}
ensures that Eq.~\eqref{eq:ind0_int1_diff_truey_taylory} satisfies
\begin{equation}\label{eq:ind0_int2_diff_truey_taylory}
\left(\frac{\sum_{j=0}^m \norm{\vec{x}(t_j) - \widetilde{x}(t_j)}^2}{\sum_{j=0}^m \norm{\vec{x}(t_j)}^2} \right)^{1/2} \leq \delta,
\end{equation}
Define the history state encoding the true solution $\vec{x}(t)$ denoted by $\ket{\Psi_x}$ and that encoding the Taylor series solution with an input of approximate forcing $\widetilde{x}(t)$ denoted by $\ket{\widetilde{\Psi}_x}$ as
\begin{equation}\label{eq:ind0_inter1_hist_states_truth_taylor}
\ket{\Psi_x} := \frac{1}{\normhist_x} \sum_{j=0}^m \norm{\vec{x}(jh)} \ket{j, x(jh)}, \quad \ket{\widetilde{\Psi}_x} := \frac{1}{\widetilde{\normhist}_x} \sum_{j=0}^m \norm{\widetilde{x}(jh)} \ket{j, \widetilde{x}(jh)}
\end{equation}
Define the vectors $\Psi_x := \normhist_x \ket{\Psi_x}$ and $\widetilde{\Psi}_x := \widetilde{\normhist}_x \ket{\widetilde{\Psi}_x}$, which we will need later. Using Claim~\ref{claim:diff_hist_states} along with Eq.~\eqref{eq:ind0_int2_diff_truey_taylory} gives us that
\begin{equation}\label{eq:inter1_hist_states_close}
\norm{\ket{\Psi_x} - \ket{\widetilde{\Psi}_x}} \leq 2 \delta. 
\end{equation}
The goal is to then produce a state close to $\ket{\widetilde{\Psi}_x}$ so we are then guaranteed to be close to the true history state.

Let $\varepsilon_2, \varepsilon_3 \in (0,1)$ be error parameters to be fixed later. Let $\mathbf{B} = -\rlcM^{-1} \rlcK$. From Claim~\ref{claim:dynamics_BE_ind0}, we obtain a $(2\alpha_K/\sigma,a_M + a_K + 1,\varepsilon_2)$-block-encoding of $\mathbf{B}$, which we denote by $U_B$, for $\varepsilon_M = o(\sigma^2 \varepsilon_2/(2 \alpha_K \log(2 \alpha_K/(\sigma \varepsilon_2)))$ and $\varepsilon_K = \sigma \varepsilon_2/4$. Let $\calL^{-1}$ be the linear operator corresponding to the linear system solve of Remark~\ref{remark:BE_ODE_solver} that would have taken $\Psi_1$ to $\widetilde{\Psi}_x$. Using Remark~\ref{remark:BE_ODE_solver} on $U_B$, we can construct an $(\alpha_3, a_3, \varepsilon_3)$-block-encoding $U_{\calL^{-1}}$ of $\calL^{-1}$ with
\begin{equation}\label{eq:ind0_inv_calL_BE}
\alpha_3 = 4 e \kappa, \quad a_3 = O(a_M + a_K + \log(T \norm{\mathbf{B}})).
\end{equation}
Noting that $\norm{\mathbf{B}} = \norm{\rlcM^{-1} \rlcK} \leq \norm{\rlcM^{-1}} \cdot \norm{\rlcK} \leq \alpha_K/\sigma$, define $\kappa := T \alpha_K \expnorm(\mathbf{B})/\sigma$. Then, this requires setting the precision $\varepsilon_2$ of $U_B$ as
\begin{equation}\label{eq:ind0_eps2_choice}
\varepsilon_2 = o(\varepsilon_3/(\kappa^2 \poly(\calC_f) \log(\kappa/\varepsilon_3)))
\end{equation}
Let $\calL_1 = (\bra{0}^{a_3} \otimes \id) U_{\calL^{-1}} (\ket{0}^{a_3} \otimes \id)$. We then add $a_3$ ancilla qubits initialized as $\ket{0}^{a_3}$ and apply $(\id_{a_1} \otimes U_{\calL^{-1}})$ on to $(\id_{a_3} \otimes W_1) \ket{0}^{a_1 + a_3} \ket{\psi_0}$\footnote{Note that the identity operators act on each others ancilla qubits, which is hard to express properly. This convention is fairly standard~\cite{gilyen2019qsvt}.}, which gives us the state
\begin{equation}\label{eq:ind0_final_state}
(\id_{a_1} \otimes U_{\calL^{-1}})(\id_{a_3} \otimes W_1) \ket{0}^{a_1 + a_3} \ket{\psi_0} = \frac{\normhist_1}{\normhist_0} \ket{0}^{a_1 + a_3} \calL_1 \ket{\psi_1} + \ket{\junk'},
\end{equation}
where $\ket{\junk'}$ is some quantum state that is orthogonal to every state with $\ket{0}^{a_1 + a_3}$ on the first register. Starting from the definition of $U_{\calL^{-1}}$, we note that
$$
\norm{\calL^{-1} - \alpha_3 \calL_1} \leq \varepsilon_3 \implies \norm{\calL^{-1} \ket{\Psi_1} - \alpha_3 \calL_1 \ket{\Psi_1}} \leq \varepsilon_3.
$$
Measuring the state from Eq.~\eqref{eq:ind0_final_state} and post-selecting for $\ket{0}^{a_1 + a_3}$ on the first register gives us the state $\ket{\widehat{\Psi}_x} = \calL_1 \Psi_1/\norm{\calL_1 \Psi_1}$ which satisfies
\begin{equation}\label{eq:ind0_diff_hatpsi_tildepsi}
\norm{\ket{\widehat{\Psi}_x} - \ket{\widetilde{\Psi}_x}} = \norm{\frac{\calL^{-1} \ket{\psi_1}}{\norm{\calL^{-1}\ket{\psi_1}}} - \frac{\alpha_3 \calL_1 \ket{\psi_1}}{\norm{\alpha_3 \calL_1 \ket{\psi_1}}}} \leq 2 \frac{\norm{\calL^{-1} \ket{\psi_1} - \alpha_3 \calL_1 \ket{\psi_1}}}{\norm{\calL^{-1} \ket{\psi_1}}} \leq \varepsilon_3 (2e \sqrt{k})
\end{equation}
where we used the fact that $\widetilde{\Psi}_x = \calL^{-1} \Psi_1$ in the first equality and used that $\spec(\calL^{-1}) \in [1/(2 \sqrt{k} e), 2 e \kappa]$ (see proof of Theorem~\ref{thm:ODE_solver} in Appendix~\ref{appsec:qalgo_ode_la}). Combining Eq.~\eqref{eq:ind0_diff_hatpsi_tildepsi} with Eq.~\eqref{eq:inter1_hist_states_close} gives us that
$$
\norm{\ket{\widehat{\Psi}_x} - \ket{\Psi_x}} \leq \norm{\ket{\widehat{\Psi}_x} - \ket{\widetilde{\Psi}_x}} + \norm{\ket{\widehat{\Psi}_x} - \ket{\widetilde{\Psi}_x}} \leq \frac{\varepsilon_3 (2e \sqrt{k})}{\alpha_3} + 2 \delta.
$$
Setting $\delta = \varepsilon/4$ and $\varepsilon_3 = \varepsilon/(4 e \sqrt{k})$ gives us the desired error of $\varepsilon$. The corresponding probability of success is
$$
p_{\mathrm{succ}} = \frac{\normhist_1^2 \norm{\calL_1 \ket{\psi_1}}^2}{\normhist_0^2} \geq \frac{\normhist_1^2}{\normhist_0^2} \left(\norm{\frac{\calL^{-1}}{\alpha_3} \ket{\psi_1}} - \norm{\calL_1 \ket{\psi_1} - \frac{\calL^{-1}}{\alpha_3}\ket{\psi_1}} \right)^2 \geq \frac{\normhist_1^2}{\normhist_0^2}\left(\frac{1}{8\sqrt{k}e^2 \kappa} - \frac{\varepsilon_3}{4e\kappa} \right)^2,
$$
where we used Eq.~\eqref{eq:ind0_diff_hatpsi_tildepsi} in the third inequality along with the fact that $\alpha_3 = 4e\kappa$ and $\norm{\calL^{-1}} \geq 1/(2\sqrt{k}e)$. Setting $\varepsilon_3 = \varepsilon/(4e\sqrt{k}) \leq 1/(4e\sqrt{k})$ ensures that
$$
p_{\mathrm{succ}} \geq \Omega\left(\frac{\normhist_1^2}{\normhist_0^2} \frac{1}{k \kappa^2} \right).
$$
Using Eq.~\eqref{eq:ind0_approx_input_encodings} and the fact that $\alpha_1 \geq \norm{\rlcM^{-1}}$, it can be shown that $\normhist_1/\normhist_0 \geq \Omega\Big(\|\rlcM^{-1} \vec{f}\|/(\alpha_1 \|\vec{f}\|)\Big) = \Omega(\sigma/\alpha_M)$ as long as $\varepsilon_1 \leq O(\sigma)$ (in addition to the condition of Eq.~\eqref{eq:ind0_eps1_choice}). This then gives us a probability of success of
\begin{equation}\label{eq:ind0_final_psucc}
p_{\mathrm{succ}} \geq \Omega\left(\sigma^2/(k \alpha_M^2 \kappa^2)\right).    
\end{equation}
This requires setting $\varepsilon_2$ (Eq.~\eqref{eq:ind0_eps2_choice}) and $\varepsilon_1$ (Eq.~\eqref{eq:ind0_eps1_choice2})
$$
\varepsilon_2 = o(\varepsilon/(\kappa^2 \poly(\calC_f) \log(\kappa/\varepsilon))), \quad \varepsilon_1 \leq  O\left(\frac{\sigma \delta \varepsilon}{T \expnorm(\mathbf{-\rlcM^{-1} \rlcK}) \norm{\vec{f}}} \right)
$$
The above conditions on $\varepsilon_1$ and $\varepsilon_2$ impose a restriction on the choice of $\varepsilon_M$ and $\varepsilon_K$ since $\varepsilon_M = o(\sigma^2 \varepsilon_2/(2 \alpha_K \log(2 \alpha_K/(\sigma \varepsilon_2)))$ and $\varepsilon_K = \sigma \varepsilon_2/4$. A sufficient condition is
$$
\varepsilon_M = \poly\left(\frac{\sigma \varepsilon}{\alpha_K T \expnorm(-\rlcM^{-1} \rlcK) \calC_f}\right), \quad \varepsilon_K = \poly\left(\frac{\varepsilon \sigma}{\alpha_K T \expnorm(-\rlcM^{-1} \rlcK) \calC_f} \right)
$$
The contribution to the time complexity is the implementation of $U_{M^{-1}}$ for $W_1$ and the dynamics $U_{\calL^{-1}}$. The gate complexity of $U_{\calL^{-1}}$ follow from Remark~\ref{remark:BE_ODE_solver}
$$
T_{\calL^{-1}} = O\Big(\kappa \cdot \poly\left(\calC_f \log(\kappa/\varepsilon) \right) \Big) \text{ uses of } U_B.
$$
Using Claim~\ref{claim:ode_BE_ind1}, we have that each query to $U_B$ requires 
$$
O\left((\alpha_M/\sigma) \log(\alpha_K/(\sigma \varepsilon_2)) \right) = \widetilde{O}\left(\kappa_M \log\left(\frac{T \alpha_K \expnorm(-\rlcM^{-1}\rlcK)}{\sigma \varepsilon}\right) \right)
$$ 
queries to $U_M$ and $O(1)$ query to $U_K$. Using Claim~\ref{claim:inv_M_BE_ind0}, we have that each query to $U_{M^{-1}}$ requires 
$$
O\left((\alpha_M/\sigma) \log(1/(\sigma \varepsilon_1))\right) = O\left(\kappa_M \log \frac{T \expnorm(-\rlcM^{-1}\rlcK)}{\sigma \varepsilon} )\right)
$$ 
queries to $U_M$ and using the specification of $\varepsilon_1$ here.

To boost the success probability from that of Eq.~\eqref{eq:ind0_final_psucc}, we use amplitude amplification (Theorem~\ref{thm:amplitude_amplification}) that involves applying $O(1/\sqrt{p_{\mathrm{succ}}})$ many iterations of the circuit of $U_{\calL^{-1}} W_1 U_0$. The overall gate complexity is then as stated. The number of queries to $O_x, O_f$ is then $O(1/\sqrt{p_{\mathrm{succ}}})$ which is given by
$$
\widetilde{O}\left(T \alpha_K \kappa_M \expnorm(-\rlcM^{-1} \rlcK) \cdot \log(1/\varepsilon)\right).
$$
This completes the proof of the theorem.
\end{proof}

\paragraph{Final time preparation.}
We can also use the algorithm of Theorem~\ref{thm:sim_ind0} to obtain a quantum state that is close to the final time state $\ket{x(T)}$ with slighly modified parameters. This is formally described below.
\begin{corollary}\label{corr:sim_T_ind0}
Let $\varepsilon, \sigma \in (0,1)$. Consider the problem setup of Definition~\ref{def:prob_setup_DAEs} and suppose that the $\dae$ has index $0$. Let the minimum singular value of $\rlcM$ satisfy $\sigma_\min(\rlcM) \geq \sigma$. Then, there exists a quantum algorithm that outputs a state $|\widehat{x}(T)\rangle$ with success probability $\geq 2/3$ such that $\left\| |\widehat{x}(T)\rangle - |x(T)\rangle \right\| \leq \varepsilon$. Define
$$
\kappa_M := \alpha_M/\sigma, \quad  g := \frac{\max_{t\in[0,T]} \norm{\vec{x}(t)}}{\norm{\vec{x}(T)}}, \quad \mathcal{C}_f := \log\left(1 + \frac{T e^2 \|\vec{f}\|}{\sigma \|\vec{x}(T)\|}\right).
$$
The algorithm uses the following complexity 
\begin{align*}
\text{Queries to }U_{M}&: \widetilde{O}\left(g T \, \expnorm(-\rlcM^{-1}\rlcK) \cdot \poly\Big(\kappa_M, \alpha_K, \calC_f, \log(1/\varepsilon) \Big) \right) \,,\\
\text{Queries to }U_{K}&: \widetilde{O}\left(g T \, \expnorm(-\rlcM^{-1}\rlcK) \cdot \poly\Big(\kappa_M, \alpha_K, \calC_f, \log(1/\varepsilon) \Big)\right) \,,\\ 
\text{Queries to }O_{f},O_x&: \widetilde{O}\left(g T \alpha_K \kappa_M \expnorm(-\rlcM^{-1} \rlcK) \cdot \log(1/\varepsilon)\right) \,.
\end{align*}
An additional gate complexity of $O(1)$ times the number of uses of $U_M$ and $U_K$ is also required. It is sufficient to set the precisions of $U_M$ and $U_K$ as
$$
\varepsilon_M = \poly\left(\frac{\sigma \varepsilon}{\alpha_K T \expnorm(-\rlcM^{-1} \rlcK) \calC_f }\right) \text{ and } \varepsilon_K = \poly\left(\frac{\sigma \varepsilon}{\alpha_K T \expnorm(-\rlcM^{-1} \rlcK) \calC_f }\right).
$$
\end{corollary}
\begin{proof}
The proof proceeds along similar lines that of Theorem~\ref{thm:sim_ind0}. We will use the same notation as there and only highlight the differences. As in the case there, we start with the state $\ket{0}^{a_1 + a_3} \ket{\psi_0}$ (Eq.~\eqref{eq:ind0_interim_init0}, and then apply $(\id_{a_3} \otimes W_1)$ (Eq.~\eqref{eq:ind0_interim_W1}) to obtain a state of the form that takes the form of Eq.~\eqref{eq:ind0_interim_W1psi0} which encodes the exact initial state $\vec{x}_0$ and the approximate forcing $\widetilde{g}$ satisfying Eq.~\eqref{eq:ind0_approx_input_encodings}.

We choose the number of terms $k$ to retain in the Taylor series differently here. Let $\delta \in (0,1)$ be an error parameter to be fixed later and define
\begin{equation}\label{eq:ind0_proj_ODE_taylor_steps1}
k = \left\lceil \frac{2\log \Omega}{\log \log \Omega} \right\rceil, \quad \text{ where } \Omega = \frac{2 m e^3}{\delta} \left(1 + \frac{T e^2 \norm{\vec{f}}}{\sigma \norm{\vec{x}(T)}} \right),
\end{equation}
This choice ensures that $(k+1)! \geq \Omega$. Using Claim~\ref{claim:diff_taylor_truth_approx_input}$(i)$ with instantiations of $\calA := -\rlcM^{-1} \rlcK$, $\eta_0 := 0$ and $\eta_f := \varepsilon_1 \norm{\vec{f}}$, we obtain
\begin{equation}\label{eq:ind0_int1_diff_truey_taylory1}
\norm{\vec{x}(T) - \widetilde{x}(T)} \leq \frac{\delta}{2} \norm{\vec{x}(T)} + 4 T  \expnorm(\mathbf{-\rlcM^{-1} \rlcK}) \varepsilon_1 \norm{\vec{f}}.
\end{equation}
Setting the block-encoding precision $\varepsilon_1$ as
\begin{equation}\label{eq:ind0_eps1_choice2}
\varepsilon_1 = \frac{\delta \cdot \mu}{8 T \expnorm(\mathbf{-\rlcM^{-1} \rlcK}) \norm{\vec{f}}}
\end{equation}
ensures that Eq.~\eqref{eq:ind0_int1_diff_truey_taylory1} satisfies
\begin{equation}\label{eq:inter1_hist_states_close1}
\norm{\vec{x}(T) - \widetilde{x}(T)} \leq \delta \norm{\vec{x}(T)} \implies \norm{\ket{x(T)} - \ket{\widetilde{x}(T)}} \leq 2\delta.
\end{equation}
The goal is to then produce a state close to $\ket{\widetilde{x}(T)}$ so we are then guaranteed to be close to the true history state. For this, we proceed similarly as in the proof of Theorem~\ref{thm:sim_ind0} but use the block-encoding $U_{\calL^{-1}}$ of $\calL^{-1}$ from Theorem~\ref{thm:ODE_solver_final_state}. It can then be checked that 
$$
\norm{\ket{x(T)} - \ket{\widehat{x}(T)}} \leq \varepsilon,
$$
by choosing the precision $\varepsilon_3$ of $U_{\calL^{-1}}$ as $\varepsilon_3 = O(\varepsilon/\sqrt{k})$. The corresponding probability of success is given by $p_{\mathrm{succ}} \geq \Omega(\sigma^2/(g^2 \alpha_M^2))$ where the factor of $\sigma/\alpha_M$ is from $\normhist_1/\normhist_0$ and the factor $g$ is as defined in the theorem statement (see also proof of Theorem~\ref{thm:ODE_solver_final_state}, Appendix~\ref{appsec:qalgo_ode_la} for how $g$ enters). Overall, we need to set
$$
\varepsilon_1 =  O\left(\frac{\sigma \delta \varepsilon}{T \expnorm(\mathbf{-\rlcM^{-1} \rlcK}) \norm{\vec{f}}} \right), \quad \varepsilon_M = \varepsilon_K = \poly\left(\frac{\sigma \varepsilon}{\alpha_K T \expnorm(-\rlcM^{-1} \rlcK) \calC_f}\right).
$$
The main contribution to the time complexity is the implementation of $U_{M^{-1}}$ for $W_1$ and the dynamics $U_{\calL^{-1}}$. The gate complexity of $U_{\calL^{-1}}$ follow from Theorem~\ref{thm:ODE_solver_final_state}. To boost the success probability from $\Omega(1/(g^2 \kappa_M^2))$ to $\Omega(1)$, we use amplitude amplification (Theorem~\ref{thm:amplitude_amplification}) that involves applying $O(g \kappa_M)$ many iterations of the circuit of $U_{\calL^{-1}} W_1 U_0$. The overall gate complexity is then as stated. This completes the proof.
\end{proof}

\subsubsection{Index one}\label{subsec:ind1_dae_solver}
For index $1$, we have the following main result.
\begin{theorem}\label{thm:sim_ind1}
Let $\varepsilon, \sigma, \sigma_1 \in (0,1)$. Consider the problem setup of Definition~\ref{def:prob_setup_DAEs} and suppose that the $\dae$ has index $1$. Let the minimum non-zero singular value of $\rlcM$ satisfy $\sigma_\min(\rlcM) \geq \sigma$ and minimum non-zero singular value of $\rlcM_1$ satisfy $\sigma_\min(\rlcM_1) \geq \sigma_1$. Define the true history state encoding the evolution of $\vec{x}(t), \, t \in [0,T]$ as 
$$
\ket{\Psi} = \normhist^{-1} \sum_{k=0}^{m} \norm{\vec{x}(k \Delta t)} \ket{k} \ket{x(k \Delta t)}, 
$$
where $\Delta t = O(\sigma_1 / \alpha_K)$ and $m = \ceil{T/\Delta t}$. Then, there exists a quantum algorithm that outputs a state $\ket{\widehat{\Psi}}$ such that $\left\| \ket{\widehat{\Psi}} - \ket{\Psi} \right\| \leq \varepsilon$ with probability $\geq 2/3$. Define 
$$
\expnorm(-\rlcP_0 \rlcM_1^{-1}\rlcK) := \sup_{t \in [0,T]} \norm{\exp(-\rlcP_0 \rlcM_1^{-1}\rlcK t)}, \quad \mu^2 = m^{-1} \sum_{j=1}^m \norm{\vec{x}(j)}^2, \quad  \calC_f := \log\left(1 + \frac{T e^2 f_{\max} }{\sigma_1 \mu}\right).
$$
The algorithm uses the following complexity 
\begin{align*}
\text{Queries to }U_{M}&: \widetilde{O}\Big(T^2 \expnorm(- \rlcP_0 \rlcM_1^{-1}\rlcK)^2 \cdot \poly\left(1/\mu, \kappa_M, \kappa_{M_1}, \alpha_K, \log(1/\varepsilon), \mathcal{C}_f)\right) \Big) \,,\\
\text{Queries to }U_{K}&: \widetilde{O}\Big(T^2 \expnorm(- \rlcP_0 \rlcM_1^{-1}\rlcK)^2 \cdot \poly\left(1/\mu, \kappa_M, \kappa_{M_1}, \alpha_K, \log(1/\varepsilon), \mathcal{C}_f)\right) \Big) \,,\\ 
\text{Queries to }O_{f},O_x&: O\left(T \alpha_K \expnorm(- \rlcP_0 \rlcM_1^{-1}\rlcK)/(\sigma_1 \mu) \cdot \log(1/\varepsilon)\right) \,,
\end{align*}
where $\kappa_M = \alpha_M/\sigma$ is an upper bound on the (effective) condition number of $\rlcM$, and $\kappa_{M_1} = (\alpha_M + \alpha_K)/\sigma_1$ is an upper bound on the condition number of $\rlcM_1$. An additional gate complexity $O(1)$ times the number of uses of $U_M$ and $U_K$ is also used. It is sufficient to set the precisions of the block-encodings of $U_M$ and $U_K$ as
$$
\varepsilon_M = \poly\left(\frac{\sigma \sigma_1 \varepsilon}{\alpha_K T \expnorm(-\rlcP_0 \rlcM_1^{-1}\rlcK) \calC_f }\right) \text{ and } \varepsilon_K = \poly\left(\frac{\sigma \sigma_1 \varepsilon}{\alpha_K T \expnorm(-\rlcP_0 \rlcM_1^{-1}\rlcK) \calC_f }\right).
$$
\end{theorem}

\begin{myalgorithm}
\begin{algorithm}[H]
    \caption{Quantum $\dae$ solver for index one} \label{algo:QDAE_solver_ind1}
    \setlength{\baselineskip}{1.8em} 
    \DontPrintSemicolon 
    \KwInput{$(\alpha_M,a_M,\varepsilon_M)$-BE $U_M$ of $\rlcM$, $(\alpha_K,a_K,\varepsilon_K)$-BE $U_K$ or $\rlcK$, oracles $O_x, O_f$ (see Definition~\ref{def:prob_setup_DAEs}), parameter $\sigma_1$ s.t. $\sigma_\min(\rlcM_1) \geq \sigma_1$.}
    \KwOutput{$\ket{\widehat{\Psi}}$ that is $\varepsilon$-close to the true history state $\ket{\Psi}$ (Theorem~\ref{thm:sim_ind1})} \vspace{2mm}
    Obtain $(\alpha_{P_0}, a_1,\varepsilon_1)$-BE $U_{P_0}$ of $\rlcP_0$ and $(\alpha_{Q_0},a_3,\varepsilon_2)$-BE $U_{Q_0}$ of $\rlcQ_0$ from Claim~\ref{claim:projs_BE_ind1} \\
    Obtain $(\alpha_{P_0 M_1^{-1}},a_2,\varepsilon_1')$-BE $U_{P_0 M_1^{-1}}$ of $\rlcP_0 \rlcM_1^{-1}$ (Claim~\ref{claim:ode_BE_ind1}) \\
    Obtain $(\alpha_{Q_0 M_1^{-1}},a_4,\varepsilon_2')$-BE $U_{Q_0 M_1^{-1}}$ of $\rlcQ_0 \rlcM_1^{-1}$ (Claim~\ref{claim:lin_alg_BE_ind1}) \\
    Obtain $(\alpha_{\calL_y^{-1}}, a_5, \varepsilon_3)$-BE $U_{\calL_y^{-1}}$ of linear system solve corresponding to the block-encoding $U_{P_0 M_1^{-1} K}$ of $\rlcP_0 \rlcM_1 \rlcK$ (Claim~\ref{claim:ode_BE_ind1}) using Remark~\ref{remark:BE_ODE_solver}. \label{algo_step:set_BE_ind_Ly} \\
    Obtain $(\alpha_F,a_6,\varepsilon_4)$-BE $U_{F}$ of $\mathbf{F} = (\id - \rlcQ_0 \rlcM_1^{-1}\rlcK)$ from Claim~\ref{claim:lin_alg_BE_ind1} \\ \vspace{2mm}
    \Comment*[l]{Define parameters for solve}
    Set $h = O(\sigma_1/\alpha_K)$ and $m = \ceil{T/h}$ \\
    Set $a= \sum_{i \in [6]}(a_i)$ \\
    Set $\normhist_0$ as in Eq.~\eqref{eq:algo_ind1_psi0} and $\normhist_0'$ as in Eq.~\eqref{eq:algo_ind1_phi0} \\
    Set $\beta_1 = \alpha_{\calL_y^{-1}} \alpha_{F} \normhist_0$, $\beta_2 = \normhist_0'$, and $\beta = \sqrt{\beta_1^2 + \beta_2^2}$ \label{algo_step:ind1_params_beta} \\ 
    \vspace{2mm}
    \Comment*[l]{Construct input state}
    Obtain the unitary $U_{1}$ that prepares $\ket{\psi_0}$ (Eq.~\eqref{eq:algo_ind1_psi0}) from Claim~\ref{claim:initial_state_prep} \\
    Obtain the unitary $U_{2}$ that prepares $\ket{\phi_0}$ (Eq.~\eqref{eq:algo_ind1_phi0}) from Claim~\ref{claim:initial_state_prep} \\ 
    Initialize state $\ket{0}^a \ket{0,0}_{\calS}$ \\
    Apply a rotation $R$ on the first qubit in $\calS$ to get $\ket{0}^a \left[\frac{\beta_1}{\beta} \ket{0,0} + \frac{\beta_2}{\beta} \ket{1,0} \right]$ \\
    Apply $\id_a \otimes (\ket{0}\bra{0} \otimes U_1 + \ket{1}\bra{1} \otimes U_2)$
    to obtain $\ket{\Psi_0} = \ket{0}^a \left[\frac{\beta_1}{\beta} \ket{0} \ket{\psi_0} + \frac{\beta_2}{\beta} \ket{1} \ket{\phi_0} \right]$ \label{algo_step:ind1_Psi0} \\
    Apply $W_1 = \ket{0}\bra{0} \otimes [\ket{0}\bra{0} \otimes U_{P_0} + (\id - \ket{0}\bra{0})\otimes U_{P_0 M_1^{-1}}] + \ket{1}\bra{1} \otimes \id$ to $\ket{\Psi_0}$ \\
    Apply $W_2 = \ket{0}\bra{0} \otimes \id + \ket{1}\bra{1} \otimes [\ket{0}\bra{0} \otimes U_{Q_0} + (\id - \ket{0}\bra{0})\otimes U_{Q_0 M_1^{-1}}]$ to $W_1 \ket{\Psi_0}$ \\
    \vspace{2mm}
    \Comment*[l]{Step 1: Differential equation solve on constraint submanifold}
    Apply $W_3 = \ket{0}\bra{0} \otimes U_{\calL_y^{-1}} + \ket{1}\bra{1} \otimes \id$ to $W_2 W_1 \ket{\Psi_0}$ \label{algo_step:ind1_ode_solve} \\ \vspace{2mm}
    \Comment*[l]{Step 2: Linear algebraic solve}
    Apply $W_4 = \ket{0}\bra{0} \otimes \Big(\ket{0}\bra{0} \otimes \id + (\id - \ket{0}\bra{0}) \otimes U_{F} \Big) + \ket{1}\bra{1} \otimes \id$ to $W_3 W_2 W_1 \ket{\Psi_0}$ \label{algo_step:ind1_la_solve} \\
    Apply $\id_a \otimes H \otimes I$ to $W_4 W_3 W_2 W_1 \ket{\Psi_0}$ \label{algo_step:ind1_add_Psiw_Phiz0} \\
    \Return $\ket{\widehat{\Psi}}$ after post-selecting for $\ket{0}^a$ over the first register.
\end{algorithm} 
\end{myalgorithm}

\paragraph{Algorithm.}
We will use Algorithm~\ref{algo:QDAE_solver_ind1} that follows the approach discussed in Section~\ref{subsec:approach} to obtain the above result. Let us introduce some relevant notation.

For the desired input state to the $\dae$ solver, we need to define separate input states corresponding to the variables $\vec{y}$ and $\vec{z}$. For $\vec{y}$, we define the following input state $\ket{\psi_0}$ with $U_1$ as its corresponding state-preparation unitary (which can be obtained from Claim~\ref{claim:initial_state_prep})
\begin{equation}\label{eq:algo_ind1_psi0}
\ket{\psi_0} = U_1 \ket{0,0,0} = \frac{1}{\normhist_0} \left[\alpha_{P_0} \norm{\vec{x}_0} \ket{0,0,\vec{x}_0} + h \alpha_{P_0 M_1^{-1}} \norm{\vec{f}} \sum_{j=0}^{m-1} \ket{j,1,f} \right],    
\end{equation}
where $\alpha_{P_0}$ is the subnormalization of the block-encoding $U_{P_0}$ of $\rlcP_0$ from Claim~\ref{claim:projs_BE_ind1}, $\alpha_{P_0 M_1^{-1}}$ is the subnormalization of the block-encoding $U_{P_0 M_1^{-1}}$ of $\rlcP_0 \rlcM_1^{-1}$ from Claim~\ref{claim:ode_BE_ind1}, and $\normhist_0 = \sqrt{\alpha_{P_0}^2 \norm{\vec{x}(0)}^2 + m h^2 \alpha_{P_0 M_1^{-1}}^2 \norm{\vec{f}}^2}$.

For $\vec{z}$, we define the following input state $\ket{\phi_0}$ with $U_2$ as its corresponding state-preparation unitary (which can be obtained from Claim~\ref{claim:initial_state_prep})
\begin{equation}\label{eq:algo_ind1_phi0}
\ket{\phi_0} = U_2 \ket{0,0,0} = \frac{1}{\normhist'_0} \left[\alpha_{Q_0} \norm{\vec{x}_0} \ket{0,0,\vec{x}_0} + \alpha_{Q_0 M_1^{-1}} \norm{\vec{f}} \sum_{j=1}^{m} \ket{j,0,f} \right],    
\end{equation}
where $\alpha_{Q_0}$ is the subnormalization of the block-encoding $U_{Q_0}$ of $\rlcQ_0$ from Claim~\ref{claim:projs_BE_ind1}, $\alpha_{Q_0 M_1^{-1}}$ is the subnormalization of the block-encoding $U_{Q_0 M_1^{-1}}$ of $\rlcQ_0 \rlcM_1^{-1}$ from Claim~\ref{claim:lin_alg_BE_ind1}, and $\normhist_0' = \sqrt{\alpha_{Q_0}^2 \norm{\vec{x}(0)}^2 + m \alpha_{Q_0 M_1^{-1}}^2 \norm{\vec{f}}^2}$.

\paragraph{Correctness.} To prove Theorem~\ref{thm:sim_ind1} that gives the correctness and complexity of Algorithm~\ref{algo:QDAE_solver_ind1}, we will prove a series of claims regarding the different steps in the algorithm.

Consider the following state (step~\ref{algo_step:ind1_Psi0} from Algorithm~\ref{algo:QDAE_solver_ind1})
\begin{equation}\label{eq:ind1_Psi0}
    \ket{\Psi_0} := \ket{0}^a \left[\frac{\beta_1}{\beta} \ket{0} \ket{\psi_0} + \frac{\beta_2}{\beta} \ket{1} \ket{\phi_0} \right],
\end{equation}
where $\ket{\psi_0}$ is as defined in Eq.~\eqref{def:ind1_psi0}, $\ket{\phi_0}$ is as defined in Eq.~\eqref{def:ind1_phi0}, $\beta_1 = \alpha_{\calL_y} \alpha_{F} \normhist_0$, $\beta_2 = \normhist_0'$, and $\beta = \sqrt{\beta_1^2 + \beta_2^2}$. Recall that $\normhist_0 = \sqrt{\alpha_{P_0}^2 \norm{\vec{x}(0)}^2 + m h^2 \alpha_{P_0 M_1^{-1}}^2 \norm{\vec{f}}^2}$ and $\normhist_0' = \sqrt{\alpha_{Q_0}^2 \norm{\vec{x}(0)}^2 + m \alpha_{Q_0 M_1^{-1}}^2 \norm{\vec{f}}^2}$.

$\ket{\Psi_0}$ is the initial state we create before applying a sequence of block-encodings on it. This state can be constructed exactly using Claim~\ref{claim:initial_state_prep} and requires $O(1)$ calls to $O_x,O_f$ and $\polylog(m)$ additional gates. It will also be convenient to define the corresponding state vectors
\begin{equation}\label{eq:ind1_Psi0_vec}
    \Psi_0 = \ket{0}^a \left[\frac{\beta_1}{\beta} \ket{0} \otimes \psi_0 + \frac{\beta_2}{\beta} \ket{1} \otimes \phi_0 \right],
\end{equation}
where 
\begin{align}
\psi_0 &= \alpha_{P_0} \ket{0,0} \otimes \vec{x}(0) + h \alpha_{P_0 M_1^{-1}} \sum_{j=0}^{m-1} \ket{j,1} \otimes \vec{f},  \\
\phi_0 &= \alpha_{Q_0} \ket{0,0} \otimes \vec{x}(0) + \alpha_{Q_0 M_1^{-1}} \sum_{j=0}^{m-1} \ket{j,1} \otimes \vec{f}.
\end{align}
Note that $\ket{\psi_0} = \psi_0 /\normhist_0$, $\ket{\phi_0} = \phi_0/\normhist_0'$, and $\ket{\Psi_0} = \Psi_0 / \norm{\Psi_0}$. 

We now comment on the state prepared as input to the $\ode$ solve in step~\ref{algo_step:ind1_ode_solve}.
\begin{claim}\label{claim:ind1_ode_input_state}
In Algorithm~\ref{algo:QDAE_solver_ind1}, the state $W_2 W_1 \ket{\Psi_0}$ is of the form
$$
W_2 W_1 \ket{0}^a \ket{\Psi_0} = \ket{0}^a \left[\frac{\normhist_1 \beta_1}{\normhist_0 \beta} \ket{0} \ket{\psi_1} + \frac{\normhist_1' \beta_2}{\normhist_0' \beta} \ket{1} \ket{\phi_1} \right] + \ket{\junk},
$$
where $\ket{\junk}$ is a state orthogonal to every state with $\ket{0}^a$ over the ancilla qubits and
$$
\ket{\psi_1} := \frac{1}{\normhist_1} \left[\norm{\widetilde{y}_0} \ket{0,0, \widetilde{y}_0} + h \norm{\widetilde{f}_y} \sum_{j=0}^{m-1} \ket{j,1,\widetilde{f}_y} \right], \ket{\phi_1} := \frac{1}{\normhist_1'} \left[\norm{\widetilde{z}_0} \ket{0,0, \widetilde{z}_0} + \norm{\widetilde{f}_z} \sum_{j=0}^{m-1} \ket{j,1,\widetilde{f}_z} \right],
$$
where $\widetilde{y}_0$ satisfies $\norm{\widetilde{y}_0 - \rlcP_0 \vec{x}_0} \leq \varepsilon_1 \norm{\vec{x}_0}$, $\widetilde{f}_y$ satisfies $\norm{\widetilde{f}_y - \rlcP_0 \rlcM_1^{-1} \vec{f}} \leq \varepsilon_1' \norm{\vec{f}}/\sigma_1$, $\widetilde{z}_0$ satisfies $\norm{\widetilde{z}_0 - \rlcQ_0 \vec{x}_0} \leq \varepsilon_2 \norm{\vec{x}_0}$, and $\widetilde{f}_z$ satisfies $\norm{\widetilde{f}_z - \rlcQ_0 \rlcM_1^{-1} \vec{f}} \leq \varepsilon_2' \norm{\vec{f}}/\sigma_1$.
\end{claim}
\begin{proof}
Consider the following state (step~\ref{algo_step:ind1_Psi0} from Algorithm~\ref{algo:QDAE_solver_ind1})
\begin{equation}\label{eq:ind1_Psi0}
    \ket{\Psi_0} := \ket{0}^a \left[\frac{\beta_1}{\beta} \ket{0} \ket{\psi_0} + \frac{\beta_2}{\beta} \ket{1} \ket{\phi_0} \right].
\end{equation}
From Lemma~\ref{lem:projected_ODE_initial_state}, we have that $W_1 \ket{0}^a \ket{\Psi_0}$ is of the following form 
$$
W_1 \ket{0}^a \ket{\Psi_0} = \ket{0}^{a} \ket{0} \frac{\beta_1}{\normhist_0 \beta} \underbrace{\left[\norm{\widetilde{y}_0} \ket{0,0, \widetilde{y}_0} + h \norm{\widetilde{f}_y} \sum_{j=0}^{m-1} \ket{j,1,\widetilde{f}_y} \right]}_{:= \psi_1} + \ket{0}^a \frac{\beta_2}{\normhist_0' \beta} \ket{1} \ket{\phi_0}  + \ket{\junk_0},
$$
where $\ket{\junk_0}$ is some state that is orthogonal to every state with $\ket{0}^a$ over the ancilla register and we will denote the vector $\psi_1$ as indicated. Then, $\normhist_1 = \norm{\psi_1}$. We can then define $\ket{\psi_1} := \psi_1/\normhist_1$ Let the true initial state and forcing over $\vec{y}$ be denoted as $\vec{y}_0 := \rlcP_0 \vec{x}_0$ and $\vec{f}_y := \rlcP_0 \rlcM_1^{-1} \vec{f}$, respectively. Then, the encoded initial condition and forcing of $\ket{\psi_1}$, which we denote by $\widetilde{y}_0$ and $\widetilde{f}_y$ respectively, satisfy
\begin{equation}\label{eq:interim_ind1_init_y_close}
\norm{\widetilde{y}_0 - \vec{y}_0} \leq \varepsilon_1 \norm{\vec{y}_0} \leq \varepsilon_1 \norm{\vec{x}_{0}}, \quad \norm{\widetilde{f}_y - \vec{f}_y} \leq \varepsilon_1' \norm{\vec{f}_y} \leq \varepsilon_1' \norm{\vec{f}}/\sigma_1,
\end{equation}
where we have used $\norm{\vec{y}_0} = \norm{\rlcP_0 \vec{x}_0} \leq \norm{\rlcP_0} \norm{\vec{x}_0} \leq \norm{\vec{x}(0)}$ and $\norm{\vec{f}_y} = \norm{\rlcP_0 \rlcM_1^{-1} \vec{f}} \leq \norm{\rlcP_0} \norm{\rlcM_1^{-1}} \norm{\vec{f}} \leq \norm{\vec{f}}/\sigma_1$.

Applying Lemma~\ref{lem:projected_ODE_initial_state} again to $W_2 W_1 \ket{0}^a \ket{\Psi_0}$ then gives us that
$$
W_1 \ket{0}^a \ket{\Psi_0} = \ket{0}^{a} \ket{0} \frac{\normhist_1 \beta_1}{\normhist_0 \beta} \ket{\psi_1} + \ket{0}^a \ket{1} \frac{\beta_2}{\normhist_0' \beta} \underbrace{\left[\norm{\widetilde{z}_0} \ket{0,0, \widetilde{z}_0} + h \norm{\widetilde{f}_z} \sum_{j=0}^{m-1} \ket{j,1,\widetilde{f}_z} \right]}_{:= \phi_1} + \ket{\junk},
$$
where $\ket{\junk}$ is some state that is orthogonal to every state with $\ket{0}^a$ over the ancilla register and the vector $\phi_1$ is as indicated. We define $\normhist_1' = \norm{\phi_1}$ and $\ket{\phi_1} := \phi_1/\normhist_1'$. Let the true initial state and forcing to the linear-algebraic solve over $\vec{z}$ be denoted as $\vec{z}_0 := \rlcQ_0 \vec{x}_0$ and $\vec{f}_z := \rlcQ_0 \rlcM_1^{-1} \vec{f}$, respectively. Then, the encoded initial condition and forcing of $\ket{\phi_1}$, which we denote by $\widetilde{z}_0$ and $\widetilde{f}_z$ respectively, satisfies
\begin{equation}\label{eq:interim_ind1_init_z_close}
\norm{\widetilde{z}_0 - \vec{z}_0} \leq \varepsilon_2 \norm{\vec{z}_0} \leq \varepsilon_2 \norm{\vec{x}_0}, \quad \norm{\widetilde{f}_z - \vec{f}_z} \leq \varepsilon_2' \norm{\vec{f}_z} \leq \varepsilon_2' \norm{\vec{f}}/\sigma_1.
\end{equation}
This completes the proof.
\end{proof}

We now comment on the error between the true history state and the final history state that would have been obtained if we were to consider the input state $W_2 W_1 \ket{\Psi_0}$ to steps~\ref{algo_step:ind1_ode_solve}--\ref{algo_step:ind1_add_Psiw_Phiz0} and the block-encodings of steps~\ref{algo_step:ind1_ode_solve}--\ref{algo_step:ind1_add_Psiw_Phiz0} were exact. For that, let the true history state corresponding to the true solution $\vec{x}$ be $\ket{\Psi_x}$ which is given by
\begin{equation}\label{eq:ind1_Psix}
\ket{\Psi_x} := \frac{1}{\normhist} \sum_{j=0}^m \norm{\vec{x}(jh)} \ket{j, x(jh)}
\end{equation}

Let us denote $\ket{\widetilde{\Psi}_x}$ as the history state obtained from the Taylor series approximation $\widetilde{x}$ over $k$ terms considering the input state $W_2 W_1 \ket{\Psi_0}$ but exact implementations of block-encodings in $W_3, W_4$ i.e., $(\alpha_{\calL^{-1}_y},a_5,0)$-BE of $\calL_y^{-1}$ and $(\alpha_F,a_6,0)$-BE of $\mathbf{F} = (\id - \rlcQ_0 \rlcM_1^{-1} \rlcK)$ be
\begin{equation}
\ket{\widetilde{\Psi}_x} := \frac{1}{\widetilde{\normhist}} \sum_{j=0}^m \norm{\widetilde{x}(jh)} \ket{j, \widetilde{x}(jh)}.
\end{equation}

\begin{claim}\label{claim:ind1_whatif_exactBE_ODE_LA}
Let $\delta' \in (0,1)$. Consider Algorithm~\ref{algo:QDAE_solver_ind1}. Suppose $W_3$ in step~\ref{algo_step:ind1_ode_solve} is defined with an $(\alpha_{\calL^{-1}_y},a_5,0)$-BE of $\calL_y^{-1}$ and $W_4$ in step~\ref{algo_step:ind1_la_solve} is defined with an $(\alpha_F,a_6,0)$-BE of $\mathbf{F}$. If $\calL_y$ is defined considering a Taylor approximation over $k$ many terms where 
$$
(k+1)! \geq \frac{16 \xi m e^3}{\delta'} \left(1 + \frac{T e^2 \norm{\vec{f}}}{\sigma_1 \mu} \right), \text{ where } \xi = (1 + \sigma_1/\alpha_K),
$$
then the state obtained from $(H\otimes \id)W_4 W_3 W_2 W_1 \ket{\Psi_0}$ is  
$$
\ket{0}^a \frac{\norm{\widetilde{\Psi}_x}}{\beta} \ket{\widetilde{\Psi}_x} + \ket{\junk}, \text{ where } \norm{\ket{\Psi_x} - \ket{\widetilde{\Psi}_x}} \leq \delta'.
$$
This is provided that the block-encoding precisions are set as follows for $\calA := -\rlcP_0 \rlcM_1^{-1}\rlcK$:
$$
\varepsilon_1 = \frac{\delta' \mu}{16 \xi (1 + 2\expnorm(\calA)) \norm{\vec{x}_0}}, \quad \varepsilon_1' = \frac{\delta' \sigma_1  \mu}{64 \xi T \expnorm(\calA) \norm{\vec{f}}},  \quad \varepsilon_2 = \frac{\delta'}{4}, \quad \varepsilon_2' = \frac{\delta' \sigma_1 \mu}{4\norm{\vec{f}}}.
$$
\end{claim}
\begin{proof}
Let $h = O(\sigma_1/\alpha_K)$ be as defined in Algorithm~\ref{algo:QDAE_solver_ind1} and denote $t_j = jh, \forall j \in [m]_0$. From Claim~\ref{claim:ind1_ode_input_state}, we know that the state $W_2 W_1 \ket{\Psi_0}$ is of the form
$$
W_2 W_1 \ket{\Psi_0} = \ket{0}^{a} \ket{0} \frac{\beta_1}{\normhist_0 \beta} \psi_1 + \ket{0}^a \ket{1} \frac{\beta_2}{\normhist_0' \beta} \phi_1 + \ket{\junk},
$$
where the vectors $\psi_1$ and $\phi_1$ are as defined in the proof of Claim~\ref{claim:ind1_ode_input_state}:
$$
\psi_1 = \norm{\widetilde{y}_0} \ket{0,0, \widetilde{y}_0} + h \norm{\widetilde{f}_y} \sum_{j=0}^{m-1} \ket{j,1,\widetilde{f}_y}.
$$
The encoded initial condition and forcing satisfy
\begin{equation}\label{eq:ind1_errors_input_y}
\norm{\widetilde{y}_0 - \rlcP_0 \vec{x}_0} \leq \varepsilon_1 \norm{\vec{x}_0}, \quad \norm{\widetilde{f}_y - \rlcP_0 \rlcM_1^{-1} \vec{f}} \leq \varepsilon_1' \norm{\vec{f}}/\sigma_1.    
\end{equation}
Let $\delta \in (0,1)$ be an error parameter to be fixed later. Let $\calL_y^{-1}$ be the linear operator of Remark~\ref{remark:BE_ODE_solver} considering the solution over $\vec{y}(t), t\in (0,T]$ is approximated by a truncated Taylor series over $k$ terms denoted by (Claim~\ref{claim:diff_taylor_truth_approx_input})
\begin{equation}\label{eq:ind1_def_params_invL}
k = \left\lceil \frac{2\log \Omega}{\log \log \Omega} \right\rceil, \quad \text{ where } \Omega = \frac{4 m e^3}{\delta} \left(1 + \frac{T e^2 \norm{\vec{f}}}{\sigma_1 \mu} \right),
\end{equation}
where $\mu^2 := m^{-1} \sum_{j=1}^m \norm{\vec{x}(t_j)}^2$, we have used that the relevant forcing here is given by $\rlcP_0 \rlcM_1^{-1} \vec{f}$ which satisfies $\norm{\rlcP_0 \rlcM_1^{-1} \vec{f}} \leq \norm{\rlcP_0} \norm{\rlcM_1^{-1}} \norm{\vec{f}} \leq \norm{\vec{f}}/\sigma_1$ and that the solution $\vec{y}(t)$ satisfies $\norm{\vec{y}(t)} = \norm{\rlcP_0 \vec{x}(t)} \leq \norm{\vec{x}(t)}$ since $\norm{\rlcP_0} \leq 1$. The above choice of $k$ ensures that $(k+1)! \geq \Omega$.

Let $\vec{y}(t), \forall t \in [0,T]$ be the true solution of the $\ode$  and let $\widetilde{y}(t), \forall t \in [0,T]$ be the solution obtained from approximating the solution by a truncated series over $k$ terms with the mis-specified initial condition $\widetilde{y}_0$ (instead of $\vec{y}_0$) and forcing $\widetilde{f}_y$ (instead of $\vec{f}_y$). Applying $W_3$ defined using $(\alpha_{\calL_y}^{-1},a_5,0)$-BE $U_{\calL_y^{-1}}$ of $\calL_y^{-1}$ to $W_2 W_1 \ket{\Psi_0}$ would produce the following vector
\begin{equation}\label{eq:ind1_tildePsiy}
W_3 W_2 W_1 \ket{\Psi_0} = \ket{0}^{a} \left[\ket{0} \frac{\beta_1}{\alpha_{\calL_y^{-1}} \normhist_0 \beta} \widetilde{\Psi}_y + \ket{1} \frac{\beta_2}{\normhist_0' \beta} \phi_1 \right] + \ket{\junk}, \text{ where } \widetilde{\Psi}_y := \sum_{j=0}^m \ket{j,0} \otimes \widetilde{y}(t_j).    
\end{equation}
Using Claim~\ref{claim:diff_taylor_truth_approx_input} with instantiations of $\calA := \rlcP_0 \rlcM_1^{-1} \rlcK$, $\eta_0 := \varepsilon_1 \norm{\vec{x}_0}$ and $\eta_f := \varepsilon_1' \norm{\vec{f}}/\sigma_1$ as in Eq.~\eqref{eq:ind1_errors_input_y}, we obtain\footnote{It can be checked in the proof of Claim~\ref{claim:diff_taylor_truth_approx_input}$(ii)$ that the following bound holds for the denominator being in terms of $\vec{x}$ as all of our conditions are based on $\vec{x}$.} 
\begin{equation}\label{eq:ind1_diff_truey_taylory}
\sum_{j=0}^m \norm{\vec{y}(t_j) - \widetilde{y}(t_j)}^2 \leq \left(\frac{\delta}{2} + \frac{(1 + 2 \expnorm(\calA))\norm{\vec{x}_0} \varepsilon_1}{\mu} + \frac{4 T \expnorm(\calA) \varepsilon_1' \norm{\vec{f}}}{\sigma_1 \mu}\right)^2 \sum_{j=0}^m \norm{\vec{x}(t_j)}^2.
\end{equation}
Setting the block-encoding precisions $\varepsilon_1$ of $U_{P_0}$ and $\varepsilon_1'$ of $U_{P_0 M_1^{-1}}$ as
\begin{equation}\label{eq:ind1_set_BE_params_inputs_y}
\varepsilon_1 = \frac{\delta \cdot \mu}{4(1 + 2\expnorm(\calA)) \norm{\vec{x}_0}}, \quad \varepsilon_1' = \frac{\sigma_1 \delta \cdot \mu}{16 T \expnorm(\calA) \norm{\vec{f}}}
\end{equation}
ensures that Eq.~\eqref{eq:ind1_diff_truey_taylory} satisfies
\begin{equation}\label{eq:ind1_diff_truey_taylory2}
\sum_{j=0}^m \norm{\vec{y}(t_j) - \widetilde{y}(t_j)}^2 \leq \delta^2 \sum_{j=0}^m \norm{\vec{x}(t_j)}^2.
\end{equation}

We now determine the state obtained by applying $W_4$ which is defined using $(\alpha_F,a_5,0)$-BE $U_{F}$ of $\mathbf{F} := (\id - \rlcQ_0 \rlcM_1^{-1} \rlcK)$, to $W_3 W_2 W_1 \ket{\Psi_0}$ (Eq.~\eqref{eq:ind1_tildePsiy}). Towards that, let us define the vector $\vec{w}(t), \forall t \in [0,T]$ as $\vec{w}(t) = \mathbf{F} \vec{y}(t)$ and $\vec{w}(0) = \vec{y}(0)$. Analogously, define the vector $\widetilde{w}(t), \forall t \in [0,T]$ as $\widetilde{w}(t) = \mathbf{F} \vec{y}(t)$ and $\widetilde{w}(0) = \widetilde{y}(0)$. The state $W_4 W_3 W_2 W_1 \ket{\Psi_0}$ can then be expressed as 
\begin{equation}\label{eq:ind1_tildePsiw}
W_4 W_3 W_2 W_1 \ket{\Psi_0} = \ket{0}^{a} \left[\ket{0} \frac{\beta_1}{\alpha_F \alpha_{\calL_y^{-1}} \normhist_0 \beta} \widetilde{\Psi}_w + \ket{1} \frac{\beta_2}{\normhist_0' \beta} \phi_1 \right] + \ket{\junk}, \text{ where } \widetilde{\Psi}_w := \sum_{j=0}^m \ket{j,0} \otimes \widetilde{w}(t_j),    
\end{equation}
and where $\widetilde{w}(t_j)$ satisfies
\begin{align}
\norm{\widetilde{w}(t_j) - \vec{w}(t_j)} 
&\leq \norm{(\id - \rlcQ_0 \rlcM_1^{-1} \rlcK)(\widetilde{y}(t_j) - \vec{y}(t_j))} \nonumber \\
& \leq \norm{(\id - \rlcQ_0 \rlcM_1^{-1} \rlcK)} \cdot \norm{\widetilde{y}(t_j) - \vec{y}(t_j)} \nonumber \\
& \leq ( \norm{\id} + \norm{\rlcQ_0} \norm{\rlcM_1^{-1}} \norm{\rlcK}) \cdot \norm{\widetilde{y}(t_j) - \vec{y}(t_j)} \nonumber \\
& \leq (1 + \alpha_K/\sigma_1) \cdot \norm{\widetilde{y}(t_j) - \vec{y}(t_j)},
\label{eq:taylorw_close_to_truew}
\end{align}
where we used the triangle inequality in the third line and the bounds $\norm{\rlcQ_0} = 1$, $\norm{\rlcM_1^{-1}} \leq 1/\sigma_1$, $\norm{\rlcK} \leq \alpha_K$ in the final line. 

We now determine the state obtained by applying $\id_a \otimes H \otimes \id$ to $W_4 W_3 W_2 W_1 \ket{\Psi_0}$ (Eq.~\eqref{eq:ind1_tildePsiw}). Towards that, we note that in step~\ref{algo_step:ind1_params_beta}, we set $\beta_1 = \alpha_{\calL_y^{-1}} \alpha_{F} \normhist_0$ and $\beta_2 = \normhist_0'$. Then, Eq.~\eqref{eq:ind1_tildePsiw} can be rewritten as
\begin{equation}\label{eq:ind1_tildePsiw}
W_4 W_3 W_2 W_1 \ket{\Psi_0} = \ket{0}^{a} \frac{1}{\beta} \left[\ket{0} \widetilde{\Psi}_w + \ket{1} \phi_1\right] + \ket{\junk}.   
\end{equation}
Then, $(\id_a \otimes H \otimes \id)W_4 W_3 W_2 W_1 \ket{\Psi_0}$ is given by
\begin{align}\label{eq:ind1_tildePsix}
(\id_a \otimes H \otimes \id)W_4 W_3 W_2 W_1 \ket{\Psi_0} 
&= \ket{0}^{a} \frac{1}{\sqrt{2} \beta} \left[\ket{0} (\widetilde{\Psi}_w + \phi_1) + \ket{1} (\widetilde{\Psi}_w - \phi_1) \right] + \ket{\junk} \nonumber \\   
&= \ket{0}^{a+1} \frac{1}{\sqrt{2} \beta} \widetilde{\Psi}_x + \ket{\junk'},   
\end{align}
where $\ket{\junk'}$ in the second line is some state that is orthogonal to every state with $\ket{0}^a$ over the ancilla qubits and $\ket{0}$ over the first register, and 
\begin{equation}
\widetilde{\Psi}_x = \ket{0} (\widetilde{y}_0 + \widetilde{z}_0) + \sum_{j=1}^m \ket{j} (\widetilde{w}(t_j) + \widetilde{f}_z)    
\end{equation}
Denoting the so obtained encoded solution (from the Talyor series approximation and approximate initial conditions) be $\widetilde{x}(t), \forall t\in[0,T]$ where $\widetilde{x}_0 = \widetilde{y}_0 + \widetilde{z}_0$ and $\widetilde{x}(t_j) = \widetilde{w}(t_j) + \widetilde{f}_z$. Recall that the true solution $\vec{x}(t)$ satisfies $\vec{x}_0 = \vec{y}_0 + \vec{z}_0$ and $\vec{x}(t_j) = \vec{w}(t_j) + f_z$. We then have that for $j \geq 1$:
\begin{equation}\label{eq:ind1_diff_taylorx_truex}
\norm{\widetilde{x}(t_j) - \vec{x}(t_j)} \leq \norm{\widetilde{w}(t_j) - \vec{w}(t_j)} + \norm{\widetilde{f}_z - \vec{f}_z} \leq (1 + \sigma_1/\alpha_K) \norm{\widetilde{y}(t_j) - \vec{y}(t_j)} + \varepsilon_2' \norm{\vec{f}}/\sigma_1,
\end{equation}
where we used Eq.~\eqref{eq:taylorw_close_to_truew} to bound the first term and the promise on $\norm{\widetilde{f}_z - \vec{f}_z}$ from Claim~\ref{claim:ind1_ode_input_state}. For $j=0$, we note that
\begin{equation}\label{eq:ind1_diff_taylor0_true0}
\norm{\widetilde{x}_0 - \vec{x}_0} \leq \norm{\widetilde{w}_0 - \vec{w}_0} + \norm{\widetilde{z}_0 - \vec{z}_0} \leq (1 + \sigma_1/\alpha_K) \norm{\widetilde{y}_0 - \vec{y}_0} + \varepsilon_2 \norm{\vec{x}_0},
\end{equation}
where we used Eq.~\eqref{eq:taylorw_close_to_truew} to bound the first term and the promise on $\norm{\widetilde{z}_0 - \vec{z}_0}$ from Claim~\ref{claim:ind1_ode_input_state}. 

Let us denote $\xi = (1+\sigma_1/\alpha_K)$. We now evaluate
\begin{align}
\sum_{j=0}^m \norm{\widetilde{x}(t_j) - \vec{x}(t_j)}^2 
&= \norm{\widetilde{x}_0 - \vec{x}_0}^2 + \sum_{j=1}^m \norm{\widetilde{x}(t_j) - \vec{x}(t_j)}^2 \nonumber \\
&\leq 2\xi^2 \norm{\widetilde{y}_0 - \vec{y}_0}^2 + 2 \varepsilon_2^2 \norm{\vec{x}_0}^2 + \sum_{j=1}^m \left[ 2 \xi^2 \norm{\widetilde{y}(t_j) - \vec{y}(t_j)}^2 + 2 \varepsilon_2^{'2} \norm{\vec{f}}^2/\sigma_1^2 \right] \nonumber \\
&= 2\xi^2 \sum_{j=1}^m \left[ \norm{\widetilde{y}_0 - \vec{y}_0}^2 \right] + 2 \varepsilon_2^2 \norm{\vec{x}_0}^2 + 2 m \varepsilon_2{'2} \norm{\vec{f}}^2/{\sigma_1}^2  \nonumber \\
&\leq 2\xi^2 \delta^2 \left(\sum_{j=0}^m \norm{\vec{x}(t_j)}^2 \right) + 2 \varepsilon_2^2 \norm{\vec{x}_0}^2 + 2 m \varepsilon_2^{'2} \norm{\vec{f}}^2/\sigma_1^2,
\label{eq:ind1_mse_x}
\end{align}
where we have used Eq.~\eqref{eq:ind1_diff_taylorx_truex} and Eq.~\eqref{eq:ind1_diff_taylor0_true0} in the second line along with $(a+b)^2 \leq 2a^2 + 2b^2$, and then used Eq.~\eqref{eq:ind1_diff_truey_taylory2} in the last line. We now set $\varepsilon_2$ and $\varepsilon_2'$ as
\begin{equation}\label{eq:ind1_first_choise_eps2}
\varepsilon_2 = \xi \delta, \quad \varepsilon_2' = \xi \delta \sigma_1 \mu/\norm{\vec{f}},
\end{equation}
where recall that $\mu^2 = m^{-1} \sum_{j=1}^m \norm{\vec{x}(t_j)}^2$. Substituting these choices from Eq.~\eqref{eq:ind1_first_choise_eps2} into Eq.~\eqref{eq:ind1_mse_x} gives us
\begin{equation}\label{eq:ind1_mse_x_semifinal}
\sum_{j=0}^m \norm{\widetilde{x}(t_j) - \vec{x}(t_j)}^2 \leq 4 \xi^2 \delta^2 \sum_{j=0}^m \norm{\vec{x}(t_j)}^2
\end{equation}

At this point, consider the history state corresponding to the true vector $\vec{x}$ as $\ket{\Psi_x}$ and that corresponding to the Taylor approximation $\widetilde{x}$ as $\ket{\widetilde{\Psi}_x}$ respectively, which are defined as
\begin{equation}\label{eq:def_psihist_w}
\ket{\Psi_x} := \frac{1}{\normhist} \sum_{j=0}^m \norm{\vec{x}(jh)} \ket{j, x(jh)}, \quad \ket{\widetilde{\Psi}_x} := \frac{1}{\widetilde{\normhist}} \sum_{j=0}^m \norm{\widetilde{x}(jh)} \ket{j, \widetilde{x}(jh)}.
\end{equation}
Using Claim~\ref{claim:diff_hist_states} and Eq.~\eqref{eq:ind1_mse_x_semifinal}, we then obtain
$$
\norm{\ket{\Psi_x} - \ket{\widetilde{\Psi}_x}} \leq 4 \xi \delta.    
$$
Setting $\delta = \delta'/(4 \xi)$ gives us the desired error. This requires setting $\varepsilon_1,\varepsilon_1'$ from Eq.~\eqref{eq:ind1_set_BE_params_inputs_y} and $\varepsilon_2,\varepsilon_2'$ from Eq.~\eqref{eq:ind1_first_choise_eps2} as
$$
\varepsilon_1 = \frac{\delta' \cdot \mu}{16 \xi (1 + 2\expnorm(\calA)) \norm{\vec{x}_0}}, \quad \varepsilon_1' = \frac{\delta' \sigma_1  \mu}{64 \xi T \expnorm(\calA) \norm{\vec{f}}},  \quad \varepsilon_2 = \frac{\delta'}{4}, \quad \varepsilon_2' = \frac{\delta' \sigma_1 \mu}{4\norm{\vec{f}}},
$$
where recall that $\xi = (1 + \sigma_1/\alpha_K)$. Note that $k$ is then defined with respect to $\delta = \delta'/(4\xi)$ in Eq.~\eqref{eq:ind1_def_params_invL}. This completes the proof.
\end{proof}

So, now the goal is to prepare a quantum state which we denote by $\ket{\widehat{\Psi}_x}$ such that it is close to $\ket{\widetilde{\Psi}_x}$ which in turn guarantees that $\ket{\widehat{\Psi}_x}$ is close to the true history state $\ket{\Psi_x}$. This mainly involves deciding the precisions of the different block-encodings involved in the $\ode$ solve over $\vec{y}$ and that corresponding to the linear-algebraic solve, which we do next.
\begin{claim}\label{claim:ind1_correctness}
Let $\varepsilon \in (0,1)$. Consider Algorithm~\ref{algo:QDAE_solver_ind1} and choices on $k, \varepsilon_1,\varepsilon_1',\varepsilon_2,\varepsilon_2'$ made in Claim~\ref{claim:ind1_whatif_exactBE_ODE_LA}. The state obtained at end of step~\ref{algo_step:ind1_add_Psiw_Phiz0} is
$$
\ket{0}^a \frac{\norm{\widehat{\Psi}_x}}{\beta} \ket{\widehat{\Psi}_x} + \ket{\junk}, \text{ where } \norm{\ket{\widehat{\Psi}_x} - \ket{\Psi_x}} \leq \varepsilon.
$$
This is provided that the block-encoding precisions are set as follows for $\calA := -\rlcP_0 \rlcM_1^{-1}\rlcK$:
$$
\varepsilon_3 = O(\min(\mu,1)\varepsilon/(\sqrt{k} \alpha_{\calL_y^{-1}}), \quad \varepsilon_4 = O(\min(\mu,1) \varepsilon/(\sqrt{k} \alpha_{\calL_y^{-1}}).
$$
Additionally, the following precisions must satisfy $\varepsilon_1 = O(T/(k \kappa \norm{\vec{x}_0}))$, $\varepsilon_1' = O(\sigma_1/(k \kappa \norm{\vec{f}} ))$, $\varepsilon_2 = O(T/(\sqrt{k} \kappa \norm{\vec{x}_0}))$, $\varepsilon_2' = O(\sigma_1/(\sqrt{k} \kappa \norm{\vec{f}} ))$, where $\kappa = T \norm{\calA} \expnorm(\calA)$ for $\calA := - \rlcP_0 \rlcM_1^{-1} \rlcK$.
\end{claim}
\begin{proof}
Let 
$$
E = \alpha_F \alpha_{\calL_y^{-1}} (\bra{0}^a \otimes \id)(\id_{a - a_6} \otimes U_F) (\id_{a - a_5} \otimes U_{\calL_y^{-1}}) (\ket{0}^a \otimes \id) - \mathbf{F} \calL_y^{-1},
$$
where recall that $\mathbf{F} = (\id - \rlcQ_0 \rlcM_1^{-1} \rlcK)$. We have from Lemma~\ref{lem:prod_BEs} that 
$$
\norm{E} \leq \alpha_{\calL^{-1}} \varepsilon_4 + \alpha_{F} \varepsilon_3.
$$
We now note that the effect of $(\id_a \otimes H \otimes \id)W_4 W_3 W_2 W_1 \ket{\Psi_0}$ is to produce the state
\begin{equation}\label{eq:ind1_rho}
\ket{\rho} = \ket{0}^{a + 1} \left[ \frac{E + \mathbf{F} \calL_y^{-1}}{\alpha_F \alpha_{\calL_y^{-1}}} \frac{\beta_1}{\sqrt{2} \normhist_0 \beta} \Psi_1 + \frac{\beta_2}{\sqrt{2} \normhist_0' \beta} \Phi_1 \right] + \ket{\junk}    
\end{equation}
where $\ket{\junk}$ is some state that is orthogonal to every state with $\ket{0}^a$ over the flag qubits and $\ket{0}$ over the first qubit of the system register. Noting that $\beta_1 = \alpha_F \alpha_{\calL_y^{-1}} \normhist_0$ and $\beta_2 = \normhist_0'$, we can simplify Eq.~\eqref{eq:ind1_rho} to
\begin{equation}\label{eq:ind1_rho2}
\ket{\rho} = \ket{0}^{a + 1} \frac{(E + \mathbf{F} \calL_y^{-1}) \Psi_1 + \Phi_1}{\sqrt{2}\beta} + \ket{\junk}.
\end{equation}
We note that $\widetilde{\Psi}_x = \mathbf{F} \calL_y^{-1} \Psi_1 + \Phi_1$ (Claim~\ref{claim:ind1_whatif_exactBE_ODE_LA}). Let us define $\widehat{\Psi}_x := (E + \mathbf{F} \calL_y^{-1}) \Psi_1 + \Phi_1$. We then can write $\ket{\rho}$ as
$$
\ket{\rho} = \ket{0}^{a + 1} \frac{\norm{\widehat{\Psi}_x}}{\sqrt{2}\beta} \ket{\widehat{\Psi}_x} + \ket{\junk}.
$$
We now evaluate the distance between $\ket{\widehat{\Psi}_x}$ and $\ket{\widetilde{\Psi}_x}$ as follows:
\begin{align}
\norm{\ket{\widehat{\Psi}_x} - \ket{\widetilde{\Psi}_x}} \leq \frac{2 \norm{\widehat{\Psi}_x - \widetilde{\Psi}_x}}{\norm{\widetilde{\Psi}_x}} \leq \frac{2 \norm{E \Psi_1}}{\norm{\widetilde{\Psi}_x}},
\end{align}
where we have used Claim~\ref{claim:diff_hist_states} in the first inequality. To obtain an upper bound on the above quantity, we lower bound $\norm{\widetilde{\Psi}_x}$ by dividing into two cases.

\emph{Case 1:} $\norm{\Psi_1} \leq 1$.

In this case, we note that $2 \norm{E \Psi_1} \leq 2 \norm{E} \leq \alpha_{\calL^{-1}} \varepsilon_4 + \alpha_{F} \varepsilon_3$. Moreover, $\norm{\widetilde{\Psi}_x} \geq \frac{1}{2} \norm{\Psi_x} \geq \sqrt{m} \mu \geq \mu$ where the first inequality follows from the proof of Claim~\ref{claim:ind1_whatif_exactBE_ODE_LA} and the following inequalities follow from $\norm{\Psi_x} = \sum_{j=0}^m \norm{\vec{x}(t_j)}^2 \geq m \mu^2 \geq \mu^2$. We then have the bound
$$
\norm{\ket{\widehat{\Psi}_x} - \ket{\widetilde{\Psi}_x}} \leq \frac{\alpha_{\calL_y^{-1}} \varepsilon_4 + \alpha_{F} \varepsilon_3}{\mu}. 
$$
Setting $\varepsilon_4 = \mu \varepsilon/(2 \alpha_{\calL_y^{-1}})$ and $\varepsilon_3 = \mu \varepsilon/(2 \alpha_F)$ gives us the desired precision.

\emph{Case 2:} $\norm{\Psi_1} \geq 1$. 

In this case, we lower bound $\norm{\widetilde{\Psi}_x}$ by evaluating 
\begin{align}
\norm{\widetilde{\Psi}_x} = \norm{\mathbf{F} \calL_y^{-1} \Psi_1 + \Phi_1} \geq \norm{\rlcP_0 (\mathbf{F} \calL_y^{-1} \Psi_1 + \Phi_1)} 
&= \norm{\rlcP_0 \calL_y^{-1} \Psi_1 + \rlcP_0 \Phi_1} \\
&\geq \norm{\rlcP_0 \calL_y^{-1} \Psi_1} - \norm{\rlcP_0 \Phi_1} \\
&\geq \norm{\calL_y^{-1} \Psi_1} - \norm{\rlcQ_0 \calL_y^{-1} \Psi_1} - \norm{\rlcP_0 \Phi_1} \\
&\geq \norm{\calL_y^{-1} \Psi_1} - \norm{\calL_y^{-1}} \norm{\rlcQ_0 \Psi_1} - \norm{\rlcP_0 \Phi_1}, \label{eq:ind1_lb_norm_tildePsix}
\end{align}
where we noted that $\rlcP_0 \mathbf{F} = \rlcP_0$ as $\rlcP_0 \rlcQ_0 = 0$ and $\norm{\rlcQ_0 \calL_y^{-1} \Psi_1} = \norm{\rlcQ_0 \calL_y^{-1} (\rlcP_0 \Psi_1 + \rlcQ_0 \Psi_1)} = \norm{\rlcQ_0 \calL_y^{-1} \rlcQ_0 \Psi_1} \leq \norm{\calL_y^{-1}} \norm{\rlcQ_0 \Psi_1}$ since $\calL_y^{-1} \rlcP_0 \Psi_1 \in \mathrm{Im}(\rlcP_0)$. To bound the second and third terms in Eq.~\eqref{eq:ind1_lb_norm_tildePsix}, we now note that Eq.~\eqref{eq:ind1_lb_norm_tildePsix}, we now note that
\begin{align*}
\norm{\rlcQ_0 \Psi_1}^2 \leq \varepsilon_1^2 \norm{\vec{x}_0}^2 + m h^2 \varepsilon_1^{'2} \norm{\vec{f}}/\sigma_1^2
\end{align*}
since $\norm{\rlcQ_0 \widetilde{y}_0} = \norm{\rlcQ_0 (\widetilde{y}_0 - \rlcP_0 \widetilde{y}_0)} \leq \varepsilon_1 \norm{\vec{x}_0}$ and similarly, $\norm{\rlcQ_0 \widetilde{f}_y} \leq \varepsilon_1' \norm{\vec{f}}/\sigma_1$. Repeating the same idea, we can show that
\begin{align*}
\norm{\rlcP_0 \Phi_1}^2 \leq \varepsilon_2^2 \norm{\vec{x}_0}^2 + m \varepsilon_2^{'2} \norm{\vec{f}}/\sigma_1^2. 
\end{align*}
Noting that $\norm{\Psi_1} \geq 1$, we have that $\norm{\rlcQ_0 \Psi_1}/\norm{\Psi_1} \leq \norm{\rlcQ_0 \Psi_1}$ and $\norm{\rlcP_0 \Phi_1}/\norm{\Psi_1} \leq \norm{\rlcP_0 \Phi_1}$. Noting that the $\spec(\calL^{-1}) \in [1/(2 \sqrt{k} e), 2 e \kappa]$ (see proof of Theorem~\ref{thm:ODE_solver} in Appendix~\ref{appsec:qalgo_ode_la}), where $\kappa = T \norm{\calA} \expnorm(\calA)$ for $\calA := - \rlcP_0 \rlcM_1^{-1} \rlcK$ here. We find that setting $\varepsilon_1 = O(T/(k \kappa \norm{\vec{x}_0}))$, $\varepsilon_1' = O(\sigma_1/(k \kappa \norm{\vec{f}} ))$, $\varepsilon_2 = O(T/(\sqrt{k} \kappa \norm{\vec{x}_0}))$, $\varepsilon_2' = O(\sigma_1/(\sqrt{k} \kappa \norm{\vec{f}} ))$ ensures that
$$
\norm{\ket{\widehat{\Psi}_x} - \ket{\widetilde{\Psi}_x}} \leq 20 \sqrt{k} \norm{E} \leq  O\left(\sqrt{k} (\alpha_{\calL_y^{-1}} \varepsilon_4 + \alpha_{F} \varepsilon_3) \right).
$$
Setting $\varepsilon_4 = O(\varepsilon/(\sqrt{k} \alpha_{\calL_y^{-1}})$ and $\varepsilon_3 = O(\varepsilon/(\sqrt{k} \alpha_{\calL_y^{-1}})$ gives us the desired error. Setting $\delta' = \varepsilon/2$ in Claim~\ref{claim:ind1_whatif_exactBE_ODE_LA} gives us the desired error.
\end{proof}

\paragraph{Complexity and overall guarantee.} We have basically shown the correctness of Algorithm~\ref{algo:QDAE_solver_ind1} in the above claims. Let us now comment on the complexity of Algorithm~\ref{algo:QDAE_solver_ind1} and thereby give a proof of Theorem~\ref{thm:sim_ind1}.
\begin{proof}[Proof of Theorem~\ref{thm:sim_ind1}]
We will use Algorithm~\ref{algo:QDAE_solver_ind1}. We choose the parameters corresponding to the block-encodings as obtained at end of Claim~\ref{claim:ind1_correctness}. Claim~\ref{claim:ind1_correctness} then guarantees that we obtain a history state $\ket{\widehat{\Psi}}$ that is $\varepsilon$-close to the true history state.  

The main contribution to the gate complexity is due to the application of the block-encodings corresponding to $\calL_y^{-1}$ for the $\ode$ solve obtained from from Remark~\ref{remark:BE_ODE_solver} and $(\id - \rlcQ_0 \rlcM_1^{-1} \rlcK)$ obtained from Claim~\ref{claim:lin_alg_BE_ind1}. The probability of success here is $(\norm{\widetilde{\Psi}_x}^2/\beta^2) \geq \Omega(m \mu^2/\beta^2)$. A coarse calculation reveals that (where $\calA := -\rlcP_0 \rlcM_1^{-1} \rlcK$)
$$
\beta^2 = O\left(\left(\kappa_M^2 + \frac{T^2\kappa_M^3\alpha_K^4 \expnorm(\calA)^2}{\sigma_1^4}
\right)\|\vec{x}_0\|^2 +
\left( \frac{T\alpha_K\kappa_M^2}{\sigma_1^3} + \frac{T^3\kappa_M^4\alpha_K^3 \expnorm(\calA)^2}{\sigma_1^5} \right)\|\vec{f}\|^2 \right),
$$
where we used $mh = T$ and $h=O(\sigma_1/\alpha_K)$. Using amplitude amplification of Theorem~\ref{thm:amplitude_amplification} to boost this to $\Omega(1)$ then gives us the final result.
\end{proof}

\paragraph{Final time preparation.}
We can also use the algorithm of Theorem~\ref{thm:sim_ind1} to obtain a quantum state that is close to the final time state $\ket{x(T)}$ with slighly modified parameters. This is formally described below. 
\begin{corollary}\label{corr:sim_T_ind1}
Let $\varepsilon, \sigma, \sigma_1 \in (0,1)$. Consider the problem setup of Definition~\ref{def:prob_setup_DAEs} and suppose that the $\dae$ has index $0$. Let the minimum non-zero singular value of $\rlcM$ satisfy $\sigma^+_\min(\rlcM) \geq \sigma$ and that of $\rlcM_1$ satisfy $\sigma_\min(\rlcM_1) \geq \sigma_1$. Then, there exists a quantum algorithm that outputs a state $|\widehat{x}(T)\rangle$ with success probability $\geq 2/3$ such that $\left\| |\widehat{x}(T)\rangle - |x(T)\rangle \right\| \leq \varepsilon$. Define
$$
\kappa_M := \alpha_M/\sigma, \quad \kappa_{M_1} = (\alpha_M + \alpha_K)/\sigma_1, \quad  g := \frac{\max_{t\in[0,T]} \norm{\vec{x}(t)}}{\norm{\vec{x}(T)}}, \quad \calC_f := \log\left(1 + \frac{T e^2 f_{\max} }{\sigma_1 \norm{\vec{x}(T)}}\right).
$$
The algorithm uses the following complexity 
\begin{align*}
\text{Queries to }U_{M}&: \widetilde{O}\Big(g T \expnorm(- \rlcP_0 \rlcM_1^{-1}\rlcK) \cdot \poly\left(\kappa_M, \kappa_{M_1}, \alpha_K, \log(1/\varepsilon), \mathcal{C}_f)\right) \Big) \,,\\
\text{Queries to }U_{K}&: \widetilde{O}\Big(g T \expnorm(- \rlcP_0 \rlcM_1^{-1}\rlcK) \cdot \poly\left(\kappa_M, \kappa_{M_1}, \alpha_K, \log(1/\varepsilon), \mathcal{C}_f)\right) \Big) \,,\\ 
\text{Queries to }O_{f},O_x&: O\left(g T \alpha_K \kappa_M \expnorm(- \rlcP_0 \rlcM_1^{-1}\rlcK)/\sigma \cdot \log(1/\varepsilon)\right) \,,
\end{align*}
An additional gate complexity $O(1)$ times the number of uses of $U_M$ and $U_K$ is also used. It is sufficient to set the precisions of the block-encodings of $U_M$ and $U_K$ as
$$
\varepsilon_M = \poly\left(\frac{\sigma \sigma_1 \varepsilon}{\alpha_K T \expnorm(-\rlcP_0 \rlcM_1^{-1}\rlcK) \calC_f }\right) \text{ and } \varepsilon_K = \poly\left(\frac{\sigma \sigma_1 \varepsilon}{\alpha_K T \expnorm(-\rlcP_0 \rlcM_1^{-1}\rlcK) \calC_f }\right).
$$
\end{corollary}
We omit the details here as it follows similarly to the proof of Theorem~\ref{thm:sim_ind1}. The key difference is to use the number of Taylor terms corresponding to Claim~\ref{claim:diff_taylor_truth_approx_input}$(i)$. 

\subsubsection{Index two}
We have the following main result for index $2$ whose proof is similar to the one for index $1$.

\begin{corollary}\label{thm:sim_ind2}
Let $\varepsilon, \sigma, \sigma_1, \sigma_2, \uptau \in (0,1)$. Consider the problem setup of Definition~\ref{def:prob_setup_DAEs} and suppose that the $\dae$ has index $2$. Let $\sigma_\min^+(\rlcM) \geq \sigma$, $\sigma_\min^+(\rlcM_1) \geq \sigma_1$, $\sigma_\min(\rlcM_1) \geq \sigma_2$ and $\sigma_\min(\rlcP_0 \widetilde{\rlcQ}_1) \geq 1/\uptau$. Define the true history state encoding the evolution of $\vec{x}(t), \, t \in [0,T]$ as 
$$
\ket{\Psi} = \normhist^{-1} \sum_{k=0}^{m} \norm{\vec{x}(k \Delta t)} \ket{k} \ket{x(k \Delta t)}, 
$$
where $\Delta t = O(\sigma_2 \uptau/ \alpha_K)$ and $m = \ceil{T/\Delta t}$. Then, there exists a quantum algorithm that outputs a state $\ket{\widehat{\Psi}}$ such that $\left\| \ket{\widehat{\Psi}} - \ket{\Psi} \right\| \leq \varepsilon$ with probability $\geq 2/3$. Define $\kappa_M = \alpha_M/\sigma$, $\kappa_{M_1} = (\alpha_M + \alpha_K)/\sigma_1$, $\gamma_f = O(\norm{\vec{f}}/(\uptau \sigma_2))$
$$
\kappa_{M_2} = O(\alpha_M + \alpha_K/\uptau)/\sigma_2, \quad \mu^2 = m^{-1} \sum_{j=1}^m \norm{\vec{x}(j)}^2, \quad  \calC_f := \log\left(1 + \frac{T e^2 \gamma_f}{\mu}\right).
$$
The algorithm uses the following complexity 
\begin{align*}
\text{Queries to }U_{M}&: \widetilde{O}\Big(T^2 \expnorm(- \rlcP_0 \rlcP_1 \rlcM_2^{-1}\rlcK)^2 \cdot \poly\left(1/\mu, \kappa_M, \kappa_{M_1}, \kappa_{M_2}, \alpha_K, \log(1/\varepsilon), \mathcal{C}_f)\right) \Big) \,,\\
\text{Queries to }U_{K}&: \widetilde{O}\Big(T^2 \expnorm(- \rlcP_0 \rlcP_1 \rlcM_2^{-1}\rlcK)^2 \cdot \poly\left(1/\mu, \kappa_M, \kappa_{M_1}, \kappa_{M_2}, \alpha_K, \log(1/\varepsilon), \mathcal{C}_f)\right) \Big) \,,\\ 
\text{Queries to }O_{f},O_x&: O\left(T \alpha_K \kappa_{M_2} \expnorm(- \rlcP_0 \rlcP_1 \rlcM_1^{-1}\rlcK) \cdot \log(1/\varepsilon)\right) \,,
\end{align*}
An additional gate complexity $O(1)$ times the number of uses of $U_M$ and $U_K$ is also used. It is sufficient to set the precisions of the block-encodings of $U_M$ and $U_K$ as
$$
\varepsilon_M =  \varepsilon_K =  \poly\left(\frac{\sigma_1 \uptau \varepsilon}{\alpha_K \kappa_M T \expnorm(-\rlcP_0 \rlcP_1 \rlcM_1^{-1}\rlcK) \calC_f }\right).
$$
\end{corollary}
We omit the details of the proof as it follows similarly to that of Theorem~\ref{thm:sim_ind1}. As mentioned in Section~\ref{subsec:approach}, we can reduce the solve for index $2$ to effectively a solve that is algorithmically identical to what is carried out for index $1$.

\paragraph{Algorithm.}
We will use Algorithm~\ref{algo:QDAE_solver_ind2} that follows the approach discussed in Section~\ref{subsec:approach} to obtain the above result. Let us introduce some relevant notation.

For the desired input state to the $\dae$ solver, we need to define separate input states corresponding to the variables $\vec{y}$ and $\vec{z}$. For $\vec{y}$, we define the following input state $\ket{\psi_0}$ with $U_1$ as its corresponding state-preparation unitary (which can be obtained from Claim~\ref{claim:initial_state_prep})
\begin{equation}\label{eq:algo_ind2_psi0}
\ket{\psi_0} = U_1 \ket{0,0,0} = \frac{1}{\normhist_0} \left[\alpha_{P_0 P_1} \norm{\vec{x}_0} \ket{0,0,\vec{x}_0} + h \alpha_{P_0 P_1 M_2^{-1}} \norm{\vec{f}} \sum_{j=0}^{m-1} \ket{j,1,f} \right],    
\end{equation}
where $\alpha_{P_0 P_1}$ is the subnormalization of the block-encoding $U_{P_0}$ of $\rlcP_0$ from Claim~\ref{claim:ode_BE_ind2}, $\alpha_{P_0 P_1 M_1^{-1}}$ is the subnormalization of the block-encoding $U_{P_0 P_1 M_2^{-1}}$ of $\rlcP_0 \rlcP_1 \rlcM_2^{-1}$ from Claim~\ref{claim:ode_BE_ind2}, and $\normhist_0 = \sqrt{\alpha_{P_0 P_1}^2 \norm{\vec{x}(0)}^2 + m h^2 \alpha_{P_0 P_1 M_2^{-1}}^2 \norm{\vec{f}}^2}$.

For $\vec{z}$, we define the following input state $\ket{\phi_0}$ with $U_2$ as its corresponding state-preparation unitary (which can be obtained from Claim~\ref{claim:initial_state_prep})
\begin{equation}\label{eq:algo_ind2_phi0}
\ket{\phi_0} = U_2 \ket{0,0,0} = \frac{1}{\normhist'_0} \left[\alpha_{\id - P_0 P_1} \norm{\vec{x}_0} \ket{0,0,\vec{x}_0} + \alpha_{F_2} \norm{\vec{f}} \sum_{j=1}^{m} \ket{j,0,f} \right],    
\end{equation}
where $\alpha_{\id - P_0 P_1}$ is the subnormalization of the block-encoding $U_{\id - P_0 P_1}$ of $\id - \rlcP_0 \rlcP_1$ from Claim~\ref{claim:lin_alg_BE_ind2}, $\alpha_{G_2}$ is the subnormalization of the block-encoding $U_{F_2}$ of $\mathbf{F}_2$ (Eq.~\eqref{eq:ind2_LA_z_simple}) from Claim~\ref{claim:lin_alg_BE_ind2}, and $\normhist_0' = \sqrt{\alpha_{\id - \rlcP_0 \rlcP_1}^2 \norm{\vec{x}(0)}^2 + m \alpha_{F_2}^2 \norm{\vec{f}}^2}$.

\begin{myalgorithm}
\begin{algorithm}[H]
    \caption{Quantum $\dae$ solver for index two} \label{algo:QDAE_solver_ind2}
    \setlength{\baselineskip}{1.8em} 
    \DontPrintSemicolon 
    \KwInput{$(\alpha_M,a_M,\varepsilon_M)$-BE $U_M$ of $\rlcM$, $(\alpha_K,a_K,\varepsilon_K)$-BE $U_K$ or $\rlcK$, oracles $O_x, O_f$ (see Definition~\ref{def:prob_setup_DAEs}), parameter $\sigma_2$ s.t. $\sigma_\min(\rlcM_2) \geq \sigma_2$.}
    \KwOutput{$\ket{\widehat{\Psi}}$ that is $\varepsilon$-close to the true history state $\ket{\Psi}$ (Theorem~\ref{thm:sim_ind1})} \vspace{2mm}
    Obtain $(\alpha_{P_0}, a_1,\varepsilon_1)$-BE $U_{P_0 P_1}$ of $\rlcP_0 \rlcP_1$ and $(\alpha_{\id - P_0 P_1},a_3,\varepsilon_2)$-BE $U_{\id - P_0 P_1}$ of $\id - \rlcP_0 \rlcP_1$ from Claim~\ref{claim:ode_BE_ind2} \\
    Obtain $(\alpha_{P_0 P_1 M_2^{-1}},a_2,\varepsilon_1')$-BE $U_{P_0 P_1 M_1^{-1}}$ of $\rlcP_0 \rlcP_1 \rlcM_1^{-1}$ (Claim~\ref{claim:ode_BE_ind2}) \\
    Obtain $(\alpha_{F_2},a_4,\varepsilon_2')$-BE $U_{F_2}$ of $\mathbf{F}_2$ (Claim~\ref{claim:lin_alg_BE_ind2}) \\
    Obtain $(\alpha_{\calL_y^{-1}}, a_5, \varepsilon_3)$-BE $U_{\calL_y^{-1}}$ of linear system solve corresponding to the block-encoding $U_{P_0 P_1 M_2^{-1} K}$ of $\rlcP_0 \rlcP_1 \rlcM_2^{-1} \rlcK$ (Claim~\ref{claim:ode_BE_ind2}) using Remark~\ref{remark:BE_ODE_solver}. \\
    Obtain $(\alpha_{G_2},a_6,\varepsilon_4)$-BE $U_{G_2}$ of $\rlcG_2$ from Claim~\ref{claim:lin_alg_BE_ind2} \\ \vspace{2mm}
    \Comment*[l]{Define parameters for solve}
    Set $h = O(\sigma_2 \uptau/\alpha_K)$ and $m = \ceil{T/h}$ \\
    Set $a= \sum_{i \in [6]}(a_i)$ \\
    Set $\normhist_0$ as in Eq.~\eqref{eq:algo_ind2_psi0} and $\normhist_0'$ as in Eq.~\eqref{eq:algo_ind2_phi0} \\
    Set $\beta_1 = \alpha_{\calL_y^{-1}} \alpha_{G_2} \normhist_0$, $\beta_2 = \normhist_0'$, and $\beta = \sqrt{\beta_1^2 + \beta_2^2}$ \\ 
    \vspace{2mm}
    \Comment*[l]{Construct input state}
    Obtain the unitary $U_{1}$ that prepares $\ket{\psi_0}$ (Eq.~\eqref{eq:algo_ind2_psi0}) from Claim~\ref{claim:initial_state_prep} \\
    Obtain the unitary $U_{2}$ that prepares $\ket{\phi_0}$ (Eq.~\eqref{eq:algo_ind2_phi0}) from Claim~\ref{claim:initial_state_prep} \\ 
    Initialize state $\ket{0}^a \ket{0,0}_{\calS}$ \\
    Apply a rotation $R$ on the first qubit in $\calS$ to get $\ket{0}^a \left[\frac{\beta_1}{\beta} \ket{0,0} + \frac{\beta_2}{\beta} \ket{1,0} \right]$ \\
    Apply $\id_a \otimes (\ket{0}\bra{0} \otimes U_1 + \ket{1}\bra{1} \otimes U_2)$
    to obtain $\ket{\Psi_0} = \ket{0}^a \left[\frac{\beta_1}{\beta} \ket{0} \ket{\psi_0} + \frac{\beta_2}{\beta} \ket{1} \ket{\phi_0} \right]$ \label{algo_step:ind1_Psi0} \\
    Apply $W_1 = \ket{0}\bra{0} \otimes [\ket{0}\bra{0} \otimes U_{P_0 P_1} + (\id - \ket{0}\bra{0})\otimes U_{P_0 P_1 M_2^{-1}}] + \ket{1}\bra{1} \otimes \id$ to $\ket{\Psi_0}$ \\
    Apply $W_2 = \ket{0}\bra{0} \otimes \id + \ket{1}\bra{1} \otimes [\ket{0}\bra{0} \otimes U_{\id - P_0 P_1} + (\id - \ket{0}\bra{0})\otimes U_{F_2}]$ to $W_1 \ket{\Psi_0}$ \\
    \vspace{2mm}
    \Comment*[l]{Step 1: Differential equation solve on constraint submanifold}
    Apply $W_3 = \ket{0}\bra{0} \otimes U_{\calL_y^{-1}} + \ket{1}\bra{1} \otimes \id$ to $W_2 W_1 \ket{\Psi_0}$ \\ \vspace{2mm}
    \Comment*[l]{Step 2: Linear algebraic solve}
    Apply $W_4 = \ket{0}\bra{0} \otimes \Big(\ket{0}\bra{0} \otimes \id + (\id - \ket{0}\bra{0}) \otimes U_{F} \Big) + \ket{1}\bra{1} \otimes \id$ to $W_3 W_2 W_1 \ket{\Psi_0}$  \\
    Apply $\id_a \otimes H \otimes I$ to $W_4 W_3 W_2 W_1 \ket{\Psi_0}$  \\
    \Return $\ket{\widehat{\Psi}}$ after post-selecting for $\ket{0}^a$ over the first register.
\end{algorithm} 
\end{myalgorithm}

\subsection{Measurement of observables}\label{sec:measure-observables}
In this section, we discuss how to extract classically relevant quantities from the output of the quantum $\dae$ solver. In particular, we will consider estimating $\vec{x}(T)^\dagger \mathbf{O} \vec{x}(T)$ for some linear operator $\mathbf{O} \in \mathbb{C}^{N \times N}$ (also called quadratic form when $\mathbf{O}$ is Hermitian). For example, in the case of $\rlc$ circuits, energy over a subset of elements takes such a form.

We now state a general lemma that comments on estimating such observables from the state that is prepared by our quantum $\dae$ solver (before post-selection and when the goal is to prepare the normalized state at time $T$).
\begin{lemma}\label{lem:meas_obs_gen}
Let $\varepsilon, \varepsilon_O, \delta, h \in (0,1)$, $b, a_O, N \in \mathbb{N}$, $\beta, T > 0$, $\alpha_O \geq 1$, and $m = \ceil{T/h}$. Define
$$
\ket{\Psi} := \frac{1}{\normhist} \sum_{j=0}^m \norm{\vec{x}(jh)} \ket{j, x(jh)}, \quad \ket{\widehat{\Psi}} := \frac{1}{\widehat{\normhist}} \sum_{j=0}^m \norm{\widehat{x}(jh)} \ket{j, \widehat{x}(jh)}.
$$
Suppose there is an algorithm $\calQ$ that prepares the state $\ket{\widehat{\phi}}$ in time $T_\calQ$, of the form
$$
\ket{\widehat{\phi}} = \ket{0}^b \frac{\widehat{\normhist}}{\beta}\ket{\widehat{\Psi}} + \ket{\junk}, \text{ where } \norm{\widehat{x}(T) - \vec{x}(T)} \leq \delta \norm{\vec{x}(T)}\footnote{We note that this is the condition that ensures $\norm{\ket{\widehat{x}(T)} - \ket{x(T)}} \leq 2 \delta$ in Corollaries~\ref{corr:sim_T_ind0},\ref{corr:rlc_sim_T_ind1} and hence, this is the right condition to consider here.}
$$
where $\ket{\junk}$ is some state that is orthogonal to every state with $\ket{0}^b$ on the first register. Suppose $\mathbf{O} \in \mathbb{C}^{N \times N}$ is a linear operator with an $(\alpha_O,a_O,\varepsilon_O)$-block-encoding $U_O$ implementable in time $T_O$. Then, there is a quantum algorithm that outputs $\widehat{\theta}$ such that $\Big| \widehat{\theta} - \vec{x}(T)^\dagger \mathbf{O} \vec{x}(T) \Big| \leq \varepsilon$, with the following complexity
\begin{align*}
    \text{sample complexity: } & O\left(\alpha_O^2 \beta^4/\varepsilon^2 \right), \\
    \text{gate complexity: } & O\left( \alpha_O^2 \beta^4 (b + \log m + T_O + T_\calQ)/\varepsilon^2 \right).
\end{align*}
It suffices to set $\varepsilon_O = \varepsilon/(3 \beta^2)$ and $\delta = \varepsilon/(9 \alpha_O \norm{\vec{x}(T)}^2)$. 
\end{lemma}
\begin{proof}
The protocol to accomplish this task will involve a Hadamard test of $\ket{\widehat{\phi}}$ and an appropriate block-encoding of $\mathbf{O}$. We will now argue why this is the case and how the estimate is produced. Let $t_j = jh, \forall j \in [m]_0$. We are given that the state $\ket{\widehat{\phi}}$ is of the form
\begin{equation}\label{eq:meas_obs_hatphi}
\ket{\phi} := \ket{0}^{b} \frac{1}{\beta} \widehat{\Psi} + \ket{\junk}, \text{ where } \widehat{\Psi} = \sum_{j=0}^m \ket{j} \otimes \widehat{x}(t_j)
\end{equation}
where $\ket{\junk}$ is some unnormalized state that is orthogonal to every state with $\ket{0}^{b}$ on the first register and we have expressed it in terms of the vector $\widehat{\Psi}$. Let us also define the state $\ket{\phi}$ as
\begin{equation}\label{eq:meas_obs_hatphi}
\ket{\phi} := \ket{0}^{b} \frac{1}{\beta} \Psi + \ket{\junk'}, \text{ where } \Psi = \sum_{j=0}^m \ket{j} \otimes \vec{x}(t_j),
\end{equation}
where $\ket{\junk'}$ is some unnormalized state (not necessarily the same as $\ket{\junk}$) that is orthogonal to every state with $\ket{0}^{b}$

We now create a block-encoding for $\mathbf{D} = \ket{0}^{b}\bra{0}^{b} \otimes \ket{T}\bra{T} \otimes \mathbf{O}$. From Claim~\ref{claim:proj_0state_BE}, we can obtain a $(1,1,0)$-block-encoding for $\ket{0}^{b}\bra{0}^{b}$, which we denote by $V_1$, with gate complexity $O(b)$. Similarly, we can obtain a $(1,1,0)$ block-encoding for $\ket{T}\bra{T}$, which we denote by $V_2$. This will have gate complexity $O(\log m)$. We can combine $V_1$ and $V_2$ with $U_O$ (using Lemma~\ref{lem:tensor_prod_BEs}) to obtain a $(\alpha_O,a_O+2,\varepsilon_O)$-block-encoding for $\mathbf{D}$, which we denote by $U_{D}$. Let us denote $\widetilde{D} = (\bra{0}^{a_O+2} \otimes \id) U_D (\ket{0}^{a_O+2} \otimes \id)$. Starting from the definition of $U_D$, we then observe that the action of $U_D$ on $\ket{\widehat{\phi}}$ satisfies
\begin{equation}\label{eq:tildeD_good}
\norm{\mathbf{D} - \alpha_O \widetilde{D}} \leq \varepsilon_O \implies \norm{\la \widehat{\phi} | \mathbf{D} | \widehat{\phi} \ra - \alpha_O \la \widehat{\phi} |\widetilde{D}| \widehat{\phi} \ra} \leq \varepsilon_O.
\end{equation}
We note that
\begin{equation}\label{eq:interim_phiDphi}
    \la \phi | \mathbf{D} | \phi \ra = \frac{\vec{x}(T)^\dagger \mathbf{O} \vec{x}(T)}{\beta^2}
\end{equation}

Let $\varepsilon_1 \in (0,1)$ be an error parameter to be determined later. We now produce an estimate $\mu$ of $\la \widehat{\phi} | \widetilde{D} | \widehat{\phi} \ra$ using the Hadamard test such that $|\mu - \la \widehat{\phi} | \widetilde{D} | \widehat{\phi} \ra| \leq \varepsilon_1$. This requires a sample complexity of $O(1/\varepsilon_1^2)$. Let us now denote $\widehat{\theta} := \alpha \mu$ where we define $\alpha := \alpha_O \beta^2$. We now argue that this is close to $\vec{x}(T)^\dagger \mathbf{O} \vec{x}(T)$ by evaluating
\begin{align}
& \Big| \widehat{\theta} - \vec{x}(T)^\dagger \mathbf{O} \vec{x}(T) \Big| \nonumber \\
\leq \,\,& \Big| \alpha \mu - \alpha \la \widehat{\phi} | \widetilde{D} | \widehat{\phi} \ra\Big| + \Big|\alpha \la \widehat{\phi} | \widetilde{D} | \widehat{\phi} \ra - \beta^2 \la \phi | \mathbf{D} | \phi \ra \Big| \nonumber  \\
\leq \,\, & \alpha \varepsilon_1 + \Big|\alpha \la \widehat{\phi} | \widetilde{D} | \widehat{\phi} \ra - \beta^2 \la \widehat{\phi}|\mathbf{D} | \widehat{\phi} \ra \Big| + \Big|\beta^2 \la \widehat{\phi}|\mathbf{D} | \widehat{\phi} \ra - \beta^2 \la \phi | \mathbf{D} | \phi \ra \Big| \nonumber  \\
= \,\, & \alpha \varepsilon_1 + \beta^2 \Big|\alpha_O \la \widehat{\phi} | \widetilde{D} | \widehat{\phi} \ra -  \la \widehat{\phi}|\mathbf{D} | \widehat{\phi} \ra \Big| + \beta^2 \Big| \la \widehat{\phi}|\mathbf{D} | \widehat{\phi} \ra - \la \phi | \mathbf{D} | \phi \ra \Big| \nonumber  \\
\leq \,\, & \alpha \varepsilon_1 + \beta^2 \varepsilon_O + \beta^2 \Big| \la \widehat{\phi}|\mathbf{D} | \widehat{\phi} \ra - \la \phi | \mathbf{D} | \phi \ra \Big|,
\label{eq:diff_est_theta}
\end{align}
where we added and subtracted $\alpha^2 \la \widehat{\phi} | \widetilde{D} | \widehat{\phi} \ra$ before applying the triangle inequality, and used Eq.~\eqref{eq:interim_phiDphi} to substitute for $\vec{x}(T)^\dagger \mathbf{O} \vec{x}(T)$ in the first line. In the second line, we added and subtracted $\beta^2 \la \widehat{\phi}|\mathbf{D} | \widehat{\phi} \ra$ in the second term before taking the triangle inequality, and used that $|\mu - \la \widehat{\phi} | \widetilde{D} | \widehat{\phi} \ra| \leq \varepsilon_1$. In third line, we used the definition of $\alpha = \alpha_O \beta^2$ and in the fourth line, used Eq.~\eqref{eq:tildeD_good}. To bound the third term from Eq.~\eqref{eq:diff_est_theta}, consider the vector $\vec{e} = \widehat{x}(T) - \vec{x}(T)$ and note that $\norm{\vec{e}} \leq \delta \norm{\vec{x}(T)}$ which follows from the promise in the lemma statement. We then have
\begin{align}
\Big| \la \widetilde{\phi}|\mathbf{D} | \widetilde{\phi} \ra - \la \phi | \mathbf{D} | \phi \ra \Big| 
&=  \frac{1}{\beta^2}\Big| \widehat{x}(T)^\dagger \mathbf{O} \widehat{x}(T) - \vec{x}(T)^\dagger \mathbf{O} \vec{x}(T) \Big| \nonumber\\
&\leq  \frac{1}{\beta^2}\Big| \widehat{x}(T)^\dagger \mathbf{O} \widehat{x}(T) - \widehat{x}(T)^\dagger \mathbf{O} \vec{x}(T) \Big| + \frac{1}{\beta^2} \Big| \widehat{x}(T)^\dagger \mathbf{O} \vec{x}(T) - \vec{x}(T)^\dagger \mathbf{O} \vec{x}(T) \Big| \nonumber\\
&\leq  \frac{1}{\beta^2}\Big| \widehat{x}(T)^\dagger \mathbf{O} (\widehat{x}(T) - \vec{x}(T)) \Big| + \frac{1}{\beta^2} \Big| (\widehat{x}(T) - \vec{x}(T))^\dagger \mathbf{O} \vec{x}(T) \Big| \nonumber\\
&\leq  \frac{1}{\beta^2}\Big( \norm{\widehat{x}(T)} \cdot \norm{\mathbf{O}} \cdot \norm{\vec{e}} + \norm{\vec{x}(T)} \cdot \norm{\mathbf{O}} \cdot \norm{\vec{e}}\Big) \nonumber\\
&\leq  \frac{1}{\beta^2}\Big( \alpha_O \delta (1 + \delta) \norm{\vec{x}(T)}^2 + \delta \alpha_O \norm{\vec{x}(T)}^2 \Big) \nonumber\\
&\leq  \frac{3 \delta \alpha_O \norm{\vec{x}(T)}^2}{\beta^2},
\label{eq:tildephi_phi_D}
\end{align}
where we have used the definitions of $\mathbf{D},\ket{\phi},\ket{\widetilde{\phi}}$ in the first line, added and subtracted $\widehat{x}(T)^\dagger \mathbf{O} \vec{x}(T)$ before applying the triangle inequality in the second line, used the submultiplicativity of norms in the fourth line, and noted that $\norm{\vec{e}} \leq \delta \norm{\vec{x}(T)}$ as well as $\norm{O} \leq \alpha_O$ and $\norm{\widehat{x}(T)} \leq \norm{e} + \norm{\vec{x}(T)} \leq (1+\delta)\norm{\vec{x}(T)}$. Substituting Eq.~\eqref{eq:tildephi_phi_D} into Eq.~\eqref{eq:diff_est_theta}, we obtain
\begin{equation}
\Big| \widehat{\theta} - \vec{x}(T)^\dagger \mathbf{O} \vec{x}(T) \Big| \leq \alpha \varepsilon_1 + \beta^2 \varepsilon_O + 3 \delta \alpha_O \norm{\vec{x}(T)}^2.
\end{equation}
Setting 
$$
\varepsilon_1 = \varepsilon/(3\alpha) = \varepsilon/(3 \alpha_O \beta^2), \quad \varepsilon_O = \varepsilon/(3 \beta^2), \quad \delta = \varepsilon/(9 \alpha_O \norm{\vec{x}(T)}^2)
$$ 
would then give us
$$
\Big| \widehat{\theta} - \vec{x}(T)^\dagger \mathbf{O} \vec{x}(T) \Big| \leq \varepsilon.
$$
We thus require a sample complexity of $O(1/\varepsilon_1^2) = O(\alpha_O^2 \beta^4/\varepsilon^2)$. This is the number of shots required from the Hadamard test. The circuit corresponding to the Hadamard circuit requires implementation of controlled-$U_D$ and preparation of $\ket{\widehat{\phi}}$ which requires calls to algorithm $\calQ$. Controlled-$U_D$ has gate complexity $O(b + \log m + T_O)$, and a call to $\calQ$ has gate complexity $T_\calQ$. The overall gate complexity is then
$$
O\left( \alpha_O^2 \beta^4 (b + \log m + T_O + T_\calQ)/\varepsilon^2 \right).
$$
This completes the proof.

\end{proof}
\emph{\bf Remark:} Note that the protocol of Lemma~\ref{lem:meas_obs_gen} does not need to know $\norm{\vec{x}(T)}$ to succeed. This is not used as part of the estimation procedure. We only need to provide an upper bound on $\norm{\vec{x}(T)}$ to comment on the specification of $\delta$ and hence, comment on the time complexity of the algorithm $\calQ$ which will be the quantum $\dae$ solver.

We now apply the above lemma to the quantum outputs from our quantum $\dae$ solver to obtain estimates of the form $\vec{x}(T)^\dagger \mathbf{O} \vec{x}(T)$. This is described formally below for different index values.
\paragraph{$\dae$ of index $0$.} We now state our main result commenting on observable estimation for index $0$.
\begin{theorem}\label{thm:output_ind0}
Let $\upsilon, \sigma \in (0,1)$. Consider the problem setup of Definition~\ref{def:prob_setup_DAEs} and suppose that the $\dae$ has index $0$. Let the minimum singular value of $\rlcM$ satisfy $\sigma_\min(\rlcM) \geq \sigma$. Let $\mathbf{O} \in \mathbb{C}^{N \times N}$ be a linear operator with an $(\alpha_O,a_O,\varepsilon_O)$-block-encoding $U_O$. Let $\calQ$ be the algorithm of Corollary~\ref{corr:sim_T_ind0}. Then, there is a quantum algorithm that outputs $\widehat{\theta}$ with success probability $\geq 2/3$ such that $\Big| \widehat{\theta} - \vec{x}(T)^\dagger O \vec{x}(T)\Big| \leq \upsilon$. The algorithm uses the following complexity 
\begin{align*}
\text{Calls to }\calQ \,\, (S_\calQ)&: O\Big(\alpha_O^2 T^4 \alpha_K^2 \expnorm(-\rlcM^{-1} \rlcK)^4 (\sigma^2 \alpha_K^2 \norm{\vec{x}_0}^4 + T^2 \norm{f}^4)/(\sigma^6 \upsilon^2) \Big) \,,\\
\text{Queries to }U_{M}&: \widetilde{O}\left(S_\calQ T \, \expnorm(-\rlcM^{-1}\rlcK) \cdot \poly\Big(\kappa_M, \alpha_K, \calC_f, \log(\alpha_O \norm{\vec{x}(T)}/\upsilon) \Big) \right) \,,\\
\text{Queries to }U_{K}&: \widetilde{O}\left(S_\calQ T \, \expnorm(-\rlcM^{-1}\rlcK) \cdot \poly\Big(\kappa_M, \alpha_K, \calC_f, \log(\alpha_O \norm{\vec{x}(T)}/\upsilon) \Big)\right) \,,\\ 
\text{Queries to }O_{f},O_x&: \widetilde{O}\left(S_\calQ T \alpha_K \kappa_M \expnorm(-\rlcM^{-1} \rlcK) \cdot \log(\alpha_O \norm{\vec{x}(T)}/\upsilon)\right) \,,
\end{align*}
where $\kappa_M := \alpha_M/\sigma$ and $\mathcal{C}_f := \log\left(1 + \frac{T e^2 \|\vec{f}\|}{\sigma \|\vec{x}(T)\|}\right)$. An additional gate complexity of $O(S_\calQ T_O)$ for implementing $U_O$ and $O(1)$ times the number of uses of $U_M$ and $U_K$, is used. It is sufficient to set the precisions of $U_M$, $U_K$, $U_O$ as
$$
\varepsilon_M = \varepsilon_K = \poly\left(\frac{\sigma \upsilon}{\alpha_O \alpha_K T \expnorm(-\rlcM^{-1} \rlcK) \calC_f \norm{\vec{x}(T)}^2}\right) \text{ and } \varepsilon_O = O\left(\frac{\upsilon \alpha_O^2}{S_Q}\right).
$$
\end{theorem}
\begin{proof}
Let $h = O(\sigma/\alpha_K)$ and $m = \ceil{T/h}$. We will use the algorithm of Corollary~\ref{corr:sim_T_ind0} to prepare a state $\ket{\widehat{\phi}}$ (before post-selection) which we define shortly and then the protocol of Lemma~\ref{lem:meas_obs_gen} which uses a Hadamard test of $\ket{\widehat{\phi}}$ and a block-encoding involving $\mathbf{O}$, which we will also define shortly. We also use the notation from Corollary~\ref{corr:sim_T_ind0} (and in turn Theorem~\ref{thm:sim_ind0}) here.

We use the algorithm of Corollary~\ref{corr:sim_T_ind0} to produce the state $\ket{\widehat{\phi}}$ before measuring which takes the following form
$$
\ket{\widehat{\phi}} = \frac{1}{\alpha_{\calL^{-1}} \calZ_0} \ket{0}^b \Psi + \ket{\junk}, \text{ where } \Psi = \sum_{j=0}^m \ket{j} \otimes \widehat{x}(t_j),
$$
where $\ket{\junk}$ is some state that is orthogonal to every state with $\ket{0}^b$ on the first register, $\normhist_0$ is the normalization of the initial state to the algorithm of Corollary~\ref{corr:sim_T_ind0} (Eq.~\eqref{eq:ind0_interim_init0}) given by $\normhist_0 = \sqrt{\norm{\vec{x}_0}^2 + m \alpha_1^2 h^2 \norm{f}^2}$, $\alpha_1 = 2/\sigma$ is the subnormalization of $U_{M^{-1}}$ (Claim~\ref{claim:inv_M_BE_ind0}), and $\alpha_{\calL^{-1}} = O(T \alpha_K \expnorm(-\rlcM^{-1} \rlcK)/\sigma)$ is the subnormalization of $U_{\calL^{-1}}$ defined with respect to the differential operator $-\rlcM^{-1} \rlcK$ (see Remark~\ref{remark:BE_ODE_solver} and Eq.~\eqref{eq:ind0_inv_calL_BE}). Let us denote $\beta = \alpha_{\calL^{-1}} \normhist_0$. We note that
$$
\beta^2 = O\left(T^2 \alpha_K \expnorm(-\rlcM^{-1} \rlcK)^2(\sigma \alpha_K \norm{\vec{x}_0}^2 + T \norm{f}^2)/\sigma^3 \right),
$$
where we used $mh = T$ and $h=O(\sigma/\alpha_K)$.

We define $\mathbf{D} = \ket{0}^b \bra{0}^b \otimes \ket{T}\bra{T} \otimes U_O$ and construct a $(\alpha_O, a_O +2, \varepsilon_O)$-block-encoding of $\mathbf{D}$, which we denote by $U_D$. We can do this as given in the proof of Lemma~\ref{lem:meas_obs_gen}. Let us denote $\widetilde{D} = (\bra{0}^{a_O + 2} \otimes \id)U_D(\ket{0}^{a_O+2} \otimes \id)$. Using Lemma~\ref{lem:meas_obs_gen}, the desired estimate is then set as $\widehat{\theta} = \alpha_O \beta^2 \mu$ where $\mu$ is a $\upsilon/(3\alpha_O \beta^2)$-estimate of $\la \widehat{\phi} | \widetilde{D} | \widehat{\phi} \ra$ which can be obtained via a Hadamard test. The corresponding sample complexity and hence calls to the algorithm of Corollary~\ref{corr:sim_T_ind0} to prepare $\ket{\widehat{\phi}}$ is
$$
O(\alpha_O^2 \beta^4/\upsilon^2) = O\Big(\alpha_O^2 T^4 \alpha_K^2 \expnorm(-\rlcM^{-1} \rlcK)^4 (\sigma^2 \alpha_K^2 \norm{\vec{x}_0}^4 + T^2 \norm{f}^4)/(\sigma^6 \upsilon^2) \Big).
$$
The accuracy $\varepsilon$ in the output of Corollary~\ref{corr:sim_T_ind0}, which appears as $\norm{\vec{x}(T) - \widetilde{x}(T)} \leq (\varepsilon/2) \vec{x}(T)$\footnote{The guarantee of Corollary~\ref{corr:sim_T_ind0} is written in terms of $\norm{\ket{x(T)} - \ket{\widehat{x}(T)}} \leq \varepsilon$ but this is obtained as a consequence of $\norm{\vec{x}(T) - \widetilde{x}(T)} \leq (\varepsilon/2) \vec{x}(T)$ (see proof).} needs to be set as $\varepsilon =O(\upsilon/(\alpha_O \norm{\vec{x}(T)}^2))$.

The corresponding gate complexity is then the sample complexity times the gate complexity of Corollary~\ref{corr:sim_T_ind0} for this value of $\varepsilon$, without taking into consideration of the probability of success and hence an improvement by a factor of $g$. This completes the proof.
\end{proof}

\paragraph{$\dae$ of index $1$.} We now state our main result commenting on observable estimation for index $1$.
\begin{theorem}\label{thm:output_ind1}
Let $\upsilon, \sigma, \sigma_1 \in (0,1)$. Consider the problem setup of Definition~\ref{def:prob_setup_DAEs} and suppose that the $\dae$ has index $1$. Let the minimum non-zero singular value of $\rlcM$ satisfy $\sigma^+_\min(\rlcM) \geq \sigma$ and that of $\rlcM_1$ satisfy $\sigma_\min(\rlcM_1) \geq \sigma_1$. Let $\mathbf{O} \in \mathbb{C}^{N \times N}$ be a linear operator with an $(\alpha_O,a_O,\varepsilon_O)$-block-encoding $U_O$. Let $\calQ$ be the algorithm of Corollary~\ref{corr:sim_T_ind1}. Then, there is a quantum algorithm that outputs $\widehat{\theta}$ with success probability $\geq 2/3$ such that $\Big| \widehat{\theta} - \vec{x}(T)^\dagger O \vec{x}(T)\Big| \leq \upsilon$. Define
$$
\kappa_M := \alpha_M/\sigma, \quad \kappa_{M_1} = (\alpha_M + \alpha_K)/\sigma_1, \quad \calC_f := \log\left(1 + \frac{T e^2 f_{\max} }{\sigma_1 \norm{\vec{x}(T)}}\right).
$$
The algorithm uses the following complexity 
\begin{align*}
\text{Calls to }\calQ \,\, (S_\calQ)&: \widetilde{O}\left(\alpha_O^2 \cdot \poly(T \kappa_M \alpha_K \expnorm(\calA)(\|\vec{x}_0\| + \|\vec{f}\|)/\sigma_1 \Big)/\upsilon^2 \right) \,,\\
\text{Queries to }U_{M}&: \widetilde{O}\Big(S_\calQ T \expnorm(- \rlcP_0 \rlcM_1^{-1}\rlcK) \cdot \poly\left(\kappa_M, \kappa_{M_1}, \alpha_K, \log(\alpha_O \norm{\vec{x}(T)}/\varepsilon), \mathcal{C}_f)\right) \Big) \,,\\
\text{Queries to }U_{K}&: \widetilde{O}\Big(S_\calQ T \expnorm(- \rlcP_0 \rlcM_1^{-1}\rlcK) \cdot \poly\left(\kappa_M, \kappa_{M_1}, \alpha_K, \log(\alpha_O \norm{\vec{x}(T)}/\varepsilon), \mathcal{C}_f)\right) \Big) \,,\\ 
\text{Queries to }O_{f},O_x&: O\left(S_\calQ T \alpha_K \expnorm(- \rlcP_0 \rlcM_1^{-1}\rlcK)/\sigma_1 \cdot \log(\alpha_O \norm{\vec{x}(T)}/\varepsilon)\right) \,.
\end{align*}
An additional gate complexity of $O(S_\calQ T_O)$ for implementing $U_O$ and $O(1)$ times the number of uses of $U_M$ and $U_K$, is used. It is sufficient to set the precisions of $U_M$, $U_K$, $U_O$ as
$$
\varepsilon_M = \varepsilon_K = \poly\left(\frac{\sigma \sigma_1 \upsilon}{\alpha_O \alpha_K T \expnorm(-\rlcP_0 \rlcM_1^{-1}\rlcK) \calC_f \norm{\vec{x}(T)}^2}\right) \text{ and } \varepsilon_O = O\left(\frac{\upsilon \alpha_O^2}{S_Q}\right).
$$
\end{theorem}
\begin{proof}
The proof proceeds similarly to that of Theorem~\ref{thm:output_ind0}. Let $h = O(\sigma_1/\alpha_K)$ and $m = \ceil{T/h}$. We will use the algorithm of Corollary~\ref{corr:sim_T_ind1} to prepare a state $\ket{\widehat{\phi}}$ (before post-selection) which we define shortly and then the protocol of Lemma~\ref{lem:meas_obs_gen} which uses a Hadamard test of $\ket{\widehat{\phi}}$ and a block-encoding involving $\mathbf{O}$, which we will also define shortly. We also use the notation from Corollary~\ref{corr:sim_T_ind1} (and in turn Theorem~\ref{thm:sim_ind1}) here.

We use the algorithm of Corollary~\ref{corr:sim_T_ind1} to produce the state $\ket{\widehat{\phi}}$ before measuring which takes the following form
$$
\ket{\widehat{\phi}} = \frac{1}{\sqrt{2} \beta} \ket{0}^{a+1} \Psi + \ket{\junk}, \text{ where } \Psi = \sum_{j=0}^m \ket{j} \otimes \widehat{x}(t_j),
$$
where $\ket{\junk}$ is some state that is orthogonal to every state with $\ket{0}^{a+1}$ on the first register and $\beta$ is the normalization of the initial state to the algorithm of Corollary~\ref{corr:sim_T_ind1} (step~\ref{algo_step:ind1_Psi0} from Algorithm~\ref{algo_step:ind1_Psi0}). In particular, $\beta = \sqrt{\beta_1^2 + \beta_2^2}$ where $\beta_1 = \alpha_{\calL_y^{-1}} \alpha_{F} \normhist_0$ and $\beta_2 = \normhist_0'$. 

The involved quantities are as follows. The normalization $\normhist_0 = \sqrt{\alpha_{P_0} \norm{\vec{x}_0}^2 + m \alpha_{P_0 M_1^{-1}}^2 h^2 \norm{f}^2}$ (Eq.~\eqref{eq:algo_ind1_psi0}) where $\alpha_{P_0} = 2 \kappa_M$ is the subnormalization of $U_{P_0}$ (Claim~\ref{claim:projs_BE_ind1}) and $\alpha_{P_0 M_1^{-1}} = 4\kappa_M/\sigma_1$ is the subnormalization of $U_{P_0 M_1^{-1}}$ (Claim~\ref{claim:ode_BE_ind1}). The normalization $\normhist_0' = \sqrt{\alpha_{Q_0}^2 \norm{\vec{x}(0)}^2 + m \alpha_{Q_0 M_1^{-1}}^2 \norm{\vec{f}}^2}$ (Eq.~\eqref{eq:algo_ind1_phi0}) where $\alpha_{Q_0} = 2\kappa_M + 1$ is the is the subnormalization of $U_{Q_0}$ (Claim~\ref{claim:projs_BE_ind1}) and $\alpha_{Q_0 M_1^{-1}} = O(\kappa_M/\sigma_1)$ is the subnormalization of $U_{Q_0 M_1^{-1}}$ (Claim~\ref{claim:ode_BE_ind1}). The parameter $\alpha_{\calL^{-1}} = O(T \alpha_K \expnorm(-\rlcP_0\rlcM_1^{-1} \rlcK)/\sigma_1)$ is the subnormalization of $U_{\calL_y^{-1}}$ defined with respect to the differential operator $-\rlcP_0 \rlcM^{-1} \rlcK$ (see Remark~\ref{remark:BE_ODE_solver} and step~\ref{algo_step:set_BE_ind_Ly} of Algorithm~\ref{algo:QDAE_solver_ind1}) and $\alpha_F = O(\kappa_M \alpha_K/\sigma_1)$ is the subnormalization of $\mathbf{F}:=(\id - \rlcQ_0 \rlcM_1^{-1} \rlcK)$ (Claim~\ref{claim:lin_alg_BE_ind1}). A coarse calculation reveals that (where $\calA := -\rlcP_0 \rlcM_1^{-1} \rlcK$)
$$
\beta^2 = O\left(\left(\kappa_M^2 + \frac{T^2\kappa_M^3\alpha_K^4 \expnorm(\calA)^2}{\sigma_1^4}
\right)\|\vec{x}_0\|^2 +
\left( \frac{T\alpha_K\kappa_M^2}{\sigma_1^3} + \frac{T^3\kappa_M^4\alpha_K^3 \expnorm(\calA)^2}{\sigma_1^5} \right)\|\vec{f}\|^2 \right),
$$
where we used $mh = T$ and $h=O(\sigma_1/\alpha_K)$.

We define $\mathbf{D} = \ket{0}^{a+1} \bra{0}^{a+1} \otimes \ket{T}\bra{T} \otimes U_O$ and construct a $(\alpha_O, a_O +2, \varepsilon_O)$-block-encoding of $\mathbf{D}$, which we denote by $U_D$. We can do this as given in the proof of Lemma~\ref{lem:meas_obs_gen}. Let us denote $\widetilde{D} = (\bra{0}^{a_O + 2} \otimes \id)U_D(\ket{0}^{a_O+2} \otimes \id)$. Using Lemma~\ref{lem:meas_obs_gen}, the desired estimate is then set as $\widehat{\theta} = \alpha_O \beta^2 \mu$ where $\mu$ is a $\upsilon/(3\alpha_O \beta^2)$-estimate of $\la \widehat{\phi} | \widetilde{D} | \widehat{\phi} \ra$ which can be obtained via a Hadamard test. The corresponding sample complexity and hence calls to the algorithm of Corollary~\ref{corr:sim_T_ind1} to prepare $\ket{\widehat{\phi}}$ is
$$
O(\alpha_O^2 \beta^4/\upsilon^2) = \widetilde{O}\left(\alpha_O^2 \poly(T \kappa_M \alpha_K \expnorm(\calA)(\|\vec{x}_0\| + \|\vec{f}\|)/\sigma_1 \Big)/\upsilon^2 \right).
$$
The accuracy $\varepsilon$ in the output of Corollary~\ref{corr:sim_T_ind1}, which appears as $\norm{\vec{x}(T) - \widetilde{x}(T)} \leq (\varepsilon/2) \vec{x}(T)$\footnote{The guarantee of Corollary~\ref{corr:sim_T_ind1} is written in terms of $\norm{\ket{x(T)} - \ket{\widehat{x}(T)}} \leq \varepsilon$ but this is obtained as a consequence of $\norm{\vec{x}(T) - \widetilde{x}(T)} \leq (\varepsilon/2) \vec{x}(T)$.} needs to be set as $\varepsilon =O(\upsilon/(\alpha_O \norm{\vec{x}(T)}^2))$.

The corresponding gate complexity is then the sample complexity times the gate complexity of Corollary~\ref{corr:sim_T_ind1} for this value of $\varepsilon$, without taking into consideration of the probability of success and hence an improvement by a factor of $g$. This completes the proof.
\end{proof}

\section{Algorithmic components for RLC circuits}\label{sec:algo_components_RLC}
In this section, we prepare for the application of the proposed quantum $\dae$ solver (Section~\ref{subsec:QDAE_solver_algo_analysis}) to simulating $\rlc$ circuits. Particularly, we will primarily focus on describing the block-encoding constructions required for the $\mna$ equations of $\rlc$ circuits. We restate the $\mna$ equations (Eq.~\eqref{eq:mna}) here for convenience:
\begin{align} 
    \label{eq:mna_algo_comp_RLC}
    \underbrace{
    \begin{bmatrix}
        \rlcA_{\rlccap} \rlcC \rlcA_{\rlccap}^T & 0 & 0 \\
        0 & \rlcL & 0 \\
        0 & 0 & 0 
    \end{bmatrix}}_{:=\rlcM}
    \frac{d}{dt} \begin{bmatrix}
        \vec{u}(t) \\ \vec{i}_\rlcind(t) \\ \vec{i}_{\rlcvs}(t)
    \end{bmatrix} + 
    \underbrace{
    \begin{bmatrix}
        \rlcA_\rlcres \rlcG \rlcA_\rlcres^T & \rlcA_{\rlcind} & \rlcA_{\rlcvs} \\
        -\rlcA_{\rlcind}^T & 0 & 0 \\
        -\rlcA_{\rlcvs}^T & 0 & 0 
    \end{bmatrix}}_{:=\rlcK}
    \begin{bmatrix}
        \vec{u}(t) \\ \vec{i}_{\rlcind}(t) \\ \vec{i}_{\rlcvs}(t)
    \end{bmatrix} = 
    \underbrace{
    \begin{bmatrix}
        - \rlcA_{\rlcjs} \vec{i}_{\rlcjs} \\ 0 \\ -\vec{v}_{\rlcvs}
    \end{bmatrix}}_{:=\vec{f}},
\end{align}
where $\rlcM$ and $\rlcK$ are as indicated.

\paragraph{Classical oracles.} We now quickly describe the classical oracles that we will use for constructing block-encodings for $\rlc$ circuits. Recall that we are working with an electrical circuit represented by a graph $\calG = (\calV,\calE)$. Let $|\calV|=N+1$ (including the reference node). When the $\rlc$ circuits have degree-$d$ (excluding the reference node), the reduced incidence matrix $\rlcA$ is at most $d$-row-sparse and $2$-column-sparse. We thus assume sparse-access oracles corresponding to the rows and columns (as described in Lemma~\ref{lem:sparse-block-enc}) for the reduced incidence matrices $\rlcA_{\rlccap}, \rlcA_{\rlcind}, \rlcA_{\rlcres}, \rlcA_{\rlcvs}, \rlcA_{\rlcjs}$, and component matrices $\rlcG, \rlcL, \rlcC$. We will denote the sparse-access oracle to the row information of any operator $\mathbf{B}$ by $O^r_{\mathbf{B}}$ and that to the column information of $\mathbf{B}$ as $O^c_{\mathbf{B}}$ e.g., $O^r_{\rlcA_{\rlccap}}$ is the sparse access oracle corresponding to the rows of $\rlcA_{\rlccap}$.\footnote{The oracles to the reduced incidence matrix could have been constructed from an oracle to the incidence matrix $\rlcA$ and an oracle that outputs the type of a component (i.e., resistor or capacitor or inductor) on an inputted edge. We omit these details here.} Note that for the specification of these oracles, an ordering is assigned to the nodes in $\calV$ and edges in $\calE$. 

\paragraph{State-preparation oracles.} We also assume access to oracles $O_f$ and $O_x$ that prepare the forcing state $\ket{f} := \vec{f}/\norm{\vec{f}}$ and the initial state $\ket{x(0)} := \vec{x}(0)/\norm{\vec{x}(0)}$ respectively.

\subsection{Properties of the MNA equations of RLC circuits}\label{subsec:props_mna}
In this section, we give spectral properties of the matrices involved in the $\mna$ equations of $\rlc$ circuits. This will allow us to comment on their block-encoding constructions in the following section.

\subsubsection{Norm of matrices}
\paragraph{Norm of $\rlcK$.}
\begin{claim}\label{claim:smax_K}
For a general $\rlc$ circuit $\calG$, the norm of the matrix $\rlcK$ can be bounded as
$$
\norm{\rlcK}_2 \leq \lambda_{\max}(\widetilde{\rlcG}) + \sqrt{\lambda_{\max}(\rlcA_{\rlcind} \rlcA_{\rlcind}^T) + \lambda_{\max}(\rlcA_{\rlcvs} \rlcA_{\rlcvs}^T)},
$$
where $\widetilde{\rlcG} = \rlcA_{\rlcres} \rlcG \rlcA_{\rlcres}^T$ is the nodal conductance matrix. When the $\rlc$ circuit $\calG$ (excluding the reference node) is promised to be degree-$d$ and satisfies Definition~\ref{def:structure_RLC_circs}, then the bound simplifies to $\norm{\rlcK}_2 \leq 2d \rlcres_\min^{-1} + \sqrt{2d}$.
\end{claim}
\begin{proof}
Note that we can decompose
$$ \rlcK = \underbrace{\begin{bmatrix} \rlcA_{\rlcres} \rlcG \rlcA_{\rlcres}^T & 0 & 0 \\ 0 & 0 & 0 \\ 0 & 0 & 0 \\ \end{bmatrix}}_{:= \rlcK_1} + \underbrace{\begin{bmatrix} 0 & \rlcA_{\rlcind} & \rlcA_{\rlcvs} \\ -\rlcA_{\rlcind}^T & 0 & 0 \\ -\rlcA_{\rlcvs}^T & 0 & 0 \\ \end{bmatrix}}_{:=\rlcK_2},$$
where $\rlcK_1$ is symmetric and $\rlcK_2$ is a skew-symmetric matrix. Using the triangle inequality, we obtain
\begin{align}\label{eq:interim_K_norm}
\norm{\rlcK}_2 \leq \norm{\rlcK_1}_2 + \norm{\rlcK_2}_2 = \lambda_{\max}(\rlcA_{\rlcres} \rlcG \rlcA_{\rlcres}^T) + \norm{\rlcK_2}_2.
\end{align}
To bound the norm of $\rlcK_2$, consider $\vec{x} = (\vec{u},\vec{i}_{\rlcind},\vec{i}_{\rlcvs})^T$ and note that
\begin{align}\label{eq:interim2_K_norm}
\norm{\rlcK_2 \vec{x}}_2^2 = \norm{ (\rlcA_{\rlcind}\vec{i}_{\rlcind} + \rlcA_{\rlcvs}\vec{i}_{\rlcvs}, -\rlcA_{\rlcind}^T \vec{u}, -\rlcA_{\rlcvs}^T\vec{u})^T }_2^2 = \norm{\rlcA_{\rlcind}\vec{i}_{\rlcind} + \rlcA_{\rlcvs}\vec{i}_{\rlcvs}}_2^2 + \norm{\rlcA_{\rlcind}^T \vec{u}}_2^2 + \norm{\rlcA_{\rlcvs}^T\vec{u}}_2^2.
\end{align}
We can bound the first term in Eq.~\eqref{eq:interim2_K_norm} as
\begin{equation} \label{eq:interim3_K_norm}
\norm{\rlcA_{\rlcind}\vec{i}_{\rlcind} + \rlcA_{\rlcvs}\vec{i}_{\rlcvs}}_2^2 \leq \left( \norm{\rlcA_{\rlcind}}_2 \|\vec{i}_{\rlcind}\|_2 + \norm{\rlcA_{\rlcvs}}_2 \|\vec{i}_{\rlcvs}\|_2 \right)^2 \leq \left(\norm{\rlcA_{\rlcind}}_2^2 + \norm{\rlcA_{\rlcvs}}_2^2 \right) \cdot \left(\|\vec{i}_{\rlcind}\|_2^2 + \|\vec{i}_{\rlcvs}\|_2^2 \right),
\end{equation}
where we have used the triangle inequality along with the sub-multiplicativity of $\norm{\cdot}_2$ in the first inequality and used Cauchy-Schwarz inequality in the second inequality. The second and third terms in Eq.~\eqref{eq:interim2_K_norm} can be bounded as
\begin{equation}\label{eq:interim4_K_norm}
\norm{\rlcA_{\rlcind}^T \vec{u}}_2^2 \leq \norm{\rlcA_{\rlcind}^T}_2^2 \|\vec{u}\|_2^2 = \norm{\rlcA_{\rlcind}}_2^2 \|\vec{u}\|_2^2, \qquad \norm{\rlcA_{\rlcvs}^T\vec{u}}_2^2 \leq \norm{\rlcA_{\rlcvs}^T}_2^2 \|\vec{u}\|_2^2 = \norm{\rlcA_{\rlcvs}}_2^2 \|\vec{u}\|_2^2.
\end{equation}
using the the sub-multiplicativity of $\norm{\cdot}_2$. Substituting Eq.~\eqref{eq:interim4_K_norm} and Eq.~\eqref{eq:interim3_K_norm} into Eq.~\eqref{eq:interim2_K_norm} gives us
$$
\norm{\rlcK_2 \vec{x}}_2^2 \leq \left(\norm{\rlcA_{\rlcind}}_2^2 + \norm{\rlcA_{\rlcvs}}_2^2 \right) \cdot \left(\|\vec{i}_{\rlcind}\|_2^2 + \|\vec{i}_{\rlcvs}\|_2^2 \right) + \norm{\rlcA_{\rlcind}}_2^2 \|\vec{u}\|_2^2 + \norm{\rlcA_{\rlcvs}}_2^2 \|\vec{u}\|_2^2 = \left(\norm{\rlcA_{\rlcind}}_2^2 + \norm{\rlcA_{\rlcvs}}_2^2 \right) \norm{\vec{x}}_2^2,
$$
which implies
\begin{equation}\label{eq:implication_norm_K2}
\norm{\rlcK_2}_2 = \sup_{\norm{\vec{x}}_2 = 1} \norm{\rlcK_2 \vec{x}}_2 \leq \sqrt{\norm{\rlcA_{\rlcind}}_2^2 + \norm{\rlcA_{\rlcvs}}_2^2}.    
\end{equation}
Substituting the above expression (Eq.~\eqref{eq:implication_norm_K2}) in Eq.~\eqref{eq:interim_K_norm} gives us
$$
\norm{\rlcK}_2 \leq \lambda_{\max}(\rlcA_{\rlcres} \rlcG \rlcA_{\rlcres}^T) + \sqrt{\norm{\rlcA_{\rlcind}}_2^2 + \norm{\rlcA_{\rlcvs}}_2^2} = \lambda_{\max}(\rlcA_{\rlcres} \rlcG \rlcA_{\rlcres}^T) + \sqrt{\lambda_{\max}(\rlcA_{\rlcind}\rlcA_{\rlcind}^T) + \lambda_{\max}(\rlcA_{\rlcvs}\rlcA_{\rlcvs}^T)} \,,
$$
which is the desired result for general $\rlc$ circuits. Using Fact~\ref{fact:spectrum_laplacian_G}, we can bound $\lambda_{\max}(\rlcA_{\rlcind}\rlcA_{\rlcind}^T) \leq 2d_{\rlcind}$ where $d_\rlcind$ is the maximal number of branches connected to any node in $\calG$ (excluding the reference node) containing inductors. Similarly, we can bound $\lambda_{\max}(\rlcA_{\rlcvs}\rlcA_{\rlcvs}^T) \leq 2d_{\rlcvs}$ where $d_\rlcind$ is the maximal number of branches connected to any node in $\calG$ (excluding the reference node) containing $\dc$ voltage sources. When we are promised that $\calG$ (excluding the reference node) is degree-$d$, we would have $d_\rlcind + d_\rlcvs \leq d$. Moreover, for such graphs $\calG$, we can bound $\lambda_{\max}(\rlcA_{\rlcres} \rlcG \rlcA_{\rlcres}^T) \leq 2d \lambda_{\max}(\rlcG)$ as we did in Claim~\ref{claim:lmax_M} and using Fact~\ref{fact:spectrum_laplacian_G}. We would then have
$$
\norm{\rlcK}_2 \leq \lambda_{\max}(\rlcA_{\rlcres} \rlcG \rlcA_{\rlcres}^T) + \sqrt{\norm{\rlcA_{\rlcind}}_2^2 + \norm{\rlcA_{\rlcvs}}_2^2} = 2d \lambda_{\max}(\rlcG) + \sqrt{2d} \,.
$$
This gives us the second desired result and completes the proof.
\end{proof}

\paragraph{Index zero.} In the case of index zero, note that the $\mna$ equations simplify to that of Eq.~\eqref{eq:mna_index0} as voltage sources are not allowed. We now provide claims which bound the norm of $\rlcM$.
\begin{claim}\label{claim:lmax_M}
For a general $\rlc$ circuit $\calG$, the norm of the matrix $\rlcM$ can be bounded as
$$\lambda_{\max}(\rlcM) \leq \max\{\lambda_{\max}(\widetilde{\rlcC}), \lambda_{\max}(\rlcL)\},$$
where $\widetilde{\rlcC} = \rlcA_{\rlccap} \rlcC \rlcA_{\rlccap}^T$ is the nodal capacitance matrix. When $\calG$ (excluding the reference node) is promised to have degree-$d$, then
$$\lambda_{\max}(\rlcM) \leq \max\{2d\lambda_{\max}(\rlcC), \lambda_{\max}(\rlcL)\}.$$    
\end{claim}
\begin{proof}
$\rlcM$ is a symmetric matrix and has a block diagonal structure. Suppose $\vec{x} = (\vec{u},\vec{i}_{\rlcind}, \vec{i}_{\rlcvs})^T$ then
$$ \vec{x}^T \rlcM \vec{x} = \vec{u}^T \rlcA_{\rlccap} \rlcC \rlcA_{\rlccap}^T \vec{u} + \vec{i}_{\rlcind}^T \rlcL \vec{i}_{\rlcind}.$$
Considering $\vec{x}$ to have unit norm and choosing $v$ (or $w$) to be an eigenvector of $\rlcA_{\rlccap} \rlcC \rlcA_{\rlccap}^T$ (or $\rlcL$) gives us that
$$\lambda_{\max}(\rlcM) \leq \max\{\lambda_{\max}(\rlcA_{\rlccap} \rlcC \rlcA_{\rlccap}^T), \lambda_{\max}(\rlcL)\}.$$
When $\calG$ (excluding the reference node) is promised to have degree-$d$, we can bound $\lambda_{\max}(\rlcA_{\rlccap} \rlcC \rlcA_{\rlccap}^T)$ as
\begin{align*}
    \lambda_{\max}(\rlcA_{\rlccap} \rlcC \rlcA_{\rlccap}^T) = \lambda_{\max}(\rlcC^{1/2} \rlcA_{\rlccap}^T \rlcA_{\rlccap} \rlcC^{1/2}) \leq \lambda_{\max}(\rlcC) \cdot \lambda_{\max}(\rlcA_{\rlccap}^T \rlcA_{\rlccap}) \leq 2d \lambda_{\max}(\rlcC),
\end{align*}
where the first equality follows from the fact that $\lambda_{\max}(\rlcB \rlcB^T) = \lambda_{\max}(\rlcB^T \rlcB)$ with $\rlcB = \rlcA_{\rlccap} \rlcC^{1/2}$ and $\rlcC$ is diagonal. The final inequality follows from $\lambda_{\max}(\rlcA_{\rlccap}^T \rlcA_{\rlccap}) = \lambda_{\max}(\rlcA_{\rlccap} \rlcA_{\rlccap}^T) \leq \lambda_{\max}(\rlcA \rlcA^T) \leq 2d$ where we used that $\rlcA_{\rlccap}$ is a submatrix of $\rlcA$ obtained by choosing certain columns of $\rlcA$, and Fact~\ref{fact:spectrum_laplacian_G}. This concludes the proof.
\end{proof}

\paragraph{Index one and two.} We now bound the norms of the matrices $\{\rlcM_j, \rlcK_j\}_{j \in [2]}$ encountered when solving $\mna$ systems (Eq.~\eqref{eq:mna} or~\eqref{eq:mna_algo_comp_RLC}) with index one and two.
\begin{claim}\label{claim:norm_M_K_ind1_ind2}
Let $\uptau > 0$. We have the following bounds on the norms of the matrices $\{\rlcM_j\}_{j \in [2]}$ for any general $\rlc$ circuit
\begin{enumerate}[$(i)$]
    \item $\norm{\rlcM_1}_2 \leq \max\{\lambda_{\max}(\widetilde{\rlcC}), \lambda_{\max}(\rlcL)\} + \lambda_{\max}(\widetilde{\rlcG}) + \sqrt{\lambda_{\max}(\rlcA_{\rlcind} \rlcA_{\rlcind}^T) + \lambda_{\max}(\rlcA_{\rlcvs} \rlcA_{\rlcvs}^T)}$,
    \item $\norm{\rlcM_2}_2 \leq \max\{\lambda_{\max}(\widetilde{\rlcC}), \lambda_{\max}(\rlcL)\} + (1 + 1/\uptau)(\lambda_{\max}(\widetilde{\rlcG}) + \sqrt{\lambda_{\max}(\rlcA_{\rlcind} \rlcA_{\rlcind}^T) + \lambda_{\max}(\rlcA_{\rlcvs} \rlcA_{\rlcvs}^T)})$ as long as $\sigma_\min(\rlcP_0 \widetilde{\rlcQ}_0) \geq \uptau$.
\end{enumerate}
where $\widetilde{\rlcC} = \rlcA_{\rlccap} \rlcC \rlcA_{\rlccap}^T$ is the nodal capacitance matrix and $\widetilde{\rlcG} = \rlcA_{\rlcres} \rlcG \rlcA_{\rlcres}^T$ is the nodal conductance matrix. When the $\rlc$ circuit $\calG$ (excluding the reference node) is promised to be degree-$d$ and satisfies Definition~\ref{def:structure_RLC_circs}, then the bounds simplify to
\begin{enumerate}[$(a)$]
    \item $\norm{\rlcM_1}_2 \leq \max\{2d \rlccap_\max, d \rlcind_\max\} + 2d \rlcres_\min^{-1} + \sqrt{2d}$,
    \item $\norm{\rlcM_2}_2 \leq \max\{2d \lambda_{\max}(\rlcC), \lambda_{\max}(\rlcL)\} + 4d \lambda_{\max}(\rlcG) + 2\sqrt{2d}$ as long as $\sigma_\min(\rlcP_0 \widetilde{\rlcQ}_0) \geq \uptau$.
\end{enumerate}
\end{claim}
\begin{proof}
The results follow from application of Claims~\ref{claim:lmax_M}--\ref{claim:smax_K} to bound the $2$-norms of $\rlcM$ and $\rlcK$ alongside the application of Claim~\ref{claim:ub_norm_Mi_Ki}.

The simplified results follow from application of Claims~\ref{claim:lmax_M}--\ref{claim:smax_K} when $\calG$ (excluding the reference node) is promised to be degree-$d$ and using Fact~\ref{fact:spectrum_RLC_components}. This completes the proof.
\end{proof}

\subsubsection{Condition number}
In the previous part of the section, we provided upper bounds on the norms of the matrices $\{\rlcM_j,\rlcK_j\}_{j \in [2]}$, which define the $\mna$ systems up to index two. We will now provide lower bounds on the $2$-norms of these matrices and thereby provide upper bounds on their condition numbers, which will be required later to comment on the complexity of the quantum algorithms for simulating the $\mna$ systems.

\paragraph{DAE of index zero.} Let us first consider the case when $\rlcM$ is promised to be non-singular i.e., the $\rlc$ circuit $\calG$ satisfies condition $(i)$ of Theorem~\ref{thm:index_mna}. If we are promised that $\lambda_{\min}(\rlcM) \geq \delta$, then an upper bound on the condition number immediately follows after using Claim~\ref{claim:lmax_M}. 

However, we show that we can also lower bound the minimum eigenvalue of $\rlcM$ by noting that Theorem~\ref{thm:index_mna} completely characterizes when $\rlcM$ in the $\mna$ equations is non-singular. This leads to the following result.
\begin{claim}\label{claim:kappa_M}
Suppose we are promised that the given $\rlc$ circuit $\calG$ satisfies condition $(i)$ of Theorem~\ref{thm:index_mna}. Let $\lambda_{\min}(\rlcC), \lambda_{\min}(\rlcL) > 0$. Then, $\rlcM$ is a non-singular matrix with
$$
\lambda_{\min}(\rlcM) \geq \min\{\lambda_{\min}(\rlcC) \cdot \sigma_{\min}(\rlcA_{\rlccap})^2, \lambda_\min(\rlcL) \}, \enspace \text{and} \enspace \kappa(\rlcM) \leq \frac{\max\{\lambda_{\max}(\widetilde{\rlcC}), \lambda_{\max}(\rlcL)\}}{\min\{\lambda_{\min}(\rlcC) \cdot \sigma_{\min}(\rlcA_{\rlccap})^2, \lambda_\min(\rlcL) \}}.
$$
When $\calG$ (excluding the reference node) is promised to be degree-$d$ and satisfies Definition~\ref{def:structure_RLC_circs}, we have the following bound:
$$
\kappa(\rlcM) \leq \frac{\max\{2d\lambda_{\max}(\rlcC), \lambda_{\max}(\rlcL)\}}{\min\{\lambda_{\min}(\rlcC) \cdot \sigma_{\min}(\rlcA_{\rlccap})^2, \lambda_\min(\rlcL) \}} 
\leq \frac{d \max\{2\rlccap_{\max}, \rlcind_\max\}}{\min\{\rlccap_\min \cdot \sigma_{\min}(\rlcA_{\rlccap})^2, \lambda_\min(\rlcL) \}}
$$
\end{claim}
\begin{proof}
As $\rlcM$ is a block-diagonal matrix, the spectrum of $\rlcM$ satisfies $\spec(\rlcM) = \spec(\rlcA_{\rlccap} \rlcC \rlcA_{\rlccap}^T) \cup \spec(\rlcL)$. We then have $\lambda_{\min}(\rlcM) \geq \min \{\lambda_{\min}(\rlcA_{\rlccap} \rlcC \rlcA_{\rlccap}^T), \lambda_{\min}(\rlcL) \}$. We now argue that $\lambda_{\min}(\rlcA_{\rlccap} \rlcC \rlcA_{\rlccap}^T)$ is guaranteed to be greater than $0$. 

For any $\vec{x} \neq 0, \norm{\vec{x}}_2=1$, we note that
$$
\vec{x}^T \rlcA_{\rlccap} \rlcC \rlcA_{\rlccap}^T \vec{x} = (\rlcA_{\rlccap}^T \vec{x})^T \rlcC (\rlcA_{\rlccap}^T \vec{x}) \geq \lambda_{\min}(\rlcC) \norm{\rlcA_{\rlccap}^T \vec{x}}_2^2 \geq \lambda_{\min}(\rlcC) \lambda_\min(\rlcA_{\rlccap} \rlcA_{\rlccap}^T),
$$
where we used the fact that $\calC$ is diagonal in the second inequality and that $\norm{\rlcA_{\rlccap}}_2^2 \|\vec{x}\|_2^2 \geq \lambda_\min(\rlcA_{\rlccap} \rlcA_{\rlccap}^T) = \sigma_{\min}(\rlcA_{\rlccap})^2$ by definition. As the graph $\calG$ satisfies $(i)$ in Theorem~\ref{thm:index_mna}, the graph $\calG$ must have a capacitive tree which implies that $\rlcA_{\rlccap}$ is a square matrix. From Fact~\ref{fact:rel_tree_A}, we have that $\rlcA_{\rlccap}$ is full rank and $\det(\rlcA_{\rlccap}) \in \{+1,-1\}$ which implies $\sigma_{\min}(\rlcA_{\rlccap}) > 0$. 

This then implies that $\lambda_{\min}(\rlcA_{\rlccap} \rlcC \rlcA_{\rlccap}^T) > 0$, and gives us the desired result on the lower bound on the minimum eigenvalue for general $\rlc$ circuits. The upper bound on $\kappa(\rlcM)$ follows from application of Claim~\ref{claim:lmax_M} and using Fact~\ref{fact:spectrum_RLC_components}.
\end{proof}

\paragraph{DAE of index one.}

We now consider the case where the $\rlc$ circuit $\calG$ satisfies the condition $(ii)$ of Theorem~\ref{thm:index_mna} and thus the $\mna$ system has index $1$. This implies that the matrix $\rlcM_1$ is non-singular. To bound the condition number of $\rlcM_1$, we require the following lemma that describes certain topological properties $\calG$ in terms of the reduced incidence matrices and projectors.
\begin{lemma}[\cite{tischendorf1999topological,estevez2000structural}]\label{lem:index1_incidence_mat_props}
Given a $\rlc$ circuit with independent voltage sources and current sources, then, the
following relations are satisfied.
\begin{enumerate}[$(i)$]
    \item The matrix $\begin{bmatrix} \rlcA_{\rlcres} & \rlcA_{\rlccap} & \rlcA_{\rlcvs} \end{bmatrix}$ has full row rank if and only if the circuit does not contain cutsets consisting of inductances and current sources only.
    \item Let $\rlcQ_{\rlccap}$ be any projector on to $\ker(\rlcA_{\rlccap})$. Then, the matrix $\rlcQ_{\rlccap}^T \rlcA_{\rlcvs}$ has full column rank if and only if the circuit does not contain loops with at least one voltage source and consisting of capacitances and voltage sources only.
\end{enumerate}
\end{lemma}
The lemma above combined with Theorem~\ref{thm:index_mna}$(ii)$ gives us guarantees on the properties of the matrices $\begin{bmatrix} \rlcA_{\rlcres} & \rlcA_{\rlccap} & \rlcA_{\rlcvs} \end{bmatrix}$ and $\rlcQ_{\rlccap}^T \rlcA_{\rlcvs}$ for index $1$. We now provide an upper bound on the condition number of $\rlcM_1$.

\begin{claim}\label{claim:kappa_M1}
Let $\lambda_{\min}^+(\rlcC), \lambda_{\min}(\rlcL) > 0$. Suppose the given $\rlc$ circuit $\calG$ satisfies condition $(ii)$ of Theorem~\ref{thm:index_mna}. Define
$$
\upsilon(\rlcA_{\rlcvs}^T \rlcQ_{\rlccap}, \widetilde{\rlcG}_{\rlccap}) = \min \left\{\frac{\sigma_{\min}(\rlcA_{\rlcvs}^T \rlcQ_{\rlccap})}{1 + (\|\widetilde{\rlcG}_{\rlccap}\|/\sigma_{\min}(\rlcA_{\rlcvs}^T \rlcQ_{\rlccap}))}, \frac{\min_{w \in \ker(\rlcA_{\rlcvs}^T \rlcQ_{\rlccap}, \norm{w}_2=1)}(w^T \widetilde{\rlcG}_{\rlccap} w)}{1 + (\|\widetilde{\rlcG}_{\rlccap}\|/\sigma_{\min}(\rlcA_{\rlcvs}^T \rlcQ_{\rlccap}))} \right\},    
$$
where $\widetilde{\rlcG}_{\rlccap} = \rlcQ_{\rlccap} \rlcA_{\rlcres} \rlcG \rlcA_{\rlcres}^T \rlcQ_{\rlccap}$. Then, $\rlcM_1$ is a non-singular matrix with 
\begin{align*}
\norm{\rlcM_1^{-1}} \leq O\left( \frac{\lambda_{\max}(\widetilde{\rlcG}) + \sqrt{\lambda_{\max}(\rlcA_{\rlcind} \rlcA_{\rlcind}^T) + \lambda_{\max}(\rlcA_{\rlcvs} \rlcA_{\rlcvs}^T)}}{\min\{ \lambda_{\min}^{+}(\widetilde{\rlcC}), \lambda_{\min}(\rlcL)\} \cdot \upsilon(\rlcA_{\rlcvs}^T \rlcQ_{\rlccap}, \widetilde{\rlcG}_{\rlccap})} \right)
\end{align*}
When $\calG$ (excluding the reference node) is promised to be degree-$d$ and satisfies Definition~\ref{def:structure_RLC_circs}, we have the following bound
$$
\norm{\rlcM_1^{-1}} \leq O\left( \frac{d \rlcres_{\min}^{-1} + \sqrt{d}}{\min\{ \lambda_{\min}^{+}(\widetilde{\rlcC}), \lambda_{\min}(\rlcL)\} \cdot \upsilon(\rlcA_{\rlcvs}^T \rlcQ_{\rlccap}, \widetilde{\rlcG}_{\rlccap})} \right)
$$
\end{claim}
\begin{proof}
Let $\rlcP = \rlcM^+ \rlcM$ and $\rlcQ = \rlcI - \rlcP$ be the orthogonal projectors as defined earlier. Let $\Pi_{\rlcP}$ and $\Pi_{\rlcQ}$ be the basis given by $\mathrm{Im}(\rlcP)$ and $\mathrm{Im}(\rlcQ)$ respectively, where $\mathrm{Im}(\cdot)$ is the image of the projector. We can then write the projectors in the basis $\{\Pi_\rlcP, \Pi_{\rlcQ}\}$ simply as
\begin{align}
\rlcP= \begin{bNiceArray}{cc}[first-col,first-row]
& \Pi_{\rlcP} & \Pi_{\rlcQ} \\
\Pi_{\rlcP} & \rlcI & 0 \\
\Pi_{\rlcQ} & 0 & 0 
\end{bNiceArray}, \qquad
\rlcQ= \begin{bNiceArray}{cc}[first-col,first-row]
& \Pi_{\rlcP} & \Pi_{\rlcQ} \\
\Pi_{\rlcP} & 0 & 0 \\
\Pi_{\rlcQ} & 0 & \rlcI
\end{bNiceArray},
\end{align}
and the matrices $\rlcM$ and $\rlcK$ as
\begin{align}
\rlcM=\begin{bNiceArray}{cc}[first-col,first-row]
& \Pi_{\rlcP} & \Pi_{\rlcQ} \\
\Pi_{\rlcP} & \mathbf{D} & 0 \\
\Pi_{\rlcQ} & 0 & 0 
\end{bNiceArray}, \qquad
\rlcK= \begin{bNiceArray}{cc}[first-col,first-row]
& \Pi_{\rlcP} & \Pi_{\rlcQ} \\
\Pi_{\rlcP} & \rlcK_{11} & \rlcK_{12} \\
\Pi_{\rlcQ} & \rlcK_{21} & \rlcK_{22}
\end{bNiceArray}\,.
\end{align}
The matrix $\rlcM_1 = \rlcM + \rlcK\rlcQ $ in this basis is then
\begin{equation}\label{eq:M1_basis_P_Q}
    \rlcM_1 = \begin{bmatrix}
        \mathbf{D} & \rlcK_{12} \\ 0 & \rlcK_{22}
    \end{bmatrix}\,.
\end{equation}
We are promised that $\rlcM_1$ is non-singular. Hence, $\mathbf{D}$ and $\rlcK_{22}$ are also necessarily non-singular. Using Schur complements, the inverse of $\rlcM_1$ can then be written as
\begin{equation}
    \rlcM_1^{-1} = \begin{bmatrix}
        \mathbf{D}^{-1} & -\mathbf{D}^{-1} \rlcK_{12} \rlcK_{22}^{-1} \\ 0 & \rlcK_{22}^{-1}
    \end{bmatrix}\,.
\end{equation}
Considering $\vec{x} = (\vec{x}_P, \vec{x}_Q)^T$ such that $\vec{x}_P \in \mathbb{R}^{\dim(\mathrm{Im}(\rlcP))}$ and $\vec{x}_Q \in \mathbb{R}^{\dim(\mathrm{Im}(\rlcQ))}$. We will now upper bound $\norm{\norm{\rlcM_1^{-1}}_2}$. Towards that end, we note that 
\begin{align*}
\norm{\rlcM_1^{-1} \vec{x}}_2 
&= \norm{\left(\mathbf{D}^{-1} \vec{x}_P - \mathbf{D}^{-1} \rlcK_{12} \rlcK_{22}^{-1} \vec{x}_Q, \rlcK_{22}^{-1} \vec{x}_Q \right)^T}_2 \\
&\leq \sqrt{2} \max \left\{ \norm{\mathbf{D}^{-1} \vec{x}_P - \mathbf{D}^{-1} \rlcK_{12} \rlcK_{22}^{-1} \vec{x}_Q}_2, \norm{\rlcK_{22}^{-1} \vec{x}_Q}_2 \right\} \\
&\leq \sqrt{2} \max \left\{ \norm{\mathbf{D}^{-1} \vec{x}_P}_2 + \norm{\mathbf{D}^{-1} \rlcK_{12} \rlcK_{22}^{-1} \vec{x}_Q}_2, \norm{\rlcK_{22}^{-1} \vec{x}_Q}_2 \right\} \\
&\leq \sqrt{2} \max \left\{ \norm{\mathbf{D}^{-1}}_2 \norm{\vec{x}_P}_2 + \norm{\mathbf{D}^{-1}}_2 \norm{\rlcK_{12}}_2 \norm{\rlcK_{22}^{-1}}_2 \norm{\vec{x}_Q}_2, \norm{\rlcK_{22}^{-1}}_2 \norm{\vec{x}_Q}_2 \right\},
\end{align*}
where we have used the triangle inequality in the third line and submultiplicativity of $\norm{\cdot}_2$ in the last line. Maximizing over all $\vec{x}$ such that $\norm{\vec{x}}_2 = 1$ (and thus $\norm{\vec{x}_P}_2, \norm{\vec{x}_Q}_2 \leq 1$) gives us
\begin{align}\label{eq:ub_norm_M1_interim1}
    \norm{\rlcM_1^{-1}}_2 \leq \sqrt{2} \max \left\{ \norm{\mathbf{D}^{-1}}_2 \left( 1 + \norm{\rlcK_{12}}_2 \norm{\rlcK_{22}^{-1}}_2\right), \norm{\rlcK_{22}^{-1}}_2\right\}
\end{align}
We now bound each of the involved norms $\norm{\mathbf{D}^{-1}}_2$, $\norm{\rlcK_{12}}_2$, and $\norm{\rlcK_{22}^{-1}}_2$ separately. We note that 
$$
\mathbf{D} = \rlcP \rlcM \rlcP = \rlcM |_{\supp(\rlcP)} \implies \lambda_{\min}(\mathbf{D}) = \min\{ \lambda_{\min}^{+}(\widetilde{\rlcC}), \lambda_{\min}(\rlcL)\},
$$
where we have used the fact that the spectrum of $\mathbf{D}$ is the non-zero spectrum of $\rlcM$ in the second inequality, denoted $\widetilde{\rlcC} = \rlcA_{\rlccap} \rlcC \rlcA_{\rlccap}^T$, and $\lambda_{\min}^{+}(\widetilde{\rlcC})$ as the minimum positive eigenvalue of the matrix $\widetilde{\rlcC}$. We thus have
\begin{equation}\label{eq:interim_ub_norm_invD}
    \norm{\mathbf{D}^{-1}}_2 \leq \frac{1}{\min\{ \lambda_{\min}^{+}(\widetilde{\rlcC}), \lambda_{\min}(\rlcL)\}}.
\end{equation}
To bound $\norm{\rlcK_{12}}_2$, we note that
\begin{equation}\label{eq:interim_ub_norm_K12}
    \norm{\rlcK_{12}}_2 = \norm{\rlcP \rlcK \rlcQ}_2 \leq \norm{\rlcK}_2 \leq \lambda_{\max}(\widetilde{\rlcG}) + \sqrt{\lambda_{\max}(\rlcA_{\rlcind} \rlcA_{\rlcind}^T) + \lambda_{\max}(\rlcA_{\rlcvs} \rlcA_{\rlcvs}^T)} \,,
\end{equation}
where the second inequality follows from the submultiplicativity of the $\norm{\cdot}_2$ and used Claim~\ref{claim:smax_K} in the final inequality. We will now bound $\norm{\rlcK_{22}^{-1}}_2$. We firstly note that
\begin{equation}\label{eq:expression_K22}
\rlcK_{22} = \rlcQ \rlcK \rlcQ = \begin{bmatrix}
    \rlcQ_{\rlccap} \rlcA_{\rlcres} \rlcG \rlcA_{\rlcres}^T \rlcQ_{\rlccap} & \rlcQ_{\rlccap} \rlcA_{\rlcvs} \\
    - \rlcA_{\rlcvs}^T \rlcQ_{\rlccap} & 0 
\end{bmatrix},
\end{equation}
where we have used that $\rlcQ_{\rlccap}$ is orthogonal i.e., $\rlcQ_{\rlccap}^T = \rlcQ_{\rlccap}$, $\rlcQ_{\rlccap} \rlcA_{\rlcvs}$ has full row rank and $-\rlcA_{\rlcvs}^T \rlcQ_{\rlccap}$ has full column rank due to Lemma~\ref{lem:index1_incidence_mat_props}. We are guaranteed that $\rlcK_{22}$ is non-singular as $\rlcM_1$ (Eq.~\eqref{eq:M1_basis_P_Q}) is non-singular. We note that $\rlcK_{22}$ is a saddle-point matrix and using analysis from \cite[Theorem~3.4]{benzi2005numerical}, it can be shown that
\begin{equation}\label{eq:lb_sigma_K22}
\sigma_{\min}(\rlcK_{22}) \geq \min \left\{\frac{\sigma_{\min}(\rlcA_{\rlcvs}^T \rlcQ_{\rlccap})}{1 + (\|\widetilde{\rlcG}_{\rlccap}\|/\sigma_{\min}(\rlcA_{\rlcvs}^T \rlcQ_{\rlccap}))}, \frac{\min_{w \in \ker(\rlcA_{\rlcvs}^T \rlcQ_{\rlccap}, \norm{w}_2=1)}(w^T \widetilde{\rlcG}_{\rlccap} w)}{1 + (\|\widetilde{\rlcG}_{\rlccap}\|/\sigma_{\min}(\rlcA_{\rlcvs}^T \rlcQ_{\rlccap}))} \right\},    
\end{equation}
where $\widetilde{\rlcG}_{\rlccap} = \rlcQ_{\rlccap} \rlcA_{\rlcres} \rlcG \rlcA_{\rlcres}^T \rlcQ_{\rlccap}$. Let us denote the lower bound of $\sigma_{\min}(\rlcK_{22})$ in Eq.~\eqref{eq:lb_sigma_K22} as $\upsilon(\rlcA_{\rlcvs}^T \rlcQ_{\rlccap}, \widetilde{\rlcG}_{\rlccap})$.

Substituting Eq.~\eqref{eq:interim_ub_norm_invD}, Eq.~\eqref{eq:interim_ub_norm_K12} and Eq.~\eqref{eq:lb_sigma_K22} into Eq.~\eqref{eq:ub_norm_M1_interim1} gives us
$$
\norm{\rlcM_1^{-1}}_2 \leq O\left( \frac{\lambda_{\max}(\widetilde{\rlcG}) + \sqrt{\lambda_{\max}(\rlcA_{\rlcind} \rlcA_{\rlcind}^T) + \lambda_{\max}(\rlcA_{\rlcvs} \rlcA_{\rlcvs}^T)}}{\min\{ \lambda_{\min}^{+}(\widetilde{\rlcC}), \lambda_{\min}(\rlcL) \upsilon(\rlcA_{\rlcvs}^T \rlcQ_{\rlccap}, \widetilde{\rlcG}_{\rlccap})} \right)
$$
The second part follows from noting that $\lambda_{\max}(\rlcA \rlcA^T) \leq 2d$ (Fact~\ref{fact:spectrum_laplacian_G}). This completes the proof.
\end{proof}

\subsection{Block-encoding incidence and nodal matrices} \label{ssec:block-enc-incidence-nodal}

\paragraph{Reduced incidence matrices.} We now describe how the block-encodings of the reduced incidence matrix $\rlcA$ and submatrices $\rlcA_\rlcind$, $\rlcA_\rlcind$ and $\rlcA_\rlcres$ which are the restrictions of $\rlcA$ to capacitative, inductive and resistive elements respectively. These use standard sparse-access oracles as defined at the beginning of the section (Section~\ref{sec:algo_components_RLC}).

\begin{claim}\label{claim:rlcA_BE}
Let $\varepsilon \in (0,1)$, $d,N \in \mathbb{N}$. Let $S = \{\rlcres, \rlcind, \rlccap, \rlcjs, \rlcvs\}$. Suppose $\calG$ is a $\rlc$ circuit over $N+1$ nodes with degree-$d$ excluding the reference node and satisfies Definition~\ref{def:structure_RLC_circs}. Let $\mathbf{A} \in \mathbb{R}^{N \times N}$ be the reduced incidence matrix of $\calG$. Then, there exists $(\sqrt{2d}, O(\log N),\varepsilon)$-block-encodings $U_{A_b}$ for the reduced incidence sub-matrices $\rlcA_b$ for all $b \in S$. It requires a single use of $O^r_{\rlcA_b},O^c_{\rlcA_b}$ oracles, two uses of $O_{\rlcA_b}$ and additional $O(\log(N) + \polylog(d/\varepsilon))$ one and two-qubit gates. 
\end{claim}
\begin{proof}
As $\rlcA$ is $d$-row-sparse and $2$-column-sparse, all its submatrices are also $d$-row-sparse and $2$-column sparse. Noting that $\rlcA$ takes values in $\{-1,0,1\}$ we then obtain $(\sqrt{2d}, O(\log N),\varepsilon)$-block-encodings for any reduced incidence submatrix using Lemma~\ref{lem:sparse-block-enc}, for say $\rlcA_\rlcres$. It requires a single use of $O^r_{\rlcA_\rlcres},O^c_{\rlcA_\rlcres}$, two uses of $O_{\rlcA_\rlcres}$ and additional $O(\log(N) + \polylog(d/\varepsilon))$ one and two-qubit gates.
\end{proof}

\paragraph{Component matrices.} We will also need block-encodings of the component matrices, $\rlcC$, $\rlcG$ ($:= \rlcR^{-1}$) and $\rlcL$. Recall that $\rlcC$ and $\rlcG$ are diagonal matrices with positive entries. As stated in Definition~\ref{def:structure_RLC_circs}, we assume that $\rlcC, \rlcG$ are diagonal matrices taking maximum values of $\rlccap_\max$ and $\rlcres_\min$. Note that $\rlcL$ is $d$-row and column sparse as any inductor in $\calG$ has at most $d_\ell \leq d - 1$ mutual inductances with other inductors. Moreover, we assume that its entries have maximum value at most $\ell_\max$. We then have the following claim regarding the block-encodings of the component matrices. 
\begin{claim}\label{claim:rlc_BEs_components}
Let $\varepsilon \in (0,1)$. Suppose $\calG$ is a $\rlc$ circuit over $N+1$ nodes with degree-$d$ excluding the reference node, satisfying Definition~\ref{def:structure_RLC_circs}. Then, there exists
\begin{enumerate}[$(i)$]
    \item $(\rlccap_\max, O(\log N), \varepsilon)$-block-encoding $U_C$ of $\rlcC$. It requires a single use of $O^r_{\rlcC},O^c_{\rlcC}$, two uses of $O_{\rlcC}$ and additional $O(\log(N) + \polylog(1/\varepsilon))$ one and two-qubit gates.
    \item $(\rlcres_\min^{-1}, O(\log N), \varepsilon)$-block-encoding $U_G$ of $\rlcG$. It requires a single use of $O^r_{\rlcG},O^c_{\rlcG}$, two uses of $O_{\rlcG}$ and additional $O(\log(N) + \polylog(1/\varepsilon))$ one and two-qubit gates.
    \item $(d \rlcind_\max, O(\log N), \varepsilon)$-block-encoding $U_L$ of $\rlcL$. It requires a single use of $O^r_{\rlcG},O^c_{\rlcG}$, two uses of $O_{\rlcG}$ and additional $O(\log(N) + \polylog(d^2/\varepsilon))$ one and two-qubit gates.
\end{enumerate}
\end{claim}
\begin{proof}
$(i)$, $(ii)$ The block-encodings follow from Lemma~\ref{lem:sparse-block-enc} noting that $\rlcC$ and $\rlcG$ are $1$-sparse diagonal matrices. Their maximum values are $\rlccap_\max$ and $\rlcres_\min$ respectively. 

$(iii)$ $\rlcL$ is $d$-row and column-sparse, with any element taking a maximum value of $\rlcind_\max$. The block-encoding then again follows from using Lemma~\ref{lem:sparse-block-enc}.
\end{proof}

\paragraph{MNA Equations.} The governing $\mna$ equations of $\rlc$ circuits are specified in terms of the matrices $\rlcM$ and $\rlcK$ (Eq.~\eqref{eq:mna}). We enumerate the cost of synthesizing their block-encodings below.
\begin{claim}\label{claim:rlc_M_K_BE}
Let $\varepsilon_M, \varepsilon_K \in (0,1)$. Suppose $\calG$ is a $\rlc$ circuit over $N+1$ nodes with degree-$d$ excluding the reference node, satisfying Definition~\ref{def:structure_RLC_circs}. Then, there exists 
\begin{enumerate}[$(i)$]
\item $(\alpha_M, a_M, \varepsilon_M)$-block-encoding $U_M$ of $\rlcM$ where 
$\alpha_M = 2dc_\max+ d \rlcind_\max$, and $a_M = O(\log N)$. It is implementable with $O(1)$ uses of $U_{A_\rlccap}$, $U_C$, and $U_L$. This requires $O(1)$ uses of their respective classical oracles and an additional gate complexity of $O(\log(N) + \polylog(d/\varepsilon))$.
\item $(\alpha_K, a_K, \varepsilon_K)$-block-encoding $U_K$ of $\rlcK$ where
$\alpha_K = 2d \rlcres_\min + \sqrt{2d}$ and $a_K = O(\log N)$. It is implementable with $O(1)$ uses of $U_G$, $U_{A_{\rlcres}}$, $U_{A_{\rlcind}}$ and $U_{A_{\rlcvs}}$. This requires $O(1)$ uses of their respective classical oracles and an additional gate complexity of $O(\log(N) + \polylog(d/\varepsilon))$.
\end{enumerate}
For LC circuits, there exists a $(4\sqrt{2d},O(\log N),\varepsilon_A)$-block-encoding $U_K$. For RLC circuits without voltage sources, there exists a $(2d \rlcres_\min^{-1} + 2\sqrt{2d},O(\log N),\varepsilon_A)$-block-encoding $U_K$.
\end{claim}
\begin{proof}
$(i)$ We note that $\rlcM = \ket{0}\bra{0} \otimes \rlcA_{\rlccap} \rlcC \rlcA_{\rlccap}^T + \ket{1}\bra{1} \otimes \rlcL$. We first obtain a block-encoding $\widetilde{\rlcC} = \rlcA_{\rlccap} \rlcC \rlcA_{\rlccap}^T$. Let $\varepsilon_1,\varepsilon_2,\varepsilon_3 \in (0,1)$ be error parameters to be determined later. Consider the $O(\sqrt{2d}, O(\log N), \varepsilon_1)$-block-encoding $U_{A_{\rlccap}}$ of $\rlcA_\rlccap$ obtained from Claim~\ref{claim:rlcA_BE} and the $O(\rlccap_\max, O(\log N), \varepsilon_2)$-block-encoding $U_C$ of $\rlcC$ from Claim~\ref{claim:rlc_BEs_components}. Using Lemma~\ref{lem:prod_BEs}, we then obtain a $(2d \rlccap_\max, O(\log N), 2 \sqrt{2d} \rlccap_\max \varepsilon_1 + 2d \varepsilon_2)$-block-encoding of $\widetilde{\rlcC}$, which we denote by $U_{\widetilde{C}}$. We also have a $(d \rlcind_\max, O(\log N), \varepsilon_3)$-block-encoding $U_L$ of $\rlcL$. From Claim~\ref{claim:proj_0state_BE}, we can obtain $(1,1,0)$-block-encodings of $\ket{0}\bra{0}$ and $\ket{1}\bra{1}$, which we denote by $U_0$ and $U_1$ respectively. We can combine $U_{\widetilde{C}}$ with $U_1$ using Lemma~\ref{lem:tensor_prod_BEs} to obtain a $O(2d \rlccap_\max, O(\log N), 2 \sqrt{2d} \rlccap_\max \varepsilon_1 + 2d \varepsilon_2)$-block-encoding of $\ket{0}\bra{0} \otimes \widetilde{\rlcC}$, which we denote by $V_1$. Similarly, we can combine $U_L$ with $U_2$ using Claim~\ref{claim:proj_0state_BE} to obtain a $(d \rlcind_\max, O(\log N), \varepsilon_3)$-block-encoding of $\ket{1}\bra{1} \otimes \rlcL$, which we denote by $V_2$. We then combine $V_1$ and $V_2$ using Lemma~\ref{lem:sum_BEs} to obtain a $(2d \rlccap_\max + d \rlcind_\max, O(\log N), \varepsilon_3 + 2 \sqrt{2d} \rlccap_\max \varepsilon_1 + 2d \varepsilon_2)$-block-encoding of $\rlcM$, which we denote by $U_M$. Setting $\varepsilon_3 = \varepsilon_M/3$, $\varepsilon_1 = \varepsilon_M/(6 \sqrt{2d})$ and $\varepsilon_2 = \varepsilon_M/(6d)$ gives us the desired precision. The gate complexity of $U_M$ is due to applications of $U_{A_\rlccap}$ (Claim~\ref{claim:rlcA_BE}), $U_C$ (Claim~\ref{claim:rlc_BEs_components}), and $U_L$ (Claim~\ref{claim:rlc_BEs_components}). This completes the proof of item $(i)$. 

$(ii)$ Let $\widetilde{\rlcG} = \rlcA_{\rlcres} \rlcG \rlcA_{\rlcres}^T$. We note that $\rlcK = \ket{0}\bra{0} \otimes \widetilde{\rlcG} + (\ket{0}\bra{1} \otimes \rlcA_{\rlcind} - \ket{1}\bra{0} \otimes \rlcA_{\rlcind}^T) + (\ket{0}\bra{2} \otimes \rlcA_{\rlcvs} - \ket{2}\bra{0} \otimes \rlcA_{\rlcvs}^T)$. Let $\varepsilon_1,\varepsilon_2,\varepsilon_3, \varepsilon_4 \in (0,1)$ be error parameters to be determined later. Using Lemma~\ref{lem:prod_BEs}, we obtain a $O(2d\rlcres_\min^{-1},O(\log N), 2\sqrt{2d} \rlcres_\min^{-1} \varepsilon_1 + 2d \varepsilon_2)$-block-encoding $U_{\widetilde{G}}$ of $\widetilde{\rlcG}$ using the $O(\rlcres_\min^{-1}, O(\log N), \varepsilon_2)$-block-encoding $U_G$ of $\rlcG$ (Claim~\ref{claim:rlc_BEs_components}) and $(\sqrt{2d}, O(\log N), \varepsilon_1)$-block-encoding $U_{A_{\rlcres}}$ of $\rlcA_{\rlcres}$ (Claim~\ref{claim:rlcA_BE}). We also have $(\sqrt{2d}, \log(N), \varepsilon_3)$-block-encoding $U_{A_{\rlcind}}$ of $\rlcA_{\rlcind}$ (Claim~\ref{claim:rlcA_BE}) and $(\sqrt{2d}, \log(N), \varepsilon_4)$-block-encoding $U_{A_{\rlcvs}}$ of $\rlcA_{\rlcvs}$ (Claim~\ref{claim:rlcA_BE}). We can also obtain $(1,1,0)$-block-encodings of $\ket{0}\bra{0}$, $\ket{0}\bra{1}$, etc., using Claim~\ref{claim:proj_0state_BE}. We can then take tensor products of these block-encodings with $U_{\widetilde{G}}$ (or $U_{A_{\rlcind}}$, $U_{A_{\rlcvs}}$) as in the expression of $\rlcK$ using Lemma~\ref{lem:tensor_prod_BEs}. Finally, we again apply Lemma~\ref{lem:sum_BEs} on the resulting block-encodings to obtain a $(\alpha_K,a_K,\varepsilon_K)$-block-encoding $U_K$ of $\rlcK$ where
$$
\alpha_K = 2d \rlcres_\min^{-1} + 4\sqrt{2d}, \quad a_K = O(\log N), \quad \varepsilon_K = 2\sqrt{2d} \rlcres_\min^{-1} \varepsilon_1 + 2d \varepsilon_2 + \varepsilon_3 + \varepsilon_4.
$$
It suffices to set $\varepsilon_1 = O(\varepsilon_K/\sqrt{d})$, $\varepsilon_2 = O(\varepsilon_K/d)$, and $\varepsilon_3 = \varepsilon_4 = O(\varepsilon_K)$ to obtain the desired precision. The main contribution to gate complexity is due to applications of $U_G$ (Claim~\ref{claim:rlc_BEs_components}) and the reduced incidence matrices of $U_{A_\rlcres}$, $U_{A_\rlcvs}$, $U_{A_\rlcind}$ (Claim~\ref{claim:rlcA_BE}). 

For LC circuits, we do not need to construct the $\widetilde{G}$ matrix which is the zero matrix. We can then obtain a $(4\sqrt{2d},O(\log N),\varepsilon_A)$-block-encoding there. For RLC circuits without voltage sources, we do not need to account for $\rlcA_{\rlcvs}$ and we obtain a $(2d \rlcres_\min^{-1} + 2\sqrt{2d},O(\log N),\varepsilon_A)$-block-encoding. This completes the proof of item $(ii)$ and the overall proof of the claim.
\end{proof}

\paragraph{Block-encoding projectors and dynamics.} 
We can now use these approximate block-encodings $U_M$ of $\rlcM$ and $U_K$ of $\rlcK$ (particular to the $\mna$ equations) for constructing block-encodings relevant for simulating the dynamics of $\rlc$ circuits, as was done in Section~\ref{subsec:useful_BEs} for general $\rlcM,\rlcK$. 

\subsection{Initial state and outputs}

\paragraph{Initial state and sources.}
As described in Problem~\ref{prob:RLC_dynamics}, we assume we are given a quantum circuit to prepare the initial state of interest $\vec{x}(0) = (\vec{u}(0),\vec{i}_{\rlcind}(0),\vec{i}_{\rlcvs})^T$ for general index and $\vec{x}(0) = (\vec{u}, \vec{i}_{\rlcind})^T$ for index $0$ where voltage sources are not permitted.

We may also be given the initial state as $(\vec{v}(0),\vec{i}_{\rlcind}(0),\vec{i}_{\rlcvs}(0))^T$ involving the branch voltages $\vec{v}(0)$ at time $0$. As our quantum algorithm works with the $\mna$ system, we require an initial state involving the node voltages. To determine the node voltages of the branch voltages, we note that they are related as $\vec{v}(0) = \rlcA^T \vec{u}(0)$. As $\rlcA$ has full row rank, we can determine $\vec{u}(0)$ uniquely by solving
$$
\vec{u}(0) = (\rlcA\rlcA^T)^{-1} \rlcA \vec{v}(0).
$$
Since we are able to construct block-encodings of $\rlcA$ efficiently, we can then determine $\vec{u}(0)$ efficiently as part of our quantum algorithm as well.

\paragraph{Outputs: dissipated power and stored energy.} 

The voltage drops across the resistors can be obtained from the nodal voltages $\vec{u}$ as 
$\vec{v}_{\rlcres} = \rlcA_{\rlcres} \vec{u}$. The dissipated power at any time is then given by
$$P_R = \vec{v}_{\rlcres}^T \rlcG \vec{v}_{\rlcres} = \vec{u}^T \rlcA_{\rlcres}^T \rlcG \rlcA_{\rlcres} \vec{u}.$$

Likewise, we can also determine the energy stored in the capacitors 
$$
E_C = \frac{1}{2} \vec{v}_{\rlccap}^T \rlcC \vec{v}_{\rlccap} = \frac{1}{2} \vec{u}_{\rlccap}^T \rlcA_\rlccap^T \rlcC \rlcA_\rlccap \vec{u}_{\rlccap}.
$$
And lastly, the energy stored in the inductors $E_L$ can be obtained from the currents across the inductors $i_l$ as
$$
E_L = \frac{1}{2} \vec{i}_\ell^T \rlcL \vec{i}_\ell.
$$
All of these quantities can be determined to within additive precision follwoing the procedure described in Section~\ref{sec:measure-observables}.

\section{Simulating dynamics of RLC circuits}\label{sec:sim_RLC}
In this section, we consider the problem of simulating $\rlc$ circuits (Problem~\ref{prob:RLC_dynamics}) and apply the quantum $\dae$ solver (Section~\ref{subsec:QDAE_solver_algo_analysis}) to the $\mna$ equations (Eq.~\eqref{eq:mna}) governing the dynamics of $\rlc$ circuits. These $\dae$s have at most index two (Theorem~\ref{thm:index_mna}) depending on the topological structure of $\calG$. We also comment on outputting physically relevant quantities such as stored energy or dissipated power (Problem~\ref{prob:classical_output}) for each case. We state the complexities in terms of queries to the block-encoding unitaries $U_M$ and $U_K$ where we recall that $\rlcM$ and $\rlcK$ define the $\dae$, see Eqn.~\eqref{eq:dae-linear}. Each query to $U_M$ or $U_K$ require $O(1)$ queries to the $\rlc$-circuit oracles $O$ as shown in Claim~\ref{claim:rlc_BEs_components}, and thus our statements below also reflect the complexities in terms of the latter.

\subsection{DAE of index zero}\label{sec:rlc_ind0}
In this section, we consider $\rlc$ circuits satisfying condition $(i)$ of Theorem~\ref{thm:index_mna} and thus, the $\mna$ equations (Eq.~\eqref{eq:mna}) are of index zero. For simulating such $\rlc$ circuits (Problem~\ref{prob:RLC_dynamics}), we consider the approach described in Section~\ref{sec:master_qa} and give the resulting theorem in Section~\ref{subsec:analysis_index0}. We also show that it is possible to obtain a quantum speedup in outputting energy corresponding to different components in the $\rlc$ circuit, thereby solving Problem~\ref{prob:classical_output}.

Finally, we show that index zero already captures the hardness of simulating $\rlc$ circuits by showing it is $\bqp$-complete in Section~\ref{sec:bqp}.

\subsubsection{Algorithm}\label{subsec:qa_index0}
The $\mna$ equations for index $0$ are
\begin{align} 
    \label{eq:rlc_mna_index0}
    \underbrace{
    \begin{bmatrix}
        \rlcA_{\rlccap} \rlcC \rlcA_{\rlccap}^T & 0 \\
        0 & \rlcL 
    \end{bmatrix}}_{:=\rlcM}
    \frac{d}{dt} \begin{bmatrix}
        \vec{u}(t) \\ \vec{i}_\rlcind(t)
    \end{bmatrix} + 
    \underbrace{
    \begin{bmatrix}
        \rlcA_\rlcres \rlcG \rlcA_\rlcres^T & \rlcA_{\rlcind} \\
        -\rlcA_{\rlcind}^T & 0
    \end{bmatrix}}_{:=\rlcK}
    \begin{bmatrix}
        \vec{u}(t) \\ \vec{i}_{\rlcind}(t)
    \end{bmatrix} = 
    \underbrace{
    \begin{bmatrix}
        - \rlcA_{\rlcjs} \vec{i}_{\rlcjs} \\ 0
    \end{bmatrix}}_{:=\vec{f}}
\end{align}
As $\rlcA_\rlccap \rlcC \rlcA_\rlccap^T$ is invertible, the governing ODE of the $\rlc$ circuits can be directly extracted:
\begin{align}
    \dot{\vec{x}}(t) = -\rlcM^{-1} \rlcK \vec{x}(t) + \rlcM^{-1} \vec{f}
\end{align}
where $\rlcM$ and $\rlcK$ are as defined in Eq.~\eqref{eq:rlc_mna_index0}. We will use the algorithm of Theorem~\ref{thm:sim_ind0} to solve the above $\ode$.

\subsubsection{Analysis}\label{subsec:analysis_index0}

\paragraph{Norm of the exponential.}

We first prove the following claim regarding general $\rlc$ circuits with $\mna$ equations of index $0$, in terms of the properties of $\rlcM$ and $\rlcK$.
\begin{claim}\label{claim:rlc_expnorm_M_index0}
    Suppose the $\rlc$ circuit $\calG$ satisfies condition $(i)$ of Theorem~\ref{thm:index_mna}. Then, the norm of the exponential of $-\rlcM^{-1}\rlcK$ satisfies
    $$\expnorm(- \rlcM^{-1} \rlcK) \leq \sqrt{\kappa(\rlcM)}.$$
\end{claim}
\begin{proof}
Firstly, we note that we can write $-\rlcM^{-1}\rlcK$ as follows:
\begin{align}\label{eq:interim_norm_ind1}
    -\rlcM^{-1} \rlcK & = -\sqrt{ \rlcM^{-1}} \sqrt{\rlcM^{-1}} \rlcK \sqrt{\rlcM^{-1}}\sqrt{\rlcM} \\
    &=\sqrt{\rlcM^{-1}}(-\sqrt{\rlcM^{-1}} \rlcK_1\sqrt{ \rlcM^{-1}}-\sqrt{\rlcM^{-1}} \rlcK_2\sqrt{\rlcM^{-1}})\sqrt{\rlcM} ,
\end{align}
where in the second line, we split $\rlcK$ into $\rlcK_1$, which is the Hermitian part of $\rlcK$ and $\rlcK_2$, which is the anti-Hermitian part of $\rlcK$, with the following matrices
$$
\rlcK_1 = \frac{\rlcK + \rlcK^\dagger}{2} = \begin{bmatrix}
        \rlcA_\rlcres \rlcG \rlcA_\rlcres^T & 0 \\
        0 & 0
    \end{bmatrix}, \enspace 
\rlcK_2 = \frac{\rlcK - \rlcK^\dagger}{2} = \begin{bmatrix}
        0 & \rlcA_{\rlcind} \\
        -\rlcA_{\rlcind}^T & 0
    \end{bmatrix}.
$$
We can then bound the exponential norm as
\begin{align}
    \expnorm(- \rlcM^{-1} \rlcK) &= \sqrt{\rlcM^{-1}}\cdot\expnorm(-\sqrt{\rlcM^{-1}} \rlcK\sqrt{ \rlcM^{-1}})\cdot\sqrt{\rlcM}\\
    &\leq \|\sqrt{ \rlcM^{-1}}\| \cdot \expnorm(-\sqrt{\rlcM^{-1}} \rlcK_1\sqrt{ \rlcM^{-1}}-\sqrt{ \rlcM^{-1}} \rlcK_2\sqrt{\rlcM^{-1}}) \cdot \|\sqrt{\rlcM}\|\\
    &\leq \kappa(\sqrt{\rlcM}) \cdot \expnorm(-\sqrt{\rlcM^{-1}} \rlcK_1 \sqrt{\rlcM^{-1}})\,,
\end{align}
where the first inequality follows from the sub-multiplicativity of norms 
and the second inequality uses the definition of $\kappa(\sqrt{\rlcM})$ along with Fact~\ref{fact:ub_exp_norm}. Next, we observe that the non-zero eigenvalues of $- \rlcK_1$ are negative as the non-zero sub-matrix $\rlcA_\rlcres \rlcG \rlcA_\rlcres^T$ is semipositive definite 
\footnote{Note that $\vec{x}^T \rlcA_\rlcres \rlcG \rlcA_\rlcres^T \vec{x} = (\sqrt{G} \rlcA_\rlcres^T \vec{x})^T \sqrt{G} \rlcA_\rlcres^T \vec{x} = \norm{\sqrt{G} \rlcA_\rlcres^T \vec{x}}_2^2 \geq 0, \,\, \forall \vec{x} \text{ s.t.} \norm{\vec{x}}_2 = 1$.}. \

This implies that the non-zero eigenvalues of $\sqrt{\rlcM^{-1}} \rlcK_1 \sqrt{\rlcM^{-1}}$ are also negative. This implies
\begin{equation}
    \expnorm(-\sqrt{ \rlcM^{-1}} \rlcK_1 \sqrt{\rlcM^{-1}})\leq 1\,.
\end{equation}
Combining the above equation with Eq.~\eqref{eq:interim_norm_ind1} gives us that
\begin{equation}
    \expnorm(-\rlcM^{-1}\rlcK) \leq \kappa(\sqrt{\rlcM}) = \sqrt{\kappa(\rlcM)}\,,
\end{equation}
where the final equality follows from the fact that $\rlcM$ is symmetric and positive definite. This concludes the proof.
\end{proof}

For $\rlc$ circuits where the corresponding graph has degree-$d$ (excluding the reference node) and correspond to index $0$, we then immediately have from Claim~\ref{claim:kappa_M} that $\rlcM^{-1} \rlcK$ satisfies
$$
\expnorm(- \rlcM^{-1} \rlcK) \leq \sqrt{\frac{d \max\{2 \rlccap_\max, \rlcind_\max \}}{\min\{\rlccap_\min \cdot \lambda_{\min}(\rlcA_{\rlccap} \rlcA_{\rlccap}^T), \lambda_\min(\rlcL) \}}}.
$$

\paragraph{Simulating degree-$d$ $\rlc$ circuits.}
\begin{theorem}\label{thm:rlc_sim_ind0_degd}
Let $\varepsilon \in (0,1)$ and $i_\max > 0$. Consider the setup of Problem~\ref{prob:RLC_dynamics}. Suppose $\calG$ is a $\rlc$ circuit with degree-$d$ excluding the reference node and satisfies condition $(i)$ of Theorem~\ref{thm:index_mna}. Let $\norm{\vec{i_{\rlcjs}}} \leq i_\max$. Then, there is a quantum algorithm that outputs $|\widehat{\Psi}\rangle$ such that $\left\| |\widehat{\Psi}\rangle - \normhist^{-1} \sum_{k=0}^{m} \norm{\vec{x}(k \Delta t)} \ket{k, x(k \Delta t)} \right\| \leq \varepsilon$ with success probability $\geq 2/3$ where $\Delta t = \sigma/(2d \rlcres_\min + \sqrt{2d})$ and $m = \ceil{T/\Delta t}$. Define $\mu^2 := m^{-1} \sum_{j=1}^m \norm{\vec{x}(j\Delta t)}^2$ and
$$
\sigma:= \min\{\rlccap_{\min} \lambda_{\min}(\rlcA_{\rlccap}\rlcA_{\rlccap}^T), \lambda_\min(\rlcL) \}, \quad \kappa_M := \frac{2 \max\{\rlccap_\max, \rlcind_\max\}}{\sigma}, \quad \mathcal{C}_f := \log\left(1 + \frac{2 T e^2 d i_\max}{\mu}\right).
$$
The algorithm has the following complexity
\begin{align*}
\text{Queries to }U_{M}&: \widetilde{O}\left(T^2 \cdot \poly\left(d, \log(N), \kappa_M, r_\min^{-1}, \log(1/\varepsilon), \mathcal{C}_f \right) \right) \,,\\
\text{Queries to }U_{K}&: \widetilde{O}\left(T^2 \cdot \poly\left(d, \log(N), \kappa_M, r_\min^{-1}, \log(1/\varepsilon), \mathcal{C}_f \right) \right) \,,\\
\text{Queries to }O_{f},O_x&: O\left(T d \rlcres_\min^{-1} \kappa_M^{3/2} \log(1/\varepsilon)\right) \,,
\end{align*}
An additional gate complexity $O(1)$ times the number of uses of $U_M$ and $U_K$ is also used. It is sufficient to set the precisions of the block-encodings of $U_M$ and $U_K$ as
$$
\varepsilon_M = \poly\left(\frac{\sigma \varepsilon}{T d \, \rlcres_\min^{-1} \, \kappa_M \, \calC_f}\right) \text{ and } \varepsilon_K = \poly\left(\frac{\sigma \varepsilon}{T d \rlcres_\min^{-1} \kappa_M \calC_f }\right).
$$
\end{theorem}
\begin{proof}
We will use the algorithm from Theorem~\ref{thm:sim_ind0}. We now estimate the cost specific to $\rlc$ circuits. We note that the relevant forcing here is $\vec{f} = \rlcA_{\rlcjs} \vec{i}_{\rlcjs}$ and satisfies $\norm{\rlcA_{\rlcjs} \vec{i}_{\rlcjs}} \leq 2d i_\max$ where we used Fact~\ref{fact:spectrum_laplacian_G}. We thus define
$$
\mathcal{C}_f := \log\left(1 + \frac{2 T e^2 d i_\max}{\mu}\right).
$$
We now use the following bounds 
\begin{align*}
    \lambda_{\min}(\rlcM) &\geq \min\{\rlccap_{\min} \lambda_{\min}(\rlcA_{\rlccap}\rlcA_{\rlccap}^T), \lambda_\min(\rlcL) \} \, \qquad &\text{(Claim~\ref{claim:kappa_M})} \\
    \kappa(\rlcM) &\leq \frac{d \max\{2 \rlccap_\max, \rlcind_\max \}}{\min\{\rlccap_\min \cdot \lambda_{\min}(\rlcA_{\rlccap} \rlcA_{\rlccap}^T), \lambda_\min(\rlcL) \}} \, \qquad &\text{(Claim~\ref{claim:kappa_M})} \\
    \alpha_K &\leq 2d \rlcres_\min^{-1} + \sqrt{2d} \, \qquad &\text{(Claim~\ref{claim:smax_K})} \\ 
    \expnorm(-\rlcM^{-1}\rlcK) &\leq \sqrt{\kappa(\rlcM)} \, \qquad &\text{(Claim~\ref{claim:rlc_expnorm_M_index0})}
\end{align*}
We denote the lower bound on $\lambda_\min(\rlcM)$ from the first line as $\sigma$ and the upper bound on $\kappa(\rlcM)$ from the second line as $\kappa_M$. Inserting these bounds into the complexities of Theorem~\ref{thm:sim_ind0} gives us the desired result.
\end{proof}

We can also prepare a state proportional to the solution at time $T$ i.e., $\ket{\psi(T)}$ using Corollary~\ref{corr:sim_T_ind0}. We then obtain the following result.
\begin{corollary}\label{corr:rlc_sim_T_ind0}
Let $\varepsilon \in (0,1)$ and $i_\max > 0$. Consider context of Theorem~\ref{thm:rlc_sim_ind0_degd}. Then, there is a quantum algorithm that outputs $\ket{\phi}$ such that $\left\| \ket{\phi} - \ket{x(T)} \right\| \leq \varepsilon$ with success probability $\geq 2/3$. Define $\sigma:= \min\{\rlccap_{\min} \lambda_{\min}(\rlcA_{\rlccap}\rlcA_{\rlccap}^T),\lambda_\min(\rlcL) \}$ and
$$
\kappa_M := \frac{2 \max\{\rlccap_\max, \rlcind_\max\}}{\sigma}, \quad g = \frac{\max_{t\in[0,T]}\|\vec{x}(t)\|}{\|\vec{x}(T)\|}, \quad \mathcal{C}_f := \log\left(1 + \frac{\sqrt{2d} T e^2 i_\max}{\norm{\vec{x}(T)}}\right).
$$
The algorithm has the following complexity
\begin{align*}
\text{Queries to }U_{M}&: \widetilde{O}\left(g T \cdot \poly\left(d, \log(N), \kappa_M, r_\min^{-1}, \log(1/\varepsilon), \mathcal{C}_f \right) \right) \,,\\
\text{Queries to }U_{K}&: \widetilde{O}\left(g T \cdot \poly\left(d, \log(N), \kappa_M, r_\min^{-1}, \log(1/\varepsilon), \mathcal{C}_f \right) \right) \,,\\
\text{Queries to }O_{f},O_x&: O\left(g T d \rlcres_\min^{-1} \kappa_M^{3/2} \log(1/\varepsilon)\right) \,,
\end{align*}
An additional gate complexity $O(1)$ times the number of uses of $U_M$ and $U_K$ is also used. It is sufficient to set the precisions of the block-encodings of $U_M$ and $U_K$ as
$$
\varepsilon_M = \poly\left(\frac{\sigma \varepsilon}{T d \, \rlcres_\min^{-1} \, \kappa_M \, \calC_f}\right) \text{ and } \varepsilon_K = \poly\left(\frac{\sigma \varepsilon}{T d \rlcres_\min^{-1} \kappa_M \calC_f }\right).
$$
\end{corollary}
Note that the corresponding cost is the same as that in Theorem~\ref{thm:rlc_sim_ind0_degd} times a factor of $g = \frac{\max_{t\in[0,T]}\|\vec{x}(t)\|}{\|\vec{x}(T)\|}$ and we save a factor of $T$. We omit the proof as it follows from Corollary~\ref{corr:sim_T_ind0} and proceeds similarly to the proof of Theorem~\ref{thm:rlc_sim_ind0_degd}.

\paragraph{Energy output.}
\begin{theorem}\label{thm:rlc_energy_ind0_degd}
Let $\upsilon \in (0,1)$ and $i_\max > 0$. Let $\varepsilon \in (0,1)$. Consider the setup of Problem~\ref{prob:classical_output}. Suppose $\calG$ is a $\rlc$ circuit with degree-$d$ excluding the reference node and satisfies condition $(i)$ of Theorem~\ref{thm:index_mna} i.e., the $\mna$ $\dae$ has index $0$. Let $\norm{\vec{i_{\rlcjs}}} \leq i_\max$. Let the algorithm of Corollary~\ref{corr:rlc_sim_T_ind0} be $\calQ$. For a specified $\calS$ subset of capacitors in $\calG$, there is a quantum algorithm that outputs $\widehat{\theta}$ with success probability $\geq 2/3$ such that $\widehat{\theta}$ is an $\upsilon$-close estimate of the energy across $\calS$ at time $T$. Define
$$
\sigma:= \min\{\rlccap_{\min} \lambda_{\min}(\rlcA_{\rlccap}\rlcA_{\rlccap}^T), \lambda_\min(\rlcL) \}, \quad \kappa_M := \frac{2 \max\{\rlccap_\max, \rlcind_\max\}}{\sigma}, \quad \mathcal{C}_f := \log\left(1 + \frac{2 T e^2 d i_\max}{\norm{\vec{x}(T)}}\right).
$$
The algorithm has the following complexity
\begin{align*}
\text{Calls to }\calQ \,\, (S_\calQ)&: \widetilde{O}\left(\poly\Big(d \, T \, \kappa_M \, \rlcres_\min^{-1}, 1/\sigma, \norm{\vec{u}(0)}, \|\vec{i}_\rlcind(0)\|, i_\max \Big)/\upsilon^2\right)\,,\\
\text{Queries to }U_{M}&: \widetilde{O}\left(S_\calQ \poly\left(T, d, \log(N), \kappa_M, r_\min^{-1}, \log(d \rlccap_\max/\upsilon), \mathcal{C}_f \right) \right) \,,\\
\text{Queries to }U_{K}&: \widetilde{O}\left(S_\calQ \poly\left(T, d, \log(N), \kappa_M, r_\min^{-1}, \log(d \rlccap_\max/\upsilon), \mathcal{C}_f \right) \right) \,,\\
\text{Queries to }O_{f},O_x&: \poly\left(S_\calQ T d \rlcres_\min^{-1} \kappa_M^{3/2} \log(d \rlccap_\max/\upsilon)\right) \,,
\end{align*}
An additional gate complexity $O(1)$ times the number of uses of $U_M$ and $U_K$ is also used. It is sufficient to set the precisions of the block-encodings of $U_M$ and $U_K$ as
$$
\varepsilon_M = \poly\left(\frac{\sigma \upsilon}{T d \, \rlcres_\min^{-1} \, \kappa_M \, \calC_f \norm{\vec{x}(T)}}\right) \text{ and } \varepsilon_K = \poly\left(\frac{\sigma \upsilon}{T d \rlcres_\min^{-1} \kappa_M \calC_f \norm{\vec{x}(T)}}\right).
$$
\end{theorem}
\begin{proof}
We use the algorithm of Theorem~\ref{thm:output_ind0} which here involves using the algorithm of Corollary~\ref{corr:rlc_sim_T_ind0} followed by the estimation protocol of Lemma~\ref{lem:meas_obs_gen}. We now estimate the cost specific to $\rlc$ circuits and for the specification of the observable $\mathbf{O}$ here. Given $\calS$, the observable $\mathbf{O}$ of interest is
$$
\mathbf{O} = \begin{bmatrix}
    \rlcA_{\rlccap,\calS} \rlcC \rlcA_{\rlccap,\calS}^T & 0 \\ 0 & 0,
\end{bmatrix}
$$
where $\rlcA_{\rlccap,\calS}$ is the reduced incidence matrix over capacitors in $\calS$. Let $\varepsilon_O$ be an error parameter to be fixed later. From the proof of Claim~\ref{claim:rlc_BEs_components}, we have a $(\alpha_O, a_O, \varepsilon_O)$-block-encoding of $\mathbf{O}$, which we denote by $U_O$, where
\begin{equation}\label{eq:params_UO_ind0}
\alpha_O = 2d \rlccap_\max, \quad a_O = O(\log N).
\end{equation}
To bound the sample and time complexity as given by Theorem~\ref{thm:output_ind0}, we note that the forcing $\vec{f} = [\rlcA_{\rlcjs} \vec{i}_{\rlcjs}, \vec{0}]^T$ satisfies $\norm{\vec{f}} \leq \norm{\rlcA_{\rlcjs} \vec{i}_{\rlcjs}} \leq \sqrt{2d} i_\max$ where we used Fact~\ref{fact:spectrum_laplacian_G}. We thus define
$$
\mathcal{C}_f := \log\left(1 + \frac{\sqrt{2d} T e^2 i_\max}{\norm{\vec{x}(T)}}\right).
$$
We now use the following bounds 
\begin{align*}
    \lambda_{\min}(\rlcM) &\geq \min\{\rlccap_{\min} \lambda_{\min}(\rlcA_{\rlccap}\rlcA_{\rlccap}^T), \lambda_\min(\rlcL) \} \, \qquad &\text{(Claim~\ref{claim:kappa_M})} \\
    \kappa(\rlcM) &\leq \frac{d \max\{2 \rlccap_\max, \rlcind_\max \}}{\min\{\rlccap_\min \cdot \lambda_{\min}(\rlcA_{\rlccap} \rlcA_{\rlccap}^T), \lambda_\min(\rlcL) \}} \, \qquad &\text{(Claim~\ref{claim:kappa_M})} \\
    \alpha_K &\leq 2d \rlcres_\min^{-1} + \sqrt{2d} \, \qquad &\text{(Claim~\ref{claim:smax_K})} \\ 
    \expnorm(-\rlcM^{-1}\rlcK) &\leq \sqrt{\kappa(\rlcM)} \, \qquad &\text{(Claim~\ref{claim:rlc_expnorm_M_index0})}
\end{align*}
We denote the lower bound on $\lambda_\min(\rlcM)$ from the first line as $\sigma$ and the upper bound on $\kappa(\rlcM)$ from the second line as $\kappa_M$. Inserting these bounds into the complexities of Theorem~\ref{thm:output_ind0} gives us the desired result. Finally, it is then required to set
$$
\varepsilon_O = \poly\left(\frac{\sigma \upsilon}{T d \rlcres_{\min}^{-1} \kappa_M (\sigma d \rlcres_\min^{-1} \norm{\vec{x}_0} + T d i_\max)} \right)
$$
This completes the proof. 
\end{proof}

For purely LC circuits, we can simplify the above results as we mention in the following remark.
\begin{remark}[LC circuits with a capacitive tree]\label{remark:results_LC_ind0}
When $\calG$ satisfies Definition~\ref{def:structure_RLC_circs} and is purely an LC circuit with a capacitive tree, the $\mna$ $\dae$ index is $0$ following Theorem~\ref{thm:index_mna}. The complexity of history state prepararation, final time preparation, and energy output is the same as that of $\rlc$ circuits except there is no dependence on $\rlcres_\min^{-1}$ as the conductance is zero.
\end{remark}

\subsection{DAE of index one}\label{sec:rlc_ind1}
In this section, we consider $\rlc$ circuits satisfying condition $(ii)$ of Theorem~\ref{thm:index_mna} and thus, the $\mna$ equations (Eq.~\eqref{eq:mna}) are of index $1$. Note that for index $1$, the $\rlc$ circuits are allowed to have voltage sources and thus the $\mna$ system is as in Eq.~\eqref{eq:mna}. We now comment on simulating such $\rlc$ circuits (Problem~\ref{prob:RLC_dynamics}) and obtaining energy outputs (Problem~\ref{prob:classical_output}).

\paragraph{Norm of the exponential.}
We will use Algorithm~\ref{algo:QDAE_solver_ind1} for simulating $\rlc$ circuits with $\mna$ $\dae$ index $1$. For the $\ode$ solve over $\vec{y} := \rlcP_0 \vec{x}$, the complexity of the Algorithm~\ref{algo:QDAE_solver_ind1} will depend on the exponential norm of the involved dynamics (Theorem~\ref{thm:sim_ind1}). The homogeneous part of the involved $\ode$ is $\dot{\vec{y}} = -\rlcP_0 \rlcM_1^{-1} \rlcK \vec{y}$ which can be rewritten as $\dot{\vec{y}} = - \rlcP_0 \rlcM_1^{-1} \rlcK \rlcP_0 \vec{y}$ since $\rlcP_0 \vec{y} = \vec{y}$. We can thus equivalently bound $\expnorm(- \rlcP_0 {\rlcM_1}^{-1} \rlcK \rlcP_0)$ as we do below. 
\begin{claim}\label{claim:expnorm_M1}
Suppose $\calG$ is a $\rlc$ circuit which satisfies condition $(ii)$ of Theorem~\ref{thm:index_mna} i.e., the $\mna$ $\dae$ is of index $1$. Then, the norm of the exponential of $- \rlcP_0{\rlcM_1}^{-1}\rlcK \rlcP_0$ satisfies
$$\expnorm(-\rlcP_0 {\rlcM_1}^{-1}\rlcK \rlcP_0) \leq \kappa(\rlcM + \rlcQ_0).$$
\end{claim}
\begin{proof}
Now we write these matrices as block $2\times 2$ matrices in the basis given by Im$(\rlcP_0)$ and Im$(\rlcI - \rlcP_0)$, where Im($\cdot$) is the image of the projector. Additionally, inside the image of $\rlcP_0$, we pick a basis in which $\rlcM$ is diagonal. Recall that $\rlcP_0$ is the projector onto the support of $\rlcM$ (as $\rlcM = \rlcM^\dagger$ here).
\begin{align}
    & \rlcP_0=\begin{bmatrix}
        \rlcI & 0 \\0&0
    \end{bmatrix} \hspace{0.2in}
    \rlcQ= \rlcI - \rlcP_0 =\begin{bmatrix}
        0 & 0 \\ 0& \rlcI
    \end{bmatrix} \\
    & \rlcM=\begin{bmatrix}
        \mathbf{D} & 0 \\ 0&0
    \end{bmatrix} \hspace{0.2in}
    \rlcK=\begin{bmatrix}
        \rlcK_{11} & \rlcK_{12} \\ \rlcK_{21} & \rlcK_{22}
    \end{bmatrix}\,.
\end{align}
Computing $\rlcM_1 = \rlcM + \rlcK\rlcQ $ in this basis gives
\begin{equation}
    \rlcM_1 = \begin{bmatrix}
        \mathbf{D} & \rlcK_{12} \\ 0 & \rlcK_{22}
    \end{bmatrix}\,.
\end{equation}
Using Schur complements, the inverse of this matrix can be written as
\begin{equation}
    \rlcM_1^{-1} = \begin{bmatrix}
        \mathbf{D}^{-1} & -\mathbf{D}^{-1} \rlcK_{12} \rlcK_{22}^{-1} \\ 0 & \rlcK_{22}^{-1}
    \end{bmatrix}.
\end{equation}
Using this, we get
\begin{equation}
    \rlcP_0 \rlcM_1^{-1} \rlcK \rlcP_0 = \begin{bmatrix}
        \mathbf{D}^{-1} \rlcK/\rlcK_{22}  & 0 \\ 0 & 0
    \end{bmatrix}\,,
\end{equation}
where the Schur complement $\rlcK/\rlcK_{22}$ is defined as
\begin{equation}
    \rlcK/\rlcK_{22} = \rlcK_{11} - \rlcK_{12} \rlcK_{22}^{-1} \rlcK_{21}\,.
\end{equation}
We now see that
\begin{equation}
    \norm{\exp(- \rlcP_0 \rlcM_1^{-1} \rlcK \rlcP_0) t} \leq \norm{\exp(-\mathbf{D}^{-1} \rlcK/ \rlcK_{22})t} \,.
\end{equation}
Now using a result about Schur complements \cite{wang1997schur}, we can say that since $-\rlcK$ has negative log-norm, so does $-\rlcK/\rlcK_{22}$. Now define the positive (diagonal) matrix
\begin{equation}
    \widetilde{\mathbf{D}} = 
    \begin{bmatrix}
    \mathbf{D} & 0 \\ 0 & \rlcI    
    \end{bmatrix}\,.
\end{equation}
Recall from above that $\mathbf{D}$ is the set of eigenvalues in the support of $\rlcM$. Since $\rlcM$ is positive, $\mathbf{D}$ is strictly positive. Therefore, we have that multiplying $-\rlcP_0 \rlcM_1^{-1} \rlcK \rlcP_0$ by the diagonal positive matrix $\widetilde{\mathbf{D}}$ gives us a matrix with negative log-norm. Now using a result from \cite{Plischke_2005}, we have that
\begin{equation}
    \norm{\exp(- \rlcP_0 \rlcM_1^{-1} \rlcK \rlcP_0)t} \leq \kappa(\widetilde{\mathbf{D}})\,.
\end{equation}
Notice that $\kappa(\widetilde{\mathbf{D}})=\kappa(\rlcM+\rlcQ)$. Using this in the equation above gives us the stated result.
\end{proof}

When the $\rlc$ circuit $\calG$ is promised to be degree-$d$ (except for the reference node) and satisfies Definition~\ref{def:structure_RLC_circs}, we have the following result.
\begin{claim}\label{claim:expnorm_M1}
Suppose $\calG$ is a degree-$d$ graph (excluding the reference node) which satisfies condition $(ii)$ of Theorem~\ref{thm:index_mna}. Let $\lambda_{\min}(\rlcC) \geq c > 0$, and $\lambda_{\min}(\rlcL) \geq \ell > 0$. Then, the norm of the exponential of $- \rlcP_0{\rlcM_1}^{-1} \rlcK$ satisfies
$$\expnorm(- \rlcP_0 {\rlcM_1}^{-1} \rlcK \rlcP_0) \leq \frac{\max\{2d \rlccap_{\max}, d \ell_{\max} \} + 2d r_{\min}^{-1} + \sqrt{2d}}{\min\{\rlccap_{\min}\lambda_\min^+(\rlcA_\rlccap \rlcA_\rlccap^T), \lambda_\min(\rlcL), 1 \}}.$$
\end{claim}

\paragraph{Simulating degree-$d$ $\rlc$ circuits.}
In the theorem below, we use the following definition introduced earlier in Section~\ref{sec:algo_components_RLC}. We define
$$
\Gamma_\calG = \min \left\{\frac{\sigma_{\min}(\rlcA_{\rlcvs}^T \rlcQ_{\rlccap})}{1 + (\|\widetilde{\rlcG}_{\rlccap}\|/\sigma_{\min}(\rlcA_{\rlcvs}^T \rlcQ_{\rlccap}))}, \frac{\min_{w \in \ker(\rlcA_{\rlcvs}^T \rlcQ_{\rlccap}, \norm{w}_2=1)}(w^T \widetilde{\rlcG}_{\rlccap} w)}{1 + (\|\widetilde{\rlcG}_{\rlccap}\|/\sigma_{\min}(\rlcA_{\rlcvs}^T \rlcQ_{\rlccap}))} \right\},
$$
where recall $\rlcQ_\rlccap$ is the orthogonal projector onto the $\ker(\rlcA_{\rlccap}^T)$ and $\widetilde{\rlcG}_{\rlccap} = \rlcQ_{\rlccap} \rlcA_{\rlcres} \rlcG \rlcA_{\rlcres}^T \rlcQ_{\rlccap}$.

\begin{theorem}\label{thm:rlc_sim_ind1_degd}
Let $\varepsilon \in (0,1)$ and $i_\max, \rlcvs_\max > 0$. Consider the setup of Problem~\ref{prob:RLC_dynamics}. Suppose $\calG$ is a $\rlc$ circuit with degree-$d$ excluding the reference node that satisfies Definition~\ref{def:structure_RLC_circs} and condition $(ii)$ of Theorem~\ref{thm:index_mna} i.e., $\mna$ $\dae$ has index $1$. Let $\|\vec{i_{\rlcjs}}\| \leq i_\max$ and $\|\vec{v}_{\rlcvs}\| \leq \rlcvs_\max$. Define $\mu^2 := m^{-1} \sum_{j=1}^m \norm{\vec{x}(j\Delta t)}^2$, $\gamma_{\rlcjs,\rlcvs} = \sqrt{2d} i_\max + \rlcvs_\max$, and
$$
\sigma:= \min\{\rlccap_{\min} \lambda^+_{\min}(\rlcA_{\rlccap}\rlcA_{\rlccap}^T), \lambda_\min(\rlcL) \}, \quad \kappa_M := \frac{2d \max\{\rlccap_\max, \rlcind_\max\}}{\sigma}, \quad \mathcal{C}_f := \log\left(1 + \frac{C T d \rlcres_{\min}^{-1} \gamma_{\rlcjs,\rlcvs}}{\mu \sigma \Gamma_\calG}\right),
$$
where $C \geq 1$ is some absolute constant. Then, there is a quantum algorithm that outputs $|\widehat{\Psi}\rangle$ such that $\left\| |\widehat{\Psi}\rangle - \normhist^{-1} \sum_{k=0}^{m} \norm{\vec{x}(k \Delta t)} \ket{k, x(k \Delta t)} \right\| \leq \varepsilon$ with success probability $\geq 2/3$ where $\Delta t = O(\sigma \Gamma_\calG/(2d \rlcres_\min^{-1} + \sqrt{2d})^2)$ and $m = \ceil{T/\Delta t}$. 
The algorithm has the following complexity
\begin{align*}
\text{Queries to }U_{M}&: \widetilde{O}\left(T^2 \cdot \poly\left(d, \log(N), 1/\mu, \kappa_M, r_\min^{-1}, 1/(\sigma \Gamma_\calG), \log(1/\varepsilon), \mathcal{C}_f \right) \right) \,,\\
\text{Queries to }U_{K}&: \widetilde{O}\left(T^2 \cdot \poly\left(d, \log(N), 1/\mu, \kappa_M, r_\min^{-1}, 1/(\sigma \Gamma_\calG), \log(1/\varepsilon), \mathcal{C}_f \right) \right) \,,\\
\text{Queries to }O_{f},O_x&: O\left(T d \rlcres_\min^{-1} \sqrt{\kappa_M}/(\sigma \mu) \log(1/\varepsilon)\right) \,,
\end{align*}
An additional gate complexity $O(1)$ times the number of uses of $U_M$ and $U_K$ is also used. It is sufficient to set the precisions of the block-encodings of $U_M$ and $U_K$ as
$$
\varepsilon_M = \varepsilon_K = \poly\left(\frac{\sigma \Gamma_\calG \varepsilon}{T d \rlcres_\min^{-1} \kappa_M \calC_f}\right).
$$
\end{theorem}
\begin{proof}
We will use Algorithm~\ref{algo:QDAE_solver_ind1} for $\dae$s of index $1$. The correctness and complexity will follow from applying Theorem~\ref{thm:sim_ind1}. We now estimate the cost specific to $\rlc$ circuits. We note that the forcing $\vec{f} = [-\rlcA_{\rlcjs} \vec{i}_{\rlcjs}, \vec{0}, -\vec{u}_\rlcvs]^T$ satisfies $\norm{\vec{f}} \leq \norm{\rlcA_{\rlcjs} \vec{i}_{\rlcjs}} + \norm{\vec{u}_\rlcvs} \leq \sqrt{2d} i_\max + \rlcvs_\max$ where we used Fact~\ref{fact:spectrum_laplacian_G}. Let us denote $\gamma_{\rlcjs,\rlcvs} = \sqrt{2d} i_\max + \rlcvs_\max$.  Let $\sigma_1 = \sigma_{\min}(\rlcM_1)$. From Claim~\ref{claim:kappa_M1}, we have that
\begin{equation}\label{eq:int1_ind1_rlc}
1/\sigma_1 = \norm{\rlcM_1^{-1}} \leq O\left( \frac{d \, \rlcres_{\min}^{-1}}{\sigma \cdot \Gamma_\calG} \right),
\end{equation}
where $\sigma := \lambda_{\min}^+(\rlcM) \geq \min\{\rlccap_{\min} \lambda^+_{\min}(\rlcA_{\rlccap}\rlcA_{\rlccap}^T), \lambda_\min(\rlcL) \}$ (Claim~\ref{claim:kappa_M}). We thus define
$$
\mathcal{C}_f := \log\left(1 + \frac{C T d \rlcres_{\min}^{-1} \gamma_{\rlcjs,\rlcvs}}{\mu \sigma \Gamma_\calG}\right),
$$
where $C \geq 1$ is some absolute constant. Let $\kappa_M = \lambda_{\max}(\rlcM)/\lambda_{\min}^+(\rlcM)$ denote the effective condition number of $\rlcM$ and $\kappa_{M_1} = \kappa(\rlcM_1)$. We now use the following bounds 
\begin{align*}
    \kappa_M &\leq \frac{d \max\{2 \rlccap_\max, \rlcind_\max \}}{\min\{\rlccap_\min \cdot \lambda^+_{\min}(\rlcA_{\rlccap} \rlcA_{\rlccap}^T), \lambda_\min(\rlcL) \}} \, \qquad &\text{(Claim~\ref{claim:kappa_M})} \\
    \kappa_{M_1} &\leq O\left(\frac{d^2 \rlcres_\min^{-1} \max\{\rlccap_\max, \rlcind_\max \}}{\sigma \Gamma_\calG }\right) \, \qquad &\text{(Claim~\ref{claim:kappa_M})} \\
    \alpha_K &\leq 2d \rlcres_\min^{-1} + \sqrt{2d} \, \qquad &\text{(Claim~\ref{claim:smax_K})} \\ 
    \expnorm(- \rlcP_0 {\rlcM_1}^{-1} \rlcK \rlcP_0) &\leq \frac{\max\{2d c_{\max}, \ell_{\max} \} + 2d r_{\min}^{-1} + \sqrt{2d}}{\min\{c_{\min}\lambda_\min^+(\rlcA_\rlccap \rlcA_\rlccap^T), \lambda_\min(\rlcL), 1 \}}\, \qquad &\text{(Claim~\ref{claim:expnorm_M1})}
\end{align*}
Inserting these bounds into the complexities of Theorem~\ref{thm:sim_ind1} gives us the desired result.
\end{proof}

Analogously, we can also prepare a the state that encodes the solution at time $T$ using Algorithm~\ref{algo:QDAE_solver_ind1} and Corollary~\ref{corr:sim_T_ind1}, which we formally state below.
\begin{corollary}\label{corr:rlc_sim_T_ind1}
Let $\varepsilon \in (0,1)$ and $i_\max, \rlcvs_\max > 0$. Consider context of Theorem~\ref{thm:rlc_sim_ind1_degd}. Then, there is a quantum algorithm that outputs $\ket{\phi}$ such that $\left\| \ket{\phi} - \ket{x(T)} \right\| \leq \varepsilon$ with success probability $\geq 2/3$. Define $\sigma:= \min\{\rlccap_{\min} \lambda^+_{\min}(\rlcA_{\rlccap}\rlcA_{\rlccap}^T),\lambda_\min(\rlcL) \}$, $\gamma_{\rlcjs,\rlcvs} = \sqrt{2d} i_\max + \rlcvs_\max$, and
$$
\kappa_M := \frac{2 \max\{\rlccap_\max, \rlcind_\max\}}{\sigma}, \quad g = \frac{\max_{t\in[0,T]}\|\vec{x}(t)\|}{\|\vec{x}(T)\|}, \quad \mathcal{C}_f := \log\left(1 + \frac{C T d \rlcres_{\min}^{-1} \gamma_{\rlcjs,\rlcvs}}{\sigma \Gamma_\calG \norm{\vec{x}(T)}}\right).
$$
The algorithm has the following complexity
\begin{align*}
\text{Queries to }U_{M}&: \widetilde{O}\left(g T \cdot \poly\left(d, \log(N), \kappa_M, r_\min^{-1}, 1/(\sigma \Gamma_\calG), \log(1/\varepsilon), \mathcal{C}_f \right) \right) \,,\\
\text{Queries to }U_{K}&: \widetilde{O}\left(g T \cdot \poly\left(d, \log(N), \kappa_M, r_\min^{-1}, 1/(\sigma \Gamma_\calG), \log(1/\varepsilon), \mathcal{C}_f \right) \right) \,,\\
\text{Queries to }O_{f},O_x&: O\left(g T d \rlcres_\min^{-1} \sqrt{\kappa_M}/\sigma \log(1/\varepsilon)\right) \,,
\end{align*}
An additional gate complexity $O(1)$ times the number of uses of $U_M$ and $U_K$ is also used. It is sufficient to set the precisions of the block-encodings of $U_M$ and $U_K$ as
$$
\varepsilon_M = \varepsilon_K = \poly\left(\frac{\sigma \Gamma_\calG \varepsilon}{T d \rlcres_\min^{-1} \kappa_M \calC_f}\right).
$$
\end{corollary}
We omit the proof as it follows from Corollary~\ref{corr:sim_T_ind1} and proceeds similarly to the proof of Theorem~\ref{thm:rlc_sim_ind1_degd}.

\paragraph{Energy output.}
\begin{theorem}\label{thm:rlc_energy_ind1_degd}
Let $\upsilon \in (0,1)$ and $i_\max, \rlcvs_\max > 0$. Let $\varepsilon \in (0,1)$. Consider the setup of Problem~\ref{prob:classical_output}. Suppose $\calG$ is a $\rlc$ circuit with degree-$d$ excluding the reference node that satisfies Definition~\ref{def:structure_RLC_circs} and condition $(ii)$ of Theorem~\ref{thm:index_mna} i.e., the $\mna$ $\dae$ has index $1$. Let $\|\vec{i_{\rlcjs}}\| \leq i_\max$ and $\|\vec{u}_\rlcvs\| \leq \rlcvs_\max$. Let the algorithm of Corollary~\ref{corr:rlc_sim_T_ind0} be $\calQ$. For a specified $\calS$ subset of capacitors in $\calG$, there is a quantum algorithm that outputs $\widehat{\theta}$ with success probability $\geq 2/3$ such that $\widehat{\theta}$ is an $\upsilon$-close estimate of the energy across $\calS$ at time $T$. Define $\gamma_{\rlcjs,\rlcvs} = \sqrt{2d} i_\max + \rlcvs_\max$ and
$$
\sigma:= \min\{\rlccap_{\min} \lambda^+_{\min}(\rlcA_{\rlccap}\rlcA_{\rlccap}^T), \lambda_\min(\rlcL) \}, \quad \kappa_M := \frac{2d \max\{\rlccap_\max, \rlcind_\max\}}{\sigma}, \quad \mathcal{C}_f := \log\left(1 + \frac{C T d \rlcres_{\min}^{-1} \gamma_{\rlcjs,\rlcvs}}{\sigma \Gamma_\calG \norm{\vec{x}(T)}}\right),
$$
where $C \geq 1$ is some absolute constant. The algorithm has the following complexity
\begin{align*}
\text{Calls to }\calQ \,\, (S_\calQ)&: \widetilde{O}\left(\alpha_O^2 \cdot \poly(T d \rlcres_\min^{-1} \kappa_M (\|\vec{x}_0\| + \gamma_{\rlcjs,\rlcvs})/(\sigma \Gamma_\calG) \Big)/\upsilon^2 \right) \,,\\
\text{Queries to }U_{M}&: \widetilde{O}\left(S_\calQ \cdot \poly\left(T, d, \log(N), \kappa_M, r_\min^{-1}, 1/(\sigma \Gamma_\calG), \log(d \rlccap_\max/\upsilon), \mathcal{C}_f \right) \right) \,,\\
\text{Queries to }U_{K}&: \widetilde{O}\left(S_\calQ \cdot \poly\left(T, d, \log(N), \kappa_M, r_\min^{-1}, 1/(\sigma \Gamma_\calG), \log(d \rlccap_\max/\upsilon), \mathcal{C}_f \right) \right) \,,\\
\text{Queries to }O_{f},O_x&: \poly\left(S_\calQ T (d \rlcres_\min + \sqrt{d}) \sqrt{\kappa_M}/\sigma \log(d \rlccap_\max/\upsilon)\right) \,,
\end{align*}
An additional gate complexity $O(1)$ times the number of uses of $U_M$ and $U_K$ is also used. It is sufficient to set the precisions of the block-encodings of $U_M$ and $U_K$ as
$$
\varepsilon_M = \varepsilon_K = \poly\left(\frac{\sigma \Gamma_\calG \upsilon}{T d \rlcres_\min^{-1} \kappa_M \calC_f \norm{\vec{x}(T)}}\right).
$$
\end{theorem}
\begin{proof}
We use the algorithm of Theorem~\ref{thm:output_ind1} which here involves using the algorithm of Corollary~\ref{corr:rlc_sim_T_ind1} followed by the estimation protocol of Lemma~\ref{lem:meas_obs_gen}. We now estimate the cost specific to $\rlc$ circuits and for the specification of the observable $\mathbf{O}$ here. Given $\calS$, the observable $\mathbf{O}$ of interest is
$$
\mathbf{O} = \begin{bmatrix}
    \rlcA_{\rlccap,\calS} \rlcC \rlcA_{\rlccap,\calS}^T & 0 \\ 0 & 0,
\end{bmatrix}
$$
where $\rlcA_{\rlccap,\calS}$ is the reduced incidence matrix over capacitors in $\calS$. Let $\varepsilon_O$ be an error parameter to be fixed later. From the proof of Claim~\ref{claim:rlc_BEs_components}, we have a $(\alpha_O, a_O, \varepsilon_O)$-block-encoding of $\mathbf{O}$, which we denote by $U_O$, where
\begin{equation}\label{eq:params_UO_ind1}
\alpha_O = 2d \rlccap_\max, \quad a_O = O(\log N).
\end{equation}
To bound the sample and time complexity as given by Theorem~\ref{thm:output_ind0}, we note that the forcing $\vec{f} = [-\rlcA_{\rlcjs} \vec{i}_{\rlcjs}, \vec{0}, -\vec{u}_\rlcvs]^T$ satisfies $\norm{\vec{f}} \leq \norm{\rlcA_{\rlcjs} \vec{i}_{\rlcjs}} + \norm{\vec{u}_\rlcvs} \leq \sqrt{2d} i_\max + \rlcvs_\max$ where we used Fact~\ref{fact:spectrum_laplacian_G}. Let us denote $\gamma_{\rlcjs,\rlcvs} = \sqrt{2d} i_\max + \rlcvs_\max$. From Claim~\ref{claim:kappa_M1}, we have that
\begin{equation}\label{eq:int1_ind1_rlc_output}
1/\sigma_1 = \norm{\rlcM_1^{-1}} \leq O\left( \frac{d \, \rlcres_{\min}^{-1}}{\sigma \cdot \Gamma_\calG} \right),
\end{equation}
where $\sigma := \lambda_{\min}^+(\rlcM) \geq \min\{\rlccap_{\min} \lambda^+_{\min}(\rlcA_{\rlccap}\rlcA_{\rlccap}^T), \lambda_\min(\rlcL) \}$ (Claim~\ref{claim:kappa_M}). We thus define
$$
\mathcal{C}_f := \log\left(1 + \frac{C T d \rlcres_{\min}^{-1} \gamma_{\rlcjs,\rlcvs}}{\sigma \Gamma_\calG \norm{\vec{x}(T)}}\right),
$$
where $C \geq 1$ is some absolute constant. Let $\kappa_M = \lambda_{\max}(\rlcM)/\lambda_{\min}^+(\rlcM)$ denote the effective condition number of $\rlcM$ and $\kappa_{M_1} = \kappa(\rlcM_1)$. We now use the following bounds 
\begin{align*}
    \kappa_M &\leq \frac{d \max\{2 \rlccap_\max, \rlcind_\max \}}{\min\{\rlccap_\min \cdot \lambda^+_{\min}(\rlcA_{\rlccap} \rlcA_{\rlccap}^T), \lambda_\min(\rlcL) \}} \, \qquad &\text{(Claim~\ref{claim:kappa_M})} \\
    \kappa_{M_1} &\leq O\left(\frac{d^2 \rlcres_\min^{-1} \max\{\rlccap_\max, \rlcind_\max \}}{\sigma \Gamma_\calG }\right) \, \qquad &\text{(Claim~\ref{claim:kappa_M})} \\
    \alpha_K &\leq 2d \rlcres_\min^{-1} + \sqrt{2d} \, \qquad &\text{(Claim~\ref{claim:smax_K})} \\ 
    \expnorm(- \rlcP_0 {\rlcM_1}^{-1} \rlcK \rlcP_0) &\leq \frac{\max\{2d c_{\max}, \ell_{\max} \} + 2d r_{\min}^{-1} + \sqrt{2d}}{\min\{c_{\min}\lambda_\min^+(\rlcA_\rlccap \rlcA_\rlccap^T), \lambda_\min(\rlcL), 1 \}}\, \qquad &\text{(Claim~\ref{claim:expnorm_M1})}
\end{align*}
Inserting these bounds into the complexities of Theorem~\ref{thm:output_ind1} gives us the desired result. Finally, it is then required to set
$$
\varepsilon_O = \poly\left(\frac{\sigma_1 \upsilon}{T \kappa_M \alpha_K \expnorm(\calA)(\|\vec{x}_0\| + \|\vec{f}\|) \Big)} \right) = \poly\left(\frac{\sigma \Gamma_\calG \upsilon}{T d \rlcres_{\min}^{-1} \kappa_M (\|\vec{x}_0\| + \gamma_{\rlcjs,\rlcvs}) \Big)} \right)
$$
This completes the proof. 

\end{proof}

\subsection{DAE of index two}\label{sec:rlc_ind2}
In this section, we consider $\rlc$ circuits satisfying condition $(iii)$ of Theorem~\ref{thm:index_mna} and thus, the $\mna$ equations (Eq.~\eqref{eq:mna}) are of index $2$. Note that is the highest index that the $\mna$ equations (Eq.~\eqref{eq:mna}) can achieve for general linear $\rlc$ circuits~(Theorem~\ref{thm:index_mna}). We now comment on simulating such $\rlc$ circuits (Problem~\ref{prob:RLC_dynamics}). The result for obtaining energy outputs (Problem~\ref{prob:classical_output}) proceeds similarly to earlier cases and we omit the details here. 

\paragraph{Norm of the exponential.}
\begin{claim}\label{claim:expnorm_M2}
Suppose $\calG$ is a $\rlc$ circuit which satisfies condition $(iii)$ of Theorem~\ref{thm:index_mna} i.e., the $\mna$ $\dae$ is of index $2$. Then, the norm of the exponential of $- \rlcP_0 \rlcP_1 \rlcM_2^{-1} \rlcK \rlcP_0 \rlcP_1$ satisfies
$$
\expnorm(-\rlcP_0 \rlcP_1 {\rlcM_1}^{-1}\rlcK \rlcP_0 \rlcP_1) \leq O(\kappa(\widetilde{\rlcM})),
$$
where $\kappa(\widetilde{\rlcM})$ is the effective condition number of $\rlcM$.
\end{claim}
\begin{proof}
We first observe that $\rlcK$ has a non-negative definite real part i.e., $\rlcK + \rlcK^\dagger \geq 0$. Expanding from the definition of $\rlcM_2$, we have
\begin{align}
    \rlcM_2 &= \rlcM_1 + \rlcK_1 \rlcQ_1\\
    &=\rlcM_1 + \rlcK \rlcP_0 \rlcQ_1\\
    &=\rlcM + \rlcK \rlcQ_0 + \rlcK \rlcP_0 \rlcQ_1\\
    &=\rlcM + \rlcK \rlcQ_0 \rlcP_1 + \rlcK \rlcQ_0 \rlcQ_1 + \rlcK \rlcP_0 \rlcQ_1\\
    &=\rlcM + \rlcK \rlcQ_0 \rlcP_1 + \rlcK \rlcQ_1\,.
\end{align}
Now let us write $\rlcM_2$ in the basis of the images of the projectors $\rlcP_0 \rlcP_1$, $\rlcQ_0 \rlcP_1$, $\rlcQ_1$. Note that each of these projectors are orthogonal to each other e.g., $(\rlcP_0 \rlcP_1) \rlcQ_1 = (\rlcP_0 \rlcP_1)(\rlcQ_0 \rlcP_1) = 0$. First, note that $\rlcM$ is a block-matrix with a single non-zero entry in the top left (corresponding to $\rlcP_0 \rlcP_1$). To see this, observe that $\rlcM_1$ only has support on the top-left $2\times 2$ submatrix since $\mathrm{Im}(\rlcQ_1)=\ker(\rlcM_1)$ and $\rlcK \rlcQ_1$ has support only on the last column. Therefore, $\rlcM_1 + \rlcK \rlcQ_1$ has zeros in the second and third rows of the first column. 

As in the index one case, we can diagonalize the top-left entry (corresponding to $\rlcM$) and have a diagonal matrix denoted $\mathbf{D}$. Now the matrix $\rlcM_2$ can be written as
\begin{equation}
    \rlcM_2 = \begin{bmatrix}
        \mathbf{D} & \rlcK_{12} & \rlcK_{13}\\
        0 & \rlcK_{22} & \rlcK_{23}\\
        0 & \rlcK_{32} & \rlcK_{33}
    \end{bmatrix}\,.
\end{equation}
Now let us rearrange the blocks to make it a $2\times 2$ block matrix. To end this end, let
\begin{equation}
    \rlcK_{1,2:3} = \begin{bmatrix}
        \rlcK_{12} & \rlcK_{13}
    \end{bmatrix}\,,\hspace{0.2in}
    \rlcK_{2:3,2:3} = \begin{bmatrix}
        \rlcK_{22} & \rlcK_{23}\\
        \rlcK_{32} & \rlcK_{33}
    \end{bmatrix}\,.
\end{equation}
Using these definitions, we have
\begin{equation}
    \rlcM_2 = \begin{bmatrix}
        \mathbf{D} & \rlcK_{1,2:3}\\
        0 & \rlcK_{2:3,2:3}
    \end{bmatrix}\,,
\end{equation}
where on the bottom left, $0$ stands for the $2\times 1$ block matrix of zeros. Therefore, we have
\begin{equation}
    \rlcM_2^{-1} = \begin{bmatrix}
        \mathbf{D}^{-1} & -\mathbf{D}^{-1} \rlcK_{1,2:3} \rlcK_{2:3,2:3}^{-1}\\
        0 & \rlcK_{2:3,2:3}^{-1}
    \end{bmatrix}\,.
\end{equation}
In this block-basis, we have
\begin{equation}
    \rlcK \rlcP_0 \rlcP_1 = \begin{bmatrix}
        \rlcK_{11} & 0\\
        \rlcK_{2:3,1} & 0
    \end{bmatrix}\,,
\end{equation}
where 
\begin{equation}
    \rlcK_{2:3,1} = \begin{bmatrix}
        \rlcK_{21}\\\rlcK_{31}
    \end{bmatrix}\,.
\end{equation}
Putting all this together, we have
\begin{equation}
    - \rlcP_0 \rlcP_1 \rlcM_2^{-1} \rlcK \rlcP_0 \rlcP_1=\begin{bmatrix}
        -\mathbf{D}^{-1} \rlcK_{11}/\rlcK_{2:3,2:3} & 0\\
        0 & 0
    \end{bmatrix}\,,
\end{equation}
where 
\begin{equation}
   \rlcK_{11}/\rlcK_{2:3,2:3} =  \rlcK_{11} - \rlcK_{1,2:3} \rlcK_{2:3,2:3}^{-1} \rlcK_{2:3,1}\,.
\end{equation}
Once again in the top left block, we have something of the form $\mathbf{D}^{-1}$ multiplied by a Schur complement of $\rlcK$. As in the case of index one, using properties of Schur complements \cite{wang1997schur}, we have that the Schur complement $-\rlcK_{11}/\rlcK_{2:3,2:3}$ has a negative log-norm since $\rlcK$ has a positive Hermitian real part. Using the result of $\cite{Plischke_2005}$ (Theorem 3.35), we have
\begin{equation}
    \|\exp(- \rlcP_0 \rlcP_1 \rlcM_2^{-1} \rlcK \rlcP_0 \rlcP_1 t)\| \leq \kappa(\widetilde{\mathbf{D}})\,,
\end{equation}
where 
\begin{equation}
    \widetilde{\mathbf{D}} = \begin{bmatrix}
        \mathbf{D} & 0 & 0\\
        0 & \id & 0\\
        0 & 0 & \id
    \end{bmatrix}\,.
\end{equation}
This completes the proof.
\end{proof}

\paragraph{Simulating degree-$d$ RLC circuits.}
We can now give a result for simulating $\rlc$ circuits which have $\mna$ $\dae$ index $2$. The proof follows from Corollary~\ref{thm:sim_ind2} and similarly to that of earlier cases so we omit the details here.
\begin{theorem}\label{thm:rlc_sim_ind2_degd}
Let $\varepsilon, \sigma_1, \sigma_2 \in (0,1)$ and $i_\max, \rlcvs_\max > 0$. Consider the setup of Problem~\ref{prob:RLC_dynamics}. Suppose $\calG$ is a $\rlc$ circuit with degree-$d$ excluding the reference node that satisfies Definition~\ref{def:structure_RLC_circs} and condition $(ii)$ of Theorem~\ref{thm:index_mna} i.e., $\mna$ $\dae$ has index $2$. Let $\sigma_\min^+(\rlcM_1) \geq \sigma_1$, $\sigma_\min(\rlcM_2) \geq \sigma_2$ and $\sigma_\min(\rlcP_0 \widetilde{\rlcQ}_1) \geq 1/\uptau$. Let $\|\vec{i_{\rlcjs}}\| \leq i_\max$ and $\|\vec{v}_{\rlcvs}\| \leq \rlcvs_\max$. Then, there is a quantum algorithm that outputs $|\widehat{\Psi}\rangle$ such that $\left\| |\widehat{\Psi}\rangle - \normhist^{-1} \sum_{k=0}^{m} \norm{\vec{x}(k \Delta t)} \ket{k, x(k \Delta t)} \right\| \leq \varepsilon$ with success probability $\geq 2/3$ where $\Delta t = O(\sigma_2 \uptau/ \alpha_K)$ and $m = \ceil{T/\Delta t}$. Define $\mu^2 := m^{-1} \sum_{j=1}^m \norm{\vec{x}(j\Delta t)}^2$, $\gamma_{\rlcjs,\rlcvs} = O((\sqrt{2d} i_\max + \rlcvs_\max)/\uptau)$, and
$$
\sigma:= \min\{\rlccap_{\min} \lambda^+_{\min}(\rlcA_{\rlccap}\rlcA_{\rlccap}^T), \lambda_\min(\rlcL) \}, \quad \kappa_M := \frac{2d \max\{\rlccap_\max, \rlcind_\max\}}{\sigma}, \quad \mathcal{C}_f := \log\left(1 + \frac{C T d \rlcres_{\min}^{-1} \gamma_{\rlcjs,\rlcvs}}{\mu \sigma \sigma_2}\right),
$$
where $C \geq 1$ is some absolute constant. The algorithm uses the following complexity 
\begin{align*}
\text{Queries to }U_{M}&: \poly\left(T, d, \log(N), 1/\mu, \kappa_M, r_\min^{-1}, 1/(\sigma \sigma_1 \sigma_2 \uptau), \log(1/\varepsilon), \mathcal{C}_f \right) \,,\\
\text{Queries to }U_{K}&: \poly\left(T, d, \log(N), 1/\mu, \kappa_M, r_\min^{-1}, 1/(\sigma \sigma_1 \sigma_2 \uptau), \log(1/\varepsilon), \mathcal{C}_f \right) \,,\\
\text{Queries to }O_{f},O_x&: \poly\left(T d \rlcres_\min^{-1} \kappa_M/(\sigma_2 \sigma_1 \uptau \mu) \log(1/\varepsilon)\right) \,,
\end{align*}
An additional gate complexity $O(1)$ times the number of uses of $U_M$ and $U_K$ is also used. It is sufficient to set the precisions of the block-encodings of $U_M$ and $U_K$ as
$$
\varepsilon_M = \varepsilon_K = \poly\left(\frac{\sigma \sigma_1 \sigma_2 \uptau \varepsilon}{T d \rlcres_\min^{-1} \kappa_M \calC_f}\right).
$$
\end{theorem}

\section{BQP Completeness}\label{sec:bqp}
In this subsection, we show that the problem of computing the energy of any subset of capacitors or inductors is $\bqp$-complete. We first give a way to embed an arbitrary coupled oscillator problem (without damping or sources) into an instance of an LC circuit. We have the following result.
\begin{theorem}
    There exists an asymptotic family of LC circuits such that it is $\bqp$-complete to decide if the energy of a single capacitor is above $2/3$ or below $1/3$ (promised that one of these cases is true), where $N$ is the number of masses in the coupled oscillator set-up.
\end{theorem}
\begin{proof}
    Consider an arbitrary sparse coupled oscillator instance given by a matrix of springs constants $\rlcK$ and a (diagonal) mass matrix $\rlcM$. In \cite{PhysRevX.13.041041}, it was shown that the evolution can be written as
    \begin{equation}
        \mathbf{M}\frac{d^2x}{dt^2} = -\mathbf{F}x\,,
    \end{equation}
    where $\mathbf{F}$ denotes the matrix of forces on the masses. It was also shown that this evolution can be recast in higher dimensions as Hamiltonian evolution. Here, we use the form of the Hamiltonian derived in \cite{krovi2024quantum} which is the following differential equation.
    \begin{equation}\label{eq:coupled_osc}
        \frac{d\tilde{x}}{dt} = \begin{bmatrix}
            0 & \mathbf{B}^\dagger\\-\mathbf{B} & 0
        \end{bmatrix}\tilde{x}\,,
    \end{equation}
    where the rectangular matrix $\mathbf{B}$ satisfies $\mathbf{B} \mathbf{B}^\dagger = \sqrt{\rlcM^{-1}}\mathbf{F}\sqrt{\mathbf{M}^{-1}}$. Given oracle access to the $\rlcK$ matrix, \cite{PhysRevX.13.041041} shows how to construct oracle access to $\mathbf{B}$ and $\mathbf{B}^\dagger$. In terms of the masses and spring constants, the matrix $\mathbf{B}$ can be written as follows. As mentioned above, $\mathbf{B}$ is a rectangular matrix (and in fact, the incidence matrix of a weighted graph). It has $N$ rows, where $N$ is the number of masses and $N(N+1)/2$ columns. Any entry can be labeled $(j,(k,l))$ where $k\leq l$. The value of this entry is given by
    \begin{equation}
        \bra{j}\mathbf{B}\ket{k,l} = \begin{cases}
            \frac{\sqrt{\kappa_{j,j}}}{\sqrt{m_j}}\,,\text{if }j=k=l\\
            \frac{\sqrt{\kappa_{j,k}}}{\sqrt{m_j}}\,,\text{if }j=k\neq l\\
            -\frac{\sqrt{\kappa_{j,k}}}{\sqrt{m_k}}\,,\text{if }j=l\neq k\\
            0\,\text{ otherwise}
        \end{cases}
    \end{equation}
    where $\kappa_{j,k}$ refer to entries of the matrix $\rlcK$ and $m_{j}$ refer to the non-zero diagonal entries of the mass matrix $M$ of the mass-and-spring system.
    
    Next, we explain how to construct an LC network and define the matrices $\rlcA_{\rlccap}$, $\rlcA_{\rlcind}$, $\rlcC$ and $\rlcL$ from the $\mathbf{B}$ matrix above. First, we set the number of nodes in the $\textsf{LC}$-network to be $N+1$ i.e., the $N$ nodes coming from the coupled oscillator set-up and an additional reference node (that will eventually be omitted from $\rlcA_{\rlccap}$). The capacitors are connected in a star-like fashion i.e., there is a capacitor connecting the reference node to every other node. There are no capacitors between the non-reference nodes. The values of these capacitors will be decided later. This configuration of capacitors ensures that the matrix $\rlcA_{\rlccap}$ defined in (\ref{eq:incidence_defn}) is the identity matrix.

    Now, the matrix $\rlcA_{\rlcind}$, which is the incidence matrix of the inductors can be constructed from the matrix $\mathbf{B}$. Define $\rlcA_{\rlcind}$ as follows.
    \begin{equation}
        \rlcA_{\rlcind}(j,(k,l))=\sgn(\mathbf{B}(j,(k,l)))\,.
    \end{equation}
    This means that the inductors are connected as follows. The diagonal terms i.e., $\rlcA_{\rlcind}(j,(j,j))$, which are unity, indicate that there is an inductor connecting the reference node to every other node (just like the capacitors and in parallel with them). The other entries of $\rlcA_{\rlcind}$ mean that there is an inductor wherever there is a spring. The values of these inductors (and the capacitors) will be explained next.

    Now recall from (\ref{eq:rlc_mna_index0}), we have
    \begin{equation}
        \begin{bmatrix}
        \rlcA_{\rlccap} \rlcC \rlcA_{\rlccap}^T & 0 \\
        0 & \rlcL 
    \end{bmatrix}
    \frac{d}{dt} \begin{bmatrix}
        \vec{u} \\ \vec{i}_\rlcind
    \end{bmatrix} + 
    \begin{bmatrix}
        0 & \rlcA_{\rlcind} \\
        -\rlcA_{\rlcind}^T & 0
    \end{bmatrix}
    \begin{bmatrix}
        \vec{u} \\ \vec{i}_\rlcind
    \end{bmatrix} = 0\,.
    \end{equation}
    Since we have that $\rlcA_{\rlccap}$ is the identity, we can re-write the above equation as 
    \begin{equation}
        \frac{dy}{dt} = \begin{bmatrix}
        0 & \sqrt{\rlcC^{-1}}\rlcA_{\rlcind}\sqrt{\rlcL^{-1}} \\
        -(\sqrt{\rlcC^{-1}}\rlcA_{\rlcind}\sqrt{\rlcL^{-1}})^T & 0
    \end{bmatrix}y\,,
    \end{equation}
    where 
    \begin{equation}
        y=\begin{bmatrix}
        \sqrt{\rlcC} \vec{u} \\ \sqrt{\rlcL} \vec{i}_\rlcind
    \end{bmatrix}\,.
    \end{equation}
    Notice that the above equation is very similar to (\ref{eq:coupled_osc}). To make them equal, we need to identify the entries of $\mathbf{B}$ with $\sqrt{\rlcC^{-1}}\rlcA_{\rlcind}\sqrt{\rlcL^{-1}}$. Recall that $\rlcC$ is an $N\times N$ diagonal matrix and $\rlcL$ is an $N(N+1)/2\times N(N+1)/2$ diagonal matrix. Therefore, we need 
    \begin{equation}
        \frac{1}{\rlcC_j\rlcL_{k,l}} = \frac{\kappa_{k,l}}{m_j}\,.
    \end{equation}
    This gives us a simple identification of springs with inductors and masses with capacitors i.e., $\rlcL_{k,l} = 1/\kappa_{k,l}$ and $\rlcC_j=m_j$.

    In \cite{PhysRevX.13.041041}, it was shown that deciding if the normalized kinetic energy of a single oscillator is above $1/\polylog N$ or below $\exp(-\sqrt{\log N})$ promised that one of these cases is true is $\bqp$ hard. The normalized kinetic energy of a single oscillator maps to the normalized energy of a single capacitor in our case. This proves the hardness in our case as well. 

    To show completeness, we use Theorem~\ref{thm:rlc_energy_ind0_degd}, specifically see Remark~\ref{remark:results_LC_ind0}. In the $\bqp$-hard instance of mass-and-spring system constructed in \cite{PhysRevX.13.041041}, the mass matrix $\rlcM = \rlcA_c \rlcC \rlcA_c^T$ is the identity and the spring constant matrix $\rlcK$ has sparsity and norm $O(1)$. This leads to a $\polylog N$ quantum algorithm to additively estimate the normalized energy of a single inductor. Through the construction above, this corresponds to the kinetic energy of a single mass in the coupled oscillator system. A constant additive approximation to the energy can be used to decide whether the normalized energy is above $2/3$ or below $1/3$ by picking the constant to be less than $1/6$.
\end{proof}


\bibliographystyle{alpha}
\bibliography{references}

\newpage 

\appendix

\section{Deferred proofs}
\subsection{Guarantees of quantum $\ode$ solver} 
\label{appsec:qalgo_ode_la}

In this section of the appendix, we provide a proof of Theorem~\ref{thm:ODE_solver}, which we restate below for convenience. The proof is almost identical to that of \cite[Theorems~6--7]{Krovi2023improvedquantum}, with the difference being that the linear operator $\bm{\calA}$ in \cite{Krovi2023improvedquantum} is assumed to be a sparse matrix whereas here, the differential operator is given to us as a block-encoding.
\qodesolver*
\begin{proof}
We will use the algorithm from \cite{Krovi2023improvedquantum} to solve the given linear $\ode$. The algorithm approximates the solution via a truncated Taylor series and then designs a linear system to be solved based on that. 

Let us denote $h = 1/\norm{\bm{\calA}}$ be a small time-step (i.e., $h = \Delta t$), $m = \ceil{T/h}$ be the number of time-steps, and $t_j = jh, \forall j \in [m]_0$. Additionally, we define the parameter $k$ as the number of terms in the truncated Taylor series which is chosen as
\begin{equation}\label{def:ODE_taylor_steps}
k = \left\lceil \frac{2\log \Omega}{\log \log \Omega} \right\rceil, \quad \text{ where } \Omega = \frac{4 m e^3}{\delta} \left(1 + \frac{T e^2 \norm{\vec{f}}}{\mu} \right),    
\end{equation}
so that $(k+1)! > \Omega$ and where $\delta \in (0,1)$ is an error parameter which we fix later. From Claim~\ref{claim:diff_taylor_truth_approx_input}$(ii)$ (with instantiations of $\eta_0 := 0$ and $\eta_f := 0$ there) ensures that the solution $\vec{x}(t)$ at time $t \in [0,T]$ and the solution $\vec{y}(t)$ obtained by the truncated Taylor series at $k$ terms, satisfies
\begin{equation}\label{eq:taylor_close_to_truth}
\left(\frac{\sum_{j=0}^m \norm{\vec{x}(t_j) - \vec{y}(t_j)}^2}{\sum_{j=0}^m \norm{\vec{x}(t_j)}^2} \right)^{1/2} \leq \frac{\delta}{2}
\end{equation}
Let us denote the true history state as $\ket{\psi}$ and the history state corresponding to the solution obtained by the truncated Taylor series as $\ket{\phi}$ which have the following definitions:
\begin{equation}\label{eq:def_true_hist_taylor_hist}
\ket{\psi} := \frac{1}{\normhist_x} \sum_{j=0}^m \norm{\vec{x}(jh)} \ket{j, x (jh)}, \enspace \ket{\phi} := \frac{1}{\normhist_y} \sum_{j=0}^m \norm{\vec{y}(jh)} \ket{j, y(jh)}
\end{equation}
Using Claim~\ref{claim:diff_hist_states}, we then have that $\ket{\phi}$ is close to $\ket{\psi}$:
\begin{equation}\label{eq:interim_phi_close_to_psi}
\norm{\ket{\psi} - \ket{\phi}} \leq \delta.
\end{equation}

It is then enough to prepare a state close to $\ket{\phi}$ to ensure that we are able to output a state close to $\ket{\psi}$. This is what the algorithm of \cite{Krovi2023improvedquantum} does as follows. The algorithm involves solving a linear system of equations defined by the following operators\footnote{In \cite{Krovi2023improvedquantum}, the linear operators below are denoted by $L, N, M_1, M_2$ respectively but we use the following notation below so as not to confuse the reader as this latter notation is used for particular operators in the context of $\rlc$ circuits and the quantum $\dae$ solver.}
\begin{align}\label{eq:ODE_solver_system}
    \calL &= \id - \calN \\ 
    \calN &= \sum_{i=0}^{m} \ket{i+1}\bra{i} \otimes \calM_2 (I - \calM_1)^{-1} \\ 
    \calM_1 &= \sum_{j=0}^{k-1} \ket{j+1}\bra{j} \otimes \frac{(\bm{\calA} h)}{{j+1}} \\ 
    \calM_2 &= \sum_{j=0}^{k-1} \ket{0}\bra{j} \otimes \id,
\end{align}
where the first register indexed by $i$ contains the time step, the second register index by $j$ is the Taylor step and the third register contains the states on which $\bm{\calA}$ acts. The history state vector $\vec{y} \in \mathbb{C}^{N m}$ of the $\ode$ is then obtained by solving the linear system
\begin{align}\label{eq:lsa-ODE}
    \calL \vec{y} = \vec{\psi}_{\mathrm{in}}, \qquad \vec{\psi}_{\mathrm{in}} = \ket{0,0} \otimes \vec{x}_0 + h\sum_{i=0}^{m-1}\ket{i,1} \otimes \vec{f},
\end{align}
where $\vec{\psi}_{\mathrm{in}}$ is the \emph{input} state vector containing information about the initial vector $\vec{x}(0)$ and the forcing $\vec{f}$. The corresponding input state is then
\begin{equation}\label{eq:interim_input_state}
\ket{\psi_{\mathrm{in}}} = \frac{1}{\normhist_{\mathrm{init}}} \Big(\norm{\vec{x}_0} \ket{0,0,x_0} + h \|\vec{f}\| \sum_{i=0}^{m-1} \ket{i,1,f} \Big),
\end{equation}
where $\normhist_{\mathrm{init}} = \sqrt{\norm{\vec{x}_0}^2 + mh^2 \|\vec{f}\|^2}$. One copy of this state can be produced with a single call to $O_x$ and $O_f$ and an additional $\polylog(m)$ elementary gates (see \cite[Lemma~14]{Krovi2023improvedquantum}).

To solve Eq.~\eqref{eq:lsa-ODE}, we need to be able to prepare a block-encoding of $\calL$ and then $\calL^{-1}$. We do this by commenting on the block-encodings of the involved operators $\calM_2$, $\calM_1$ and $\calN$. Let $\varepsilon_1,\varepsilon_2,\varepsilon_3 \in (0,1)$ be error parameters to be fixed later. We first give a block-encoding $\calM_2$, which we show can be done exactly. Let $\mathbf{B} = \sum_{j=0}^{k-1} \ket{0}\bra{j}$. Consider the $\ell = \ceil{\log k}$-qubit unitary $W$ that has the following action
\begin{equation}\label{eq:interim_def_W}
    W \ket{0^\ell} = \frac{1}{\sqrt{k}} \sum_{j=0}^{k - 1} \ket{j}.
\end{equation}
A circuit for $W$ with gate complexity $O(\log k)$ can be determined using~\cite{shukla2024efficient}. Let $R_0 = 2\ket{0^\ell}\bra{0^\ell} - \id$ be the reflection about $\ket{0^\ell}$ and $V = \ket{0}_a\bra{0}_a \otimes \id + \ket{1}_a\bra{1}_a \otimes R_0$ be the controlled-reflection operator where $a$ is ancillary qubit. Consider the following $\ell$-qubit circuit 
$$
U_B = (H_a \otimes \id)V(H_a \otimes \id)(\id_a \otimes W^\dagger).
$$
We can then show that $U_B$ is a $(\sqrt{k},1,0)$-block-encoding of $\mathbf{B}$ by evaluating
\begin{align*}
(\bra{0}_a \otimes \id)U_B(\ket{0}_a \otimes \id) &= (\bra{+}_a \otimes \id)V(\ket{+}_a \otimes \id) W^\dagger 
= \Big(\frac{\id}{2} + \frac{2 \ket{0^\ell}\bra{0^\ell} - \id}{2} \Big) W^\dagger
= \ket{0^\ell}\bra{0^\ell}W^\dagger = \sqrt{k} \mathbf{B},
\end{align*}
where we used the definition of $V$ in the third equality and noted that $\sqrt{k}\bra{0^\ell} W^\dagger = \sum_{j=0}^{k-1} \bra{j}$ from the definition of $W$ in Eq.~\eqref{eq:interim_def_W}. We then obtain a $(\sqrt{k},1,0)$-block-encoding $U_{\calM_2}$ of $\calM_2$ by taking a tensor product of $U_B$ with $\id$. The corresponding gate complexity is $O(\log k)$ which is mainly due to controlled-$R_0$ operator and assuming access to $O(\log k)$ clean ancillas.

To block-encode $\calM_1$, we first block-encode $\mathbf{B}_1 := \sum_{j=0}^{k-1} \ket{j+1}\bra{j}/(j+1)$ which we will then combine with $U_{\calA}$. We note that $\mathbf{B}_1$ has row and column sparsity of $1$. We can then use Lemma~\ref{lem:sparse-block-enc} to obtain a $(1,\ell + 3, \varepsilon_1)$-block-encoding of $\mathbf{B}_1$, which we denote by $U_{B_1}$, with gate complexity $O(\ell + \log^{2.5}(1/\varepsilon_1))$ while using $O(\ell, \log^{2.5}(1/\varepsilon_1))$ ancilla qubits. We note that $U_{\calA}$ which is a $(\alpha_\calA,a_\calA,\varepsilon_{\calA})$-block-encoding of $\bm{\calA}$ is also a $(h\alpha_\calA,a_\calA,h\varepsilon_{\calA})$-block-encoding of $\bm{\calA} h$. We can then use Lemma~\ref{lem:tensor_prod_BEs} to combine $U_{B_1}$ and $U_{\calA}$ to obtain a $(h \alpha_{\calA}, \ell + a_\calA + 3, h \alpha_{A}\varepsilon_1 + h\varepsilon_{\calA} + h \varepsilon_1\varepsilon_{\calA})$-block-encoding of $\calM_1 = \mathbf{B} \otimes (\bm{\calA} h)$. Using Lemma~\ref{lem:sum_BEs}, we can combine $U_{\calM_1}$ with $\id$ to obtain a $(1 + h\alpha_\calA, \ell + a_\calA + 4, h \alpha_{A}\varepsilon_1 + h\varepsilon_{\calA} + h \varepsilon_1\varepsilon_{\calA})$-block-encoding of $\id - \calM_1$, which we denote by $U_{\id - \calM_1}$, with gate complexity $O(T_{\calA} + \ell + \log^{2.5}(1/\varepsilon_1))$.

We can obtain a block-encoding of $(\id - \calM_1)^{-1}$ using Lemma~\ref{lem:pseudoinverse_BE}. It was shown in \cite[Theorem~6]{Krovi2023improvedquantum} that $\spec({\id - \calM_1}) \in [1/k,2]$. Lemma~\ref{lem:pseudoinverse_BE} then gives us a $(2k, \ell + a_{\calA} + 5, \varepsilon_2)$-block-encoding of $(\id - \calM_1)^{-1}$, which denote by $U_{(\id - \calM_1)^{-1}}$, provided $h \alpha_{A}\varepsilon_1 + h\varepsilon_{\calA} + h \varepsilon_1\varepsilon_{\calA} = o(\varepsilon_2/(k^2 \log(k/\varepsilon_2))$. Noting that $h \alpha_{\calA} \leq 1$, it is sufficient to ensure 
\begin{equation}\label{eq:cond_eps_inv_idminusM1}
\varepsilon_1 + \varepsilon_{\calA} + \varepsilon_1 \varepsilon_{\calA} = o(\varepsilon_2/(k^2 \log(k/\varepsilon_2)).    
\end{equation}
The corresponding gate complexity is 
\begin{equation}\label{eq:cost_inv_idminusM1}
O\left(k \log(k/\varepsilon_2) \Big(T_{\calA} + \ell + \log^{2.5}(1/\varepsilon_1)\Big) \right).
\end{equation}
Using Lemma~\ref{lem:prod_BEs}$(ii)$, we can obtain a $(2k\sqrt{k}, \ell + a_{\calA} + 6, \sqrt{k}\varepsilon_2)$-block-encoding for $\mathbf{D} := \calM_2(\id - \calM_1)^{-1}$, which we denote by $U_D$, by combining $U_{\calM_2}$ and $U_{(\id - \calM_1)^{-1}}$. The corresponding gate complexity is dominated by that of $U_{(\id - \calM_1)^{-1}}$ and is asymptotically same as that of Eq.~\eqref{eq:cost_inv_idminusM1}. 

Towards completing the block-encoding of $\calN$, we also need a block-encoding for $\mathbf{B}_2 := \sum_{i=0}^m \ket{i+1}\bra{i}$. We can use Lemma~\ref{lem:sparse-block-enc} again to obtain a block-encoding for $\mathbf{B}_2$ which has row and column sparsity of $1$. We then obtain a $(1, \ceil{\log m} + 3, \varepsilon_3)$-block-encoding for $\mathbf{B}_2$, which we denote by $U_{B_2}$. The corresponding gate complexity is $O(\log m, \log^{2.5}(1/\varepsilon_3))$. We can now combine $U_{D}$ with $U_{B_2}$ using Lemma~\ref{lem:tensor_prod_BEs} to obtain a $(\alpha_\calN, a_\calN, \varepsilon_{\calN})$-block-encoding for $\calN := \mathbf{B}_2 \otimes \mathbf{D}$, which we denote by $U_{\calN}$ with
$$
\alpha_{\calN} = 2k \sqrt{k}, \enspace a_\calN = \ell + a_{\calA} + \ceil{\log m} + 6, \enspace \varepsilon_\calN = 2k\sqrt{k}\varepsilon_3 + 2\sqrt{k}\varepsilon_2.
$$
Above, we used that $\varepsilon_3 < 1$ to simplify the expression of $\varepsilon_\calN$. The corresponding gate complexity is 
$$
O\left(k \log(k/\varepsilon_2) \Big(T_{\calA} + \ell + \log^{2.5}(1/\varepsilon_1)\Big) + \log(m) \log^{2.5}(1/\varepsilon_3) \right).
$$ 
Using Lemma~\ref{lem:sum_BEs} again to combine $U_\calN$ with $\id$, we obtain a $(\alpha_\calL, a_\calL, \varepsilon_\calL)$-block-encoding of $\calL = \id - \calN$, which we denote by $U_{\calL}$, with 
$$
\alpha_\calL = 2k\sqrt{k} + 1, \enspace a_\calL = \ell + a_{\calA} + \ceil{\log m} + 7, \enspace \varepsilon_\calL = 2k\sqrt{k}\varepsilon_3 + 2\sqrt{k} \varepsilon_2.
$$
The corresponding gate complexity, which we denote by $T_{\calL}$, is asymptotically same as that of $U_\calN$ and therefore
\begin{equation}\label{eq:cost_L}
T_\calL = O\left(k \log(k/\varepsilon_2) \Big(T_{\calA} + \ell + \log^{2.5}(1/\varepsilon_1)\Big) + \log(m) \log^{2.5}(1/\varepsilon_3) \right).    
\end{equation}

To obtain a block-encoding of $\calL^{-1}$, we use Lemma~\ref{lem:pseudoinverse_BE}. It was shown in \cite[Theorem~4]{Krovi2023improvedquantum} that the spectrum of $\calL$ satisfies
\begin{equation}\label{eq:spectrum_L}
\spec(\calL) \in [(2 e m \expnorm(\bm{\calA)})^{-1}, 1 + \sqrt{k}e]
\end{equation}
and the condition number of $\calL$, which we denote by $\kappa_{\calL}$, is thus bounded as
\begin{equation}\label{eq:bound_kappa_L}
    \kappa_{\calL} \leq 4 \sqrt{k} e^2 m \expnorm(\bm{\calA}).
\end{equation}
Using Lemma~\ref{lem:pseudoinverse_BE} on $U_\calL$, we can then obtain a $(\alpha', a', \varepsilon')$-block-encoding for $\calL^{-1}$, which we denote by $U_{\calL^{-1}}$ with
\begin{equation}\label{eq:parameters_invL_BE}
\alpha' = 4 em \expnorm(\bm{\calA}), \enspace a' = \alpha_\calL + 1 = \ell + a_{\calA} + \ceil{\log m} + 8,
\end{equation}
and $\varepsilon' \in (0,1)$ is a desired error parameter to be fixed later, provided 
\begin{equation}\label{eq:condition_eps_L}
\varepsilon_{\calL} = o\left(\frac{\varepsilon'}{\sqrt{k} m^2 \expnorm(\bm{\calA})^2 \log(m \expnorm(\bm{\calA})/\varepsilon') }\right)    
\end{equation}
The corresponding gate complexity is
\begin{equation}\label{eq:cost_inv_L}
O\left( T_\calL \sqrt{k} m \expnorm(\bm{\calA}) \log\Big(\sqrt{k} m \expnorm(\bm{\calA})/\varepsilon'\Big) \right),
\end{equation}
where $T_\calL$ is given by Eq.~\eqref{eq:cost_L}.

Let $\widetilde{\calL} = (\bra{0}^{a'} \otimes \id \otimes \bra{0} \otimes \id)U_{\calL^{-1}}(\ket{0}^{a'} \otimes \id \otimes \ket{0} \otimes \id)$, which is the action of implementing $U_{\calL^{-1}}$ followed by post-selecting on the register over the $a'$ qubits and the register over the Taylor index. We can then say that the application of $U_{\calL^{-1}}$ on the input state $\ket{\psi_{\mathrm{in}}}$ and before measuring is
$$
\ket{0}^{a'+1} \frac{\widetilde{L}}{\alpha'} \ket{\psi_{\mathrm{in}}} + \ket{\perp},
$$
where $\ket{\perp}$ is some unnormalized state orthogonal to every state with $\ket{0}^{a'}$ over the first register and $\ket{0}$ over the Taylor index. Let $\delta' \in (0,1)$ be an error parameter to be fixed shortly. Using Claim~\ref{claim:state_approx_BE}, we then produce the following state denoted $\ket{\widetilde{\phi}}$ after post-selecting for $\ket{0}^{a'+1}$
\begin{equation}
\ket{\widetilde{\phi}} = \frac{\widetilde{L}\ket{\psi_{\mathrm{in}}}}{\norm{\widetilde{L}\ket{\psi_{\mathrm{in}}}}},
\end{equation}
which is promised to satisfy
\begin{equation}\label{eq:tildephi_close_to_phi}
\norm{ \ket{\widetilde{\phi}} - \ket{\phi}} \leq \delta',
\end{equation}
where $\ket{\phi}$ is as defined in Eq.~\eqref{eq:def_true_hist_taylor_hist} and provided 
\begin{equation}\label{eq:constraint_varepsprime}
\varepsilon' \leq \frac{\delta'}{8 \sqrt{k} e^2 m \expnorm(\bm{\calA})},
\end{equation}
where we used $\spec(\calL^{-1}) \in [1/(2\sqrt{k}e), 2em\expnorm(\bm{\calA})]$ following Eq.~\eqref{eq:spectrum_L}. The probability of success of obtaining $\ket{\widetilde{\phi}}$ is then
\begin{equation}\label{eq:succ_prob_ODE_solver}
p_{\mathrm{succ}} \geq 1/(2^8 k e^4 m^2 \expnorm(\bm{\calA})^2).
\end{equation}
We can now observe that $\ket{\widetilde{\phi}}$ is close to the true history state $\ket{\psi}$ by combining Eq.~\eqref{eq:tildephi_close_to_phi} and Eq.~\eqref{eq:interim_phi_close_to_psi} as follows
\begin{equation}\label{eq:correctness_hist_state}
    \norm{\ket{\widetilde{\phi}} - \ket{\psi}} \leq \norm{\ket{\widetilde{\phi}} - \ket{\phi}} + \norm{\ket{\phi} - \ket{\psi}} \leq \delta' + \delta.
\end{equation}
Setting $\delta' = \delta = \varepsilon/2$ gives us the desired accuracy. This in turn requires setting the following error parameters from Eq.~\eqref{eq:constraint_varepsprime} and Eq.~\eqref{eq:condition_eps_L} introduced so far as
\begin{equation}\label{eq:choices_error_params}
\varepsilon' \leq O\left(\frac{\varepsilon}{T \norm{\bm{\calA}} \expnorm(\bm{\calA}) \sqrt{\log(\Omega)}}\right), \quad \varepsilon_\calL \leq O\left( \frac{\varepsilon}{T^3 \norm{\bm{\calA}}^3 \expnorm(\bm{\calA})^3 \log(\Omega) \log(T \norm{\bm{\calA}} \expnorm(\bm{\calA})/\varepsilon)} \right) 
\end{equation}
where we used
$$
m = \ceil{T \norm{\bm{\calA}}}, \quad \Omega = \frac{4 m e^3}{\varepsilon}\left(1 + \frac{T e^2 \norm{\vec{f}}}{\mu} \right), \quad k = O(\log \Omega).
$$
Noting that $\varepsilon_\calL = 2k\sqrt{k} \varepsilon_3 + 2 \sqrt{k} \varepsilon_2$ where recall $\varepsilon_3$ was the precision of $U_{B_2}$ and $\varepsilon_2$ was the precision of $U_{(\id-\calM_1)^{-1}}$, we set these errors as
\begin{equation}\label{eq:choices_more_error_params}
\varepsilon_3 = \frac{\varepsilon_{\calL}}{4k\sqrt{k}} \leq O\left( \frac{\varepsilon}{\kappa^3 \log^{2.5}(\Omega) \log(\kappa/\varepsilon)} \right), \quad \varepsilon_2 = \frac{\varepsilon_{\calL}}{4 \sqrt{k}} \leq O\left( \frac{\varepsilon}{\kappa^3 \log^{1.5}(\Omega) \log(\kappa/\varepsilon)} \right),
\end{equation}
where we defined $\kappa:= T \norm{\bm{\calA}} \expnorm(\bm{\calA)}$ for simplifying the above expressions. Further upstream, we set $\varepsilon_1$ and $\varepsilon_{\calA}$ so that it satisfies Eq.~\eqref{eq:cond_eps_inv_idminusM1}\footnote{We note that to satisfy Eq.~\eqref{eq:cond_eps_inv_idminusM1}, it suffices to set $(\varepsilon_1 + \varepsilon_\calA)^2 \leq O(\varepsilon_2^2/k^3)$ since $1/x < 1/\log(x), \forall x > 1$, and hence $\varepsilon_1 = \varepsilon_\calA \leq O(\varepsilon_2/k^{1.5})$} as
\begin{equation}\label{eq:choice_eps1}
\varepsilon_1 \leq O\left( \frac{\varepsilon}{\kappa^3 \log^{3}(\Omega) \log(\kappa/\varepsilon)} \right), \quad \varepsilon_{\calA} \leq O\left( \frac{\varepsilon}{\kappa^3 \log^{3}(\Omega) \log(\kappa/\varepsilon)} \right).
\end{equation}
Considering all the above choices of error parameters, the cost of the algorithm is due to the preparation of $\ket{\psi_{\mathrm{in}}}$ and due to the gate complexity of $U_{\calL^{-1}}$ from Eq.~\eqref{eq:cost_inv_L}, which is
\begin{align}\label{eq:cost_inv_L_final}
& O\left( T_\calL \sqrt{\log(\Omega)} \kappa \log\Big(\log \Omega \kappa^2 /\varepsilon\Big) + \polylog(T \norm{\bm{\calA}}) + T_x + T_f\right) \\
=\,& \widetilde{O}\Big( (T_\calA + 1) \kappa \cdot \poly(\log \kappa, \log \Omega, \log(1/\varepsilon)) + T_x + T_f \Big),
\end{align}
where we used Eq.~\eqref{eq:cost_L} for the cost of $T_\calL$ and denoted the cost of one call to $O_x,O_f$ as $T_x,T_f$. To boost the success probability of obtaining $\ket{\widetilde{\phi}}$ which is given by Eq.~\eqref{eq:succ_prob_ODE_solver} to $\Omega(1)$, we use amplitude amplification (Theorem~\ref{thm:amplitude_amplification}) and as a result the overall cost then is $O(1/\sqrt{p_{\mathrm{succ}}})$ times that of Eq.~\eqref{eq:cost_inv_L_final}:
\begin{equation}\label{eq:complexity_ODE_solver}
\widetilde{O}\Big( (T_\calA + 1) \kappa^2 \poly(\log \kappa, \log \Omega, \log(1/\varepsilon)) + (T_x + T_f) \kappa \log \Omega \Big),
\end{equation}
where recall that $\kappa = T \norm{\bm{\calA}} \expnorm(\bm{\calA})$. To summarize, the correctness of the algorithm was shown in Eq.~\eqref{eq:correctness_hist_state} and the complexity of the algorithm was found to be as in Eq.~\eqref{eq:complexity_ODE_solver}.
\end{proof}

We can now use the above result to prepare a quantum state proportional to the solution of the $\ode$ at a particular time and thereby prove Theorem~\ref{thm:ODE_solver_final_state}. The key difference is to use Claim~\ref{claim:diff_taylor_truth_approx_input}$(i)$ as the goal is to produce a state close to the final time solution and a padded matrix of $\calN$ (Eq.~\eqref{eq:ODE_solver_system}) as done in \cite{Krovi2023improvedquantum} to increase the success probability of $\ket{x(T)}$. We omit the details here as it follows similarly to the result above.

\subsection{Proofs for DAE of index two}\label{appsec:proof_dae_index2}
In this section, we provide the proof of Theorem~\ref{thm:dae_index2_ode}. The approach is very similar to that in \cite[Section~3.2]{schulz2003four}. Recall our definitions from Eq.~\ref{eq:hierarchy_matrix_pencil}:
\begin{align*}
    & \rlcM_0 = \rlcM, \, \rlcK_0 = \rlcK, \\
    & \rlcM_1 = \rlcM + \rlcK \rlcQ_0, \,  \rlcK_1 = \rlcK_0 \rlcP_0, \\
    & \rlcM_2 = \rlcM_1 + \rlcK_1 \rlcQ_1, \,  \rlcK_2 = \rlcK_1 \rlcP_1.
\end{align*}

We restate the theorem here for convenience.
\daeindextwo*

We require the following useful claim regarding the projectors.
\begin{claim}\label{claim:projectors_index2_props}
The following is true:
$$
\rlcP_0 \rlcP_1 = \rlcP_1 \rlcP_0 + \rlcQ_0 \rlcQ_1
$$
\end{claim}
\begin{proof}
    Note that $\rlcP_1$ can be expressed in two ways as follows:
    $$\rlcP_0 \rlcP_1 + \rlcQ_0 \rlcP_1 = \rlcP_1 = \rlcP_1 \rlcP_0 + \rlcP_1 \rlcQ_0.$$
    We thus have 
    $$\rlcP_0 \rlcP_1 + \rlcQ_0(\id-\rlcQ_1) = \rlcP_1 \rlcP_0 + (\id-\rlcQ_1)\rlcQ_0 \implies \rlcP_0 \rlcP_1 = \rlcP_1 \rlcP_0 + \rlcQ_0 \rlcQ_1,$$
    where we used that $\rlcQ_1 \rlcQ_0 = 0$ by construction. This gives us the desired result.
\end{proof}

We will use the following claim to bring the equations into their desired forms.
\begin{claim}\label{claim:dae_index2_props}
The following are true:
\begin{enumerate}[(a)]
    \item $\rlcM_2^{-1} \rlcM = \rlcP_1 \rlcP_0$
    \item $\rlcM_2^{-1} \rlcK = \rlcM_2^{-1} \rlcK_2 + \rlcQ_1 + \rlcQ_0$
\end{enumerate}
\end{claim}
\begin{proof}
$(a)$ We note that 
\begin{align*}
    &\rlcM_2 \rlcP_1 \rlcP_0 = (\rlcM_1 + \rlcK_1 \rlcQ_1) \rlcP_1 \rlcP_0 = \rlcM_1 \rlcP_1 \rlcP_0 + \rlcK_1 \rlcQ_1 \rlcP_1 \rlcP_0 \nonumber\\
    &= \rlcM_1 (\id-\rlcQ_1) \rlcP_0 = \rlcM_1 \rlcP_0 = (\rlcM+\rlcK\rlcQ_0)\rlcP_0 = \rlcM\,,
\end{align*}
where we used that $\rlcQ_1 \rlcP_1 = 0$ in the second equality, expressed $\rlcP_1 = \id-\rlcQ_1$ in the third equality, used $\rlcM_1 \rlcQ_1 = 0$ in the fourth equality by definition of $\rlcQ_1$ as the projector onto $\ker(\rlcM_1)$, expanded $\rlcM_1$ in the fifth equality, and finally used $\rlcQ_0 \rlcP_0 = 0$. This proves statement $(a)$.

$(c)$ Noting that $\rlcK = \rlcK \rlcP_0 + \rlcK\rlcQ_0 = \rlcK \rlcP_0 \rlcP_1 + \rlcK \rlcP_0 \rlcQ_1 + \rlcK \rlcQ_0 = \rlcK_2 + \rlcK_1 \rlcQ_1 + \rlcK \rlcQ_0$, we can write
\begin{equation}\label{eq:invM2_K}
    \rlcM_2^{-1} \rlcK = \rlcM_2^{-1} \rlcK_2 + \rlcM_2^{-1} \rlcK_1 \rlcQ_1 + \rlcM_2^{-1} \rlcK \rlcQ_0.    
\end{equation}
To simplify $\rlcM_2^{-1} \rlcK_1 \rlcQ_1$ , we observe
\begin{equation}\label{eq:M2_Q0}
    \rlcM_2 \rlcQ_0 = (\rlcM_1 + \rlcK_1 \rlcQ_1) \rlcQ_0 = \rlcM_1 \rlcQ_0 = (\rlcM + \rlcK\rlcQ_0) \rlcQ_0 = \rlcK\rlcQ_0 \implies \rlcM_2^{-1} \rlcK \rlcQ_0 = \rlcQ_0,
\end{equation}
where we used $\rlcQ_1 \rlcQ_0 = 0$ in the second equality, expanded $\rlcM_1$ in the third equality and used $\rlcM \rlcQ_0 = 0$ in the fourth equality by definition of $\rlcQ_0$ as the projector onto $\ker(\rlcM)$. To simplify $\rlcM_2^{-1} \rlcK \rlcQ_0$, we observe
\begin{equation}\label{eq:M2_Q1}
    \rlcM_2 \rlcQ_1 = (\rlcM_1 + \rlcK_1 \rlcQ_1) \rlcQ_1 = \rlcK_1 \rlcQ_1 \implies \rlcM_2^{-1} \rlcK_1 \rlcQ_1 = \rlcQ_1,
\end{equation}
where we expanded $\rlcM_2$ in the first equality and used used $\rlcM_1 \rlcQ_1 = 0$ in the second equality by definition of $\rlcQ_1$ as the projector onto $\ker(\rlcM_1)$. Substituting the implications from Eqs.~\eqref{eq:M2_Q0}--\eqref{eq:M2_Q1} into Eq.~\eqref{eq:invM2_K} gives us the desired statement.
\end{proof}

\begin{proof}
    We can write the solution $\vec{x}(t)$ as (where we have dropped the time dependence for brevity)
    \begin{equation}\label{eq:decomposition_x_index2}
        \vec{x} = \underbrace{\rlcP_0 \rlcP_1 \vec{x}}_{:=\vec{y}} + \underbrace{\rlcQ_0 \rlcP_1 \vec{x}}_{:=\vec{z}_1} + \underbrace{\rlcQ_1 \vec{x}}_{:=\vec{z}_2},
    \end{equation}
    where we used $\id = \rlcP_1 + \rlcQ_1 = \rlcP_0 \rlcP_1 + \rlcQ_0 \rlcP_1 + \rlcQ_1$ and the definitions of $\vec{y},\vec{z}_1,\vec{z}_2$ are inline. We will now derive the explicit ODE dictating the time evolution of $\vec{y}$ starting from the implicit DAE in hand:
    \begin{align}
        \rlcM \dot{\vec{x}} + \rlcK\vec{x} &= \vec{f} \\
        \rlcM_2^{-1} \rlcM \dot{\vec{x}} + \rlcM_2^{-1} \rlcK \vec{x} &= \rlcM_2^{-1} \vec{f} \\
        \rlcP_1 \rlcP_0 \dot{\vec{x}} + \rlcM_2^{-1} \rlcK_2 \vec{x} + \rlcQ_1 \vec{x} + \rlcQ_0 \vec{x} &= \rlcM_2^{-1} \vec{f}
        \label{eq:main_interim_index2}
    \end{align}
    where we multiplied by $\rlcM_2^{-1}$ in the first line, used Claim~\ref{claim:dae_index2_props}$(a)$ to simplify $\rlcM_2^{-1}\rlcM$ and Claim~\ref{claim:dae_index2_props}$(c)$ to expand $\rlcM_2^{-1}\rlcK$ in the third line. To obtain equations for $\vec{y},\vec{z}_1,\vec{z}_2$, we will multiply Eq.~\eqref{eq:main_interim_index2} by $\rlcP_0 \rlcP_1$, $\rlcQ_0 \rlcP_1$ and $\rlcQ_1$ respectively. Firstly, multiplying Eq.~\eqref{eq:main_interim_index2} by $\rlcP_0 \rlcP_1$ gives us
    \begin{align}\label{eq:interim_eq_ode_mna_dae_index2}
        \rlcP_0 \rlcP_1 \rlcP_0 \dot{\vec{x}} + \rlcP_0 \rlcP_1 \rlcM_2^{-1} \rlcK_2 \vec{x} + \rlcP_0 \rlcP_1 \rlcQ_0 \vec{x} &= \rlcP_0 \rlcP_1 \rlcM_2^{-1} \vec{f} \\
        \rlcP_0 \rlcP_1 \rlcP_0 \dot{\vec{x}} + \rlcP_0 \rlcP_1 \rlcM_2^{-1} \rlcK \vec{y} &= \rlcP_0 \rlcP_1 \rlcM_2^{-1} \vec{f},
    \end{align}
    where we used $\rlcP_1 \rlcQ_1 = 0$ in the first line and that $\rlcK_2 \vec{x} = \rlcK \rlcP_0 \rlcP_1 \vec{x} = \rlcK \vec{y}$ to simplify the second term along with the fact that $\rlcP_0 \rlcP_1 \rlcQ_0 = \rlcP_0 (\id-\rlcQ_1) \rlcQ_0 = \rlcP_0 \rlcQ_0 - \rlcP_0 \rlcQ_1 \rlcQ_0 = 0$. To prove the equation regarding the evolution of $\vec{y}$ as stated in the theorem, we now analyze the term $\rlcP_0 \rlcP_1 \rlcP_0 \dot{\vec{x}}$ in Eq.~\eqref{eq:interim_eq_ode_mna_dae_index2} upon substituting $\vec{x} = \rlcP_0 \rlcP_1 \vec{x} + \rlcQ_0 \rlcP_1 \vec{x} + \rlcQ_1 \vec{x}$:
    \begin{align}
        \rlcP_0 \rlcP_1 \rlcP_0 \dot{\vec{x}} &= \rlcP_0 \rlcP_1 \rlcP_0 \rlcP_1 \dot{\vec{x}} + \rlcP_0 \rlcP_1 \rlcP_0 \rlcQ_0 \rlcP_1 \dot{\vec{x}} + \rlcP_0 \rlcP_1 \rlcP_0 \rlcQ_1 \dot{\vec{x}} \\
        &= \dot{\vec{y}} + \rlcP_0 \rlcP_1 \rlcP_0 \rlcQ_1 \dot{\vec{x}} \\
        &= \dot{\vec{y}} + \rlcP_0 (\rlcP_0 \rlcP_1 - \rlcQ_0 \rlcQ_1) \rlcQ_1 \dot{\vec{x}} \\
        &= \dot{\vec{y}} + \rlcP_0 \rlcP_1 \rlcQ_1 \dot{\vec{x}} + \rlcP_0 \rlcQ_0 \rlcQ_1 \dot{\vec{x}} \\
        &= \dot{\vec{y}},
    \end{align}
    where we used $(\rlcP_0 \rlcP_1) (\rlcP_0 \rlcP_1) \vec{x} = (\rlcP_0 \rlcP_1) \vec{y} = \vec{y}$ from the definition of $\vec{y}$ and $\rlcP_0 \rlcQ_0=0$ in the second line, applied Claim~\ref{claim:projectors_index2_props} in the third line, and used $\rlcP_1 \rlcQ_1 = 0$ along with $\rlcP_0 \rlcQ_0 = 0$ in the fifth line. Substituting back in Eq.~\eqref{eq:interim_eq_ode_mna_dae_index2} gives us the desired ODE governing the evolution of $\vec{y}$:
    \begin{equation}
        \dot{\vec{y}} =  - \rlcP_0 \rlcP_1 \rlcM_2^{-1} \rlcK \vec{y} + \rlcP_0 \rlcP_1 \rlcM_2^{-1} \vec{f},
    \end{equation}
    To obtain the equation governing the evolution of $\vec{z}_2$, we multiply Eq.~\eqref{eq:main_interim_index2} by $\rlcQ_1$ as follows:
    \begin{align}
        \rlcQ_1 \rlcM_2^{-1} \rlcK_2 \vec{x} + \rlcQ_1 \vec{x} &= \rlcQ_1 \rlcM_2^{-1} \vec{f} \\
        \rlcQ_1 \rlcM_2^{-1} \rlcK \vec{y} + \vec{z}_2 &= \rlcQ_1 \rlcM_2^{-1} \vec{f} \\
        \vec{z}_2 &= -\rlcQ_1 \rlcM_2^{-1} \rlcK \vec{y} + \rlcQ_1 \rlcM_2^{-1} \vec{f}
    \end{align}
    where we used $\rlcQ_1 \rlcQ_0 = 0$ by construction and $\rlcQ_1 \rlcP_1 = 0$ by definition of the projectors in the first line, used the fact that $\rlcK_2 \vec{x} = \rlcK \rlcP_0 \rlcP_1 \vec{x} = \rlcK\vec{y}$ by the definition of $\vec{y}$ in the second line, and $\rlcQ_1 \vec{x} = \vec{z}_2$ by definition again.

    Finally, to obtain the equation governing the evolution of $\vec{z}_1$, we multiply Eq.~\eqref{eq:main_interim_index2} by $\rlcQ_0 \rlcP_1$ as follows:
    \begin{align}\label{eq:interim_eq_ode_mna_dae_index2}
        \rlcQ_0 \rlcP_1 \rlcP_0 \dot{\vec{x}} + \rlcQ_0 \rlcP_1 \rlcM_2^{-1} \rlcK \vec{y} + \rlcQ_0 \rlcP_1 \rlcQ_0 \vec{x} &= \rlcQ_0 \rlcP_1 \rlcM_2^{-1} \vec{f} \\
        \rlcQ_0 (\rlcP_0 \rlcP_1 - \rlcQ_0 \rlcQ_1) \dot{\vec{x}} + \rlcQ_0 \rlcP_1 \rlcM_2^{-1} \rlcK \vec{y} + \rlcQ_0 (\id-\rlcQ_1) \rlcQ_0 \vec{x} &= \rlcQ_0 \rlcP_1 \rlcM_2^{-1} \vec{f} \\
        - \rlcQ_0 \dot{\vec{z}}_2 + \rlcQ_0 \rlcP_1 \rlcM_2^{-1} \rlcK \vec{y} + \rlcQ_0 \vec{x} &= \rlcQ_0 \rlcP_1 \rlcM_2^{-1} \vec{f} \\
        - \rlcQ_0 \dot{\vec{z}}_2 + \rlcQ_0 \rlcP_1 \rlcM_2^{-1} \rlcK \vec{y} + \rlcQ_0 (\rlcP_0 \rlcP_1 \vec{x} + \rlcQ_0 \rlcP_1 \vec{x} + \rlcQ_1 \vec{x}) &= \rlcQ_0 \rlcP_1 \rlcM_2^{-1} \vec{f} \\
        - \rlcQ_0 \dot{\vec{z}}_2 + \rlcQ_0 \rlcP_1 \rlcM_2^{-1} \rlcK \vec{y} + \vec{z}_1 + \rlcQ_0 \vec{z}_2 &= \rlcQ_0 \rlcP_1 \rlcM_2^{-1} \vec{f},
    \end{align}
    where we used $\rlcP_1 \rlcQ_1 = 0$ by definition of the projectors in the first line and the fact that $\rlcK_2 \vec{x} = \rlcK \rlcP_0 \rlcP_1 \vec{x} = \rlcK\vec{y}$ by the definition of $\vec{y}$. We used Claim~\ref{claim:projectors_index2_props} to expand $\rlcP_1 \rlcP_0$ in the second line, $\rlcQ_1 \dot{\vec{x}} = \dot{\vec{z}}_2$ by definition of $\vec{z}_2$ in the third line, and Eq.~\eqref{eq:decomposition_x_index2} to expand $\vec{x}$ in the fourth line. Reorganizing the last equation gives us the following, which is the desired result:
    $$
    \vec{z}_1 = \rlcQ_0 \dot{\vec{z}}_2 - \rlcQ_0 \rlcP_1 \rlcM_2^{-1} \rlcK \vec{y} - \rlcQ_0 \vec{z}_2 + \rlcQ_0 \rlcP_1 \rlcM_2^{-1} \vec{f}.
    $$

    To obtain an equation for $\vec{z}_1(t)$ that does not depend on time-derivatives of any variable, we note that the equation for $\vec{z}_1(t)$ can be expressed just in terms of $\vec{y}$ by using the following relation for $\dot{\vec{z}}_2$
\begin{align}
    \dot{\vec{z}}_2 &= -\rlcQ_1 \rlcM_2^{-1} \rlcK \dot{\vec{y}} \nonumber \\    
    &= -\rlcQ_1 \rlcM_2^{-1} \rlcK \left(- \rlcP_0 \rlcP_1 \rlcM_2^{-1} \rlcK \vec{y} + \rlcP_0 \rlcP_1 \rlcM_2^{-1} \vec{f} \right) \nonumber \\    
    &= \rlcQ_1 \rlcM_2^{-1} \rlcK_2 \rlcM_2^{-1} \rlcK \vec{y} - \rlcQ_1 \rlcM_2^{-1} \rlcK_2 \rlcM_2^{-1} \vec{f} \nonumber \\
    &= \rlcQ_1 (\rlcM_2^{-1} \rlcK_2)^2 \vec{y} - \rlcQ_1 \rlcM_2^{-1} \rlcK_2 \rlcM_2^{-1} \vec{f},
\end{align}
where we used $\dot{\vec{f}}=0$ as it is assumed to be time-independent, used Theorem~\ref{thm:dae_index2_ode} for $\dot{\vec{y}}$ in the second line along with $\rlcK_2 = \rlcK \rlcP_0 \rlcP_1$ and finally noted that $\rlcP_0 \rlcP_1 \vec{y} = \vec{y}$. Substituting the above expression of $\vec{\dot{z}}_2$ into the expression of $\vec{z}_1(t)$ from Theorem~\ref{thm:dae_index2_ode}, we obtain
\begin{align}
    \vec{z}_1 &= \rlcQ_0 \left( \rlcQ_1 (\rlcM_2^{-1} \rlcK_2)^2 \vec{y} - \rlcQ_1 \rlcM_2^{-1} \rlcK_2 \rlcM_2^{-1} \vec{f} \right)  - \rlcQ_0 \rlcP_1 \rlcM_2^{-1} \rlcK \vec{y} \nonumber \\ \nonumber
    & \qquad - \rlcQ_0 \left(-\rlcQ_1 \rlcM_2^{-1} \rlcK \vec{y} + \rlcQ_1 \rlcM_2^{-1} \vec{f} \right) + \rlcQ_0 \rlcP_1 \rlcM_2^{-1} \vec{f} \\
    &= \rlcQ_0 \rlcQ_1 \left(\rlcM_2^{-1}\rlcK_2\right)^2 \vec{y} + \rlcQ_0 (2 \rlcQ_1 - \id) \rlcM_2^{-1} \rlcK\vec{y} - \rlcQ_0 \rlcQ_1 \rlcM_2^{-1}\rlcK_2 \rlcM_2^{-1} \vec{f} + \rlcQ_0 (\id - 2\rlcQ_1) \rlcM_2^{-1} \vec{f}
\end{align}
That concludes the proof of the set of equations corresponding to $\dae$ of index two.
\end{proof}

\end{document}